\documentclass[english, 12pt, letter]{article}
\usepackage{amsmath,amssymb,amsthm}
\usepackage{mathtools}
\usepackage{natbib}
\usepackage[nodisplayskipstretch]{setspace}
\usepackage{bm}
\usepackage{threeparttable}
\usepackage{graphicx}
\usepackage{enumerate}
\usepackage{commath}			% for \norm{...}
\usepackage{bibentry}
\usepackage{booktabs}
\usepackage[titletoc,title]{appendix}
\usepackage{multirow}
\usepackage{arydshln}
\usepackage{tabularx}
\usepackage{dsfont}
\usepackage{rotating} % Rotating table
\usepackage[font={small}]{subcaption}
\usepackage{placeins}
\usepackage[left=1in, right=1in, top=1in, bottom=1in]{geometry}

\usepackage{algorithm}
\usepackage{algpseudocode}

\usepackage[colorinlistoftodos,textsize=small]{todonotes}
\allowdisplaybreaks

\usepackage{xr-hyper}  % for cross-referencing
\usepackage{url}
\usepackage[
	breaklinks,
	colorlinks=true,
	pagebackref=false,
	pdfauthor={},%
	pdftitle={},% 
	citecolor=blue,
	linkcolor=blue,
	urlcolor=blue
	]{hyperref}      

\newtheorem{lem}{Lemma}
\newtheorem{assumption}{Assumption}
\newtheorem{prop}{Proposition}

\newtheorem{thm}{Theorem}

\newtheoremstyle{remark2}{1ex}{1ex}%
      {}% Body font
      {}% Indent amount (empty = no indent, \parindent = para indent)
      {\bf}% Thm head font. notice \YYshape command works unlike \textyy
      {.}% Punctuation after thm head
      {5pt}% Space after thm head (\newline = linebreak)
      {\thmname{#1}\thmnumber{ #2}\thmnote{ \slshape{(#3)}}} % Thm head spec, so numbering first
\theoremstyle{remark2}
\newtheorem{rem}{Remark}

\newtheorem{example}{Example}

\newtheoremstyle{remark3}{1ex}{1ex}%
      {}% Body font
      {}% Indent amount (empty = no indent, \parindent = para indent)
      {\bf}% Thm head font. notice \YYshape command works unlike \textyy
      {.}% Punctuation after thm head
      {5pt}% Space after thm head (\newline = linebreak)
      {\thmname{#1}\thmnumber{.#2}\thmnote{ \slshape{(#3)}}} % Thm head spec, so numbering first
\theoremstyle{remark3}
\newtheorem{lemV}{Lemma V}
\newtheorem{lemM}{Lemma M}

\newtheorem{assumptionM}{Assumption LM}

\newcommand{\rank}{\operatorname{rank}}

\newcommand{\1}{\mathds{1}}

\usepackage{babel}
\input{ee.sty}

\@ifundefined{bibfont}{  } %%%% if then else
    {  }

\makeatletter
\renewenvironment{proof}[1][\bfseries\proofname]{\par
   \pushQED{\qed}%
   \normalfont \topsep6\p@\@plus6\p@\relax
   \trivlist
   \item[\hskip\labelsep
     #1\@addpunct{:}]\ignorespaces
}{%
   \popQED\endtrivlist\@endpefalse
}
\makeatother

\newcommand{\Comments}{1}
\newcommand{\mynote}[2]{\ifnum\Comments=1\textcolor{#1}{#2}\fi}
\newcommand{\mytodo}[2]{\ifnum\Comments=1%
  \todo[linecolor=#1!80!black,backgroundcolor=#1,bordercolor=#1!80!black]{#2}\fi}

\ifnum\Comments=1               % fix margins for todonotes
\fi

\ifnum\Comments=1               % fix margins for todonotes
\fi

\renewcommand\appendixpagename{Appendix}

\newcommand{\RORAC}{\operatorname{RORAC}}

\newcommand{\RC}{\operatorname{RC}}
\newcommand{\MES}{\operatorname{MES}}

\newcommand{\VaR}{\operatorname{VaR}}
\newcommand{\FCI}{\operatorname{FCI}}

\newcommand{\VIX}{\operatorname{VIX}}

\newcommand{\ES}{\operatorname{ES}}

\newcommand{\inter}{\operatorname{int}}
\newcommand{\D}{\,\mathrm{d}}
\newcommand{\lex}{\operatorname{lex}}
\renewcommand{\E}{\mathbb{E}}
\renewcommand{\P}{\mathbb{P}}

\renewcommand{\b}{\beta}

\newcommand{\Vv}{\mathsf{\mathbf{v}}}
\newcommand{\Mm}{\mathsf{\mathbf{m}}}

\begin{document}

\baselineskip18pt
% \addtolength{\floatsep}{-5cm}
% \renewcommand\floatpagefraction{.8}
% \renewcommand\textfraction{.05}
\renewcommand\floatpagefraction{.9}
\renewcommand\topfraction{.9}
\renewcommand\bottomfraction{.9}
\renewcommand\textfraction{.1}
\setcounter{totalnumber}{50}
\setcounter{topnumber}{50}
\setcounter{bottomnumber}{50}
\abovedisplayskip1.5ex plus1ex minus1ex
\belowdisplayskip1.5ex plus1ex minus1ex
\abovedisplayshortskip1.5ex plus1ex minus1ex
\belowdisplayshortskip1.5ex plus1ex minus1ex

% TO UNBLIND: lines 187, 192, 721

\title{Regressions under Adverse Conditions%\thanks{The authors would like to thank the AE, two anonymous referees, Tobias Fissler and Christoph Hanck for their insightful comments that significantly improved the quality of the paper. Both authors gratefully acknowledge support of the Deutsche Forschungsgemeinschaft (DFG, German Research Foundation) through grants 502572912 (first author), and 460479886 and 531866675 (second author).}
}

%\begin{notthisone}
\author{
	Timo Dimitriadis\thanks{Faculty of Economics and Business, Goethe University Frankfurt, 60629 Frankfurt am Main, Germany, and Heidelberg Institute for Theoretical Studies (HITS), \href{mailto:dimitriadis@econ.uni-frankfurt.de}{dimitriadis@econ.uni-frankfurt.de}.}
\and 
	Yannick Hoga\thanks{Faculty of Economics and Business Administration, University of Duisburg-Essen, Universit\"atsstra\ss e 12, D--45117 Essen, Germany, \href{mailto:yannick.hoga@vwl.uni-due.de}{yannick.hoga@vwl.uni-due.de}.}
}

 \date{\today}
%\date{February 1, 2025}
%\end{notthisone}
\maketitle

\begin{abstract}
	\noindent 
	\singlespacing
	We introduce a new regression method that relates the mean of an outcome variable to covariates, under the ``adverse condition'' that a distress variable falls in its tail.
	This allows to tailor classical mean regressions to adverse scenarios, which receive increasing interest in economics and finance, among many others.
	% This allows to tailor classical mean regressions to adverse economic scenarios, which receive increasing interest in managing macroeconomic and financial risks, among many others.
	In the terminology of the systemic risk literature, our method can be interpreted as a regression for the Marginal Expected Shortfall.
%	The MES is the expectation of a random variable given that another random variable falls in its tail.
%	We introduce a new econometric tool called \textit{marginal expected shortfall (MES) regressions}. 
%	The MES is known from the systemic risk literature  expectation of a random variable, given that 
%	These relate the expectation of an outcome variable to covariates \textit{conditional} on some other variable falling in its tail, which is known as the MES.
	We propose a two-step procedure to estimate the new models, show consistency and asymptotic normality of the estimator, and propose feasible inference under weak conditions that allow for cross-sectional and time series applications.
	Simulations verify the accuracy of the asymptotic approximations of the two-step estimator.
	Two empirical applications show that our regressions under adverse conditions are a valuable tool in such diverse fields as the study of the relation between systemic risk and asset price bubbles, and dissecting macroeconomic growth vulnerabilities into individual components.\\
%	We illustrate the usefulness of our new methodology in three empirical applications---one concerning optimal allocation of portfolio risks and another concerning the identification of the truly significant drivers of systemic risk in the financial system. 

	\noindent \textbf{Keywords:} Estimation, Growth-at-Risk, Marginal Expected Shortfall, Regression Modeling, Systemic Risk \\	
	\noindent \textbf{JEL classification:} C18, C51, C58, E44
	% C18 (Methodological Issues); C51	(Model Construction and Estimation); C58 (Financial Econometrics); E44 (Financial Markets and the Macroeconomy)
\end{abstract}

% \thispagestyle{empty}
% \clearpage
% \addtocounter{page}{-1}

\pagebreak 
\section{Introduction}
\label{Motivation}

% \doublespacing
\onehalfspacing

In this paper, we add to the econometric toolbox by introducing \textit{regressions under adverse conditions}.
While the classical ordinary least squares (OLS) regression relates the mean of an outcome variable to covariates, our new regression relates the mean of an outcome variable to covariates only when some other (distress) variable falls in its tail.
Here, a variable falling in its tail means an exceedance of a large conditional quantile, which is a natural statistical interpretation of an adverse condition going back to \citet{KB78}.
%which is denoted as the Value-at-Risk (VaR) in the risk management literature.
Technically, our methodology can be interpreted as a regression for \citeauthor{Aea17}'s \citeyearpar{Aea17} systemic risk measure Marginal Expected Shortfall (MES), which is precisely defined as the conditional mean of some outcome variable given that a distress variable exceeds a large quantile.
Therefore, we use the terms \textit{regression under adverse conditions} and \textit{MES regression} interchangeably. %that we use throughout the paper.
% We call the conditioning variable a ``distress'' variable and the quantile exceedance the ``stress event''.

There are various applications where the relation between an outcome variable and explanatory variables is of particular interest if some distress variable experiences an extreme outcome. 
First, in the systemic risk literature, the relationship between a bank's stock performance and bank characteristics is mainly of interest if the financial system is in distress \citep{AB16, Aea17, BRS20, BDP20}.
In this case, the bank's sensitivity to financial turmoil  (i.e., its systemic riskiness) is related to bank-level data, yielding insights into the determinants of systemic risk. 
Our regressions under adverse conditions can deliver such insights.
Second, in the recent ``Growth-at-Risk'' literature \citep{ABG19,BS21,Aea22}, interest focuses on the risk to a country's gross domestic product (GDP) growth.
%  where risk is measured as the VaR or the expected shortfall (ES). 
Our MES regressions allow growth risk---measured with the Expected Shortfall (ES)---to be dissected into different parts that can be attributed to specific sub-components (such as sub-regions), therefore contributing to the study of the interconnections of macroeconomic risks.
While we apply our regressions under adverse conditions to real data in these two fields in Section~\ref{sec:EmpiricalApplication}, we also explore some further uses in portfolio optimization in Appendix~\ref{sec:add appl} of the Online Supplement.

The three main contributions of this paper are as follows.
First, we propose a general modeling framework for the---typically dynamic---conditional MES.
Our models are flexible enough to incorporate possible non-linearities in the relationship between the MES and the regressors. 
Our framework also allows for contemporaneous and/or lagged covariates, which may be serially dependent.
Second, we provide feasible inference tools for MES regressions, including parameter and asymptotic variance estimation. 
This allows practitioners to identify the regressors with significant explanatory content. Third, we illustrate the usefulness of our regressions under adverse conditions in the three diverse contexts mentioned above.

From the technical side, standard (one-step) M-estimation of MES regression models is infeasible as there exists no \textit{scalar} consistent loss function that is uniquely minimized in expectation by the MES \citep{FH24, DFZ_CharMest}.
This contrasts with classical mean regressions, which can be estimated based on the squared error loss, as this loss is minimized (in expectation) by the mean.

To make up for this defect of the MES, we draw on \textit{bivariate} ``multi-objective'' consistent loss functions for the MES (together with the quantile) proposed by \citet[Theorem~4.2]{FH24}.
Following tradition in the risk management literature, we denote quantiles as the Value-at-Risk (VaR).
These bivariate loss functions lead to a \textit{two-step} M-estimator:
In the first step, we estimate the parameters of the VaR model for the distress variable by minimizing the quantile loss, exactly as in quantile regressions.
In the second step, we obtain parameter estimates of the MES model by minimizing the squared error loss only for those observations with a distress variable larger than the (estimated) VaR.
This event is typically called a ``VaR exceedance''.

In developing asymptotic theory for our two-step estimator, we draw on classic results of \citet{NeweyMcFadden1994} for consistency and on invariance principles of \citet{DMR95} for asymptotic normality. 
Note that for asymptotic normality we cannot rely on standard results for two-step M-estimation \citep[as in][Sec.~6]{NeweyMcFadden1994} because our second-step objective function is discontinuous due to the truncation at the VaR.
As expected for two-step estimators, the estimation precision in the second step is influenced by the first-step estimator.
We further propose methods for feasible inference based on consistent estimation of the asymptotic variance-covariance matrix.

We illustrate the good finite-sample performance the two-step M-estimator and the associated inference methods in simulations, covering typical sample sizes and probability levels (for the quantile model) that we employ in our empirical applications.
We further compare our regression method with an \textit{ad hoc} estimator that is used in \citet{BRS20, BDP20}, \citet{Berger2020} and \citet{Karolyi2023} to estimate MES models and show that this method leads to inconsistent parameter estimates together with unrealistically small standard errors.

In Section~\ref{sec:EmpiricalApplication}, we apply our regressions under adverse conditions in the two above sketched fields of finding covariates associated with systemic risk, and in dissecting an economic region's GDP growth vulnerability into contributions of individual countries.
First, we reconsider an analysis similar to \citet[Table~7]{BRS20} to analyze the contemporaneous relation of asset price bubbles and the systemic riskiness (as measured by their MES) of the three systemically most relevant US banks.
In doing so, we compare our MES regression estimator with the ad hoc estimation method used in \citet{BRS20}.
Overall, we find that the MES is negatively affected by bust periods, but the effect of boom periods is not statistically significant, yielding---as in our simulations---different results compared to \citet{BRS20}. 

Second, the recent ``Growth-at-Risk'' (GaR) literature, started by \citet{ABG19}, analyzes downside risks to future GDP growth as a function of current economic and financial conditions.
We apply the GaR methodology to the economic region consisting of Germany, France and the United Kingdom (UK)---the three largest economies in Europe.
Following \citet{ABG19}, we find that financial and economic conditions pose similar risks to future growth of the whole economic region.
Our regressions under adverse conditions now allow us to investigate how these total (economic and financial) risk factors can be attributed to the individual countries.
We obtain the striking finding that---in contrast to France and Germany---the risk to growth in the UK is mainly determined by \textit{financial} conditions, which may be explained by the importance of the financial marketplace London.
Similarly, we find that---in contrast to France and the UK---current \textit{economic} conditions in the joint economic region are the main risk to future German growth, which may be due to Germany's strong export-oriented manufacturing base.

Our regressions under adverse conditions are related to quantile regressions of \citet{KB78} in that they model a tail functional, and as the (preliminary) first part VaR regression is simply a (possibly nonlinear) quantile regression.
In that context we also refer to \citet{EM04} and \citet{CL23} for time series applications of quantile regressions, to \citet{Hog24} for MES forecasting models based on CCC--GARCH-type filters, and to \cite{DH24} for dynamic models for the related systemic risk measure CoVaR.
Our method is more closely related to recently developed regressions for the ES of \citet{DimiBayer2019}, \citet{PZC19}, \citet{GBP21}, \citet{Bar23+} and \citet{FMW23}.
% The ES is defined as the (conditional) mean of an outcome variable \textit{given} that this variable exceeds its conditional quantile.
Our and their methods share the property of relying on an ancillary quantile regression. 
However, while ES regressions can be estimated jointly with VaR/quantile regressions (essentially due to the joint elicitability of the pair (VaR, ES) due to \citet{FZ16a}), our MES regressions require a two-step M-estimator (due to (VaR, MES) only being multi-objective elicitable in the sense of \citet{FH24}).
As the MES simplifies to the ES when the distress and outcome variables coincide, our MES regressions may be seen as a multivariate extension of ES regressions, which opens up many new fields of application, as outlined above and as demonstrated in more detail in the empirical applications.

Our MES regressions are further related to threshold models (or also: sample split models), which consider (possibly distinct) regression models for the two subsamples where some threshold variable is above or below a deterministic threshold parameter \citep{Han99,Han00,CH04}.
The main differences between our MES regressions and threshold models are as follows. 
First, in the latter the threshold is a fixed scalar, whereas our models allow the threshold (i.e., the quantile of the distress variable) to be influenced by the user via the quantile level. 
This flexibility is important, e.g., in the portfolio application, where the quantile level (most often 97.5\%) is predetermined by the regulator \citep{BCBS19}. 
Second, while the threshold is deterministic in sample split models, we allow the stochastic threshold (i.e., the VaR) to be modeled dynamically.
%Second, MES regressions only focus on the case of an extreme conditioning variable, whereas threshold regressions consider models for the cases where the threshold variable is above \textit{and} below the threshold parameter. 
Third, the models we consider are more general in that our framework covers non-linear relationships, whereas threshold regressions are mostly confined to linear models.

The remainder of the paper is structured as follows. Section~\ref{Modeling and Estimation} introduces MES regression models and our proposed two-step estimator. Its asymptotic properties and feasible inference methods are presented in Section~\ref{Asymptotic Properties}. The simulations in Section~\ref{sec:Simulations} verify the good finite-sample performance of our estimator.
Section~\ref{sec:EmpiricalApplication} demonstrates the usefulness of MES regressions in two empirical applications and the final Section~\ref{sec:Conclusion} concludes.
The proofs of all technical results are relegated to the Supplementary Material that contains the Appendices~\ref{sec:thm1}--\ref{sec:AddEmpResults}.

\section{Modeling and Estimation}
\label{Modeling and Estimation}

\subsection{The Definition of MES}

Consider some time series $\big\{\mV_t=(\mW_t^\prime, X_t,Y_t)^\prime\big\}_{t\in\mathbb{N}}$, where $\mW_t$ are exogenous variables, $X_t$ denotes the distress or conditioning variable (e.g., a measure of macroeconomic or financial conditions) and $Y_t$ stands for the variable of interest (e.g., economic or financial losses). 
Let $\mathcal{F}_{t}=\sigma(\mW_{t}, \mV_{t-1}, \mV_{t-2}, \ldots)$ be the information set generated by the contemporaneous $\mW_t$ and past $\mV_t$.
The inclusion of $\mW_t$ in $\mathcal{F}_t$ allows us to incorporate contemporaneous variables in MES regressions.
For instance, in studying the relationship between asset price bubbles and systemic risk, \citet{BRS20} use a contemporaneous bubble indicator in an MES-type regression; see also the empirical application in Section~\ref{sec:ApplMESReg}.

For $\b\in(0,1)$ we define the VaR as the generalized inverse $\VaR_{t,\beta}:=\VaR_{\beta}(X_t\mid\mathcal{F}_t):=F_{X_t\mid\mathcal{F}_{t}}^{\leftarrow}(\b)$, where $F_{X_t\mid\mathcal{F}_{t}}$ denotes the conditional cumulative distribution function (c.d.f.) of $X_t$ given $\mathcal{F}_t$.
The distress event that quantifies the ``adverse conditions'' in the definition of the MES is that the stress variable $X_t$ exceeds its VaR, i.e., $\{X_t\geq\VaR_{t,\beta}\}$.
With our orientation of $X_t$ denoting (economic or financial) losses, we commonly consider values for $\b$ close to one.
Left-tail truncations in the sense of  $\{X_t\leq\VaR_{t,\beta}\}$ can be analyzed within our framework by simply considering $-X_t$.
We formally define our regression target under adverse conditions as $\MES_{t,\b}:=\MES_{\b}(Y_t\mid\mathcal{F}_{t}):=\E_t\big[Y_t \mid X_t\geq\VaR_{\beta}(X_t\mid\mathcal{F}_t)\big]$, where $\E_t[\cdot]:=\E[\ \cdot\mid\mathcal{F}_t]$ and we suppress the dependence of the MES on $X_t$ for brevity.

\subsection{Joint Models for VaR and MES}

Let $\mZ_t=(Z_{1,t},\ldots,Z_{k,t})^\prime$ be the observable covariates affecting the VaR and the MES (i.e., $\VaR_{t,\beta}$ and $\MES_{t,\b}$).
We assume that $\mZ_t$ is $\mathcal{F}_{t}$-measurable, such that it can, e.g., be a subset of the contemporaneous $\mW_t$ and past $\mV_t$.
The empirical applications in Section~\ref{sec:EmpiricalApplication} provide some concrete examples for $\mZ_t$.
We introduce the functions $v_t(\vtheta^{v})=v(\mZ_t;\, \vtheta^v)$ and $m_t(\vtheta^{m})=m(\mZ_t;\, \vtheta^m)$, where $\vtheta^v$ ($\vtheta^m$) is a generic vector from some parameter space $\mTheta^v\subset\mathbb{R}^p$ ($\mTheta^m\subset\mathbb{R}^q$), and $v(\,\cdot\,;\vtheta^v)$ and $m(\,\cdot\,;\vtheta^m)$ are measurable functions for all $\vtheta^v\in\mTheta^v$ and $\vtheta^m\in\mTheta^m$, respectively. 
We assume that there exist unique true parameters $\vtheta^v_0 \in \mTheta^{v}$ and $\vtheta^m_0 \in \mTheta^{m}$, such that almost surely (a.s.)
\begin{align}
	\label{eqn:TrueModelParameters}
	\begin{pmatrix}
		\VaR_{\beta}(X_t\mid\mathcal{F}_t)\\
		\MES_{\b}(Y_t\mid\mathcal{F}_{t})
	\end{pmatrix}
	= \begin{pmatrix}
		v_t(\vtheta_{0}^{v})\\
		m_t(\vtheta_{0}^{m})
	\end{pmatrix},\qquad t\in\mathbb{N}.
\end{align}
Thus, $v_t(\vtheta_{0}^{v})$ is simply a (possibly non-linear) quantile regression model that relates the $\beta$-quantile of the distress variable $X_t$ to covariates $\mZ_t$.
This may be seen as an auxiliary regression step that allows to identify the ``adverse condition'' $\{X_t\geq\VaR_{\beta}(X_t\mid\mathcal{F}_t)\}$ in the actual quantity of interest $\MES_{\b}(Y_t\mid\mathcal{F}_{t})=\E_t\big[Y_t \mid X_t\geq\VaR_{t,\beta}\big]$, which is modeled via $m_t(\vtheta_{0}^{m})$ and serves to explain the interrelation between $Y_t$ and $X_t$ (in the form of the MES) with covariates.
Therefore, we call the models $v_t(\vtheta_{0}^{v})$ and $m_t(\vtheta_{0}^{m})$ in \eqref{eqn:TrueModelParameters} a \textit{regression under adverse conditions} or also an \textit{MES regression model}.
The leading special case of \eqref{eqn:TrueModelParameters} is the linear model shown in  Example~\ref{ex:1} below, where $v_t(\vtheta^{v})$ and $m_t(\vtheta^{m})$ are linear in the parameters $\vtheta^v$ and $\vtheta^m$, respectively.

The fact that the conditional-on-$\mathcal{F}_t$ quantities $\VaR_{t,\beta}$ and $\MES_{t,\b}$ only depend on $\mZ_t$ is, of course, a modeling assumption. 
Except for the restriction that $\mZ_t$ be $\mathbb{R}^{k}$-valued, this is quite a flexible assumption. It covers purely predictive regressions (where $\mZ_t$ only contains information available at time $t-h$, say) and also static regressions (where $\mZ_t$ only contains contemporaneous variables) and any combination of the two.

Note that we consider the case of separated parameters in \eqref{eqn:TrueModelParameters}, where the parameters pertain only to the VaR model or the MES model.
This is vital for our two-step M-estimator introduced below.

% We also mention that our framework models the interconnectedness between $X_t$ and $Y_t$ through the model for the MES, i.e., through $\MES_{t,\b}=m(\mZ_t;\vtheta_0^m)$
We also mention that our framework inherently models the interconnectedness between $X_t$ and $Y_t$ through the (parameter-free) definition of the MES.
This contrasts with, e.g., \citet[Section III.B]{AB16}, whose estimation method (for the closely related CoVaR) can only capture linear effects of the VaR of $X_t$ onto $Y_t$.

\begin{example}\label{ex:1}
Here, we provide an example of a model that gives rise to a linear regression under adverse conditions in \eqref{eqn:TrueModelParameters}. Consider the model
\begin{align*}
	X_t &= \mZ_t^{v\prime}\vtheta_0^{v}+\varepsilon_{X,t},\\
	Y_t &= \mZ_t^{m\prime}\vtheta_0^{m}+\varepsilon_{Y,t},
\end{align*}
where $\mZ_t^{v}$ and $\mZ_t^{m}$ are ($p$ and $q$-dimensional) vectors containing some (or all) components of $\mZ_t$,
% that are hence $\mathcal{F}_{t}$-measurable.
%where $\mZ_t=(\mZ_t^{v\prime}, \mZ_t^{m\prime})^\prime$ is assumed to be $\mathcal{F}_{t}$-measurable,
and we assume that
$\VaR_{\b}(\varepsilon_{X,t}\mid\mathcal{F}_t)=F_{\varepsilon_{X,t}\mid\mathcal{F}_t}^{\leftarrow}(\beta)=0$, and $\MES_{\b}(\varepsilon_{Y,t}  \mid \mathcal{F}_t) = \E_t\big[\varepsilon_{Y,t}\mid \varepsilon_{X,t}\geq \VaR_{\b}(\varepsilon_{X,t}\mid\mathcal{F}_t)\big]=0$. 
Then, exploiting the $\mathcal{F}_t$-measurability of $\mZ_t$, it is easy to check that \eqref{eqn:TrueModelParameters} holds with linear $v_t(\vtheta_{0}^{v})=\mZ_t^{v\prime}\vtheta_0^{v}$ and linear $m_t(\vtheta_{0}^{m})=\mZ_t^{m\prime}\vtheta_0^{m}$. 
Recall that $\VaR_\b(\varepsilon_{X,t}\mid\mathcal{F}_t)=0$ is the standard assumption on the errors in quantile regressions. The requirement that %$\E_t\big[\varepsilon_{Y,t}\mid \varepsilon_{X,t}\geq0\big]=0$ 
$\MES_{\b}(\varepsilon_{Y,t}  \mid \mathcal{F}_t) = 0$ may then be viewed as the analog assumption in our MES regressions.
Figure~\ref{fig:MESillu} graphically illustrates a linear regression under adverse conditions with one covariate.
The top panel shows the quantile regression line $v_t(\vtheta_0^v)$ in blue and the bottom panel the MES regression line $m_t(\vtheta_0^m)$ in red.
The black points in both panels indicate the observations $(X_t,Y_t)^\prime$ with a VaR exceedance, such that $X_t\geq v_t(\vtheta_0^v)$.
It is obvious from this figure that the auxiliary quantile regression in the top plot is key in identifying the relevant points for the MES regression in the bottom plot.
\end{example}

\begin{figure}
	\centering
	\includegraphics[width=\linewidth]{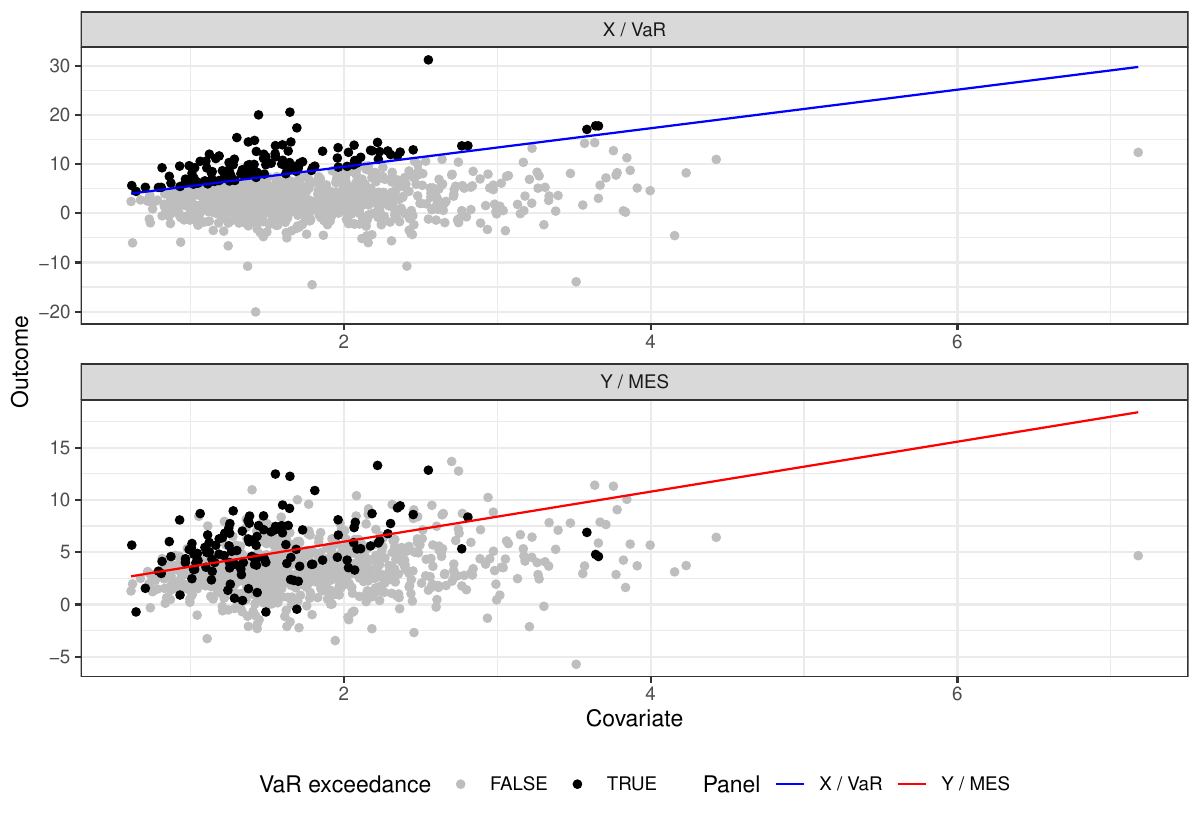}
	\caption{Top panel: Quantile regression line $v_t(\vtheta_0^v)$ in blue. 
	Bottom panel: MES regression line $m_t(\vtheta_0^m)$ in red. 
	Points in black in both panels correspond to observations with a VaR exceedance (i.e., $X_t\geq v_t(\vtheta_0^v)$).}
	\label{fig:MESillu}
\end{figure}

While ``usual'' regressions simply relate some outcome $Y_t$ to covariates, our MES regressions allow us to investigate such a relationship \textit{conditional} on an extreme realization of $X_t$. This is of interest, e.g., when relating a bank's losses $Y_t$ to institution characteristics, where the precise relationship is mainly of interest during times of large system-wide losses $X_t$. We use the model of Example~\ref{ex:1} to investigate this in the empirical application in Section~\ref{sec:ApplMESReg}, but also in the other applications in Section~\ref{sec:GDP_MES} and Appendix~\ref{sec:add appl}.
We discuss the relation of our MES regressions to OLS regressions that use $X_t$ as an explanatory variable in Appendix~\ref{sec:RelationOLSXCovariate}.

\begin{rem}
	\label{rem:MarginalEffect}
	Using the notation $\mZ_t^m = (Z^m_{1,t},\dots, Z^m_{q,t})^\prime$ and $\vtheta^m = (\theta_1^m, \dots, \theta_q^m)^\prime$,
	our models in \eqref{eqn:TrueModelParameters} allow for the definition of a \emph{marginal effect under adverse conditions} with respect to $Z^m_{j,t}$ for any $j \in \{1,\dots,q\}$ via $\frac{\partial}{\partial Z^m_{j,t}} m_t(\vtheta^m)$.
	This captures the impact of a unit change in $Z^m_{j,t}$ on the conditional MES given all other variables remain constant, which is analogous to the classical marginal effect for mean or quantile regressions.
	For the linear models of Example~\ref{ex:1}, the marginal effect under adverse conditions is simply $\theta^m_j$.
\end{rem}

\subsection{Parameter Estimation}
\label{sec:ParameterEstimation}

We now introduce our estimators of the unknown model parameters $\vtheta_{0}^{v}$ and $\vtheta_{0}^{m}$. 
Typically, (one-step) M-estimation requires a to-be-minimized loss function (or also: objective function). Yet, as pointed out in the Introduction, there is no real-valued loss function associated with the pair (VaR, MES). The existence of such a (strictly consistent) loss function is, however, necessary for consistent M-estimation of semiparametric models such as ours \citep[Theorem 1]{DFZ_CharMest}.

To overcome this lack, we draw on \citet{FH24}.
Denoting by $\1_{\{\cdot\}}$ the indicator function, they show that (under some smoothness conditions) the $\mathbb{R}^2$-valued loss function
\begin{equation}\label{eq:loss}
	\mS\Bigg(\begin{pmatrix}v\\ m\end{pmatrix}, \begin{pmatrix}x\\ y\end{pmatrix}\Bigg) = \begin{pmatrix}S^{\VaR}(v,x)\\ S^{\MES}\big((v,m)^\prime, (x,y)^\prime\big) \end{pmatrix}=
	\begin{pmatrix}
		[\1_{\{x\leq v\}}-\b][v-x]\\
		\frac{1}{2}\1_{\{x>v\}}[y-m]^2
	\end{pmatrix}
\end{equation}
is minimized in expectation by the true VaR and MES with respect to the lexicographic order. That is, for all $v,m\in\mathbb{R}$,
\[
	\E\Bigg[\mS\Bigg(\begin{pmatrix}\VaR_\b\\ \MES_\b\end{pmatrix}, \begin{pmatrix}X\\ Y\end{pmatrix}\Bigg)\Bigg]\preceq_{\lex}\E\Bigg[\mS\Bigg(\begin{pmatrix}v\\ m\end{pmatrix}, \begin{pmatrix}X\\ Y\end{pmatrix}\Bigg)\Bigg],
\]
where $(x_1,x_2)^\prime\preceq_{\lex} (y_1,y_2)^\prime$ if $x_1<y_1$ or ($x_1=y_1$ and $x_2\leq y_2$), and $\VaR_\b$ and $\MES_\b$ denote the true VaR and MES of the joint distribution of $(X,Y)^\prime$. Note that $S^{\VaR}(\cdot,\cdot)$ in \eqref{eq:loss} is the standard pinball-loss function known from quantile regression. Clearly, $S^{\MES}(\cdot,\cdot)$ resembles the squared error loss familiar from usual mean regressions, except that the indicator $\1_{\{x>v\}}$ restricts the evaluation to observations with VaR exceedances in the first component. 
The pre-factor $1/2$ is simply a convenient normalization, enabling an easier representation of some subsequent quantities.
In the related literature, loss functions are often also called \textit{scoring functions} \citep{Gne11} and we use these two terms interchangeably.

To introduce our two-step estimator, suppose that we have a sample $\big\{(X_t,Y_t,\mZ_t^\prime)^\prime\big\}_{t=1,\ldots,n}$ of size $n$.
In a first step, the lexicographic order suggests to minimize the expected score of the VaR component. Approximating the expected score by the sample average of the scores, the first-step estimator is
\[
	\widehat{\vtheta}_{n}^{v} = \argmin_{\vtheta^{v}\in\mTheta^{v}}\frac{1}{n}\sum_{t=1}^{n}S^{\VaR}\big(v_t(\vtheta^{v}), X_t\big).
\]
\citet{OH16} study this estimator in their nonlinear quantile regressions for deterministic regressors.

Given $\widehat{\vtheta}_n^v$, the lexicographic order similarly suggests in a second step to minimize the sample average of the MES losses via
\begin{align}
	\label{eq:MES_estimator}
	\widehat{\vtheta}_{n}^{m} = \argmin_{\vtheta^{m}\in\mTheta^{m}}\frac{1}{n}\sum_{t=1}^{n}S^{\MES}\Big(\big(v_t(\widehat{\vtheta}_{n}^{v}), m_t(\vtheta^{m})\big)^\prime, \big(X_t, Y_t\big)^\prime\Big).
\end{align}
For this two-step estimator to be feasible, the requirement that the VaR model does not depend on the MES model is essential. Section~\ref{Asymptotic Properties} shows that the estimation precision of $\widehat{\vtheta}_{n}^{m}$ is affected by relying on $\widehat{\vtheta}_n^v$ in the second step, which is the usual case for two-step estimators \citep[Section 6]{NeweyMcFadden1994}.

\section{Asymptotics for the Two-Step Estimators}\label{Asymptotic Properties}

% \subsection{Asymptotic Normality}\label{Asymptotic Normality}

We now show the joint asymptotic normality of our two-step M-estimators $\widehat{\vtheta}_{n}^{v}$ and $\widehat{\vtheta}_{n}^{m}$ under the regularity conditions in Assumptions~\ref{ass:1}--\ref{ass:7} given below. Before presenting Assumptions~\ref{ass:1}--\ref{ass:7}, we have to introduce some notation. The joint c.d.f.~of $(X_{t}, Y_t)^\prime\mid\mathcal{F}_t$ is denoted by $F_t(\cdot,\cdot)$, and its Lebesgue density (which we assume exists) by $f_t(\cdot,\cdot)$. Similarly, $F_t^{W}(\cdot)$ ($f_t^{W}(\cdot)$) denotes the distribution (density) function of $W_t\mid\mathcal{F}_t$ for $W\in\{X,Y\}$. For sufficiently smooth functions $\mathbb{R}^{p}\ni\vtheta\mapsto f(\vtheta) \in \mathbb{R}$, we denote the $(p\times1)$-gradient by $\nabla f(\vtheta)$, its transpose by $\nabla' f(\vtheta)$ and the $(p\times p)$-Hessian by $\nabla^2 f(\vtheta)$. 
All (in-)equalities involving random quantities are meant to hold almost surely.

The asymptotic variance-covariance matrices of our estimators depend on the following matrices:
\begin{align*}
	\mV&=\b(1-\b)\E\big[\nabla v_t(\vtheta_0^v)\nabla^\prime v_t(\vtheta_0^v)\big]\in\mathbb{R}^{p\times p},\\
	\mLambda & =\E\Big[f_t^{X}\big(v_t(\vtheta_0^v)\big)\nabla v_t(\vtheta_0^v)\nabla^\prime v_t(\vtheta_0^v)\Big]\in\mathbb{R}^{p\times p},\\
	\mM^{\ast}&= \E\Big[\big\{Y_t -m_t(\vtheta_0^m)\big\}^2 \1_{\{X_t>v_t(\vtheta_0^v)\}} \nabla m_t(\vtheta_0^m) \nabla^\prime m_t(\vtheta_0^m)\Big]\in\mathbb{R}^{q\times q},\\
	\mLambda_{(1)} &= (1-\b)\E\big[\nabla m_t(\vtheta_0^m)\nabla^\prime m_t(\vtheta_0^m)\big]\in\mathbb{R}^{q\times q},\\
	\mLambda_{(2)} &=\E\bigg[\Big\{\int_{-\infty}^{\infty} y f_t\big(v_t(\vtheta_0^v),y\big)\D y - m_t(\vtheta_0^m)f_{t}^{X}\big(v_t(\vtheta_0^v)\big)\Big\}\nabla m_t(\vtheta_0^m)\nabla^\prime v_t(\vtheta_0^v)\bigg]\in\mathbb{R}^{q\times p}.
\end{align*}
All these matrices exist by virtue of Assumptions~\ref{ass:1}--\ref{ass:7}, which we state now. For this, let $K<\infty$ be some large universal positive constant, and denote by $\norm{\vx}$ the Euclidean norm when $\vx$ is a vector, and the Frobenius norm when $\vx$ is matrix-valued.

\begin{assumption}\label{ass:1}
\renewcommand{\theenumi}{(\roman{enumi})}
\begin{enumerate}
\item\label{it:1i} There exist unique parameters $\vtheta_0^v$ and $\vtheta_0^m$, such that \eqref{eqn:TrueModelParameters} holds a.s.~for all $t\in\mathbb{N}$.

\item\label{it:1ii}
The parameter space $\mTheta=\mTheta^{v}\times\mTheta^{m}$ is compact, where $\mTheta^{v}\subset\mathbb{R}^p$ and $\mTheta^{m}\subset\mathbb{R}^q$.

\item\label{it:1iii} 
$\vtheta_0^v\in\inter(\mTheta^v)$ and $\vtheta_{0}^{m}\in\inter(\mTheta^m)$, where $\inter(\cdot)$ denotes the interior of a set.

\end{enumerate}
\end{assumption}

\begin{assumption}\label{ass:2}
$\big\{(X_t, Y_t, \mZ_t^\prime)^\prime\big\}_{t\in\mathbb{N}}$ is strictly stationary and $\beta$-mixing with coefficients $\beta(\cdot)$ of size $-r/(r-1)$ for some $r>1$.
\end{assumption}

\begin{assumption}\label{ass:3}
\renewcommand{\theenumi}{(\roman{enumi})}
\begin{enumerate}
	\item\label{it:3i} For all $t\in\mathbb{N}$, $F_t(\cdot,\cdot)$ belongs to a class of distributions on $\mathbb{R}^2$ that possess a positive Lebesgue density $f_t(x,y)$ for all $(x,y)^\prime\in\mathbb{R}^2$ such that $F_t(x,y)\in(0,1)$.

	\item\label{it:3ii} $\big|f_t^{X}(x)-f_t^{X}(x^\prime)\big|\leq K|x-x^\prime|$ and $\sup_{x\in\mathbb{R}}f_t^{X}(x)\leq K$.
	
	\item\label{it:3iii} The conditional density $f_t(x,y)$ is differentiable in its first argument, with partial derivative denoted by $\partial_1 f_t(x,y)$. 
	
	\item\label{it:3iv} $\sup_{x\in\mathbb{R}}\big|\int_{-\infty}^{\infty}y\partial_1 f_t(x,y)\D y\big|\leq F(\mathcal{F}_t)$ and $\sup_{x\in\mathbb{R}}\int_{-\infty}^{\infty}|y| f_t(x,y)\D y\leq F_1(\mathcal{F}_t)$ for $\mathcal{F}_t$-measurable random variables $F(\mathcal{F}_t)$ and $F_1(\mathcal{F}_t)$.

\end{enumerate}
\end{assumption}

\begin{assumption}\label{ass:4}
\renewcommand{\theenumi}{(\roman{enumi})}
\begin{enumerate}
\item\label{it:4i} For all $t\in\mathbb{N}$, $v_t(\cdot)$ and $m_t(\cdot)$ are twice differentiable on $\inter(\mTheta^v)$ and $\inter(\mTheta^m)$ with gradients $\nabla v_t(\cdot)$ and $\nabla m_t(\cdot)$, and Hessians $\nabla^2 v_t(\cdot)$ and $\nabla^2 m_t(\cdot)$.

\item\label{it:4ii} $\sup_{\vtheta^v\in\mTheta^v}\big|v_t(\vtheta^v)\big|\leq V(\mZ_t)$ and $\sup_{\vtheta^m\in\mTheta^m}\big|m_t(\vtheta^m)\big|\leq M(\mZ_t)$ for some measurable functions $V(\cdot)$ and $M(\cdot)$.

\item\label{it:4iii} $\sup_{\vtheta^v\in\mTheta^v}\big\Vert\nabla v_t(\vtheta^v)\big\Vert\leq V_1(\mZ_t)$ and $\sup_{\vtheta^m\in\mTheta^m}\big\Vert\nabla m_t(\vtheta^m)\big\Vert\leq M_1(\mZ_t)$ for some measurable functions $V_1(\cdot)$ and $M_1(\cdot)$.

\item\label{it:4iv} $\sup_{\vtheta^v\in\mTheta^v}\norm{\nabla^2v_t(\vtheta^v)}\leq V_2(\mZ_t)$ and $\sup_{\vtheta^m\in\mTheta^m}\norm{\nabla^2m_t(\vtheta^m)}\leq M_2(\mZ_t)$ for some measurable functions $V_2(\cdot)$ and $M_2(\cdot)$.

\item\label{it:4v} There exist neighborhoods of $\vtheta_0^v$ and $\vtheta_0^m$, such that $\norm{\nabla^2 v_t(\vtau^v) - \nabla^2 v_t(\vtheta^v)}\leq V_3(\mZ_t)\norm{\vtau^v-\vtheta^v}$ and $\norm{\nabla^2 m_t(\vtau^m) - \nabla^2 m_t(\vtheta^m)}\leq M_3(\mZ_t)\norm{\vtau^m-\vtheta^m}$ for all elements $\vtheta^v$, $\vtau^v$ and $\vtheta^m$, $\vtau^m$ of the respective neighborhoods, and measurable functions $V_3(\cdot)$ and $M_3(\cdot)$.

\end{enumerate}
\end{assumption}

\begin{assumption}\label{ass:5}
 For $r>1$ from Assumption~\ref{ass:2} and some $\iota>0$, it holds that $\E\big[V(\mZ_t)\big]\leq K$, $\E\big[V_1^{4r}(\mZ_t)\big]\leq K$, $\E\big[V_2^{2r}(\mZ_t)\big]\leq K$, $\E\big[V_3(\mZ_t)\big]\leq K$, $\E\big[M^{4r+\iota}(\mZ_t)\big]\leq K$, $\E\big[M_1^{4r+\iota}(\mZ_t)\big]\leq K$, $\E\big[M_2^{4r}(\mZ_t)\big]\leq K$, $\E\big[M_3^{4r/(4r-1)}(\mZ_t)\big]\leq K$, $\E\big[F^{4r/(4r-3)}(\mathcal{F}_t)\big]\leq K$, $\E\big[F_1^{4r/(4r-3)}(\mathcal{F}_t)\big]\leq K$, $\E|X_t|\leq K$, $\E|Y_t|^{4r+\iota}\leq K$.
\end{assumption}

\begin{assumption}\label{ass:6}
The matrices  $\mLambda$ and $\mLambda_{(1)}$ are positive definite.
\end{assumption}

\begin{assumption}\label{ass:7}
$\sup_{\vtheta^v\in\mTheta^v}\sum_{t=1}^{n}\1_{\{X_t=v_t(\vtheta^v)\}}\leq K$ a.s.~for all $n\in\mathbb{N}$.
\end{assumption}

Assumption~\ref{ass:1}~\ref{it:1i} ensures identification of the true parameters. 
Compactness in Assumption~\ref{ass:1}~\ref{it:1ii} is a standard requirement in extremum estimation; see \citet{NeweyMcFadden1994}. 
Assumption~\ref{ass:1}~\ref{it:1iii} forces the true parameters to be interior to the respective parameter spaces, which is a key requirement for proving asymptotic normality. Assumption~\ref{ass:2} is a standard stationarity and mixing condition ensuring that suitable laws of large numbers and central limit theorems apply. 
Often, the $\beta$-mixing coefficients decrease exponentially, such that Assumption~\ref{ass:2} holds for \textit{any} $r>1$ arbitrarily close to unity; see, e.g., \citet{CC02} and \citet{Lie05} for ARMA and GARCH processes. 
Assumption~\ref{ass:3}~\ref{it:3i} ensures strict (multi-objective) consistency of the scoring function given in \eqref{eq:loss}; see \citet[Theorem~4.2]{FH24}. 
The final items \ref{it:3ii}--\ref{it:3iv} of Assumption~\ref{ass:3} may be interpreted as smoothness conditions for the conditional density.
Since the VaR and MES only depend on $\mZ_t$ under our modeling framework, it will most often be the case that the conditional density $f_t(\cdot,\cdot)$ will also be $\mZ_t$-measurable; see, e.g., Appendix~\ref{sec:ex verif}. Nonetheless, the random variables $F(\mathcal{F}_{t})$ and $F_1(\mathcal{F}_{t})$ appearing in Assumption~\ref{ass:3}~\ref{it:3iv} are in principle allowed to be $\mathcal{F}_{t}$-measurable.
Assumption~\ref{ass:4} provides smoothness conditions on the VaR and MES model, and
Assumption~\ref{ass:5} collects moment bounds.
When mixing is exponential (corresponding to $r=1$ in Assumption~\ref{ass:2}), the moment bounds of Assumption~\ref{ass:5} are weakest, such that there is the usual moment-memory tradeoff. Note that the moment conditions for the VaR model are weaker than those for the MES model, because the latter is estimated based on a \textit{squared} error-type loss.
Assumption~\ref{ass:6} ensures that the asymptotic variance-covariance matrix in Theorem~\ref{thm:an} is well-defined, which can however be degenerate if $\mV$ or $\mM^\ast$ are singular; also see \citet[footnote 20]{Hansen:82}.
Assumption~\ref{ass:7} is identical to Assumption~2~(G) in \citet{PZC19}. It prevents an infinite number of $X_t$'s to lie on any possible ``quantile regression line'' $v_t(\vtheta^v)$. Appendix~\ref{sec:ex verif} shows that $K=\dim(\vtheta^v)$ for linear models.

\begin{thm}\label{thm:an}
Suppose Assumptions~\ref{ass:1}--\ref{ass:7} hold. Then, as $n\to\infty$, 
\begin{align*}
\sqrt{n}\begin{pmatrix}
\widehat{\vtheta}_n^{v}-\vtheta_0^v\\
\widehat{\vtheta}_n^{m}-\vtheta_0^m
\end{pmatrix}
\overset{d}{\longrightarrow}N\big(\vzeros, \mGamma \mM \mGamma^\prime),
\end{align*}
where 
\[
\mGamma =\begin{pmatrix}
	\mLambda^{-1} & \vzeros\\
	-\mLambda_{(1)}^{-1}\mLambda_{(2)}\mLambda^{-1} & \mLambda_{(1)}^{-1}\end{pmatrix}\in\mathbb{R}^{(p+q) \times (p+q)}\qquad\text{and}\qquad \mM=\begin{pmatrix}\mV &\vzeros \\ \vzeros & \mM^\ast\end{pmatrix}\in\mathbb{R}^{(p+q)\times(p+q)}.
\]
\end{thm}

\sloppy
The proof of Theorem~\ref{thm:an} is in Appendix~\ref{sec:thm1} and draws on \citet{DLS23}.
In the special case $\mLambda_{(2)}=\vzeros$, the first-step estimator does not have an impact on the asymptotic variance of the second-step estimator \citep[Section 6.2]{NeweyMcFadden1994}.
To gain some intuition for this, rewrite the curly bracket in  $\mLambda_{(2)}$ as $f_t^X(v_t(\vtheta_0^v))  \left( \mathbb{E} \big[ Y_t \big| X_t = v_t(\vtheta_0^v) \big] - \mathbb{E} \big[ Y_t \big| X_t \ge v_t(\vtheta_0^v) \big] \right)$, which measures how the expectation of $Y_t$ reacts upon moving from $\{X_t = v_t(\vtheta_0^v)\}$ further to the tail through $\{ X_t \ge v_t(\vtheta_0^v)\}$.
Hence, $\mLambda_{(2)}=\vzeros$ can be interpreted as a form of ``conditional tail uncorrelatedness'' of $X_t$ and $Y_t$, as $\mLambda_{(2)}$ is a weighted expectation of the above quantity.

\begin{rem}[Expected Shortfall Regressions]
	\label{rem:ESreg}
	For $X_t = Y_t$, the definition of the MES simplifies to the Expected Shortfall, $\ES_{t,\b} := \ES_{\b}(X_t\mid\mathcal{F}_t) := \E_t\big[X_t \mid X_t \geq \VaR_{t,\beta}\big]$.
	For the ES, associated regression models (jointly with the conditional VaR) have been proposed recently. 
	E.g., \citet{DimiBayer2019} and \citet{PZC19} consider \emph{joint} M-estimation with a VaR model, which is feasible as (VaR, ES) are \emph{jointly elicitable}, as opposed to the pair (VaR, MES), which is merely multi-objective elicitable.
	One can also use a two-step estimator for VaR and ES models, whose second step is essentially given by setting $Y_t = X_t$ in the right-hand side of \eqref{eq:MES_estimator}.
	Its asymptotic variance-covariance matrix resembles the one presented in Theorem \ref{thm:an} by again setting $Y_t = X_t$ in all matrices but $\mLambda_{(2)}$, where an ES-specific version is 
	\begin{align*}
		\mLambda_{(2)}^\text{ES} =\E\Big[ \big\{v_t(\vtheta_0^v)- m_t(\vtheta_0^m) \big\} f_{t}^{X}\big(v_t(\vtheta_0^v)\big) \nabla m_t(\vtheta_0^m)\nabla^\prime v_t(\vtheta_0^v)\Big].
	\end{align*}
	This arises as $\int_{-\infty}^{\infty} y f_t\big(v_t(\vtheta_0^v),y\big) \D y =  f_{t}^{X}\big(v_t(\vtheta_0^v)\big) \mathbb{E}_t \big[ Y_t \mid X_t = v_t(\vtheta^v_0) \big]$, and the latter conditional expectation equals $v_t(\vtheta^v_0)$ if $Y_t = X_t$.
	Notice, however, that the joint distribution of $(Y_t, X_t)^\prime$ is degenerate if $Y_t = X_t$, such that our Assumption \ref{ass:3} is violated. 
	Therefore, Assumption \ref{ass:3} is replaced in ES models by the requirement that the (univariate) conditional distribution of $X_t$ is sufficiently smooth around its $\beta$-quantile; see, e.g., \citet[Assumption 2 (B)]{PZC19}.
\end{rem}

For feasible inference, we have to estimate the asymptotic variance-covariance matrix of Theorem~\ref{thm:an} consistently.
In Appendix~\ref{sec:thm3}, we propose suitable estimators and show their consistency in Theorem~\ref{thm:lrv} under the additional Assumptions~\ref{ass:9}--\ref{ass:8}.

Appendix~\ref{sec:ex verif} verifies Assumptions~\ref{ass:1}--\ref{ass:7} and also Assumptions~\ref{ass:9}--\ref{ass:11} for the linear MES regression model of Example~\ref{ex:1} subject to some regularity conditions.

\section{Simulations}
\label{sec:Simulations}

\subsection{The Data-Generating Process}

We consider a time series regression as a data-generating process (DGP), which exhibits the desirable features of a  time-varying conditional mean and covariance matrix and at the same time results in a correctly specified regression for the MES (and the VaR) that allows for a comparison with the estimation used in \cite{BRS20, BDP20} in Section \ref{sec:BrunnermeierMESReg}.

Using the notation of Example \ref{ex:1}, the covariates $\mZ^v_{t} = \mZ^m_{t} =(1, Z_{1,t}, Z_{2,t})^\prime$ follow the (transformed) auto-regressions
\begin{align*}
	Z_{1,t} &= 0.3 + 0.4 \exp(\xi_t)  \qquad \text{with } \qquad \xi_t = 0.6 \xi_{t-1} + \varepsilon_{\xi,t}, \\
	Z_{2,t} &= 0.75 Z_{2,t-1} + \varepsilon_{Z,t},
\end{align*}
where $\varepsilon_{Z,t}, \varepsilon_{\xi,t} \stackrel{\text{i.i.d.}}{\sim} {N}(0,1)$, independently of each other for all $t=1,\ldots,n$.
The exponential transformation guarantees positivity of $Z_{1,t}$, which is required to introduce heteroskedasticity in the process
% We simulate from the heteroskedastic process
\begin{align}
	\label{eqn:DGPCrossSectional}
	\big(X_t, Y_t \big)'
	&= (\gamma_1,\gamma_1)^\prime + (\gamma_2, \gamma_2)^\prime  Z_{1,t} + (\gamma_3,\gamma_3)^\prime  Z_{2,t} + \big(\gamma_4 + \gamma_5  Z_{1,t}  \big) \boldsymbol{\varepsilon}_t, \qquad t=1,\ldots,n.
	%	  \qquad \text{where} \\
	%	&V_t \stackrel{iid}{\sim} t_{\nu} \big( \boldsymbol{0}, \boldsymbol{\Sigma} \big) \qquad \text{and} \qquad  \varepsilon_{Z,t} \stackrel{iid}{\sim} {N}(0,1).
\end{align}
The multivariate $t$-distributed innovations $\boldsymbol{\varepsilon}_t =(\varepsilon_{1,t}, \varepsilon_{2,t})^\prime\stackrel{\text{i.i.d.}}{\sim} t_{6} \big( \boldsymbol{0}, \boldsymbol{\Sigma} =  \begin{psmallmatrix}
	1 & 1.2 \\
	1.2 & 4	
\end{psmallmatrix}\big)$
%$\nu = 8$, and
%$\boldsymbol{\Sigma} = 
%\begin{psmallmatrix}
%	1 & 1 \\
%	1 & 4	
%\end{psmallmatrix}$ 
%
%(such that $\operatorname{Corr}(V_{1,t}, V_{2,t}) = 0.6$) 
are independent of the $\{\varepsilon_{Z,t}\}$ and $\{\varepsilon_{\xi,t}\}$.
We choose $\vgamma=(\gamma_1, \ldots, \gamma_5)^\prime =(1,\ 1.5,\ 2,\ 0.25,\ 0.5)^\prime$ as the true values. 
For simplicity, our DGP in \eqref{eqn:DGPCrossSectional} implies that $X_t$ and $Y_t$ are driven by the same factors and only the heteroskedastic shock differs between the variables.
This DGP results in the linear (VaR, MES) model 
\begin{align*}
	v_t(\vtheta^v) = \theta^v_1  + \theta^v_2 Z_{1,t} +  \theta^v_3 Z_{2,t}
	\qquad \text{ and } \qquad 
	m_t(\vtheta^m) = \theta^m_1  + \theta^m_2 Z_{1,t} +  \theta^m_3 Z_{2,t}.
\end{align*}
The true parameter values $\vtheta_0^v=(\theta_{1,0}^v,\, \theta_{2,0}^v,\, \theta_{3,0}^v)^\prime$ and $\vtheta_0^m=(\theta_{1,0}^m,\, \theta_{2,0}^m,\, \theta_{3,0}^m)^\prime$ are given by
\begin{align*}
	\vtheta^v_{0} = 
	\big(\gamma_1 + \gamma_4 \tilde q_\beta, \; 
	\gamma_2 + \gamma_5 \tilde q_\beta,  \; 
	\gamma_3 \big)^\prime
	\qquad \text{and} \qquad
	\vtheta^m_{0} = 
	\big(\gamma_1 + \gamma_4 \tilde m_\beta, \;
	\gamma_2 + \gamma_5 \tilde m_\beta, \;
	\gamma_3  \big)^\prime,
\end{align*}
%
%\begin{alignat*}{3}
%	&\vtheta_{1,0} = \gamma_1 + \gamma_4 q_t(\beta), \qquad
%	&&\vtheta_{2,0} = \gamma_2 + \gamma_5 q_t(\beta), \qquad
%	&&\vtheta_{3,0} = \gamma_3 \\
%	&\vtheta_{4,0} = \gamma_1 + \gamma_4 c_t(\alpha|\beta), \qquad
%	&&\vtheta_{5,0} = \gamma_2 + \gamma_5 c_t(\alpha|\beta), \qquad
%	&&\vtheta_{6,0} = \gamma_3 
%\end{alignat*}
where $\tilde m_\beta$ is the $\beta$-MES of the $t_{6} \big( \boldsymbol{0}, \boldsymbol{\Sigma} \big)$-distribution, and $\tilde q_\beta$ is the $\beta$-quantile of its first component. 
While the \emph{true} parameters for $\theta_{3,0}^v$ and  $\theta_{3,0}^m$ coincide, this specification does not violate the separated parameter condition in \eqref{eqn:TrueModelParameters} as we do not impose the condition $\theta_{3}^v = \theta_{3}^m$ in the estimation.

\subsection{Simulation Results}

We estimate the correctly specified six-parameter (VaR, MES) regression model based on the joint covariates $\mZ^v_{t} = \mZ^m_{t}$ given above.
%The additional choice $\gamma_{4p} = (1,1.5,0,0.25,0.5)$ results in $\vtheta_{3,0} = \vtheta_{6,0} = 0$ and hence, we estimate a four parameter model excluding the covariate $Z_{2,t}$.
Table~\ref{tab:SimResultsCrossSectionalReg} shows simulation results for the estimated parameters and their standard deviations. 
Results are based on $5000$ Monte Carlo replications of the DGP in \eqref{eqn:DGPCrossSectional} for the probability levels  $\beta \in \{0.9,\ 0.95,\ 0.975\}$, which we also use in our two applications in Sections \ref{sec:ApplMESReg}--\ref{sec:GDP_MES}.
We consider the typical sample sizes $n \in \{500,\ 1000,\ 2000,\ 4000\}$.
The asymptotic variances are estimated as detailed in Appendix~\ref{sec:thm3}.
A formal description of the table columns is given in the table caption.

\begin{table}[tb]
	\centering
	\scriptsize
	\begin{tabular}{rr lrrrr lrrrr lrrrr}
		\toprule
		\multicolumn{2}{l}{\textbf{VaR}} & &  \multicolumn{4}{c}{$\theta_1^v$} & & \multicolumn{4}{c}{$\theta_2^v$} & & \multicolumn{4}{c}{$\theta_3^v$}  \\
		\cmidrule{4-7}  	\cmidrule{9-12} 	 \cmidrule{14-17} 
		$\beta$ & $n$ & & 	
		Bias & $\hat \sigma_\text{emp}$ & $\hat \sigma_\text{asy}$ & CI  & &
		Bias & $\hat \sigma_\text{emp}$ & $\hat \sigma_\text{asy}$ & CI  & &
		Bias & $\hat \sigma_\text{emp}$ & $\hat \sigma_\text{asy}$ & CI   \\
		\midrule
		\multirow{4}{*}{0.9}  & 500 &  & 0.049 & 0.581 & 0.585 & 0.92 &  & $-$0.033 & 0.489 & 0.503 & 0.91 &  & $-$0.003 & 0.163 & 0.159 & 0.93 \\ 
		& 1000 &  & 0.023 & 0.410 & 0.409 & 0.93 &  & $-$0.015 & 0.346 & 0.352 & 0.92 &  & $-$0.000 & 0.114 & 0.112 & 0.94 \\ 
		& 2000 &  & 0.012 & 0.288 & 0.287 & 0.93 &  & $-$0.007 & 0.244 & 0.247 & 0.93 &  & $-$0.000 & 0.081 & 0.079 & 0.94 \\ 
		& 4000 &  & 0.009 & 0.208 & 0.203 & 0.93 &  & $-$0.007 & 0.176 & 0.174 & 0.94 &  & $-$0.001 & 0.057 & 0.056 & 0.94 \\ 
		\addlinespace
		\multirow{4}{*}{0.95} & 500 &  & 0.104 & 0.824 & 0.850 & 0.89 &  & $-$0.073 & 0.687 & 0.730 & 0.87 &  & $-$0.001 & 0.232 & 0.230 & 0.91 \\ 
		& 1000 &  & 0.051 & 0.593 & 0.586 & 0.90 &  & $-$0.031 & 0.500 & 0.504 & 0.88 &  & 0.000 & 0.164 & 0.159 & 0.91 \\ 
		& 2000 &  & 0.034 & 0.412 & 0.410 & 0.92 &  & $-$0.024 & 0.346 & 0.353 & 0.91 &  & $-$0.001 & 0.115 & 0.112 & 0.93 \\ 
		& 4000 &  & 0.015 & 0.294 & 0.290 & 0.93 &  & $-$0.008 & 0.248 & 0.249 & 0.93 &  & 0.000 & 0.079 & 0.079 & 0.94 \\ 
		\addlinespace
		\multirow{4}{*}{0.975}& 500 &  & 0.214 & 1.210 & 1.425 & 0.87 &  & $-$0.130 & 1.007 & 1.179 & 0.82 &  & $-$0.001 & 0.338 & 0.377 & 0.89 \\ 
		& 1000 &  & 0.108 & 0.860 & 0.936 & 0.88 &  & $-$0.064 & 0.720 & 0.797 & 0.85 &  & 0.003 & 0.236 & 0.249 & 0.90 \\ 
		& 2000 &  & 0.061 & 0.604 & 0.626 & 0.90 &  & $-$0.035 & 0.516 & 0.536 & 0.88 &  & $-$0.003 & 0.169 & 0.169 & 0.92 \\ 
		& 4000 &  & 0.025 & 0.427 & 0.437 & 0.91 &  & $-$0.015 & 0.361 & 0.375 & 0.90 &  & 0.001 & 0.117 & 0.118 & 0.93 \\ 
		\midrule
		\midrule
		\\
		\multicolumn{2}{l}{\textbf{MES}} & &  \multicolumn{4}{c}{$\theta_1^m$} & & \multicolumn{4}{c}{$\theta_2^m$} & & \multicolumn{4}{c}{$\theta_3^m$}  \\
		\cmidrule{4-7}  	\cmidrule{9-12} 	 \cmidrule{14-17} 
		$\beta$ & $n$ & & 	
		Bias & $\hat \sigma_\text{emp}$ & $\hat \sigma_\text{asy}$ & CI  & &
		Bias & $\hat \sigma_\text{emp}$ & $\hat \sigma_\text{asy}$ & CI  & &
		Bias & $\hat \sigma_\text{emp}$ & $\hat \sigma_\text{asy}$ & CI   \\
		\midrule
		\multirow{4}{*}{0.9}& 500 &  & 0.030 & 0.822 & 0.763 & 0.89 &  & $-$0.031 & 0.629 & 0.575 & 0.87 &  & 0.003 & 0.207 & 0.201 & 0.93 \\ 
		& 1000 &  & 0.018 & 0.614 & 0.573 & 0.90 &  & $-$0.018 & 0.466 & 0.433 & 0.88 &  & 0.005 & 0.144 & 0.141 & 0.94 \\ 
		& 2000 &  & 0.010 & 0.454 & 0.434 & 0.92 &  & $-$0.009 & 0.343 & 0.327 & 0.91 &  & $-$0.001 & 0.099 & 0.100 & 0.95 \\ 
		& 4000 &  & 0.006 & 0.325 & 0.317 & 0.93 &  & $-$0.006 & 0.243 & 0.239 & 0.93 &  & $-$0.000 & 0.070 & 0.071 & 0.95 \\ 
		\addlinespace
		\multirow{4}{*}{0.95}& 500 &  & 0.045 & 1.211 & 1.100 & 0.86 &  & $-$0.060 & 0.935 & 0.815 & 0.81 &  & $-$0.002 & 0.323 & 0.303 & 0.91 \\ 
		& 1000 &  & 0.039 & 0.906 & 0.849 & 0.88 &  & $-$0.035 & 0.698 & 0.640 & 0.86 &  & 0.001 & 0.220 & 0.218 & 0.94 \\ 
		& 2000 &  & 0.029 & 0.657 & 0.656 & 0.91 &  & $-$0.027 & 0.500 & 0.494 & 0.90 &  & $-$0.002 & 0.158 & 0.155 & 0.94 \\ 
		& 4000 &  & 0.010 & 0.493 & 0.487 & 0.92 &  & $-$0.010 & 0.373 & 0.366 & 0.92 &  & 0.001 & 0.112 & 0.111 & 0.94 \\ 
		\addlinespace
		\multirow{4}{*}{0.975} & 500 &  & 0.073 & 1.891 & 1.762 & 0.84 &  & $-$0.099 & 1.443 & 1.261 & 0.76 &  & $-$0.001 & 0.543 & 0.478 & 0.88 \\ 
		& 1000 &  & 0.062 & 1.334 & 1.256 & 0.87 &  & $-$0.061 & 1.019 & 0.931 & 0.82 &  & $-$0.005 & 0.360 & 0.336 & 0.90 \\ 
		& 2000 &  & 0.016 & 0.967 & 0.930 & 0.89 &  & $-$0.026 & 0.742 & 0.702 & 0.87 &  & $-$0.003 & 0.239 & 0.241 & 0.94 \\ 
		& 4000 &  & 0.028 & 0.741 & 0.734 & 0.91 &  & $-$0.026 & 0.565 & 0.552 & 0.90 &  & $-$0.004 & 0.174 & 0.173 & 0.94 \\ 
		\bottomrule
	\end{tabular}
	\caption{Simulation results for the parameter estimates of a (correctly specified) linear MES regression model based on the DGP in \eqref{eqn:DGPCrossSectional} and $5000$ simulation replications.
		The columns ``Bias'' show the average bias of the parameter estimates,  ``$\hat \sigma_\text{emp}$'' report the empirical standard deviation of the parameter estimates, ``$\hat \sigma_\text{asy}$'' the mean of the estimated standard deviations (computed from the diagonal elements of our estimate of $\mGamma\mM\mGamma^\prime$) and the columns ``CI'' show the coverage rates of $95\%$-confidence intervals.}
	\label{tab:SimResultsCrossSectionalReg}
\end{table}

We find that both the VaR and MES parameters are estimated consistently as the empirical bias and the empirical standard deviation shrink as $n$ grows in all settings.
As expected, the estimates are more accurate for moderate (smaller) probability levels $\beta$, but even the largest choice $\beta = 0.975$ leads to fairly accurate estimates in large samples.
Overall, the MES parameters are estimated with similar accuracy as the corresponding VaR parameters, and so are their asymptotic variances.
These results are confirmed by the coverage rates of the resulting confidence intervals, where the empirical coverage rates approach the nominal rate of $95\%$.
We observe some undercoverage in finite samples that, however, vanishes with decreasing $\beta$ and increasing sample size $n$.
%For the probability levels $\beta \in \{0.95, 0.975\}$ that we use in our empirical application, sample sizes

\subsection{Comparison with Existing Approaches to MES Regressions}
\label{sec:BrunnermeierMESReg}

\cite{BRS20, BDP20}, \citet{Berger2020} and \citet{Karolyi2023} among others use the following approach to MES regressions. 
Define the empirical MES estimate over a rolling window of $S \in \mathbb{N}$ (typically, $S=250$) past observations, i.e.,
\begin{align}
	\label{eq:Y_MES_transform}
	Y_t^\ast = 	\bigg(\sum_{s=t-S}^{t} \mathds{1}_{\{ X_s \ge \widehat{Q}_\beta(X_{(t-S):t})\}} \bigg)^{-1}  \sum_{s=t-S}^{t}  Y_s \mathds{1}_{\{ X_s \ge \widehat{Q}_\beta(X_{(t-S):t}) \}},
\end{align}
where $\widehat{Q}_\beta(X_{(t-S):t})$ denotes the sample quantile of $\{X_{t-S},\dots,X_t\}$ in the rolling window.
\cite{BRS20, BDP20} then relate the transformed response $Y_t^\ast$ to covariates $\mZ_t^m$ in a standard OLS mean regression; see \citet[equations (2) and (3)]{BRS20} and \citet[Section 2.3]{BDP20} for details.
We henceforth call this method the ``ad hoc MES regression''.
A mean regression of $Y_t^\ast$ on $\mZ_t^m$ estimates the conditional expectation $ \mathbb{E} \big[ Y_t^\ast  \mid \mZ_{t}^m \big]$, whose interpretation is unclear and which is in general different from the \textit{de facto} target of an MES regression, that is, $\MES_{t,\b}=\E\big[Y_t \mid X_t\geq\VaR_{t,\beta}, \mZ_{t}^m \big]$.
A further drawback of the ad hoc MES regressions is that---in contrast to our MES regression---the covariates $\mZ_t^m$ cannot contain past values of $X_t$ or $Y_t$ as these would influence both the left-hand and right-hand sides of the associated OLS regression, as $Y_t^\ast$ is a rolling average over past values of $Y_t$ (truncated by using lagged values of $X_t$).

To illustrate in a simpler context, such an ad hoc regression for the $\beta$-quantile would relate the moving window sample quantile $Y_t^{\dagger}  = \widehat{Q}_\beta(X_{(t-S):t})$ to covariates. 
This contrasts with quantile regression of \citet{KB78}, which relates $Y_t$ directly to the regressors by using the quantile-specific check loss function given in the first row of \eqref{eq:loss}.

\begin{figure}
	\centering
	\includegraphics[width=\linewidth]{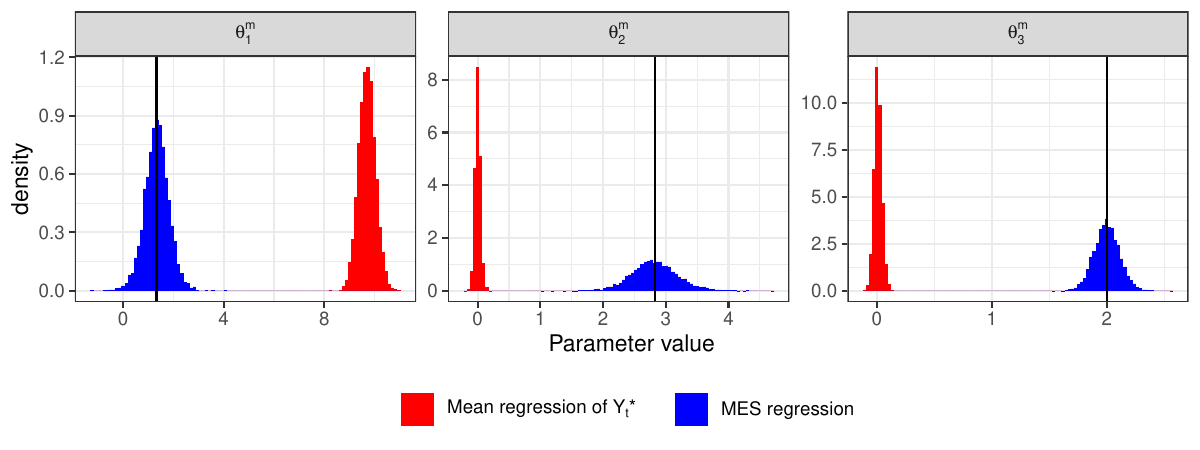}
	\caption{Histograms (over the $M=5000$ simulation replications) of our MES regression parameter estimates in blue and the estimates using the technique of \cite{BRS20,BDP20} in red.
	We fix $\beta = 0.95$ and $n=4000$.
	The true values are marked by the vertical black lines.}
	\label{fig:ParamHist_Brunnermeier}
\end{figure}

Figure \ref{fig:ParamHist_Brunnermeier} illustrates how severe the discrepancy between our and the ad hoc version of the MES regression is for the DGP in \eqref{eqn:DGPCrossSectional} for $\beta = 0.95$.
We do so by estimating a mean regression of the target $Y_t^\ast$ on the same covariates $(1,Z_{1,t}, Z_{2,t})'$ and by computing our MES estimate $\widehat{\vtheta}_n^m$.
While the parameter estimates of our MES regression are normally distributed around the true parameter values (given by the black lines), the estimates based on the ad hoc estimation method are far from the true values.
Most strikingly, the ad hoc estimates of the two slope parameters associated with the covariates $Z_{1,t}$ and $Z_{2,t}$ vary around zero because the contemporaneous effect of the covariates on the \textit{rolling window} quantity $Y_t^\ast$ is negligible. 
This illustrates that employing the ad hoc MES regression procedure does not capture the effect of covariates on the MES functional, but rather for some functional that is inherently hard to interpret; see \eqref{eq:Y_MES_transform}.

\section{Empirical Applications}
\label{sec:EmpiricalApplication}

We illustrate the versatility of our MES regressions in three applications. 
These concern explanatory variables of systemic risk in the banking sector in Section~\ref{sec:ApplMESReg} and dissecting GDP growth vulnerabilities among the three biggest European economies in Section~\ref{sec:GDP_MES}.
A further application concerning (equal risk contribution) portfolios is deferred to Appendix~\ref{sec:add appl}.

%Replication material and data sources are available under \href{https://github.com/TimoDimi/replication_MES}{https://github.com/TimoDimi/replication\_MES}. 

\subsection{Systemic Risk Regressions}
\label{sec:ApplMESReg}

With the introduction of RiskMetrics \citep{RM96}, financial risk management---based mainly on the VaR---was beginning to be firmly established in the 1990s. 
Subsequently, the importance of adequately managing individual financial risks was also reflected in official regulations by the \citet{BCBS96}. 
As these regulations did not prevent the financial crisis of 2007--08, more attention (of regulators as well as industry) has been paid to the \textit{systemic} nature of financial risks. This means that more scrutiny was applied in studying the interconnectedness of individual institutions (in the context of the financial system as a whole) or the interconnectedness of individual trading desks (in the context of managing bank-wide risks). 
As a consequence of this development, a huge literature on systemic risk and its determinants has emerged \citep[e.g.,][]{ABT12,AB16,Aea17,BE17,BRS20,BDP20}.

For instance, \citet{BRS20} investigate the contemporaneous relation of asset price bubbles and systemic risk by, among others, considering the (conditional) MES in their Section 4.2.
% However, classical regression procedures---such as those applied by the aforementioned authors---do neither yield consistent parameter estimates nor valid inference as illustrated in our simulations in Section \ref{sec:BrunnermeierMESReg}.
However, our simulations in Section \ref{sec:BrunnermeierMESReg} show that the ad hoc estimation method for MES regressions---as applied by the aforementioned authors---does not necessarily provide consistent parameter estimates nor allows for valid inference.
Therefore, we now consider a simplified form of their analysis for the conditional MES by using our regressions under adverse conditions.
Here, we mainly aim at identifying which covariates drive the conditional MES, and as in  \citet{BRS20}, we particularly focus on stock market boom and bust indicators.

Specifically, we focus on the three US banks that are classified as systemically most risky  according to the \citet{FSB22} such that $Y_t$ equals the daily log-losses  (i.e., the negative log-returns) of either the Bank of America Corporation (BAC), Citigroup (C) or JPMorgan Chase (JPM).
We provide corresponding results for the other five US banks listed as a global systemically important bank (G-SIB) by the \citet{FSB22} in Appendix~\ref{sec:AddEmpResults}.
Throughout, we use the daily log-losses of the S\&P~500 Financials for $X_t$ and set $\beta = 0.95$.
%Hence, the MES measures the risk in the financial system \textit{conditional} on the respective large bank being in distress.
Hence, the MES measures the mean loss of the bank \textit{conditional} on the financial system being in distress.

We estimate the joint linear model  $ \big( v_t(\vtheta^v), m_t(\vtheta^m) \big)' = \big( \mZ_{t}^{v\,\prime} \vtheta^v, \mZ_{t}^{m\,\prime} \vtheta^m\big)'$, where $ \vtheta^v,  \vtheta^m \in \mathbb{R}^6$. For reasons of data availability and concerns of non-stationarity of some regressors, we use a restricted set of variables in $\mZ_{t}^{v} = \mZ_{t}^{m} = \mZ_t = (1, Z_{t,1},\ldots,Z_{t,5})^\prime$,
% Specifically, the covariate vectors $\mZ_{t}^v = \mZ_{t}^m = (1, Z_{t,1},\ldots,Z_{t,5})^\prime$ 
which contain an intercept, the change in spread between Moody's Baa-rated bonds and the 10-year Treasury bill rate (\textit{Change Spread}), the spread between the 3-month LIBOR and 3-month Treasury bills (\textit{TED Spread}), the VIX index as a forward-looking measure of market volatility, and separate indicator variables for stock market booms (\textit{SM Boom}) and busts (\textit{SM Bust}) from \citet{BRS20}.
To study the contemporaneous relation between systemic risk and bubbles, we follow \citet[Eq.~(3)]{BRS20} by lagging all explanatory variables by one time period, except for the boom and bust indicators for which we use contemporaneous values.

To allow for an immediate comparison of the estimation methods, we also estimate the ``ad hoc MES regression'' described in Section~\ref{sec:BrunnermeierMESReg}, which is a classical OLS  regression of the transformed target variable defined in \eqref{eq:Y_MES_transform}.
We estimate the models based on daily data and use the largest common available sample ranging from May 5, 1993 until December 31, 2015, yielding a total of $n=5,555$ trading days.

\begin{table}[tb]
	\centering
	% \scriptsize
	\footnotesize
	\resizebox{\columnwidth}{!}{
	\begin{tabular}{c c l c rrr c rrr c rrr}
		\toprule
		&&&& \multicolumn{3}{c}{VaR Model} &&  \multicolumn{3}{c}{MES Model} &&  \multicolumn{3}{c}{Ad Hoc MES Regression} \\
		\cmidrule{5-7}   	\cmidrule{9-11}   	\cmidrule{13-15}  
		Bank & & Covariate && Est. & SE & $p$-val  &&  Est. & SE & $p$-val &&  Est. & SE & $p$-val   \\
		\midrule
		\addlinespace		
		\multirow{6}{*}{BAC}  &  & Intercept &  & $-$1.734 & 0.185 & 0.000 &  & $-$5.535 & 1.438 & 0.000 &  & $-$1.325 & 0.156 & 0.000 \\ 
		 &  & Change Spread &  & 0.591 & 0.107 & 0.000 &  & 2.184 & 0.752 & 0.004 &  & 1.965 & 0.090 & 0.000 \\ 
		 &  & TED Spread &  & 1.751 & 0.182 & 0.000 &  & 2.985 & 1.027 & 0.004 &  & $-$1.865 & 0.124 & 0.000 \\ 
		 &  & VIX &  & 0.092 & 0.015 & 0.000 &  & 0.147 & 0.058 & 0.011 &  & 0.068 & 0.009 & 0.000 \\ 
		 &  & ST Boom &  & $-$0.194 & 0.125 & 0.120 &  & 0.140 & 0.611 & 0.819 &  & 0.566 & 0.131 & 0.000 \\ 
		 &  & ST Bust &  & $-$0.571 & 0.169 & 0.001 &  & $-$1.483 & 0.604 & 0.014 &  & $-$0.280 & 0.181 & 0.121 \\ 
		\midrule
		\addlinespace		 
		\multirow{6}{*}{C}  &  & Intercept &  & $-$1.734 & 0.185 & 0.000 &  & $-$4.654 & 1.880 & 0.013 &  & $-$1.386 & 0.120 & 0.000 \\ 
		 &  & Change Spread &  & 0.591 & 0.107 & 0.000 &  & 1.141 & 0.920 & 0.215 &  & 1.791 & 0.070 & 0.000 \\ 
		 &  & TED Spread &  & 1.751 & 0.182 & 0.000 &  & 2.110 & 1.106 & 0.057 &  & $-$1.642 & 0.095 & 0.000 \\ 
		 &  & VIX &  & 0.092 & 0.015 & 0.000 &  & 0.263 & 0.063 & 0.000 &  & 0.092 & 0.007 & 0.000 \\ 
		 &  & ST Boom &  & $-$0.194 & 0.125 & 0.120 &  & $-$0.321 & 0.635 & 0.613 &  & 0.687 & 0.101 & 0.000 \\ 
		 &  & ST Bust &  & $-$0.571 & 0.169 & 0.001 &  & $-$1.476 & 0.608 & 0.015 &  & $-$0.328 & 0.139 & 0.018 \\ 
		\midrule
		\addlinespace		
		\multirow{6}{*}{JPM} &  & Intercept &  & $-$1.734 & 0.185 & 0.000 &  & $-$2.601 & 0.873 & 0.003 &  & $-$0.234 & 0.084 & 0.005 \\ 
		 &  & Change Spread &  & 0.591 & 0.107 & 0.000 &  & 0.868 & 0.431 & 0.044 &  & 1.208 & 0.049 & 0.000 \\ 
		 &  & TED Spread &  & 1.751 & 0.182 & 0.000 &  & 2.455 & 0.927 & 0.008 &  & $-$1.668 & 0.067 & 0.000 \\ 
		 &  & VIX &  & 0.092 & 0.015 & 0.000 &  & 0.151 & 0.049 & 0.002 &  & 0.090 & 0.005 & 0.000 \\ 
		 &  & ST Boom &  & $-$0.194 & 0.125 & 0.120 &  & $-$0.151 & 0.444 & 0.734 &  & 0.422 & 0.070 & 0.000 \\ 
		 &  & ST Bust &  & $-$0.571 & 0.169 & 0.001 &  & $-$0.025 & 0.713 & 0.972 &  & 0.821 & 0.097 & 0.000 \\ 
		\bottomrule 
	\end{tabular}
	}
	\caption{VaR and MES regression results for $\beta = 0.95$ for $X_t$ being the log-losses of the S\&P~500 Financials and $Y_t$ being the log-losses of Bank of America Corporation (BAC), Citigroup (C) and JPMorgan Chase (JPM) in the respective panels. The vertical panel entitled ``Ad Hoc MES Regression'' reports the results of the method described around  \eqref{eq:Y_MES_transform}. The respective parameter estimates are given in the columns ``Est.'', their standard errors in the columns ``SE'', and associated $t$-test $p$-values in the columns ``$p$-val''.}
	\label{tab:PredictiveRegression}
\end{table}

Table~\ref{tab:PredictiveRegression} displays the parameter estimates and the appertaining standard errors and $t$-test $p$-values for our MES regression using the inference methods developed in Section \ref{Asymptotic Properties}. Also shown are the corresponding results of the ad hoc MES regression using standard inference methods for the OLS estimator.
As in our simulations, we see that the parameter estimates and the standard errors of our MES regression and the ad hoc method based on a mean regression of the ``empirically filtered'' MES differ substantially.
While the latter suggests a clearly significant influence of the bubble indicators in five out of six cases, our MES regression shows a different picture indicating that only the bust indicators have a (negative) significant effect for BAC and C, but not for JPM.
Hence, given the functional of interest is the conditional MES, the ad hoc estimation method (likely) falsely classifies the boom indicators as being significant.
Closely related, the substantially smaller standard errors of the ad hoc MES regression (caused by the OLS regression that omits the necessary ``truncation'' in the lower row of \eqref{eq:loss}) are especially concerning in providing overconfident and possibly false conclusions on the drivers of systemic risk.

\begin{figure}[tb]
	\centering
	\includegraphics[width=\linewidth]{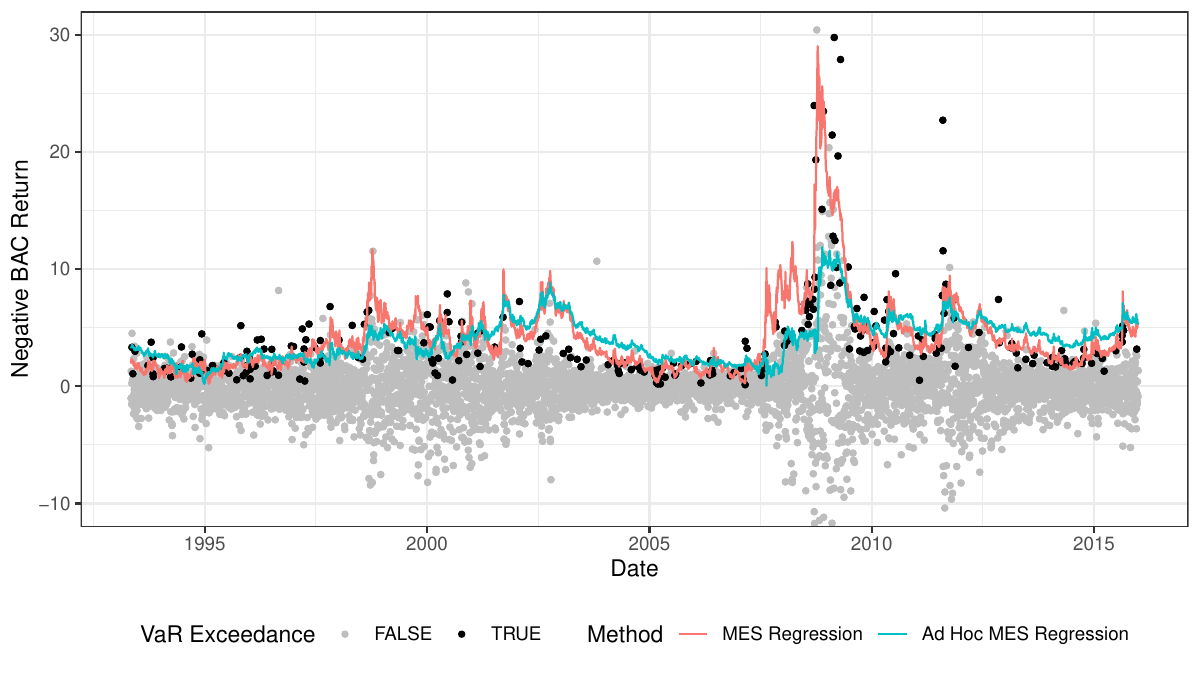}
	\caption{Estimated sample paths of our MES regression and the ``ad hoc MES regression'' employed in \citet{BRS20} for $Y_t$ being negative returns of the Bank of America (BAC), $X_t$  negative returns of the S\&P\,500 Financials and $\beta = 0.95$.
	Returns on days with an exceedance of the estimated conditional VaR are displayed in black while all other returns are shown in gray.}
	\label{fig:MES_BRS_regressions}
\end{figure}

While MES regressions estimate the conditional MES as quantified in Theorem~\ref{thm:an}, the ad hoc MES regression estimates the unconventional target functional of a conditional mean (through an OLS regression) of a rolling window MES estimate.
A comparison of the respective model predictions in Figure~\ref{fig:MES_BRS_regressions} shows that these can differ quite drastically, especially in turbulent times as in the year 2008: 
In contrast to our method, the ad hoc MES regression fails to quickly adapt to the rising levels of (systemic) risk due to its inherent dependence on the distant past through the rolling-window MES estimate in \eqref{eq:Y_MES_transform}.
%Hence, we suggest to clearly motivate the functional of interest and choose the suitable regression method, which---when effects on the conditional MES are of interest---is provided by our two-step estimator.
Hence, we suggest to clearly motivate the functional of interest and choose the appropriate methodological framework, which---when effects on the conditional MES are of interest---is provided by our regressions under adverse conditions.

The appropriateness of our (linear) regression under adverse conditions for the data is reinforced by in-sample model diagnostics.
In classical correctly specified \emph{mean} regressions, the expectation of the classical model residuals is zero conditional on known information, such as time or the model predictions.
This is often analyzed by non-parametrically regressing the residuals on these quantities and plotting the results.
In standard mean regressions, the classical model residuals $V(m,y)=m-y$ can be obtained as the derivative with respect to the model predictions $m$ of the (scaled) squared error loss $S(m,y)=\frac{1}{2}(m-y)^2$.
% For mean regressions, the derivative of the (scaled) squared error loss $S(m,y)=\frac{1}{2}(m-y)^2$ (with respect to the forecast $m$) leads to the identification function $V(m,y)=m-y$.  Evaluated at the prediction and the outcome, this leads to the ``classical'' residual $V(\hat{y}_t,y_t)=\hat{y}_t - y_t$. Hence, in mean regressions the ``classical'' residual is identical to the ``generalized'' residual.}

For our more complicated target functional (VaR, MES), we follow \citet{Pohle_GenResid} and replace the model residuals by so-called ``generalized model residuals'', whose conditional expectation---as above---equals zero for correctly specified models.
The generalized model residuals correspond to the identification functions (known from forecast evaluation) evaluated at the model predictions and the outcomes.
An identification function for the pair (VaR, MES) is given by $V\big( (v,m)', (x,y)' \big) = \big(\1_{\{x\leq v\}}-\b, \, \1_{\{x>v\}}(m-y) \big)'$, which arises (almost everywhere) as the derivative of the lexicographical loss in \eqref{eq:loss} with respect to $v$ and $m$ \citep{FH24}.
This construction parallels the derivative of the squared error mentioned above. 
% \footnote{For mean regressions, the derivative of the (scaled) squared error loss $S(m,y)=\frac{1}{2}(m-y)^2$ (with respect to the forecast $m$) leads to the identification function $V(m,y)=m-y$.  Evaluated at the prediction and the outcome, this leads to the ``classical'' residual $V(\hat{y}_t,y_t)=\hat{y}_t - y_t$. Hence, in mean regressions the ``classical'' residual is identical to the ``generalized'' residual.}
In the following, we non-parametrically regress the generalized model residuals on time and the model predictions.
Plots of the regression curve that deviate from the zero line are indicative of model misspecification (because then some combination of the covariates could predict the generalized residual---in violation of correct specification).

\begin{figure}[tb]
	\centering
	\includegraphics[width=\linewidth]{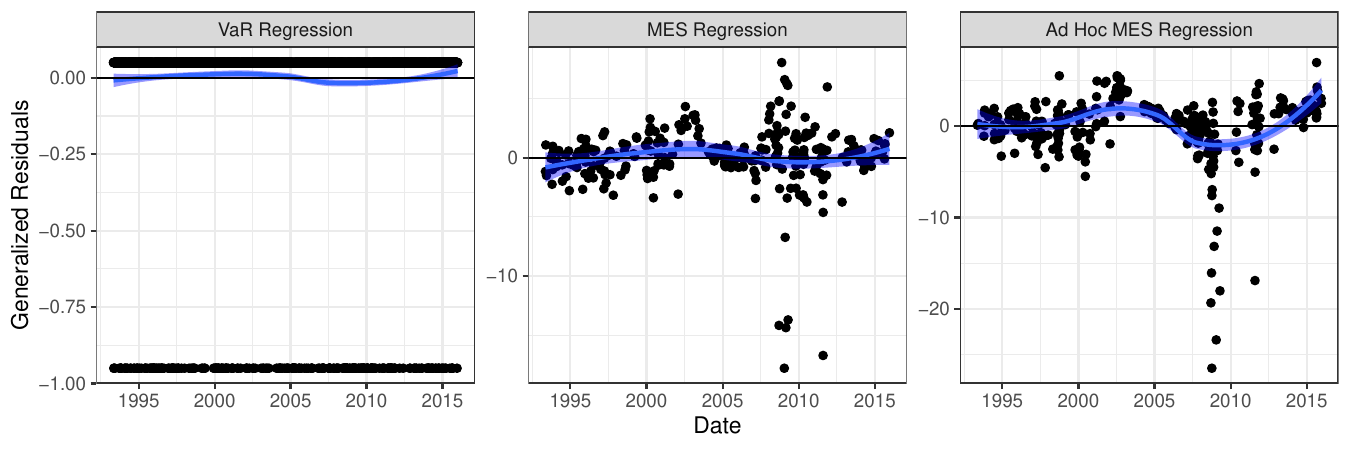}
	\includegraphics[width=\linewidth]{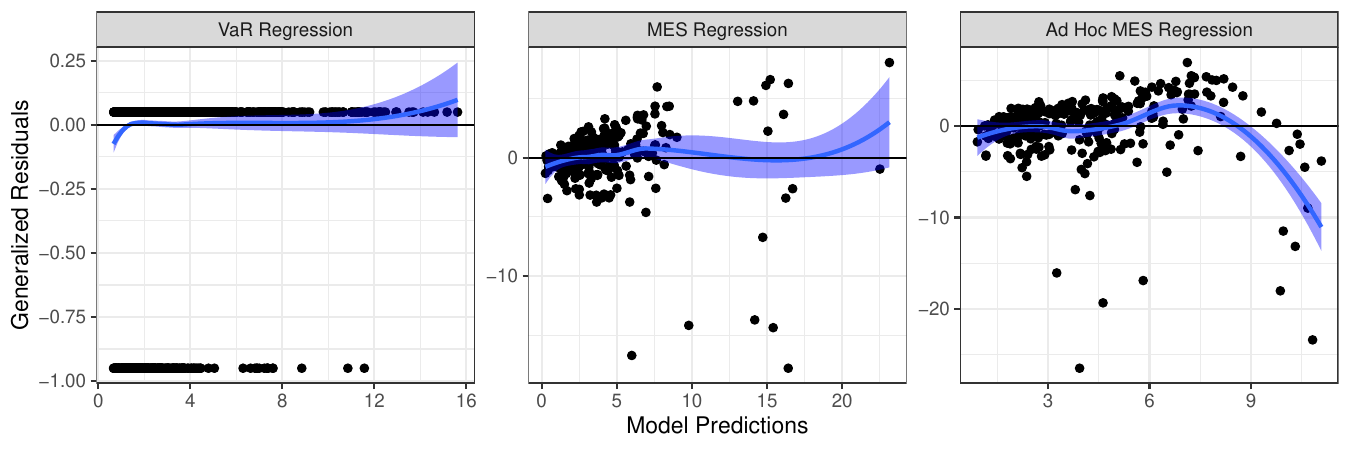}
	\caption{In-sample model diagnostics of the VaR and MES predictions of our regression and the ``ad hoc MES regression'' of \citet{BRS20, BDP20} for $Y_t$ being negative returns of the Bank of America (BAC), $X_t$  negative returns of the S\&P\,500 Financials and $\beta = 0.95$. 
	The generalized residuals are plotted and non-parametrically regressed (in blue) against time (upper row) and the model predictions (lower row).}
	\label{fig:ModelDiagnostics_MES_BRS_regressions}
\end{figure}

Figure~\ref{fig:ModelDiagnostics_MES_BRS_regressions} shows generalized residuals of the VaR and MES parts of our regression, plotted against time and the model predictions, respectively.
Also shown are the corresponding generalized MES residuals from the ad hoc MES regression. 
The blue line and its pointwise 95\%-confidence band stem from a polynomial (mean) regression with automated parameter choices in the \texttt{loess} function of the statistical software \texttt{R}.
Further details on the implementation are provided in Appendix~\ref{sec:AddEmpResults}.
The zero line is almost always contained in the blue band for the predictions from our MES regression, implying no evidence against model misspecification.
In contrast, the ad hoc MES regression clearly violates this condition. 
Most troubling from a risk management perspective, the violations are most apparent during the financial crisis and for large model predictions, i.e., when accuracy of the method is most important.

\subsection{Dissecting GDP Growth Vulnerabilities}
\label{sec:GDP_MES}

In recent years, many researchers have quantified macroeconomic risks via the ES \citep{ABG19,Pea20,Aea21,DDP24}.
For instance, in their influential paper, \citet[Section II.B]{ABG19} use the ES to measure US GDP growth vulnerability.
Using an associated regression method, they then relate future ES to current macroeconomic and financial conditions by using US GDP growth and the US national financial conditions index (NFCI) as covariates.

In this section, we first introduce a slight (methodological) variation of \citeauthor{ABG19}'s \citeyearpar{ABG19} procedure by estimating a joint VaR and ES regression.
Then, we detail how our MES regressions can be used to gain additional insights on the decomposition of growth vulnerabilities in Europe.
Here, our interest lies in both, understanding the driving forces of the conditional MES (see Table \ref{tab:GDP_MES}) and in providing an accurate model fit (see Figure~\ref{fig:MES_GDP}).

Specifically, we denote by $X_t$ the negative  quarterly year-on-year GDP growth of the economic area consisting of the three biggest European economies, Germany, France and the United Kingdom (UK).
Formally, we estimate the parameters  $\vtheta^v = (\theta_1^v, \theta_2^v, \theta_3^v)^\prime$ and  $\vtheta^e = (\theta_1^e, \theta_2^e, \theta_3^e)^\prime$ of the linear (VaR, ES) regression
\begin{align}
		\VaR_{\beta}(X_t \mid \mathcal{F}_{t}) &= \theta_1^v + \theta_2^v \FCI_{t-1} + \theta_3^v X_{t-1}, \label{eq:VaR GaR}\\
		\ES_{\beta}(X_t \mid \mathcal{F}_{t}) = \E_t [X_t\mid X_t\geq\VaR_{t,\beta}] &= \theta_1^e + \theta_2^e \FCI_{t-1} + \theta_3^e X_{t-1},\label{eq:ES GaR}
\end{align}
where $\FCI_{t-1}$ denotes the lagged ``equally weighted'' financial conditions index for advanced economies of \citet{Arrigoni2022} and $\mathcal{F}_t$ only contains time-$(t-1)$ information in the form of past $\FCI_{t-1}$ and $X_{t-1}$.
\citet[Section II.B]{ABG19} estimate the ES regression in \eqref{eq:ES GaR} by averaging over quantile regressions for multiple levels.
In contrast, we employ the two-step M-estimator for the VaR and ES regression parameters that arises by setting $Y_t = X_t$ in \eqref{eq:MES_estimator}; cf.~Remark \ref{rem:ESreg}.
In the numerical optimization, we additionally employ the \emph{non-crossing constraint} that $\ES_{\b}(X_t\mid \mathcal{F}_{t}) \ge \VaR_{\beta}(X_t\mid \mathcal{F}_{t})$ for all $t$.

\begin{table}[tb]
	\centering
	\begin{tabular}{ll c rr c rr c rr} 
		\toprule
		&&& \multicolumn{2}{c}{$\theta_{1}$} && \multicolumn{2}{c}{$\theta_{2}$} && \multicolumn{2}{c}{$\theta_{3}$} \\
		\cmidrule{4-5}    \cmidrule{7-8}    \cmidrule{10-11}   
		Country & Risk Measure &&   Est. & SE &&  Est. & SE  &&  Est. & SE  \\
		\midrule
		\multirow{2}{*}{Joint Region} & VaR && 0.602 & 0.212  &&   0.673 & 	0.097  && 0.884 & 0.100 \\
		& ES	&& 0.794 &  && 0.663 &  && 0.856 &  \\
		%& ES (unconstr.)	&& 0.615 & 0.148 && 0.730 & 0.120 && 0.755 & 0.106 \\
		\midrule
		Germany 	& MES	 &&  2.160 & 0.261 && 0.352 & 0.160 && 1.350     &   0.077 \\
		France 			& MES	&&  $-0.557$ & 0.222 && 0.818  & 0.226  && 0.315    &  0.211 \\
		United Kingdom 	 & MES	&&   $-0.723$ & 0.335 && 1.370  & 0.312 && 0.122     &  0.255 \\
		\bottomrule
	\end{tabular}	
	\caption{VaR, ES and MES regression parameter estimates at the level $\beta=0.9$ for the models in \eqref{eq:VaR GaR}--\eqref{eq:ES GaR} and \eqref{eq:GDPmodel} for $n=99$ observations of negative quarterly GDP growth between  1995 Q1 and 2019 Q4 for Germany, France, the United Kingdom (as $Y_{t,d}$) and their joint economic region (as $X_t$). The columns ``Est.'' report the parameter estimates and ``SE'' the estimated standard errors.
	The generic parameter names $\theta_{j}$, $j=1,2,3$ in the column headings refer to $\theta^v_{j}$, $\theta^e_{j}$ and $\theta^m_{j,d}$ from \eqref{eq:VaR GaR}, \eqref{eq:ES GaR} and \eqref{eq:GDPmodel}, respectively.}
	\label{tab:GDP_MES}
\end{table}

The upper panel of Table~\ref{tab:GDP_MES} (labeled ``Joint Region'') shows the parameter estimates together with their standard errors.
Notice that we do not report standard errors for the ES parameters in this panel, because no valid inference methods exist for a constrained two-step estimation of the conditional ES (as opposed to the unconstrained case discussed in Remark~\ref{rem:ESreg}).
%Notice that no valid inference methods exist for a \emph{constrained} two-step estimator of the conditional ES (opposed to the unconstrained case discussed in Remark~\ref{rem:ESreg}) such that we do not report standard errors of the ES parameters in Table~\ref{tab:GDP_MES}.
The estimates are based on a sample of $n=99$ quarterly observations from 1995 Q1 until 2019 Q4, which (almost) corresponds to the available sample of the financial conditions index of \citet{Arrigoni2022}.
% \footnote{As the financial conditions index of \citet{Arrigoni2022} is only available until May 2020, we deliberately exclude the Brexit in January 2020 and the COVID period. Using only the beginning of the latter period with its extreme values of GDP growth might substantially bias our (M)ES regressions. We have $n=99$ despite having 25 full years of data as we lose one observation by using lagged values $X_{t-1}$ in \eqref{eq:VaR GaR} and \eqref{eq:ES GaR}.}
We find that the slope coefficients for both the VaR and ES model are positive and (for the VaR) statistically significant  at any commonly used significance level.
This confirms the result of \citet{ABG19} that future growth risk loads roughly equally on both financial and macroeconomic conditions (as captured by $\FCI_{t-1}$ and $X_{t-1}$ in \eqref{eq:VaR GaR}--\eqref{eq:ES GaR}).

We now demonstrate how our MES regressions can be used to provide a more refined picture of growth vulnerabilities. 
Specifically, as we consider an entire economic region, the question naturally arises how the driving forces of growth vulnerabilities are distributed among the individual countries.

Formally, the common negative GDP growth $X_t = \sum_{d=1}^D w_{t,d} Y_{t,d}$ can be dissected into negative GDP growth, $Y_{t,d}$, of the $D=3$ individual countries. 
The time-varying weights $\vw_t = (w_{t,1}, \dots, w_{t,D})^\prime$ are given by the countries' relative share of GDP in the entire economic region, such that $w_{t,d}\in(0,1)$ and $\sum_{d=1}^{D}w_{t,d}=1$.
% determined by the relative GDP value of the individual countries among the total GDP of the entire region.
We partition the ES risk of $X_t$ from \eqref{eq:ES GaR} into its systemic MES components by writing
\begin{equation}\label{eq:ES decomp}
\E_t[X_t\mid X_t\geq\VaR_{t,\beta}] = \sum_{d=1}^{D} w_{t,d} \, \E_t[Y_{t,d}\mid X_t\geq\VaR_{t,\beta}].
\end{equation}
Then, we model the individual components of the right-hand side with our MES regression
\begin{align}
	\label{eq:GDPmodel}
	\E_t [Y_{t,d}\mid X_t\geq\VaR_{t,\beta}] &= \theta_{1,d}^m + \theta_{2,d}^m  \FCI_{t-1} + \theta_{3,d}^m X_{t-1},
\end{align}
jointly with the VaR regression in \eqref{eq:VaR GaR}.
%We use corresponding quarterly year-over-year GDP growth rates for the three individual countries.
The conditional MES $\E_t[Y_{t,d}\mid X_t\geq\VaR_{t,\beta}]$ quantifies the expected (negative) GDP growth of country $d$ given that the economic region $X_t$ is in distress and, hence, measures the systemic vulnerability of economy $d$. 
The MES regression in \eqref{eq:GDPmodel} therefore assesses how this systemic vulnerability is affected by past economic and financial conditions. 
By virtue of \eqref{eq:ES decomp}, the three MES regressions in \eqref{eq:GDPmodel} provide a complete dissection of the aggregated growth vulnerabilities captured by the \citet{ABG19} regression in \eqref{eq:ES GaR}.

\begin{figure}[tb]
	\centering
	\includegraphics[width=\linewidth]{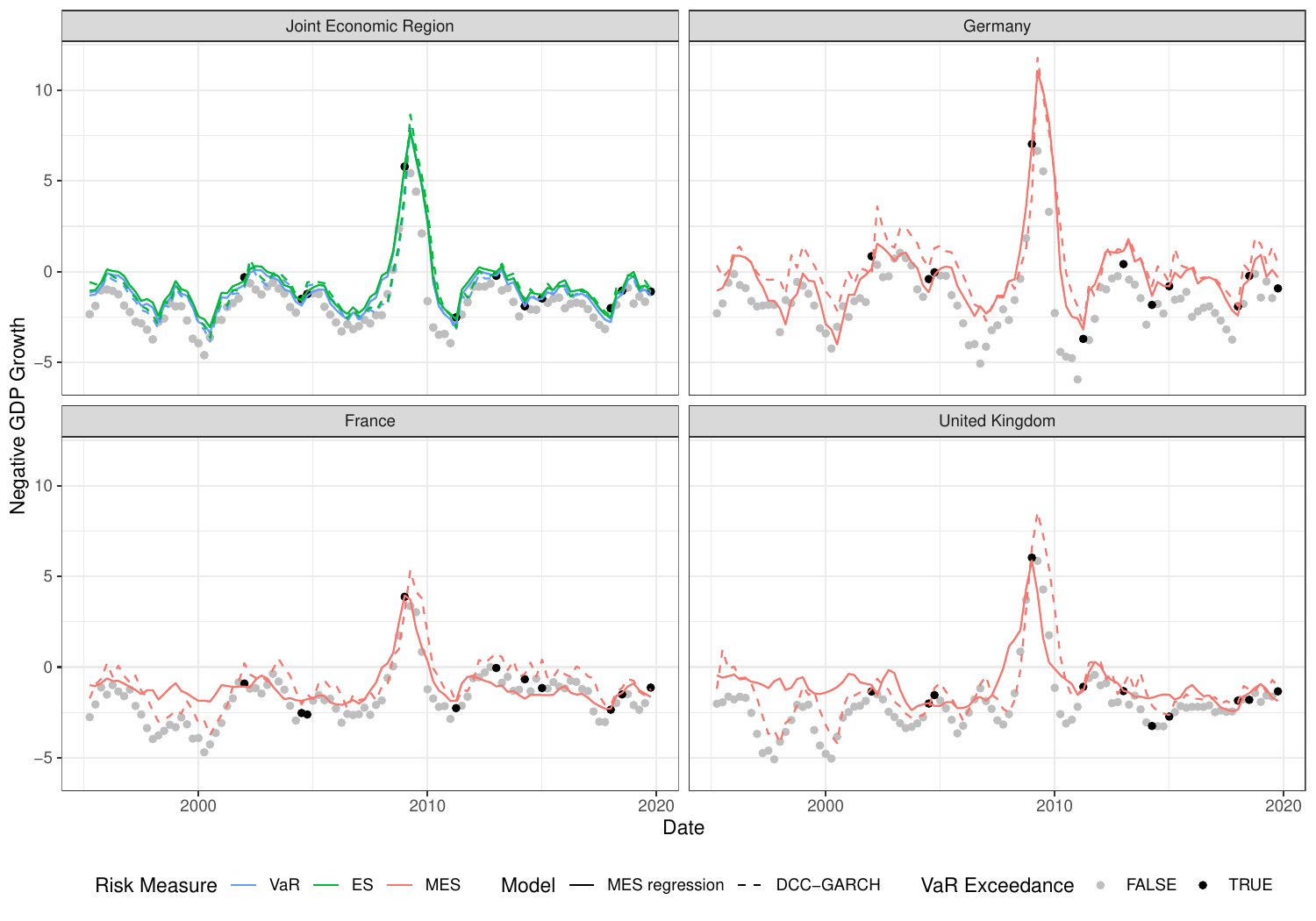}
	\caption{Estimated VaR and ES sample paths at level $\beta=0.9$ for our regression (solid lines) and a Gaussian DCC--GARCH model (dashed lines)  for the negative GDP growth of the joint economic region together with the estimated MES for Germany, France and the United Kingdom in the respective panels. 	
	Returns on quarters with a exceedance of the estimated conditional VaR (of our regression) are displayed as black points while all other returns are shown in gray color.}
	\label{fig:MES_GDP}
\end{figure}

The lower panel of Table~\ref{tab:GDP_MES} displays the MES parameter estimates together with their standard errors.
Four out of six MES slope parameters, which equal the marginal effect under adverse conditions (see Remark~\ref{rem:MarginalEffect}), are significant at the $5\%$-level and we find the striking result that the explanatory power of the different variables varies strongly between the countries.
This is in stark contrast to the classic (VaR, ES) regression of \citet{ABG19}, where financial and economic conditions are equally important determinants of growth risk.
For instance, the MES of the UK mainly loads on the financial conditions indicator, which may be explained by a high dependence on the financial sector of the UK with London as one of the most important financial hubs in the world.
In contrast, the systemic vulnerability of the German economy is primarily driven by past economic conditions, which may be due to its strong export-oriented manufacturing sector.

For weights $\vw_t $ being constant over time, the true ES regression parameters are given by the weighted average of the true MES regression parameters; see \eqref{eq:ES decomp}.
While this relationship is blurred in finite samples and due to the time-varying weights in this application (the weights for Germany vary between 0.4 and 0.45 and the ones for France and the UK between 0.26 and 0.31 in our sample),
% \footnote{The weights for Germany vary between 0.4 and 0.45 and the ones for France and the UK between 0.26 and 0.31 in our sample.} 
it can still be recognized in the parameter estimates in Table \ref{tab:GDP_MES}. 
E.g., for the slope coefficient pertaining to $\FCI$, we have $0.663 \approx 0.7757 = 0.426 \cdot 0.352 + 0.291 \cdot 0.818 + 0.283 \cdot 1.370$, where 0.426, 0.291, and 0.283 are the average weights for Germany, France, and the UK.

Figure~\ref{fig:MES_GDP} plots the implied VaR, ES and MES estimates for the regressions in \eqref{eq:VaR GaR}--\eqref{eq:ES GaR} and \eqref{eq:GDPmodel} for the joint economic region together with the three individual economies with solid lines.
For comparison, the dashed lines show the implied VaR, ES and MES estimates of a four-dimensional DCC--GARCH model with Gaussian innovations \citep{Eng02}.
The display shows very plausible model fits of both methods, which is especially remarkable given that we estimate risk measure models in the tails with only $n=99$ observations.
%For example, we observe eleven VaR exceedances in the sample, which is very close to expected number of 9.9 exceedances.

\begin{figure}[tb]
	\centering
	\includegraphics[width=\linewidth]{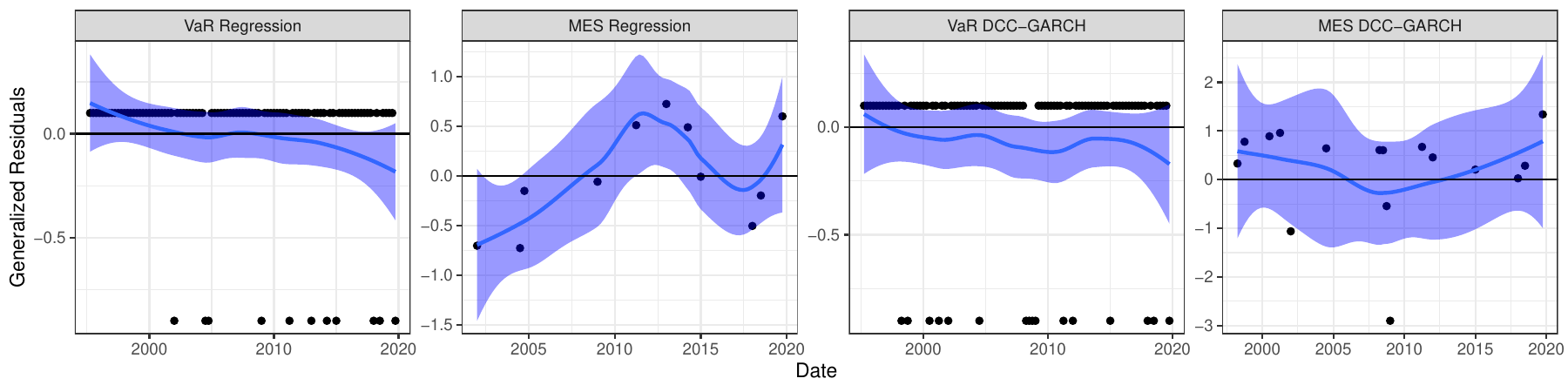}
	\includegraphics[width=\linewidth]{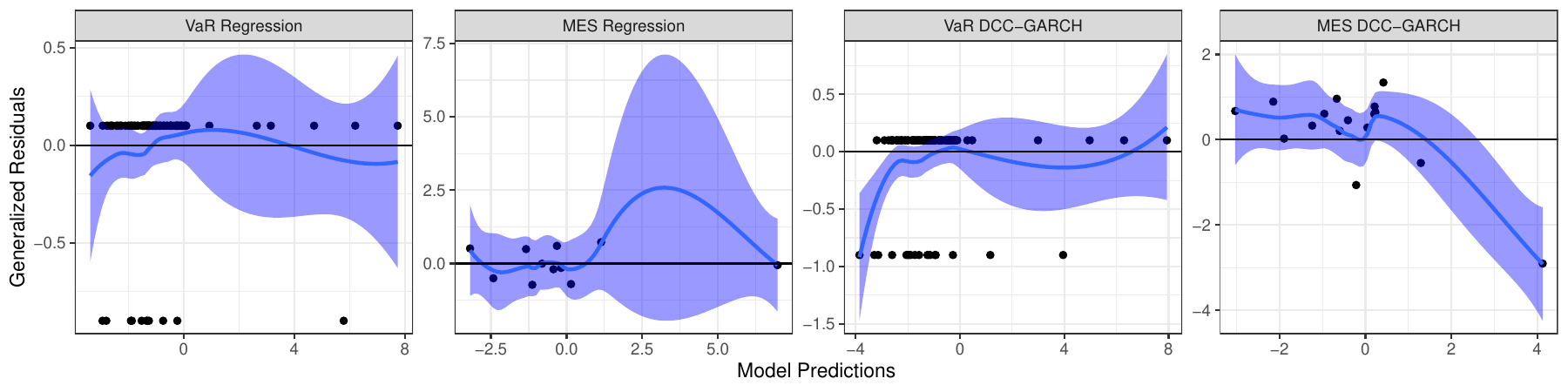}
	\caption{In-sample model diagnostics of the VaR and MES predictions of our regression and a Gaussian DCC--GARCH model for $\beta = 0.9$. Here, $Y_t$ is the negative GDP growth of Germany and $X_t$ the joint economic region.
	The generalized residuals are plotted and non-parametrically regressed (in blue) against time (upper row) and the model predictions (lower row).}
	\label{fig:ModelDiagnostics_GDP_MES_DCC}
\end{figure}

Figure~\ref{fig:ModelDiagnostics_GDP_MES_DCC} further illustrates the good in-sample model fit by plotting (VaR and MES) generalized model residuals (see Section~\ref{sec:ApplMESReg}) against time and the model predictions themselves.
As the zero line is barely excluded from the confidence bands, we obtain (almost) no evidence against model misspecification for all models.
However, due to the limited sample size, this was to be expected.

To summarize, while in the standard GaR literature growth vulnerabilities are the main object of study, our MES regressions allow to decompose these into \textit{systemic} vulnerabilities of the different countries, therefore allowing a more fine-grained picture to emerge.
We mention that our MES regressions may also be fruitfully used for single countries, such as the US. 
In such cases, one may decompose total growth into different economic sectors (e.g., manufacturing, service and agricultural) to shed more light on growth vulnerabilities in these different parts of the economy.
We leave such explorations for future study.

\section{Conclusion}
\label{sec:Conclusion}

Regressions typically model some functional of the \textit{scalar} outcome $Y_t$ (say, the mean in mean regressions or quantiles in quantile regressions) as a function of covariates. 
In contrast, in this paper we model a functional of the \textit{bivariate} outcome $(X_t,Y_t)^\prime$ (viz.~the MES) as a function of covariates.
This leads to what we term \emph{regressions under adverse conditions}, where the mean of $Y_t$, given the adverse condition that a distress variable $X_t$ exceeds its conditional quantile, is related to covariates.
We develop a two-step estimator and prove its asymptotic normality under general mixing conditions allowing for cross-sectional and time series applications, where we overcome the technical difficulty of having a discontinuous second-step objective function.
We also propose feasible inference that works well as shown in our simulations.
Our estimator particularly improves upon a recently proposed \textit{ad hoc} estimation method of \citet{BRS20, BDP20}, which we illustrate in simulations and an empirical application.

% While our regression target---the conditional expectation of $Y_t$ given that $X_t$ exceeds its conditional quantile---is originally known as the MES in the (financial) systemic risk literature, we illustrate its value in three diverse applications focusing on the relation between systemic risk and asset price bubbles \citep{BRS20}, portfolio risk allocation \citep{Bea18}, and dissecting GDP growth risks among countries \citep{ABG19}.

While our regression target---the conditional expectation of $Y_t$ given that $X_t$ exceeds its conditional quantile---is originally known as the MES in the (financial) systemic risk literature, we illustrate its value in three diverse applications focusing on the relation between systemic risk and asset price bubbles \citep{BRS20}, dissecting GDP growth risks among countries \citep{ABG19}, and in Appendix \ref{sec:add appl} by analyzing portfolio risk allocation \citep{Bea18}.

Moreover, we envision a multitude of further possible fields of applications such as in sensitivity analysis \citep[Sec.~1.2]{Asimit2019} or micro-econometric applications, where the MES could, e.g., be interpreted as the mean (personal or firm) income given the stress event of high inflation rates, or other challenging economic settings.
We also refer to \citet{HTZ23}, \citet{WTH24} and the references therein for the growing interest in quantile-truncated expectations in the micro-econometric literature.

A particular advantage of the MES as our regression target is its additivity \citep{CIM13} that---similar to our Section \ref{sec:GDP_MES}---allows for ES risk decompositions into different (additive) sub-components such as economic sectors.
In a similar fashion, recently analyzed inflation risks \citep{LL24} could be decomposed into product categories.
As such disaggregations are only possible by the additivity of the ES and MES risk measures as conditional expectations, these considerations provide further arguments for use of the ES over the VaR (and the MES over the CoVaR) in the recent debate of which risk measure is most suitable in practice \citep{EKT15, BCBS19, WangZitikis2021}.

\section*{Supplementary Materials}

The Supplementary Material contains the proof of Theorem~\ref{thm:an}, and shows how to estimate the asymptotic variance-covariance matrix of Theorem~\ref{thm:an} consistently.
It also verifies our main assumptions for the linear model of Example~\ref{ex:1} and provides an additional empirical application to portfolio optimization.

Replication material for the simulations and the applications in Section \ref{sec:EmpiricalApplication} (together with the data) is available under \href{https://github.com/TimoDimi/replication_MES}{https://github.com/TimoDimi/replication\_MES}.
It draws on the corresponding open source package \texttt{SystemicRisk} \citep{package_SystemicRisk} implemented in the statistical software \texttt{R} \citep{R2022}. 

%Furthermore, we supply \texttt{R} code that replicates the simulations in Section~\ref{sec:Simulations}.
%The supplement also contains \texttt{R} code and data to reproduce the empirical applications in Section~\ref{sec:EmpiricalApplication}.

\section*{Acknowledgments}

We would like to thank the AE, two anonymous referees, Markus Brunnermeier, Tobias Fissler, Christoph Hanck, Simon Rother and Isabel Schnabel for their insightful comments that significantly improved the quality of the paper.
We further thank Simon Rother for providing us with the bubble indicators used in \citet{BRS20}.

%
%
%\section*{Disclosure Statement}
%
%The authors report there are no competing interests to declare.
%
%

\section*{Funding}

The authors gratefully acknowledge support of the Deutsche Forschungsgemeinschaft (DFG, German Research Foundation) through grants 502572912 (Timo Dimitriadis), 460479886 and 531866675 (Yannick Hoga).

\pagebreak 
\appendix
\renewcommand\appendixpagename{Supplementary Material}
\appendixpage

In this Supplementary Material, we prove Theorem~\ref{thm:an} in Section~\ref{sec:thm1}.
In Section~\ref{sec:thm3}, we propose an estimator of the asymptotic variance-covariance matrix of Theorem~\ref{thm:an} and show its consistency.
% under the additional Assumptions~\ref{ass:9}--\ref{ass:8}.
Section~\ref{sec:ex verif} verifies the main Assumptions~\ref{ass:1}--\ref{ass:11} for the linear model of Example~\ref{ex:1}.
Section~\ref{sec:RelationOLSXCovariate} discusses the relation of our MES regressions to OLS regressions that use $X_t$ as an explanatory variable.
Section~\ref{sec:add appl} illustrates the use of MES regressions in real-world portfolio optimization and Section~\ref{sec:AddEmpResults} provides additional results for the applications in Section \ref{sec:EmpiricalApplication} of the main article.\\

\section*{Notation}

We use the following notational conventions throughout this supplement. The probability space where all random variables are defined on is $\big(\Omega, \mathcal{F}, \P\big)$. We denote by $C>0$ a large positive constant that may change from line to line. If not specified otherwise, all convergences are to be understood with respect to $n\to\infty$. We also write $\E_{t}[\cdot]=\E[\cdot\mid\mathcal{F}_t]$ and $\P_{t}\{\cdot\}=\P\{\cdot\mid\mathcal{F}_t\}$ for short. We exploit without further mention that the Frobenius norm is submultiplicative, such that $\norm{\mA\mB}\leq\norm{\mA}\cdot\norm{\mB}$ for conformable matrices $\mA$ and $\mB$. As is common, we define $\norm{X}_r=\big\{\E|X|^{r}\big\}^{1/r}$ to be the $L^{r}$-norm for $r\geq1$. For a real-valued, differentiable function $f(\cdot)$, we denote the $j$-th element of the gradient $\nabla f(\cdot)$ by $\nabla_j f(\cdot)$.

\section{Proof of Theorem~\ref{thm:an}}
\label{sec:thm1}

We split the proof of Theorem~\ref{thm:an} into four parts presented in Sections~\ref{Consistency of VaR and MES Parameter Estimators}--\ref{Joint Asymptotic Normality of the VaR and MES Parameter Estimator}.
Section~\ref{Consistency of VaR and MES Parameter Estimators} proves the preliminary results of consistency of $\widehat{\vtheta}_n^v$ and $\widehat{\vtheta}_n^m$. Using this, Section~\ref{Asymptotic Normality of the VaR Parameter Estimator} shows the asymptotic normality of $\widehat{\vtheta}_n^v$ and Section~\ref{Asymptotic Normality of the MES Parameter Estimator} that of $\widehat{\vtheta}_n^m$. Finally, Section~\ref{Joint Asymptotic Normality of the VaR and MES Parameter Estimator} combines these two results to derive the joint asymptotic normality of $\widehat{\vtheta}_n^{v}$ and $\widehat{\vtheta}_n^{m}$. 

\subsection{Consistency of VaR and MES Parameter Estimators}\label{Consistency of VaR and MES Parameter Estimators}

\renewcommand{\theequation}{A.\arabic{equation}}	% different equation numbering in appendix
\setcounter{equation}{0}

The purpose of this subsection is to establish the consistency of $\widehat{\vtheta}_{n}^{v}$ and $\widehat{\vtheta}_{n}^{m}$. 
We will use these consistency results in the asymptotic normality proofs in the subsequent sections.

\begin{prop}\label{prop:cons}
	Suppose Assumptions~\ref{ass:1}--\ref{ass:7} hold. Then, as $n\to\infty$, $\widehat{\vtheta}_{n}^{v}\overset{\P}{\longrightarrow}\vtheta_0^{v}$ and $\widehat{\vtheta}_{n}^{m}\overset{\P}{\longrightarrow}\vtheta_0^{m}$.
\end{prop}	

The proof of Proposition~\ref{prop:cons} relies on Lemmas~\ref{lem:se VaR}--\ref{lem:Type IV}, which can be found in Section~\ref{SR for cons} (including their proofs). A careful reading of the proof of Proposition~\ref{prop:cons} shows that consistency actually holds under a strict subset of the conditions of that proposition. However, for the sake of brevity, we refrain from a more precise statement. We prove both convergences in Proposition~\ref{prop:cons} by verifying conditions (i)--(iv) of Theorem~2.1 in \citet{NeweyMcFadden1994}. First, we prove $\widehat{\vtheta}_{n}^{v}\overset{\P}{\longrightarrow}\vtheta_0^{v}$:

\begin{proof}[{\textbf{Proof of Proposition~\ref{prop:cons} ($\widehat{\vtheta}_n^v\overset{\P}{\longrightarrow}\vtheta_0^{v}$):}}]
	Define
	\begin{align*}
		Q_0^{v}(\vtheta^v) &= \E\Big[S^{\VaR}\big(v_t(\vtheta^v), X_t\big)\Big],\\
		\widehat{Q}_n^v(\vtheta^v)& = \frac{1}{n}\sum_{t=1}^{n} S^{\VaR}\big(v_t(\vtheta^v), X_t\big),
	\end{align*}
	where $Q_0^{v}(\vtheta^v)$ does not depend on $t$ due to stationarity of $\big\{(X_t,\mZ_t^\prime)^\prime\big\}$ (see Assumption~\ref{ass:2}).
	
	For condition (i) of Theorem~2.1 in \citet{NeweyMcFadden1994}, we have to show that $Q_0^{v}(\vtheta^v)$ is uniquely minimized at $\vtheta^v=\vtheta_0^v$. 
	By Assumption~\ref{ass:3}~\ref{it:3i} and from \citet[Theorem~4.2]{FH24} it follows that $\E_{t}\big[S^{\VaR}\big(\cdot, X_t\big)\big]$ is uniquely minimized at $\VaR_{t,\b}$, which equals $v_t(\vtheta_0^v)$ under correct specification. 
	By Assumption~\ref{ass:1}~\ref{it:1i}, $\vtheta_0^v$ is then the unique minimizer of $\E_{t}\big[S^{\VaR}\big(v_t(\cdot), X_t\big)\big]$. By the law of the iterated expectations (LIE) we have that
	\[
	\E\Big[S^{\VaR}\big(v_t(\vtheta^v), X_t\big)\Big] = \E\bigg\{\E_{t}\Big[S^{\VaR}\big(v_t(\vtheta^v), X_t\big)\Big]\bigg\},
	\]
	which implies that $\vtheta_0^v$ is also the unique minimizer of $Q_0^{v}(\cdot)$, as desired.
	
	The compactness requirement of condition (ii) of Theorem~2.1 in \citet{NeweyMcFadden1994} is immediate from Assumption~\ref{ass:1}~\ref{it:1ii}.
	
	Condition (iii), i.e., the continuity of $\vtheta^v\mapsto Q_0^{v}(\vtheta^v)$, holds because for any $\vtheta^v\in\mTheta^v$,
	\begin{align*}
		\lim_{\vtheta\to\vtheta^v}Q_0^v(\vtheta) &= \lim_{\vtheta\to\vtheta^v}\E\Big[S^{\VaR}\big(v_t(\vtheta), X_t\big)\Big]\\
		&= \E\Big[\lim_{\vtheta\to\vtheta^v}S^{\VaR}\big(v_t(\vtheta), X_t\big)\Big]\\
		&= \E\Big[S^{\VaR}\big(v_t(\vtheta^v), X_t\big)\Big]\\
		& =Q_0^v(\vtheta^v),
	\end{align*}
	where the second step follows from the dominated convergence theorem (DCT), and the third step from the continuity of the map $v\mapsto S^{\VaR}(v,x)$ together with the continuity of $v_t(\cdot)$ (by Assumption~\ref{ass:4}~\ref{it:4i}). Thus, $\vtheta^v\mapsto Q_0^v(\vtheta^v)$ is continuous (by the sequential criterion). We may apply the DCT here, because $S^{\VaR}\big(v_t(\vtheta), X_t\big)$ is dominated by 
	\begin{equation}\label{eq:A2}
		\big|S^{\VaR}\big(v_t(\vtheta), X_t\big)\big|\leq\big|X_t - v_t(\vtheta)\big|\leq \big|X_t\big| + \big|v_t(\vtheta)\big|\leq |X_t| + V(\mZ_t),
	\end{equation}
	where the final term is integrable due to Assumption~\ref{ass:5}.
	
	It remains to verify condition (iv) of Theorem~2.1, i.e., the uniform law of large numbers (ULLN)
	\begin{equation}\label{eq:(A.2++)}
		\sup_{\vtheta^v\in\mTheta^v}\big|\widehat{Q}_n^v(\vtheta^v)-Q_0^{v}(\vtheta^v)\big|=o_{\P}(1).
	\end{equation}
	To task this, we use Theorem~21.9 of \citet{Dav94}, which shows that the following conditions are sufficient for \eqref{eq:(A.2++)}:
	\begin{enumerate}
		\item[(a)] $\widehat{Q}_n^v(\vtheta^v)-Q_0^{v}(\vtheta^v)=o_{\P}(1)$ for all $\vtheta^v\in\mTheta^v$;
		
		\item[(b)] $\big\{\widehat{Q}_n^v(\cdot)-Q_0^{v}(\cdot)\big\}$ is stochastically equicontinuous (s.e.).
	\end{enumerate}
	We first prove (a), for which we verify the conditions of the law of large numbers (LLN) of \citet[Theorem 3.34]{Whi01}. First, by \citet[Eq.~(1.11)]{Bra05} and Assumption~\ref{ass:2}, $(X_t,\mZ_t^\prime)^\prime$ is $\alpha$-mixing of size $-r/(r-1)$. As a measurable function of $(X_t,\mZ_t^\prime)^\prime$, $S^{\VaR}\big(v_t(\vtheta^v), X_t\big)$ is also $\alpha$-mixing of size $-r/(r-1)$ \citep[Theorem~3.49]{Whi01}. Furthermore, $S^{\VaR}\big(v_t(\vtheta^v), X_t\big)$ is stationary as a measurable function of the stationary $(X_t,\mZ_t^\prime)^\prime$. Combining these two facts with \citet[Proposition 3.44]{Whi01} implies ergodicity of $S^{\VaR}\big(v_t(\vtheta^v), X_t\big)$.
	
	Second, by \eqref{eq:A2},
	\begin{equation*}
		\E\big|S^{\VaR}\big(v_t(\vtheta^v), X_t\big)\big| \leq \E\big|X_t\big| + \E\big| V(\mZ_t)\big| \leq 2K<\infty.
	\end{equation*}
	Combining this with the ergodicity of $S^{\VaR}\big(v_t(\vtheta^v), X_t\big)$, the pointwise LLN of (a) follows from the ergodic theorem \citep[Theorem 3.34]{Whi01}.
	
	The s.e.~requirement of (b) directly follows from Lemma~\ref{lem:se VaR} (in Subsection~\ref{SR for cons} below) and \citet[Lemma~2 (a)]{And92}, such that the desired ULLN follows.
	
	Having verified all conditions, $\widehat{\vtheta}_{n}^{v}\overset{\P}{\longrightarrow}\vtheta_0^{v}$ now follows from Theorem~2.1 in \citet{NeweyMcFadden1994}.
\end{proof}

\begin{proof}[{\textbf{Proof of Proposition~\ref{prop:cons} ($\widehat{\vtheta}_n^m\overset{\P}{\longrightarrow}\vtheta_0^{m}$):}}]
	We now show that $\widehat{\vtheta}_{n}^{m}\overset{\P}{\longrightarrow}\vtheta_0^{m}$, which follows along similar lines as the proof of $\widehat{\vtheta}_{n}^{v}\overset{\P}{\longrightarrow}\vtheta_0^{v}$. 
	Define
	\begin{align*}
		Q_0^{m}(\vtheta^m) &= \E\Big[S^{\MES}\big((v_t(\vtheta_0^v), m_t(\vtheta^m))^\prime, (X_t, Y_t)^\prime\big)\Big],\\
		\widehat{Q}_n^m(\vtheta^m)& = \frac{1}{n}\sum_{t=1}^{n} S^{\MES}\big((v_t(\widehat{\vtheta}_n^v), m_t(\vtheta^m))^\prime, (X_t, Y_t)^\prime\big),
	\end{align*}
	where we suppress the dependence on the VaR parameter for notational brevity.
	Once again, $Q_0^{m}(\vtheta^m)$ does not depend on $t$ by Assumption~\ref{ass:2}.
	
	For condition (i) of Theorem~2.1, we have to show that $Q_0^{m}(\vtheta^m)$ is uniquely minimized at $\vtheta^m=\vtheta_0^m$. 
	By Assumption~\ref{ass:3}~\ref{it:3i} and from \citet[Theorem~4.2]{FH24} it follows that $\E_{t}\big[S^{\MES}\big((v_t(\vtheta_0^v), \cdot)^\prime, (X_t, Y_t)^\prime\big)\big]$ is uniquely minimized at $\MES_{t,\b}$, which equals $m_t(\vtheta_0^m)$ under correct specification. 
	By Assumption~\ref{ass:1}~\ref{it:1i}, $\vtheta_0^m$ is then the unique minimizer of $\E_{t}\big[S^{\MES}\big((v_t(\vtheta_0^v), m_t(\cdot))^\prime, (X_t, Y_t)^\prime\big)\big]$. By the LIE we have that
	\[
	\E\Big[S^{\MES}\big((v_t(\vtheta_0^v), m_t(\vtheta^m))^\prime, (X_t, Y_t)^\prime\big)\Big] = \E\bigg\{\E_{t}\Big[S^{\MES}\big((v_t(\vtheta_0^v), m_t(\vtheta^m))^\prime, (X_t, Y_t)^\prime\big)\Big]\bigg\},
	\]
	which implies that $\vtheta_0^m$ is also the unique minimizer of $Q_0^{m}(\cdot)$, as desired.
	
	The compactness requirement of condition (ii) of Theorem~2.1 in \citet{NeweyMcFadden1994} is immediate from Assumption~\ref{ass:1}~\ref{it:1ii}.
	
	Condition (iii), i.e., the continuity of $\vtheta^m\mapsto Q_0^{m}(\vtheta^m)$, holds because for any $\vtheta^m\in\mTheta^m$,
	\begin{align*}
		\lim_{\vtheta\to\vtheta^m}Q_0^m(\vtheta) &= \lim_{\vtheta\to\vtheta^m}\E\Big[S^{\MES}\big((v_t(\vtheta_0^v), m_t(\vtheta))^\prime, (X_t, Y_t)^\prime\big)\Big]\\
		&= \E\Big[\lim_{\vtheta\to\vtheta^m}S^{\MES}\big((v_t(\vtheta_0^v), m_t(\vtheta))^\prime, (X_t, Y_t)^\prime\big)\Big]\\
		&= \E\Big[S^{\MES}\big((v_t(\vtheta_0^v), m_t(\vtheta^m))^\prime, (X_t, Y_t)^\prime\big)\Big]\\
		& =Q_0^m(\vtheta^m),
	\end{align*}
	where the second step follows from the DCT, and the third step from the continuity of the map $m\mapsto S^{\MES}\big((v,m)^\prime,(x, y)^\prime\big)$ together with continuity of $m_t(\cdot)$ (by Assumption~\ref{ass:4}~\ref{it:4i}). Thus, $\vtheta^m\mapsto Q_0^m(\vtheta^m)$ is continuous (by the sequential criterion). Note that we may apply the DCT, because $S^{\MES}\big((v_t(\vtheta_0^v), m_t(\vtheta^m))^\prime, (X_t, Y_t)^\prime\big)$ is dominated by 
	\begin{align}
		\big|S^{\MES}\big((v_t(\vtheta_0^v), m_t(\vtheta^m))^\prime, (X_t, Y_t)^\prime\big)\big|&=\frac{1}{2}\big|\1_{\{X_t> v_t(\vtheta_0^v)\}}\big(Y_t-m_t(\vtheta^m)\big)^2\big|\notag\\
		&\leq |Y_t|^2 +  |m_t(\vtheta^m)|^2\notag\\
		&\leq  Y_t^2  +  M^2(\mZ_t),\label{eq:(8.1)}
	\end{align}
	where the first inequality follows from the $c_r$-inequality \citep[e.g.,][Theorem~9.28]{Dav94}, and the integrability of the variables in \eqref{eq:(8.1)} follows from Assumption~\ref{ass:5}.
	
	It remains to verify condition (iv) of Theorem~2.1, i.e., the ULLN
	\[
	\sup_{\vtheta^m\in\mTheta^m}\big|\widehat{Q}_n^m(\vtheta^m)-Q_0^{m}(\vtheta^m)\big|=o_{\P}(1).
	\]
	Since $Q_0^m(\vtheta^m)=\frac{1}{n}\sum_{t=1}^{n}\E\Big[S^{\MES}\big((v_t(\vtheta_0^v), m_t(\vtheta^m))^\prime, (X_t, Y_t)^\prime\big)\Big]$ (from Assumption~\ref{ass:2}), we obtain the bound
	\begin{align*}
		&\sup_{\vtheta^m\in\mTheta^m}\big|\widehat{Q}_n^m(\vtheta^m)-Q_0^{m}(\vtheta^m)\big| \\
		&\leq \sup_{\vtheta^m\in\mTheta^m}\bigg|\frac{1}{n}\sum_{t=1}^{n}S^{\MES}\big((v_t(\widehat{\vtheta}_n^v), m_t(\vtheta^m))^\prime, (X_t, Y_t)^\prime\big) - \E\Big[S^{\MES}\big((v_t(\widehat{\vtheta}_n^v), m_t(\vtheta^m))^\prime, (X_t, Y_t)^\prime\big)\Big]\bigg|\\
		&\hspace{0.5cm} + \sup_{\vtheta^m\in\mTheta^m}\bigg|\frac{1}{n}\sum_{t=1}^{n}\E\Big[S^{\MES}\big((v_t(\widehat{\vtheta}_n^v), m_t(\vtheta^m))^\prime, (X_t, Y_t)^\prime\big)\Big] - \E\Big[S^{\MES}\big((v_t(\vtheta_0^v), m_t(\vtheta^m))^\prime, (X_t, Y_t)^\prime\big)\Big]\bigg|\\
		&=: A_{1n} + B_{1n}.
	\end{align*}
	Note that 
	\[
	\Big|S^{\MES}\big((v_t(\widehat{\vtheta}_n^v), m_t(\vtheta^m))^\prime, (X_t, Y_t)^\prime\big)\Big|\leq Y_t^2  +  M^2(\mZ_t)
	\]
	by similar arguments leading to \eqref{eq:(8.1)}, such that $\E\Big[S^{\MES}\big((v_t(\widehat{\vtheta}_n^v), m_t(\vtheta^m))^\prime, (X_t, Y_t)^\prime\big)\Big]$ exists by virtue of Assumption~\ref{ass:5}.
	We show in turn that $A_{1n}=o_{\P}(1)$ and $B_{1n}=o_{\P}(1)$. First, we show that $A_{1n}=o_{\P}(1)$. From $\widehat{\vtheta}_n^v\overset{\P}{\longrightarrow}\vtheta_0^v$ (from the first part of this proof) and Assumption~\ref{ass:1} \ref{it:1ii}, $(\widehat{\vtheta}_n^v,\vtheta^m)^\prime\in\mTheta$ with probability approaching one (w.p.a.~1) as $n\to\infty$, such that w.p.a.~1,
	\[
	A_{1n}\leq\sup_{\vtheta\in\mTheta}\bigg|\frac{1}{n}\sum_{t=1}^{n}\Big\{S_t^{\MES}(\vtheta) - \E\big[S_t^{\MES}(\vtheta)\big]\Big\}\bigg|,
	\]
	where $S_t^{\MES}(\vtheta)$ is defined in Lemma~\ref{lem:Type IV} in Section~\ref{SR for cons} below.
	The arguments used to show that the right-hand side is $o_{\P}(1)$ resemble those of \citet[Proof of Lemma~4]{DLS23}. By Lemma~\ref{lem:Type IV}, the class of functions $S_t^{\MES}(\vtheta)$ forms a type IV class with index $p=2r$ and, hence, satisfies ``Ossiander's $L^{2r}$-entropy condition'' by Theorem~5 of \citet{And94} (with ``$L^{2r}$-envelope'' given by the supremum). 
	
	Moreover, by the $c_r$-inequality and Assumption~\ref{ass:5},
	\begin{align*}
		\E\bigg[\sup_{\vtheta\in\mTheta}\big|S_t^{\MES}(\vtheta)\big|^{2r}\bigg] &\leq C \E\bigg[\sup_{\vtheta\in\mTheta}\big|\1_{\{X_t> v_t(\vtheta^v)\}}\big\{Y_t-m_t(\vtheta^m)\big\}^{4r}\big|\bigg] \\
		&\leq C\E|Y_t|^{4r} + C\E\Big[\sup_{\vtheta^m\in\mTheta^m}m_t^{4r}(\vtheta^m)\Big]\\
		&\leq C\E|Y_t|^{4r} + C\E\big[M^{4r}(\mZ_t)\big]\\
		&\leq C<\infty.
	\end{align*}
	Also note from Assumption~\ref{ass:2} that $\sum_{t=1}^{\infty}t^{1/(r-1)}\beta(t)<\infty$ for the $\beta$-mixing coefficients $\beta(\cdot)$ of the $\{(X_t, Y_t, \mZ_t^\prime)^\prime\}$ and, hence, the same holds for the $\beta$-mixing coefficients of the $\{S_t^{\MES}(\vtheta)\}$, which are bounded by those of $\{(X_t, Y_t, \mZ_t^\prime)^\prime\}$ (by standard mixing inequalities used before).
	
	Therefore, Theorem~1 and Application~1 of \citet{DMR95} imply that a functional central limit theorem (FCLT) holds for $S_t^{\MES}(\vtheta)$, such that 
	\begin{equation}\label{eq:FCLT Thm 1}
		\sup_{\vtheta\in\mTheta}\bigg|\frac{1}{\sqrt{n}}\sum_{t=1}^{n}\Big\{S_t^{\MES}(\vtheta) - \E\big[S_t^{\MES}(\vtheta)\big]\Big\} - B^{(n)}(\vtheta)\bigg|=o_{\P}(1)
	\end{equation}
	for some sequence of Gaussian processes $B^{(n)}(\cdot)$ with a.s.~continuous sample paths and some covariance function $\Gamma(\cdot,\cdot)$. In particular, this implies that
	\[
	\sup_{\vtheta\in\mTheta}\bigg|\frac{1}{n}\sum_{t=1}^{n}\Big\{S_t^{\MES}(\vtheta) - \E\big[S_t^{\MES}(\vtheta)\big]\Big\}\bigg|=o_{\P}(1),
	\]
	whence $A_{1n}=o_{\P}(1)$ follows. 
	
	Next, we show that $B_{1n}=o_{\P}(1)$. For any $\varepsilon>0$ and $\delta>0$ it holds that
	\begin{align}
		\P\big\{B_{1n}>\varepsilon\big\} & \leq \P\Big\{B_{1n}>\varepsilon,\ \big\lVert\widehat{\vtheta}_n^{v}-\vtheta_0^v\big\rVert\leq\delta\Big\} + \P\Big\{\big\lVert\widehat{\vtheta}_n^{v}-\vtheta_0^v\big\rVert>\delta\Big\}\notag\\
		&\leq \P\Bigg\{\frac{1}{n}\sum_{t=1}^{n}\E\bigg[\sup_{\substack{\vtheta^{m}\in\mTheta^{m}\\ \Vert\vtheta^v-\vtheta_0^v\Vert\leq\delta}}\Big|S^{\MES}\big((v_t(\vtheta^v), m_t(\vtheta^m))^\prime, (X_t, Y_t)^\prime\big)\notag\\
		&\hspace{4cm}-S^{\MES}\big((v_t(\vtheta_0^v), m_t(\vtheta^m))^\prime, (X_t, Y_t)^\prime\big)\Big|\bigg]>\varepsilon\Bigg\} + o(1).\label{eq:tbb}
	\end{align}
	Define the $\mZ_{t}$-measurable quantities
	\begin{align*}
		\underline{\vtheta}^v &= \argmin_{\Vert\vtheta^v-\vtheta_0^v\Vert\leq\delta}v_t(\vtheta^v),\\
		\overline{\vtheta}^v  &= \argmax_{\Vert\vtheta^v-\vtheta_0^v\Vert\leq\delta}v_t(\vtheta^v),
	\end{align*}
	which exist by continuity of $v_t(\cdot)$ (see Assumption~\ref{ass:4}~\ref{it:4i}).
	For notational brevity and because it does not affect any of the subsequent arguments, we suppress the dependence of $\underline{\vtheta}^v$ and $\overline{\vtheta}^v$ on $t$.
	Use the $c_r$-inequality to derive that
	\begin{align*}
		&\sup_{\substack{\vtheta^{m}\in\mTheta^{m}\\ \Vert\vtheta^v-\vtheta_0^v\Vert\leq\delta}}\Big|S^{\MES}\big((v_t(\vtheta^v), m_t(\vtheta^m))^\prime, (X_t, X_t)^\prime\big)-S^{\MES}\big((v_t(\vtheta_0^v), m_t(\vtheta^m))^\prime, (X_t, X_t)^\prime\big)\Big|\\
		&=\frac{1}{2}\sup_{\substack{\vtheta^{m}\in\mTheta^{m}\\ \Vert\vtheta^v-\vtheta_0^v\Vert\leq\delta}}\big|\1_{\{X_t>v_t(\vtheta^v)\}}-\1_{\{X_t>v_t(\vtheta_0^v)\}}\big|\big\{Y_t - m_t(\vtheta^m)\big\}^2\\
		&\leq \sup_{\substack{\vtheta^{m}\in\mTheta^{m}\\ \Vert\vtheta^v-\vtheta_0^v\Vert\leq\delta}}\big|\1_{\{X_t>v_t(\vtheta^v)\}}-\1_{\{X_t>v_t(\vtheta_0^v)\}}\big|m_t^2(\vtheta^m) + \sup_{\substack{\vtheta^{m}\in\mTheta^{m}\\ \Vert\vtheta^v-\vtheta_0^v\Vert\leq\delta}}\big|\1_{\{X_t>v_t(\vtheta^v)\}}-\1_{\{X_t>v_t(\vtheta_0^v)\}}\big|Y_t^2\\
		&\leq \big|\1_{\{X_t>v_t(\overline{\vtheta}^v)\}}-\1_{\{X_t>v_t(\underline{\vtheta}^v)\}}\big|\sup_{\vtheta^{m}\in\mTheta^{m}} m_t^2(\vtheta^m) + \big|\1_{\{X_t>v_t(\overline{\vtheta}^v)\}}-\1_{\{X_t>v_t(\underline{\vtheta}^v)\}}\big|Y_t^2.
	\end{align*}
	Therefore,
	\begin{align*}
		\E&\Bigg[\sup_{\substack{\vtheta^{m}\in\mTheta^{m}\\ \Vert\vtheta^v-\vtheta_0^v\Vert\leq\delta}}\Big|S^{\MES}\big((v_t(\vtheta^v), m_t(\vtheta^m))^\prime, (X_t, X_t)^\prime\big)-S^{\MES}\big((v_t(\vtheta_0^v), m_t(\vtheta^m))^\prime, (X_t, X_t)^\prime\big)\Big|\Bigg]\\
		&\leq \E\Big[\big|\1_{\{X_t>v_t(\overline{\vtheta}^v)\}}-\1_{\{X_t>v_t(\underline{\vtheta}^v)\}}\big|M^2(\mZ_t)\Big] + \E\Big[\big|\1_{\{X_t>v_t(\overline{\vtheta}^v)\}}-\1_{\{X_t>v_t(\underline{\vtheta}^v)\}}\big|Y_t^2\Big]\\	
		&=:B_{11t} + B_{12t}.
	\end{align*}
	For $B_{11t}$, we obtain using the LIE (in the first step), Assumption~\ref{ass:3} \ref{it:3ii} (in the third step), the mean value theorem (in the fourth step) and Assumptions~\ref{ass:4}--\ref{ass:5} (in the fifth and sixth step) that
	\begin{align}
		B_{11t}	&= \E\bigg\{M^2(\mZ_t)\E_{t}\Big[\big|\1_{\{X_t>v_t(\overline{\vtheta}^v)\}}-\1_{\{X_t>v_t(\underline{\vtheta}^v)\}}\big|\Big]\bigg\}\notag\\
		&\leq \E\bigg\{M^2(\mZ_t)\Big|\int_{v_t(\underline{\vtheta}^v)}^{v_t(\overline{\vtheta}^v)}f_t^{X}(x)\D x\Big|\bigg\}\notag\\
		&\leq \E\Big\{M^2(\mZ_t) K \big|v_t(\overline{\vtheta}^v)-v_t(\underline{\vtheta}^v)\big|\Big\}\notag\\
		&\leq K\E\Big\{M^2(\mZ_t)  \norm{\nabla v_{t}(\vtheta^{\ast})}\cdot\big\Vert\overline{\vtheta}^v-\underline{\vtheta}^v\big\Vert\Big\}\notag\\
		&\leq K\E\Big\{M^2(\mZ_t)  V_1(\mZ_t)\Big\}\delta\notag\\
		&\leq K \norm{M(\mZ_t)}_4^2 \norm{V_1(\mZ_t)}_2\delta\notag\\
		&\leq C\delta,\notag
	\end{align}
	where $\vtheta^{\ast}$ is some mean value between $\underline{\vtheta}^v$ and $\overline{\vtheta}^v$. Hence, choosing $\delta>0$ sufficiently small, we can ensure that $B_{11t}<\varepsilon/2$ for any fixed $\varepsilon>0$.

	For $B_{12t}$, we get from Hölder's inequality that for $p=2r/(2r-1)$ and $q=2r$,
	\begin{equation}\label{eq:B12t h}
		B_{12t} \leq \Big\{\E\big|\1_{\{X_t>v_t(\overline{\vtheta}^v)\}}-\1_{\{X_t>v_t(\underline{\vtheta}^v)\}}\big|^{p}\Big\}^{1/p}\Big\{\E|Y_t|^{2q}\Big\}^{1/q}.
	\end{equation}
	By the LIE and the mean value theorem (MVT),
	\begin{align}
		\E\big|\1_{\{X_t>v_t(\overline{\vtheta}^v)\}}-\1_{\{X_t>v_t(\underline{\vtheta}^v)\}}\big|^{p}	&= \E\Big[\E_{t}\big|\1_{\{X_t>v_t(\overline{\vtheta}^v)\}}-\1_{\{X_t>v_t(\underline{\vtheta}^v)\}}\big|^{p}\Big]\notag\\
		&= \E\Big[\E_{t}\big|\1_{\{X_t>v_t(\overline{\vtheta}^v)\}}-\1_{\{X_t>v_t(\underline{\vtheta}^v)\}}\big|\Big]\notag\\
		&=\E\Big|F_t^{X}\big(v_t(\overline{\vtheta}^v)\big) - F_t^{X}\big(v_t(\underline{\vtheta}^v)\big)\Big|\notag\\
		&=\E\Big[f_t^{X}\big(v_t(\vtheta^\ast)\big)\big|v_t(\overline{\vtheta}^v) - v_t(\underline{\vtheta}^v)\big|\Big]\notag\\
		&\leq \E\Big[K\big|\nabla v_t(\vtheta^\ast)(\overline{\vtheta}^v - \underline{\vtheta}^v)\big|\Big]\notag\\
		&\leq \E\Big[KV_1(\mZ_{t})\big\Vert\overline{\vtheta}^v - \underline{\vtheta}^v\big\Vert\Big]\notag\\
		&\leq K\E\big[V_1(\mZ_{t})\big]2\delta,\label{eq:help ind1}
	\end{align}
	where $\vtheta^\ast$ (which may change from line to line) lies on the line connecting $\underline{\vtheta}^v$ and $\overline{\vtheta}^v$, and the penultimate step follows from Assumption~\ref{ass:4}~\ref{it:4iii}.
	Plugging this into \eqref{eq:B12t h} gives
	\begin{equation}\label{eq:B12t}
		B_{12t}\leq C\Big\{\E\big[V_1(\mZ_{t})\big]\Big\}^{1/p}\Big\{\E|Y_t|^{2q}\Big\}^{1/q}\delta^{1/p}.
	\end{equation}
	Once again, this and Assumption~\ref{ass:5} allow us to conclude that $B_{12t}<\varepsilon/2$ by choosing $\delta>0$ sufficiently small. Therefore, for a suitable choice of $\delta>0$, the first right-hand side term in \eqref{eq:tbb} can be shown to equal zero, such that $B_{1n}=o_{\P}(1)$ follows. This establishes condition (iv). That $\widehat{\vtheta}_n^m\overset{\P}{\longrightarrow}\vtheta_0^m$ now follows from Theorem~2.1 in \citet{NeweyMcFadden1994}.
\end{proof}

\subsubsection{Supplementary Results for Consistency}\label{SR for cons}

\begin{lem}\label{lem:se VaR}
	Suppose Assumptions~\ref{ass:1}--\ref{ass:7} hold. Then, for all $\vtheta^v, \widetilde{\vtheta}^v\in\mTheta^v$,
	\[
	\big|S^{\VaR}\big(v_t(\vtheta^v), X_t\big) - S^{\VaR}\big(v_t(\widetilde{\vtheta}^v), X_t\big)\big|\leq B_t \big\Vert\vtheta^v - \widetilde{\vtheta}^v\big\Vert
	\]
	a.s.~for some random variables $B_t>0$, not depending on $\vtheta^v$ and $\widetilde{\vtheta}^v$, with $\sup_{n\in\mathbb{N}}\frac{1}{n}\sum_{t=1}^{n}\E\big[B_t\big]<\infty$.
\end{lem}

\begin{proof}
	Write
	\begin{align}
		& \big|S^{\VaR}\big(v_t(\vtheta^v), X_t\big) - S^{\VaR}\big(v_t(\widetilde{\vtheta}^v), X_t\big)\big|\notag\\
		& =\Big|\big[\1_{\{X_t\leq v_t(\vtheta^v)\}}-\b\big]\big[v_t(\vtheta^v) - X_t\big] - \big[\1_{\{X_t\leq v_t(\widetilde{\vtheta}^v)\}}-\b\big]\big[v_t(\widetilde{\vtheta}^v) - X_t\big]\Big|\notag\\
		& = \Big|\1_{\{X_t\leq v_t(\vtheta^v)\}}\big[v_t(\vtheta^v) - X_t\big] - \1_{\{X_t\leq v_t(\widetilde{\vtheta}^v)\}}\big[v_t(\widetilde{\vtheta}^v) - X_t\big] + \b\big[v_t(\widetilde{\vtheta}^v) - v_t(\vtheta^v)\big]\Big|\notag\\
		& \leq \Big|\1_{\{X_t\leq v_t(\vtheta^v)\}}\big[v_t(\vtheta^v) - X_t\big] - \1_{\{X_t\leq v_t(\widetilde{\vtheta}^v)\}}\big[v_t(\widetilde{\vtheta}^v) - X_t\big]\Big| + \b\big|v_t(\widetilde{\vtheta}^v) - v_t(\vtheta^v)\big|\notag\\
		&= A_{2t} + B_{2t}.\label{eq:(5.1)}
	\end{align}
	
	Consider the two terms separately. First, by the MVT and Assumption~\ref{ass:4}~\ref{it:4iii},
	\[
	B_{2t} \leq \b\sup_{\vtheta^v\in\mTheta^v}\big\Vert\nabla v_t(\vtheta^v)\big\Vert\cdot \big\Vert\vtheta^v-\widetilde{\vtheta}^v\big\Vert\leq \b V_1(\mZ_t) \big\Vert\vtheta^v-\widetilde{\vtheta}^v\big\Vert.
	\]
	Second,
	\begin{align*}
		A_{2t} &= \frac{1}{2}\Big|v_t(\vtheta^v) - X_t + \big|v_t(\vtheta^v) - X_t\big| - \big\{v_t(\widetilde{\vtheta}^v) - X_t + \big|v_t(\widetilde{\vtheta}^v) - X_t\big|\big\}\Big|\\
		& \leq \frac{1}{2}\Big|v_t(\vtheta^v) - X_t - \big[v_t(\widetilde{\vtheta}^v) - X_t\big]\Big| + \frac{1}{2}\Big| \big|v_t(\vtheta^v) - X_t\big| - \big|v_t(\widetilde{\vtheta}^v) - X_t\big|\Big|\\
		& \leq \big|v_t(\vtheta^v) - v_t(\widetilde{\vtheta}^v)\big|\\
		&\leq \sup_{\vtheta^v\in\mTheta^v}\big\Vert\nabla v_t(\vtheta^v)\big\Vert\cdot \big\Vert\vtheta^v-\widetilde{\vtheta}^v\big\Vert\\
		&\leq V_1(\mZ_t)\big\Vert\vtheta^v-\widetilde{\vtheta}^v\big\Vert,
	\end{align*}
	where we used that $\big||a|-|b|\big|\leq|a-b|$ in the third step and Assumption~\ref{ass:4} \ref{it:4iii} combined with the MVT in the final two steps.
	
	Therefore, we obtain from \eqref{eq:(5.1)} that
	\begin{align*}
		\Big|S^{\VaR}\big(v_t(\vtheta^v), X_t\big) - S^{\VaR}\big(v_t(\widetilde{\vtheta}^v), X_t\big)\Big| & \leq(1+\b)V_1(\mZ_t) \Vert\vtheta^v-\widetilde{\vtheta}^v\Vert\\
		&=:B_t\Vert\vtheta^v-\widetilde{\vtheta}^v\Vert.
	\end{align*}
	From Assumption~\ref{ass:5},
	\[
	\frac{1}{n}\sum_{t=1}^{n}\E\big[B_t\big]= (1+\b)\frac{1}{n}\sum_{t=1}^{n}\E\big[V_1(\mZ_t)\big]\leq (1+\b)C<\infty,
	\]
	such that the conclusion follows.
\end{proof}

\begin{lem}\label{lem:Type IV}
	Suppose Assumptions~\ref{ass:1}--\ref{ass:7} hold. Then, the class of functions $S_t^{\MES}\big(\vtheta=(\vtheta^{v\prime}, \vtheta^{m\prime})^\prime\big):=S^{\MES}\big((v_t(\vtheta^v), m_t(\vtheta^m))^\prime, (X_t, Y_t)^\prime\big)$ forms a \textit{type IV class} in the sense of \citet[p.~2278]{And94} with index $p=2r$ (where $r>1$ is from Assumption~\ref{ass:2}).
\end{lem}

\begin{proof}
	To prove the lemma, we must show that there exist constants $C>0$ and $\psi>0$, such that
	\begin{equation}\label{eq:(11.0)}
		\sup_{t\in\mathbb{N}}\bigg\Vert\sup_{\Vert\widetilde{\vtheta} - \vtheta\Vert\leq\delta}\Big|S_t^{\MES}(\widetilde{\vtheta}) - S_t^{\MES}(\vtheta)\Big|\bigg\Vert_{2r}\leq C\delta^\psi
	\end{equation}
	for all $\vtheta=(\vtheta^{v\prime}, \vtheta^{m\prime})^\prime\in\mTheta$ and all $\delta>0$ in some neighborhood of 0.
	
	Fix $\vtheta\in\mTheta$. By definition of $S_t^{\MES}(\cdot)$ we have that 
	\begin{align}
		\Big|S_t^{\MES}(\widetilde{\vtheta}) - S_t^{\MES}(\vtheta)\Big| &= \frac{1}{2}\Big|\1_{\{X_t> v_t(\vtheta^v)\}}\big[Y_t - m_t(\vtheta^m)\big]^2 - \1_{\{X_t> v_t(\widetilde{\vtheta}^v)\}}\big[Y_t - m_t(\widetilde{\vtheta}^m)\big]^2\Big|\notag\\
		&=\frac{1}{2}\Big|\1_{\{X_t> v_t(\vtheta^v)\}}\Big[\big\{Y_t - m_t(\vtheta^m)\big\}^2 - \big\{Y_t - m_t(\widetilde{\vtheta}^m)\big\}^2\Big] \notag\\
		&\hspace{2cm} + \Big[\1_{\{X_t> v_t(\vtheta^v)\}} - \1_{\{X_t> v_t(\widetilde{\vtheta}^v)\}}\Big]\big\{Y_t - m_t(\widetilde{\vtheta}^m)\big\}^2\Big|\notag\\
		&\leq\frac{1}{2}\Big|\1_{\{X_t> v_t(\vtheta^v)\}}\Big[\big\{Y_t - m_t(\vtheta^m)\big\}^2 - \big\{Y_t - m_t(\widetilde{\vtheta}^m)\big\}^2\Big]\Big| \notag\\
		&\hspace{2cm} + \frac{1}{2}\Big|\Big[\1_{\{X_t> v_t(\vtheta^v)\}} - \1_{\{X_t> v_t(\widetilde{\vtheta}^v)\}}\Big]\big\{Y_t - m_t(\widetilde{\vtheta}^m)\big\}^2\Big|\notag\\
		&\leq \Big|\1_{\{X_t> v_t(\vtheta^v)\}}\Big[Y_t\big\{m_t(\widetilde{\vtheta}^m) - m_t(\vtheta^m)\big\} + \frac{1}{2}\big\{m_t^2(\vtheta^m) - m_t^2(\widetilde{\vtheta}^m)\big\}\Big]\Big| \notag\\
		&\hspace{2cm} + \big|\1_{\{X_t> v_t(\vtheta^v)\}} - \1_{\{X_t> v_t(\widetilde{\vtheta}^v)\}}\big|Y_t^2\notag\\
		&\hspace{2cm} + \big|\1_{\{X_t> v_t(\vtheta^v)\}} - \1_{\{X_t> v_t(\widetilde{\vtheta}^v)\}}\big|m_t^2(\widetilde{\vtheta}^m)\notag\\
		&=:A_{3t} + B_{3t} + C_{3t}.\label{eq:(11.1)}
	\end{align}
	
	We treat these three terms separately. Define the $\mZ_t$-measurable quantities
	\begin{align}
		\underline{\vtheta}^{v} & =\argmin_{\Vert\widetilde{\vtheta} - \vtheta\Vert\leq\delta} v_t(\widetilde{\vtheta}^v),&& \overline{\vtheta}^{v}  =\argmax_{\Vert\widetilde{\vtheta} - \vtheta\Vert\leq\delta} v_t(\widetilde{\vtheta}^v),\label{eq:theta v ast}\\
		\underline{\vtheta}^{m} & =\argmin_{\Vert\widetilde{\vtheta} - \vtheta\Vert\leq\delta} m_t(\widetilde{\vtheta}^m),&& 	\overline{\vtheta}^{m} =\argmax_{\Vert\widetilde{\vtheta} - \vtheta\Vert\leq\delta} m_t(\widetilde{\vtheta}^m),\\
		\underline{\vtheta}^{m^2} & =\argmin_{\Vert\widetilde{\vtheta} - \vtheta\Vert\leq\delta} m_t^2(\widetilde{\vtheta}^m),&& 	\overline{\vtheta}^{m^2} =\argmax_{\Vert\widetilde{\vtheta} - \vtheta\Vert\leq\delta} m_t^2(\widetilde{\vtheta}^m),\notag
	\end{align}
	where we again suppress the dependence of the defined quantities on $t$.
	By continuity of $v_t(\cdot)$ and $m_t(\cdot)$, these quantities are well-defined.

	First, use Minkowski's inequality to write
	\begin{align}
		\bigg\Vert\sup_{\Vert\widetilde{\vtheta} - \vtheta\Vert\leq\delta} A_{3t}\bigg\Vert_{2r} &= \bigg\Vert\sup_{\Vert\widetilde{\vtheta} - \vtheta\Vert\leq\delta}\Big|Y_t\big\{m_t(\widetilde{\vtheta}^m) - m_t(\vtheta^m)\big\} + \frac{1}{2}\big\{ m_t^2(\vtheta^m) - m_t^2(\widetilde{\vtheta}^m)\big\}\Big|\bigg\Vert_{2r}\notag\\
		&\leq \bigg\Vert Y_t \sup_{\Vert\widetilde{\vtheta} - \vtheta\Vert\leq\delta}\big|m_t(\widetilde{\vtheta}^m) - m_t(\vtheta^m)\big|\bigg\Vert_{2r} + \frac{1}{2}\bigg\Vert \sup_{\Vert\widetilde{\vtheta} - \vtheta\Vert\leq\delta}\big|m_t^2(\widetilde{\vtheta}^m) - m_t^2(\vtheta^m)\big|\bigg\Vert_{2r}\notag\\
		&\leq \Big\Vert Y_t\big[m_t(\overline{\vtheta}^{m}) - m_t(\underline{\vtheta}^m)\big]\Big\Vert_{2r} + \frac{1}{2}\Big\Vert m_t^2(\overline{\vtheta}^{m^2}) - m_t^2(\underline{\vtheta}^{m^2})\Big\Vert_{2r}\notag\\
		&= \Big\Vert Y_t \nabla m_t(\vtheta^{\ast})(\overline{\vtheta}^{m} - \underline{\vtheta}^m)\Big\Vert_{2r} + \frac{1}{2}\Big\Vert 2m_t(\vtheta^{\ast})\nabla m_t(\vtheta^{\ast}) (\overline{\vtheta}^{m^2} - \underline{\vtheta}^{m^2}) \Big\Vert_{2r}\notag\\
		&\leq \big\Vert Y_t M_1(\mZ_t)\big\Vert_{2r}2\delta + \big\Vert M(\mZ_t)M_1(\mZ_t)\big\Vert_{2r}2\delta\notag\\
		&\leq \big\Vert Y_t\big\Vert_{4r} \big\Vert M_1(\mZ_t)\big\Vert_{4r} 2\delta + \big\Vert M(\mZ_t)\big\Vert_{4r} \big\Vert M_1(\mZ_t)\big\Vert_{4r}2\delta\notag\\
		&\leq  C\delta,\label{eq:A3t}
	\end{align}
	where $\vtheta^{\ast}$ is some mean value that may change from appearance to appearance.
	
	Furthermore, by similar arguments used to derive \eqref{eq:B12t},
	\begin{align}
		\bigg\Vert \sup_{\Vert\widetilde{\vtheta} - \vtheta\Vert\leq\delta} B_{3t}\bigg\Vert_{2r} &= \bigg\Vert \sup_{\Vert\widetilde{\vtheta} - \vtheta\Vert\leq\delta}\Big|\1_{\{X_t> v_t(\vtheta^v)\}} - \1_{\{X_t> v_t(\widetilde{\vtheta}^v)\}}\Big|Y_t^2\bigg\Vert_{2r}\notag\\
		&\leq \big\Vert \1_{\{X_t> v_t(\overline{\vtheta}^v)\}} - \1_{\{X_t> v_t(\underline{\vtheta}^v)\}} \big\Vert_{p} \big\Vert Y_t^2 \big\Vert_{q}\notag\\
		&\leq C\delta^{1/p}\label{eq:B3t}
	\end{align}
	for $p=2r(4r+\iota)/\iota$ and $q=2r+\iota/2$ (such that $1/p+1/q=1/(2r)$).

	Now, turn to $C_{3t}$. 
	%We have that
	%\begin{align}
	%\E_{t}\bigg[\sup_{\Vert\widetilde{\vtheta} - \vtheta\Vert\leq\delta}\Big| \1_{\{X_t> v_t(\vtheta^v)\}} - \1_{\{X_t> v_t(\widetilde{\vtheta}^v)\}}\Big|\bigg] &\leq \E_{t}\Big[\1_{\{v_t(\underline{\vtheta}^v)\leq X_t\leq v_t(\overline{\vtheta}^{v})\}}\Big]\notag\\
	%&=\int_{v_t(\underline{\vtheta}^{v})}^{v_t(\overline{\vtheta}^{v})}f_t^{X}(x)\D x\notag\\
	%&\leq K\big|v_t(\overline{\vtheta}^{v}) - v_t(\underline{\vtheta}^{v})\big|\notag\\
	%&=K\big|\nabla v_t(\vtheta^{\ast})(\overline{\vtheta}^{v} - \underline{\vtheta}^{v})\big|\notag\\
	%&\leq K V_1(\mZ_t)\Vert\overline{\vtheta}^{v} - \underline{\vtheta}^{v}\Vert\notag\\
	%&\leq KV_1(\mZ_t)\delta,\label{eq:p.13}
	%\end{align}
	%where the third step follows from Assumption~\ref{ass:3}~\ref{it:3ii}, the fourth step from the MVT (with $\vtheta^{\ast}$ some value on the line connecting $\underline{\vtheta}^{v}$ and $\overline{\vtheta}^{v}$), and the fifth step from Assumption~\ref{ass:4}~\ref{it:4iii}. 
	Using similar arguments to those used to obtain \eqref{eq:B3t}, we get that
	\begin{align}
		\bigg\Vert \sup_{\Vert\widetilde{\vtheta} - \vtheta\Vert\leq\delta} C_{3t}\bigg\Vert_{2r} &= \bigg\Vert \sup_{\Vert\widetilde{\vtheta} - \vtheta\Vert\leq\delta}\Big|\1_{\{X_t> v_t(\vtheta^v)\}} - \1_{\{X_t> v_t(\widetilde{\vtheta}^v)\}}\Big|m_t^{2}(\widetilde{\vtheta}^m)\bigg\Vert_{2r} \notag\\
		&\leq \big\Vert \1_{\{X_t> v_t(\overline{\vtheta}^v)\}} - \1_{\{X_t> v_t(\underline{\vtheta}^v)\}} \big\Vert_{p} \big\Vert M^2(\mZ_t) \big\Vert_{q}\notag\\
		&\leq C\delta^{1/p}\label{eq:C3t}
	\end{align}
	for the above $p$ and $q$.
	Combining \eqref{eq:A3t}--\eqref{eq:C3t} with \eqref{eq:(11.1)} implies that \eqref{eq:(11.0)} holds with $\psi=1/p$.
\end{proof}

\subsection{Asymptotic Normality of the VaR Parameter Estimator}\label{Asymptotic Normality of the VaR Parameter Estimator}

In proving asymptotic normality of $\widehat{\vtheta}_n^v$, we use standard arguments from the extremum estimation literature. We present it here nonetheless since some of the following results are needed in establishing the joint asymptotic normality of $(\widehat{\vtheta}_n^{v\prime}, \widehat{\vtheta}_n^{m\prime})^\prime$ in Sections~\ref{Asymptotic Normality of the MES Parameter Estimator}--\ref{Joint Asymptotic Normality of the VaR and MES Parameter Estimator}. We mention that the results given here seem to require somewhat more primitive conditions than those given in the related literature on nonlinear quantile regression under weak dependence \citep{OH16}.

Before giving the formal proof, we collect some results that will be used in the sequel. 
Define
\begin{equation}\label{eq:gt}
	\Vv_t(\vtheta^v) := \nabla v_t(\vtheta^v)\big[\1_{\{X_t\leq v_t(\vtheta^v)\}}-\b\big].
\end{equation}
This quantity may be interpreted as the derivative of the VaR score with respect to the model parameters, i.e., $\frac{\partial}{\partial \vtheta^v}S^{\VaR}\big(v_t(\vtheta^v),X_t\big)$.
To see this, note that for $v\neq x$ it holds that 
\begin{equation*}
	\frac{\partial}{\partial v}S^{\VaR}(v,x) = \1_{\{x\leq v\}}-\b.
\end{equation*}
Thus, by the chain rule, we have for $v_t(\vtheta^v)\neq X_t$ that
\begin{equation*}
	\frac{\partial}{\partial \vtheta^v}S^{\VaR}\big(v_t(\vtheta^v),X_t\big)
	= \nabla v_t(\vtheta^v)\big[\1_{\{X_t\leq v_t(\vtheta^v)\}}-\b\big]=\Vv_t(\vtheta^v).
\end{equation*}
Note that $v_t(\vtheta^v)\neq X_t$ holds a.s., because, as $\varepsilon\downarrow0$,
\begin{align*}
	0 &\leq \P\big\{X_t=v_t(\vtheta^v)\big\} \\
	&\leq \P\big\{|X_t-v_t(\vtheta^v)|\leq\varepsilon\big\}\\
	&= \P\big\{v_t(\vtheta^v) - \varepsilon\leq X_t\leq v_t(\vtheta^v) + \varepsilon\big\}\\
	&=\E\big[\1_{\{v_t(\vtheta^v) - \varepsilon\leq X_t\leq v_t(\vtheta^v) + \varepsilon\}}\big]\\
	&=\E\Big[\E_{t}\big\{\1_{\{v_t(\vtheta^v) - \varepsilon\leq X_t\leq v_t(\vtheta^v) + \varepsilon\}}\big\}\Big]\\
	&=\E\Big[\P_{t}\big\{v_t(\vtheta^v) - \varepsilon\leq X_t\leq v_t(\vtheta^v) + \varepsilon\big\}\Big]\\
	&= \E\bigg[\int_{v_t(\vtheta^v) - \varepsilon}^{v_t(\vtheta^v) + \varepsilon}f_t^{X}(x)\D x\bigg]\\
	&\leq \E[2K\varepsilon]\longrightarrow0,
\end{align*}
where the final inequality follows from Assumption~\ref{ass:3}~\ref{it:3ii}.
By \eqref{eq:gt}, we have
\begin{align}
	\E\big[\Vv_{t}(\vtheta^v)\big] &=\E\bigg\{\nabla v_t(\vtheta^v)\E_{t}\Big[\1_{\{X_t\leq v_t(\vtheta^v)\}}-\b\Big]\bigg\}\notag\\
	&= \E\bigg\{\nabla v_t(\vtheta^v)\Big[F_t^{X}\big(v_t(\vtheta^v)\big)-\b\Big]\bigg\}.\label{eq:(A.18+)}
\end{align}
Finally, Assumptions~\ref{ass:4}--\ref{ass:5} and the DCT allow us to interchange differentiation and expectation to yield that
\begin{equation}
	\frac{\partial}{\partial \vtheta^v}\E\big[\Vv_{t}(\vtheta^v)\big] =\E\bigg\{\nabla^2 v_t(\vtheta^v)\Big[F_t^{X}\big(v_t(\vtheta^v)\big)-\b\Big] + \nabla v_t(\vtheta^v)\nabla^\prime v_t(\vtheta^v)f_t^{X}\big(v_t(\vtheta^v)\big)\bigg\}.\label{eq:Lambda v}
\end{equation}
Evaluating this quantity at the true parameters gives
\begin{align}
	\mLambda &= \frac{\partial}{\partial \vtheta^v}\E\big[\Vv_{t}(\vtheta^v)\big]\Big\vert_{\vtheta^v=\vtheta_0^v} =\E\Big\{\nabla v_t(\vtheta_0^v)\nabla^\prime v_t(\vtheta_0^v)f_t^{X}\big(v_t(\vtheta_0^v)\big)\Big\},
\end{align}
since $F_t^{X}\big(v_t(\vtheta_0^v)\big)=\b$ by correct specification.
By virtue of Assumption~\ref{ass:2}, $\mLambda$ does not depend on $t$.

Similarly as in \citet{DLS23}, the key step in the proof is to apply the FCLT of \citet{DMR95}. Define
\begin{align*}
	\widehat{\vlambda}_n(\vtheta^v) &= \frac{1}{n}\sum_{t=1}^{n}\Vv_{t}(\vtheta^v),\\
	\vlambda(\vtheta^v) &= \E\big[\Vv_{t}(\vtheta^v)\big],\\
	\mLambda(\vtheta^\ast) &= \frac{\partial}{\partial\vtheta^v}\E\big[\Vv_{t}(\vtheta^v)\big]\Big\vert_{\vtheta^v=\vtheta^\ast},
\end{align*}
and note that $\mLambda=\mLambda(\vtheta_0^v)$.
Clearly, $\vlambda(\vtheta^v)$ and $\mLambda(\vtheta^\ast)$ do not depend on $t$ by Assumption~\ref{ass:2}.

We first give a general outline of the proof of the asymptotic normality of $\widehat{\vtheta}_n^v$, filling in some details in the subsequent Lemmas~V.\ref{lem:1+}--V.\ref{lem:3} in Section~\ref{Supplementary Results for VaR Parameter Estimator}.

\begin{proof}[{\textbf{Proof of Theorem~\ref{thm:an} (Asymptotic normality of $\widehat{\vtheta}_n^{v}$):}}]
	In this first part of the proof we have to show that, as $n\to\infty$,
	\begin{equation}\label{eq:AsNorVaR}
		\sqrt{n}\big(\widehat{\vtheta}_n^{v}-\vtheta_0^v\big)\overset{d}{\longrightarrow}N\big(\vzeros,\mLambda^{-1}\mV \mLambda^{-1}\big).
	\end{equation}
	Under the stated assumptions, the MVT implies that for all $i=1,\ldots,p$,
	\[
	\lambda^{(i)}(\widehat{\vtheta}_n^{v}) = \lambda^{(i)}(\vtheta_0^{v}) + \mLambda^{(i,\cdot)}(\vtheta^{\ast}_{i}) (\widehat{\vtheta}_n^v- \vtheta_0^v),
	\]
	where $\vlambda(\cdot)=\big(\lambda^{(1)}(\cdot),\ldots,\lambda^{(p)}(\cdot)\big)^\prime$, $\mLambda^{(i,\cdot)}(\cdot)$ denotes the $i$-th row of $\mLambda(\cdot)$ and $\vtheta_{i}^{\ast}$ lies on the line connecting $\vtheta_0^v$ and $\widehat{\vtheta}_n^v$. To economize on notation, we shall slightly abuse notation (here and elsewhere) by writing this as
	\begin{equation}\label{eq:NE MVT}
		\vlambda(\widehat{\vtheta}_n^{v}) = \vlambda(\vtheta_0^{v}) + \mLambda(\vtheta^{\ast}) (\widehat{\vtheta}_n^v- \vtheta_0^v)
	\end{equation}
	for some $\vtheta^{\ast}$ between $\widehat{\vtheta}_n^{v}$ and $\vtheta_0^{v}$.
	In doing so, keep in mind that the value of $\vtheta^{\ast}$ is in fact different from row to row in $\mLambda(\vtheta^{\ast})$. 
	However, this does not change any of the subsequent arguments.
	
	We have by \eqref{eq:(A.18+)} that
	\begin{equation}\label{eq:lambda nought}
		\vlambda(\vtheta_0^{v}) = \E\bigg\{\nabla v_t(\vtheta_0^{v})\Big[F_t^{X}\big(v_t(\vtheta_0^{v})\big)-\b\Big]\bigg\} =\vzeros,
	\end{equation}
	since $F_t^{X}\big(v_t(\vtheta_0^{v})\big)=\b$ by correct specification. Plugging this into \eqref{eq:NE MVT} gives
	\begin{equation*}
		\vlambda(\widehat{\vtheta}_n^{v}) = \mLambda(\vtheta^{\ast}) (\widehat{\vtheta}_n^v- \vtheta_0^v).
	\end{equation*}
	Expand the left-hand side of the above equation to obtain that
	\begin{equation}\label{eq:exp VaR}
		\widehat{\vlambda}_n(\widehat{\vtheta}_n^{v}) - \Big\{\big[\widehat{\vlambda}_n(\widehat{\vtheta}_n^{v}) - \vlambda(\widehat{\vtheta}_n^{v})\big] - \big[\widehat{\vlambda}_n(\vtheta_0^{v}) - \vlambda(\vtheta_0^{v})\big]\Big\} - \big[\widehat{\vlambda}_n(\vtheta_0^{v}) - \vlambda(\vtheta_0^{v})\big] = \mLambda(\vtheta^{\ast}) (\widehat{\vtheta}_n^v- \vtheta_0^v).
	\end{equation}
	To establish asymptotic normality of the VaR parameter estimator in \eqref{eq:AsNorVaR}, we therefore have to show that
	\begin{enumerate}
		\item[(i)] $\mLambda^{-1}(\vtheta^{\ast})=\mLambda^{-1} + o_{\P}(1)$;
		\item[(ii)] $\sqrt{n}\widehat{\vlambda}_n(\widehat{\vtheta}_n^{v})=o_{\P}(1)$;
		\item[(iii)] $\sqrt{n}\Big\{\big[\widehat{\vlambda}_n(\widehat{\vtheta}_n^{v}) - \vlambda(\widehat{\vtheta}_n^{v})\big] - \big[\widehat{\vlambda}_n(\vtheta_0^{v}) - \vlambda(\vtheta_0^{v})\big]\Big\}=o_{\P}(1)$;
		\item[(iv)] $\sqrt{n}\big[\widehat{\vlambda}_n(\vtheta_0^{v}) - \vlambda(\vtheta_0^{v})\big]\overset{d}{\longrightarrow}N(\vzeros, \mV)$, as $n\to\infty$.
	\end{enumerate}

	Claim (i) is verified in Lemma~V.\ref{lem:1+}. For this, note that since $\vtheta^{\ast}$ is a mean value between $\widehat{\vtheta}_n^v$ and $\vtheta_0^v$, and $\widehat{\vtheta}_n^v\overset{\P}{\longrightarrow}\vtheta_0^v$ (from Proposition~\ref{prop:cons}), it also follows that $\vtheta^{\ast}\overset{\P}{\longrightarrow}\vtheta_0^v$.
	
	Claim (ii) is verified in Lemma~V.\ref{lem:2}
	
	Finally, claims (iii) and (iv) follow from the FCLT of Lemma~V.\ref{lem:3}. Specifically, (iii) can be derived as follows. Since $\P\big\{\Vert\widehat{\vtheta}_n^v - \vtheta_0^v\Vert\leq\delta\big\}\underset{(n\to\infty)}{\longrightarrow}1$ for any $\delta>0$ due to Proposition~\ref{prop:cons}, we have w.p.a.~1,
	\begin{align}
		&\bigg\Vert\sqrt{n}\Big\{\big[\widehat{\vlambda}_n(\widehat{\vtheta}_n^{v}) - \vlambda(\widehat{\vtheta}_n^{v})\big] - \big[\widehat{\vlambda}_n(\vtheta_0^{v}) - \vlambda(\vtheta_0^{v})\big]\Big\}\bigg\Vert \notag\\
		&\leq \sup_{\Vert\vtheta^v-\vtheta_0^v\Vert\leq\delta} \bigg\Vert \sqrt{n}\Big\{\big[\widehat{\vlambda}_n(\vtheta^{v}) - \vlambda(\vtheta^{v})\big] -\big[\widehat{\vlambda}_n(\vtheta_0^{v}) - \vlambda(\vtheta_0^{v})\big]\Big\}\bigg\Vert\notag\\
		&\leq \sup_{\Vert\vtheta^v-\vtheta_0^v\Vert\leq\delta} \Big\Vert \sqrt{n}\big[\widehat{\vlambda}_n(\vtheta^{v}) - \vlambda(\vtheta^{v})\big] - \mB^{(n)}(\vtheta^v)\Big\Vert\notag\\
		& \hspace{2cm} + \sup_{\Vert\vtheta^v-\vtheta_0^v\Vert\leq\delta} \Big\Vert\sqrt{n}\big[\widehat{\vlambda}_n(\vtheta_0^{v}) - \vlambda(\vtheta_0^{v})\big] - \mB^{(n)}(\vtheta_0^v) \Big\Vert\notag\\
		&\hspace{2cm} + \sup_{\Vert\vtheta^v-\vtheta_0^v\Vert\leq\delta} \big\Vert\mB^{(n)}(\vtheta^v) - \mB^{(n)}(\vtheta_0^v)\big\Vert\notag\\
		&=o_{\P}(1) + o_{\P}(1) + \sup_{\Vert\vtheta^v-\vtheta_0^v\Vert\leq\delta} \big\Vert\mB^{(n)}(\vtheta^v) - \mB^{(n)}(\vtheta_0^v)\big\Vert,\label{eq:(p.16)}
	\end{align}
	where the final step follows from Lemma~V.\ref{lem:3}. By (a.s.) continuity of the sample paths of $\mB^{(n)}(\cdot)$ and the fact that $\delta>0$ can be chosen arbitrarily small, it also follows that $\sup_{\Vert\vtheta^v-\vtheta_0^v\Vert\leq\delta} \big\Vert\mB^{(n)}(\vtheta^v) - \mB^{(n)}(\vtheta_0^v)\big\Vert=o_{\P}(1)$, as $\delta\downarrow0$. Hence, (iii) follows.
	
	Claim (iv) directly follows from Lemma~V.\ref{lem:3}, which implies that
	\[
	\sqrt{n}\big[\widehat{\vlambda}_n(\vtheta_0^{v}) - \vlambda(\vtheta_0^{v})\big]-\mB^{(n)}(\vtheta_0^v)=o_{\P}(1),
	\]
	where the variance-covariance matrix of $\mB^{(n)}(\vtheta_0^v)$ is given by
	\[
	\mGamma(\vtheta_0^v,\vtheta_0^v)=\E\big[\Vv_1(\vtheta_0^v)\Vv_1^\prime(\vtheta_0^v)\big] + \sum_{j=1}^{\infty}\Big\{\E\big[\Vv_1(\vtheta_0^v)\Vv_{1+j}^\prime(\vtheta_0^v)\big] + \E\big[\Vv_{1+j}(\vtheta_0^v)\Vv_1^\prime(\vtheta_0^v)\big]\Big\}.
	\]
	Now, from \eqref{eq:gt} and the LIE
	\begin{align}
		\E\big[\Vv_1(\vtheta_0^v)\Vv_1^\prime(\vtheta_0^v)\big] &= \E\Big[\nabla v_1(\vtheta_0^v)\nabla^\prime v_1(\vtheta_0^v)\big(\1_{\{X_1\leq v_1(\vtheta^v)\}}-\b\big)^2\Big] \notag\\
		&= \E\Big[\nabla v_1(\vtheta_0^v)\nabla^\prime v_1(\vtheta_0^v)\E_1\big(\1_{\{X_1\leq v_1(\vtheta_0^v)\}}-2\b\1_{\{X_1\leq v_1(\vtheta_0^v)\}} +\b^2\big)\Big] \notag\\
		&= \E\Big[\nabla v_1(\vtheta_0^v)\nabla^\prime v_1(\vtheta_0^v)\big(\b-2\b^2 +\b^2\big)\Big] \notag\\
		&= \b(1-\b)\E\big[\nabla v_1(\vtheta_0^v)\nabla^\prime v_1(\vtheta_0^v)\big]=\mV\label{eq:V lrv 1}
	\end{align}
	and
	\begin{align}
		\E\big[\Vv_1(\vtheta_0^v)\Vv_{1+j}^\prime(\vtheta_0^v)\big] &= \E\Big\{\Vv_1(\vtheta_0^v)\E_{1+j}\big[\Vv_{1+j}^\prime(\vtheta_0^v)\big]\Big\}\notag\\
		&=\E\Big\{\Vv_1(\vtheta_0^v)\vzeros\Big\}=\vzeros,\label{eq:V lrv 2}
	\end{align}
	such that $\mGamma(\vtheta_0^v, \vtheta_0^v) = \mV$.
	
	In sum, the desired result follows.
\end{proof}

\subsubsection{Supplementary Results for VaR Parameter Estimator}\label{Supplementary Results for VaR Parameter Estimator}

\begin{lemV}\label{lem:1+}
	Suppose Assumptions~\ref{ass:1}--\ref{ass:7} hold. Then, as $n\to\infty$, $\mLambda^{-1}(\vtheta_n^v)\overset{\P}{\longrightarrow}\mLambda^{-1}$ for any $\vtheta_n^v$ with $\vtheta_n^v\overset{\P}{\longrightarrow}\vtheta_0^v$.
\end{lemV}

\begin{proof}
	We first show that $\norm{\mLambda(\vtau) - \mLambda(\vtheta)} \leq C \norm{\vtau - \vtheta}$ for all $\vtau,\vtheta\in\mathcal{N}(\vtheta_0^v)$, where $\mathcal{N}(\vtheta_0^v)$ is some neighborhood of $\vtheta_0^v$ such that Assumption~\ref{ass:4}~\ref{it:4v} holds.
	%For (i), note that because of the continuity of $v_t(\vtheta^v)$, $\nabla v_t(\vtheta^v)$ and $f_t^{X}(\cdot)$ (see Assumption~\ref{ass:an} \ref{it:diff} and \ref{it:Lipschitz}), $\mLambda_n(\vtheta^v)$ is continuous in $\vtheta^v$. Then, for any $\vtheta^v\in\mTheta^v$ it follows that
	%\begin{align*}
	%\lim_{\vtheta\to\vtheta^v}\mLambda(\vtheta) & =\lim_{\vtheta\to\vtheta^v}\lim_{n\to\infty}\mLambda_{n}(\vtheta)\\
	%&= \lim_{n\to\infty}\lim_{\vtheta\to\vtheta^v}\mLambda_{n}(\vtheta)\\
	%&= \lim_{n\to\infty}\mLambda_{n}(\vtheta^v)\\
	%&= \mLambda(\vtheta^v),
	%\end{align*}
	%where the interchange of the limits in the second step is allowed due to the dominated convergence theorem, and the third step follows from continuity of $\mLambda_n(\cdot)$. Therefore, continuity of $\mLambda(\cdot)$ follows.
	%
	%Item (ii) is immediate from the definition of $\mLambda(\vtheta^v)$.
	Use \eqref{eq:Lambda v} to write
	\begin{align}
		\big\Vert& \mLambda(\vtau) - \mLambda(\vtheta)\big\Vert\notag\\
		&= \bigg\Vert \E\Big\{\nabla^2 v_t(\vtau)\Big[F_t^{X}\big(v_t(\vtau)\big)-\b\Big] - \nabla^2 v_t(\vtheta)\Big[F_t^{X}\big(v_t(\vtheta)\big)-\b\Big]\notag\\
		&	\hspace{0.5cm} + \nabla v_t(\vtau)\nabla^\prime v_t(\vtau)f_t^{X}\big(v_t(\vtau)\big) - \nabla v_t(\vtheta)\nabla^\prime v_t(\vtheta)f_t^{X}\big(v_t(\vtheta)\big)\Big\}\bigg\Vert\notag\\
		&\leq  \E\Big\Vert\nabla^2 v_t(\vtau)\Big[F_t^{X}\big(v_t(\vtau)\big)-F_t^{X}\big(v_t(\vtheta)\big)\Big]\Big\Vert + \E\Big\Vert F_t^{X}\big(v_t(\vtheta)\big)\big[\nabla^2 v_t(\vtau) - \nabla^2 v_t(\vtheta)\big]\Big\Vert \notag\\
		& \hspace{0.5cm} +\b \E\Big\Vert \nabla^2 v_t(\vtau) - \nabla^2 v_t(\vtheta)\Big\Vert + \E\Big\Vert\nabla v_t(\vtau)\nabla^\prime v_t(\vtau)f_t^{X}\big(v_t(\vtau)\big)- \nabla v_t(\vtheta)\nabla^\prime v_t(\vtheta)f_t^{X}\big(v_t(\vtheta)\big)\Big\Vert\notag\\
		&=:A_{4} + B_{4} + C_{4} + D_{4}.\label{eq:decomp Lambda}
	\end{align}
	
	A mean value expansion around $\vtheta_0^v$ and Assumptions~\ref{ass:4}--\ref{ass:5} imply that
	\begin{align}
		A_{4} &=\E\norm{\nabla^2 v_t(\vtau)\Big[ F_t^{X}\big(v_t(\vtau)\big)-F_t^{X}\big(v_t(\vtheta)\big)\Big]} \\
		&\leq \E\norm{V_2(\mZ_t)f_t^{X}\big(v_t(\vtheta^{\ast})\big)\nabla^\prime v_t(\vtheta^{\ast})(\vtau-\vtheta)}\notag\\
		&\leq K \E\big[V_1(\mZ_t)V_2(\mZ_t)\big]\norm{\vtau-\vtheta}\notag\\
		&\leq K \big\Vert V_1(\mZ_t)\big\Vert_{2} \big\Vert V_2(\mZ_t)\big\Vert_{2}\norm{\vtau-\vtheta}\notag\\
		&\leq C \norm{\vtau-\vtheta},\label{eq:first term}
	\end{align}
	where $\vtheta^{\ast}$ is some value on the line connecting $\vtau$ and $\vtheta$, and the penultimate step uses H\"{o}lder's inequality.
	
	By Assumption~\ref{ass:4}~\ref{it:4v}, we have for $B_{4}$ and $C_{4}$ that
	\begin{align}
		B_{4} &= \E\Big\Vert F_t^{X}\big(v_t(\vtheta)\big)\big[\nabla^2 v_t(\vtau) - \nabla^2 v_t(\vtheta)\big]\Big\Vert \leq \E\big\Vert \nabla^2 v_t(\vtau) - \nabla^2 v_t(\vtheta)\big\Vert\leq C \big\Vert \vtau-\vtheta\big\Vert,\label{eq:(A.28)}\\
		C_{4} &=	\b \E\big\Vert \nabla^2 v_t(\vtau) - \nabla^2 v_t(\vtheta)\big\Vert \leq \b \E\big[V_3(\mZ_t)\big]\norm{\vtau - \vtheta}\leq C \norm{\vtau - \vtheta}.\label{eq:nought term}
	\end{align}
	
	For $D_{4}$, we use a mean value expansion around $\vtheta$ to obtain for some $\vtheta^{\ast}$ between $\vtau$ and $\vtheta$ (where $\vtheta^{\ast}$ may vary from line to line) that
	\begin{align}
		D_{4}&=\E\Big\Vert\nabla v_t(\vtau)\nabla^\prime v_t(\vtau)f_t^{X}\big(v_t(\vtau)\big) - \nabla v_t(\vtheta)\nabla^\prime v_t(\vtheta)f_t^{X}\big(v_t(\vtheta)\big)\Big\Vert\notag\\
		&=\E\Big\Vert\nabla v_t(\vtau)\nabla^\prime v_t(\vtau)f_t^{X}\big(v_t(\vtau)\big) - \nabla v_t(\vtheta)\nabla^\prime v_t(\vtau)f_t^{X}\big(v_t(\vtau)\big)\notag\\
		&\hspace{1.8cm} + \nabla v_t(\vtheta)\nabla^\prime v_t(\vtau)f_t^{X}\big(v_t(\vtau)\big) - \nabla v_t(\vtheta)\nabla^\prime v_t(\vtheta)f_t^{X}\big(v_t(\vtau)\big)\notag\\
		&\hspace{1.8cm} + \nabla v_t(\vtheta)\nabla^\prime v_t(\vtheta)f_t^{X}\big(v_t(\vtau)\big) - \nabla v_t(\vtheta)\nabla^\prime v_t(\vtheta)f_t^{X}\big(v_t(\vtheta)\big) \Big\Vert\notag\\
		&= \E\Big\Vert\nabla^2 v_t(\vtheta^{\ast})(\vtau - \vtheta)\nabla^\prime v_t(\vtau)f_t^{X}\big(v_t(\vtau)\big)\notag\\
		&\hspace{1.8cm} + \nabla v_t(\vtheta)(\vtau - \vtheta)^\prime [\nabla^2 v_t(\vtheta^{\ast})]^\prime f_t^{X}\big(v_t(\vtau)\big) \notag\\
		&\hspace{1.8cm} + \nabla v_t(\vtheta)\nabla^\prime v_t(\vtheta)\big\{f_t^{X}\big(v_t(\vtau)\big) - f_t^{X}\big(v_t(\vtheta)\big)\big\} \Big\Vert\notag\\
		&\leq \E\big[KV_2(\mZ_t)V_1(\mZ_t) + KV_1(\mZ_t)V_2(\mZ_t) + K V_1^3(\mZ_t) \big]\norm{\vtau - \vtheta}\notag\\
		&\leq K\Big\{2\big\Vert V_1(\mZ_t) \big\Vert_2\big\Vert V_2(\mZ_t) \big\Vert_{2}  + \E\big[V_1^3(\mZ_t)\big]\Big\}\norm{\vtau - \vtheta}\notag\\
		&\leq C\norm{\vtau - \vtheta},\label{eq:second term}
	\end{align}
	where we used Assumptions~\ref{ass:4}--\ref{ass:5}.
	
	Plugging \eqref{eq:first term}--\eqref{eq:second term} into \eqref{eq:decomp Lambda} yields for all $\vtau,\vtheta\in\mathcal{N}(\vtheta_0^v)$ that
	\begin{equation}\label{eq:cty Lambda_n}
		\norm{\mLambda(\vtau) - \mLambda(\vtheta)} \leq C \norm{\vtau - \vtheta}.
	\end{equation}
	Using this, we obtain for $\delta>0$ with $\big\{\vtheta\in\mTheta^v\colon\norm{\vtheta-\vtheta_0^v}\leq\delta\big\}\subset\mathcal{N}(\vtheta_0^v)$ that
	\begin{align*}
		\P\Big\{\big\Vert\mLambda(\vtheta_n^v) -\mLambda \big\Vert >\varepsilon\Big\}&=\P\Big\{\big\Vert\mLambda(\vtheta_n^v) -\mLambda(\vtheta_0^v)\big\Vert >\varepsilon\Big\}\\
		&\leq \P\Big\{\big\Vert\mLambda(\vtheta_n^v) -\mLambda(\vtheta_0^v)\big\Vert >\varepsilon,\ \big\Vert\vtheta_n^v - \vtheta_0^v\big\Vert\leq\delta\Big\} + \P\Big\{\big\Vert\vtheta_n^v - \vtheta_0^v\big\Vert>\delta\Big\}\\
		&\leq \P\bigg\{\sup_{\Vert\vtheta-\vtheta_0^v\Vert\leq\delta}\big\Vert\mLambda(\vtheta) -\mLambda(\vtheta_0^v)\big\Vert >\varepsilon\bigg\} + o(1)\\
		&\leq \P\bigg\{\sup_{\Vert\vtheta-\vtheta_0^v\Vert\leq\delta}C\big\Vert\vtheta -\vtheta_0^v\big\Vert >\varepsilon\bigg\} + o(1)\\
		&=0+o(1),
	\end{align*}
	where the final line additionally requires $\delta<\varepsilon/C$. This shows that $\mLambda(\vtheta_n^v)\overset{\P}{\longrightarrow}\mLambda$. By Assumption~\ref{ass:6} and continuity of $\mLambda(\cdot)$ (from \eqref{eq:cty Lambda_n}), $\mLambda(\vtheta^v)$ is non-singular in a neighborhood of $\vtheta_0^v$. Therefore, the continuous mapping theorem (CMT) applied to $\mLambda(\vtheta_n^v)\overset{\P}{\longrightarrow}\mLambda$ implies that $\mLambda^{-1}(\vtheta_n^v)\overset{\P}{\longrightarrow}\mLambda^{-1}$, as desired.%; see also Lemma~A.1 in \citet{Woo94}.
\end{proof}

\begin{lemV}\label{lem:2}
	Suppose Assumptions~\ref{ass:1}--\ref{ass:7} hold. Then, as $n\to\infty$,
	\[
	\sqrt{n}\widehat{\vlambda}_n(\widehat{\vtheta}_n^v)=o_{\P}(1).
	\]
\end{lemV}

\begin{proof}
	We adopt the proof strategy of \citet[Proof of Lemma~A.2]{RC80}, which was subsequently used by, e.g., \citet[Proof of Theorem~2]{EM04} and \citet[Proof of Lemma~2]{PZC19}.
	
	Recall from Assumption~\ref{ass:1}~\ref{it:1ii} that $\mTheta^{v}\subset\mathbb{R}^{p}$, such that $\vtheta^v$ is a $p$-dimensional parameter vector. Let $\ve_1,\ldots,\ve_p$ denote the standard basis of $\mathbb{R}^{p}$ and define for $a\in\mathbb{R}$
	\[
	S_{j,n}^{\VaR}(a) := \frac{1}{\sqrt{n}}\sum_{t=1}^{n}S^{\VaR}\big(v_t(\widehat{\vtheta}_n^{v}+a \ve_j), X_t\big),\qquad j=1,\ldots,p.
	\]
	Let $G_{j,n}(a)$ be the right partial derivative of $S_{j,n}^{\VaR}(a)$, such that (see \eqref{eq:gt})
	\[
	G_{j,n}(a)=\frac{1}{\sqrt{n}}\sum_{t=1}^{n}\nabla_j v_t(\widehat{\vtheta}_n^{v}+a \ve_j)\big[\1_{\{X_t\leq v_t(\widehat{\vtheta}_n^{v}+a \ve_j)\}}-\b\big],
	\]
	where $\nabla_j v_t(\cdot)$ is the $j$-th component of $\nabla v_t(\cdot)$. Then, $G_{j,n}(0)=\lim_{\xi\downarrow0}G_{j,n}(\xi)$ is the right partial derivative of
	\[
	S_n^{\VaR}(\vtheta^{v}) := \frac{1}{\sqrt{n}}\sum_{t=1}^{n}S^{\VaR}\big(v_t(\vtheta^v), X_t\big)
	\]
	at $\widehat{\vtheta}_n^{v}$ in the direction $\theta_j^{v}$, where $\vtheta^{v}=(\theta_1^{v},\ldots,\theta_p^{v})^\prime$. 
	Correspondingly, $\lim_{\xi\downarrow0}G_{j,n}(-\xi)$ is the left partial derivative. 
	Because $S_n^{\VaR}(\cdot)$ achieves its minimum at $\widehat{\vtheta}_n^{v}$, the left derivative must be non-positive and the right derivative must be non-negative, such that for any sufficiently small $\xi>0$,
	\[
	G_{j,n}(-\xi) \leq G_{j,n}(0)\leq G_{j,n}(\xi).
	\]
	Thus,
	\begin{align}
		\big|G_{j,n}(0)\big| &\leq G_{j,n}(\xi) - G_{j,n}(-\xi)\notag\\
		&= \frac{1}{\sqrt{n}}\sum_{t=1}^{n}\nabla_j v_t(\widehat{\vtheta}_n^{v}+\xi \ve_j)\big[\1_{\{X_t\leq v_t(\widehat{\vtheta}_n^{v}+\xi \ve_j)\}}-\b\big]\notag\\
		&\hspace{2cm} - \frac{1}{\sqrt{n}}\sum_{t=1}^{n}\nabla_j v_t(\widehat{\vtheta}_n^{v}-\xi \ve_j)\big[\1_{\{X_t\leq v_t(\widehat{\vtheta}_n^{v}-\xi \ve_j)\}}-\b\big]\notag\\
		&=\frac{1}{\sqrt{n}}\sum_{t=1}^{n}\Big[\nabla_j v_t(\widehat{\vtheta}_n^{v}+\xi \ve_j)\1_{\{X_t\leq v_t(\widehat{\vtheta}_n^{v}+\xi \ve_j)\}}-\nabla_j v_t(\widehat{\vtheta}_n^{v}-\xi \ve_j)\1_{\{X_t\leq v_t(\widehat{\vtheta}_n^{v}-\xi \ve_j)\}}\Big]\label{eq:(A.33m)}\\
		&\hspace{2cm} - \frac{\b}{\sqrt{n}}\sum_{t=1}^{n}\big[\nabla_j v_t(\widehat{\vtheta}_n^{v}+\xi \ve_j) - \nabla_j v_t(\widehat{\vtheta}_n^{v}-\xi \ve_j)\big].\label{eq:second sum}
	\end{align}
	By continuity of $\nabla v_t(\cdot)$ (see Assumption~\ref{ass:4} \ref{it:4i}) it follows for \eqref{eq:second sum} that a.s.
	\begin{equation}\label{eq:fi re}
		\lim_{\xi\downarrow0}\frac{\b}{\sqrt{n}}\sum_{t=1}^{n}\big[\nabla_j v_t(\widehat{\vtheta}_n^{v}+\xi \ve_j) - \nabla_j v_t(\widehat{\vtheta}_n^{v}-\xi \ve_j)\big]=0.
	\end{equation}
	For the term in \eqref{eq:(A.33m)}, we consider the three cases $X_t<v_t(\widehat{\vtheta}_n^v)$, $X_t=v_t(\widehat{\vtheta}_n^v)$ and $X_t>v_t(\widehat{\vtheta}_n^v)$ for each $t=1,\ldots,n$.
	In case $X_t<v_t(\widehat{\vtheta}_n^v)$ both indicators equal one for sufficiently small $\xi>0$ (by continuity of $v_t(\cdot)$), such that
	\begin{multline}\label{eq:1f}
		\Big[\nabla_j v_t(\widehat{\vtheta}_n^{v}+\xi \ve_j)\1_{\{X_t\leq v_t(\widehat{\vtheta}_n^{v}+\xi \ve_j)\}}-\nabla_j v_t(\widehat{\vtheta}_n^{v}-\xi \ve_j)\1_{\{X_t\leq v_t(\widehat{\vtheta}_n^{v}-\xi \ve_j)\}}\Big]\\
		=\big[\nabla_j v_t(\widehat{\vtheta}_n^{v}+\xi \ve_j)-\nabla_j v_t(\widehat{\vtheta}_n^{v}-\xi \ve_j)\big]\overset{\text{a.s.}}{\underset{(\xi\downarrow0)}{\longrightarrow}}0.
	\end{multline}
	When $X_t>v_t(\widehat{\vtheta}_n^v)$ both indicators equal zero for sufficiently small $\xi>0$, such that
	\begin{equation}\label{eq:2f}
		\Big[\nabla_j v_t(\widehat{\vtheta}_n^{v}+\xi \ve_j)\1_{\{X_t\leq v_t(\widehat{\vtheta}_n^{v}+\xi \ve_j)\}}-\nabla_j v_t(\widehat{\vtheta}_n^{v}-\xi \ve_j)\1_{\{X_t\leq v_t(\widehat{\vtheta}_n^{v}-\xi \ve_j)\}}\Big]\overset{\text{a.s.}}{\underset{(\xi\downarrow0)}{\longrightarrow}}0.
	\end{equation}
	If $X_t=v_t(\widehat{\vtheta}_n^v)$, then it may still happen that both indicators have equal values in the limit as $\xi\downarrow0$ (namely when $v_t(\cdot)$ has a local extremum in $\widehat{\vtheta}_n^v$), in which case the above arguments show that
	\begin{equation}\label{eq:3f}
		\Big[\nabla_j v_t(\widehat{\vtheta}_n^{v}+\xi \ve_j)\1_{\{X_t\leq v_t(\widehat{\vtheta}_n^{v}+\xi \ve_j)\}}-\nabla_j v_t(\widehat{\vtheta}_n^{v}-\xi \ve_j)\1_{\{X_t\leq v_t(\widehat{\vtheta}_n^{v}-\xi \ve_j)\}}\Big]\overset{\text{a.s.}}{\underset{(\xi\downarrow0)}{\longrightarrow}}0.
	\end{equation}
	So it remains to consider the case where the indicators differ from each other in the limit for $\xi\downarrow0$ (i.e., are equal to 0 and 1).
	Then,
	\begin{multline}\label{eq:4f}
		\lim_{\xi\downarrow0}\Big|\nabla_j v_t(\widehat{\vtheta}_n^{v}+\xi \ve_j)\1_{\{X_t\leq v_t(\widehat{\vtheta}_n^{v}+\xi \ve_j)\}}-\nabla_j v_t(\widehat{\vtheta}_n^{v}-\xi \ve_j)\1_{\{X_t\leq v_t(\widehat{\vtheta}_n^{v}-\xi \ve_j)\}}\Big|\\
		\leq \big|\nabla_j v_t(\widehat{\vtheta}_n^{v})\big| \1_{\{X_t= v_t(\widehat{\vtheta}_n^{v})\}}.
	\end{multline}
	Therefore, for the term in \eqref{eq:(A.33m)} it follows from \eqref{eq:1f}--\eqref{eq:4f} that a.s.
	\begin{multline}\label{eq:se re}
		\lim_{\xi\downarrow0}\frac{1}{\sqrt{n}}\sum_{t=1}^{n}\Big[\nabla_j v_t(\widehat{\vtheta}_n^{v}+\xi \ve_j)\1_{\{X_t\leq v_t(\widehat{\vtheta}_n^{v}+\xi \ve_j)\}}-\nabla_j v_t(\widehat{\vtheta}_n^{v}-\xi \ve_j)\1_{\{X_t\leq v_t(\widehat{\vtheta}_n^{v}-\xi \ve_j)\}}\Big]\\
		\leq \frac{1}{\sqrt{n}}\sum_{t=1}^{n}\big|\nabla_j v_t(\widehat{\vtheta}_n^{v})\big|\1_{\{X_t= v_t(\widehat{\vtheta}_n^{v})\}}.
	\end{multline}
	
	Combining \eqref{eq:fi re} and \eqref{eq:se re}, we obtain upon letting $\xi\downarrow0$ that a.s.
	\begin{align}
		\big|G_{j,n}(0)\big| &\leq \frac{1}{\sqrt{n}}\sum_{t=1}^{n}\big|\nabla_j v_t(\widehat{\vtheta}_n^{v})\big|\1_{\{X_t= v_t(\widehat{\vtheta}_n^{v})\}}\notag\\
		&\leq \frac{1}{\sqrt{n}}\Big[\max_{t=1,\ldots,n}V_1(\mZ_t)\Big]\sum_{t=1}^{n}\1_{\{X_t= v_t(\widehat{\vtheta}_n^{v})\}}\notag\\
		&\leq \frac{1}{\sqrt{n}}\Big[\max_{t=1,\ldots,n}V_1(\mZ_t)\Big]\sup_{\vtheta^v\in\mTheta^v}\sum_{t=1}^{n}\1_{\{X_t= v_t(\vtheta^{v})\}}.\label{eq:(N.7.1)}
	\end{align}
	Subadditivity, Markov's inequality and Assumption~\ref{ass:5} imply in turn that
	\begin{align*}
		\P\Big\{n^{-1/2}\max_{t=1,\ldots,n}V_1(\mZ_t)>\varepsilon\Big\} &\leq \sum_{t=1}^{n}\P\Big\{V_1(\mZ_t)>\varepsilon n^{1/2}\Big\}\\
		&\leq \sum_{t=1}^{n}\varepsilon^{-3}n^{-3/2}\E\big[V_1^3(\mZ_t)\big]\\
		& =O(n^{-1/2})\\
		&=o(1),
	\end{align*}
	such that $n^{-1/2}\max_{t=1,\ldots,n}V_1(\mZ_t)=o_{\P}(1)$.
	Plugging this and Assumption~\ref{ass:7} into \eqref{eq:(N.7.1)} yields that
	\[
	\big|G_{j,n}(0)\big| = o_{\P}(1)O(1)=o_{\P}(1).
	\]
	As this holds for every $j=1,\ldots,p$, we get that 
	\[
	\frac{1}{\sqrt{n}}\sum_{t=1}^{n}\Vv_{t}(\widehat{\vtheta}_n^{v})=\frac{1}{\sqrt{n}}\sum_{t=1}^{n}\nabla v_t(\widehat{\vtheta}_n^{v})\big[\1_{\{X_t\leq v_t(\widehat{\vtheta}_n^{v})\}} - \b\big]=o_{\P}(1),
	\]
	which is just the conclusion.
\end{proof}

For the proof of Lemma~V.\ref{lem:3} below, we require the following preliminary result.

\begin{lem}\label{lem:type IV Vv}
	Suppose Assumptions~\ref{ass:1}--\ref{ass:7} hold. Then, $\Vv_t(\cdot)$ forms a type IV class in the sense of \citet[p.~2278]{And94} with index $p=2r$ (where $r>1$ is from Assumption~\ref{ass:2}).
\end{lem}

\begin{proof}
	Write $\Vv_t(\cdot)=\big(\mathsf{v}_t^{(1)}(\cdot), \ldots, \mathsf{v}_t^{(p)}(\cdot)\big)^\prime$.
	To prove the lemma, we have to show that there exist constants $C>0$ and $\psi>0$, such that
	\begin{equation}\label{eq:TSL3}
		\sup_{t\in\mathbb{N}}\bigg\Vert\sup_{\Vert\widetilde{\vtheta}^v - \vtheta^v\Vert\leq\delta}\big|\mathsf{v}_t^{(i)}(\widetilde{\vtheta}^v) - \mathsf{v}_t^{(i)}(\vtheta^v)\big|\bigg\Vert_{2r}\leq C\delta^\psi
	\end{equation}
	for all $i=1,\ldots,p$, all $\vtheta^v\in\mTheta^v$ and all $\delta>0$ in some neighborhood of 0.
	Fix $i\in\{1,\ldots,p\}$, $\vtheta^v\in\mTheta^v$ and $\delta>0$  such that $\{\widetilde{\vtheta}^v\colon \Vert\widetilde{\vtheta}^v-\vtheta^v\Vert\leq\delta\}\subset\mTheta^v$. Then,
	\begin{align*}
		\bigg\Vert & \sup_{\Vert\widetilde{\vtheta}^v - \vtheta^v\Vert\leq\delta}\big|\mathsf{v}_t^{(i)}(\widetilde{\vtheta}^v) - \mathsf{v}_t^{(i)}(\vtheta^v)\big|\bigg\Vert_{2r} \\
		&=  \bigg\Vert\sup_{\Vert\widetilde{\vtheta}^v - \vtheta^v\Vert\leq\delta}\Big|\nabla_i v_t(\widetilde{\vtheta}^v)\big[\1_{\{X_t\leq v_t(\widetilde{\vtheta}^v)\}} - \b\big] - \nabla_i v_t(\vtheta^v)\big[\1_{\{X_t\leq v_t(\vtheta^v)\}} - \b\big]\Big|\bigg\Vert_{2r} \\
		&= \bigg\Vert\sup_{\Vert\widetilde{\vtheta}^v - \vtheta^v\Vert\leq\delta}\Big|\big[\nabla_i v_t(\widetilde{\vtheta}^v) - \nabla_i v_t(\vtheta^v)\big]\big[\1_{\{X_t\leq v_t(\widetilde{\vtheta}^v)\}} - \b\big] \\
		&\hspace{7cm} + \nabla_i v_t(\vtheta^v)\big[\1_{\{X_t\leq v_t(\widetilde{\vtheta}^v)\}} - \1_{\{X_t\leq v_t(\vtheta^v)\}}\big]\Big|\bigg\Vert_{2r} \\
		&\leq \bigg\Vert\sup_{\Vert\widetilde{\vtheta}^v - \vtheta^v\Vert\leq\delta}\Big|\big[\nabla_i v_t(\widetilde{\vtheta}^v) - \nabla_i v_t(\vtheta^v)\big]\big[\1_{\{X_t\leq v_t(\widetilde{\vtheta}^v)\}} - \b\big]\Big|\bigg\Vert_{2r}\\
		&\hspace{4cm} + \bigg\Vert\sup_{\Vert\widetilde{\vtheta}^v - \vtheta^v\Vert\leq\delta}\Big|\nabla_i v_t(\vtheta^v)\big[\1_{\{X_t\leq v_t(\widetilde{\vtheta}^v)\}} - \1_{\{X_t\leq v_t(\vtheta^v)\}}\big]\Big|\bigg\Vert_{2r}\\
		&=:A_{5t} + B_{5t},
	\end{align*}
	where the third step follows from Minkowski's inequality.
	
	Consider the two terms separately. Use the MVT and Assumptions~\ref{ass:4}~\ref{it:4iv} and \ref{ass:5} to write
	\begin{align*}
		A_{5t} &\leq \bigg\Vert\sup_{\Vert\widetilde{\vtheta}^v - \vtheta^v\Vert\leq\delta}\Big|\nabla_{(i,\cdot)}^{2} v_t(\vtheta^{\ast})\big(\widetilde{\vtheta}^v-\vtheta^v\big)\Big|\bigg\Vert_{2r}\\
		& \leq \bigg\Vert\sup_{\vtheta^v\in\mTheta^v}\big|\nabla_{(i,\cdot)}^{2} v_t(\vtheta^v)\big|\bigg\Vert_{2r}\delta\\
		&\leq \big\Vert V_2(\mZ_t)\big\Vert_{2r}\delta\\
		&\leq C\delta,
	\end{align*}
	where $\nabla_{(i,\cdot)}^{2}$ denotes the $i$-th row of the Hessian.
	Furthermore, arguing similarly as for \eqref{eq:C3t},
	\begin{align*}
		B_{5t} &\leq \Big\Vert V_1(\mZ_t)\big[\1_{\{X_t\leq v_t(\overline{\vtheta}^v)\}} - \1_{\{X_t\leq v_t(\underline{\vtheta}^v)\}}\big]\Big\Vert_{2r}\\
		&\leq \big\Vert V_1(\mZ_t) \big\Vert_{4r} \big\Vert \1_{\{X_t\leq v_t(\overline{\vtheta}^v)\}} - \1_{\{X_t\leq v_t(\underline{\vtheta}^v)\}}\big\Vert_{4r}\\
		&\leq C\delta^{1/(4r)}.
	\end{align*}
	Overall, \eqref{eq:TSL3} follows with $\psi=1/(4r)$.
\end{proof}

\begin{lemV}\label{lem:3}
	Suppose Assumptions~\ref{ass:1}--\ref{ass:7} hold. Then, (on a possibly enlarged probability space) there exists a sequence of zero-mean Gaussian processes $\mB^{(n)}(\cdot)$ with a.s.~continuous sample paths and covariance function $\mGamma(\cdot,\cdot)$, such that
	\begin{equation*}
		\sup_{\vtheta^v\in\mTheta^v}\Big\Vert\sqrt{n} \big[\widehat{\vlambda}_n(\vtheta^v) - \vlambda(\vtheta^v)\big]- \mB^{(n)}(\vtheta^v)\Big\Vert=o_{\P}(1),
	\end{equation*}
	where the covariance function satisfies
	\[
	\mGamma(\vtheta^v,\vtheta^v)=\E\big[\Vv_1(\vtheta^v)\Vv_1^\prime(\vtheta^v)\big] + \sum_{j=1}^{\infty}\Big\{\E\big[\Vv_1(\vtheta^v)\Vv_{1+j}^\prime(\vtheta^v)\big] + \E\big[\Vv_{1+j}(\vtheta^v)\Vv_1^\prime(\vtheta^v)\big]\Big\}.
	\]
\end{lemV}

\begin{proof}
	The arguments resemble those used to show \eqref{eq:FCLT Thm 1} in the proof of Proposition~\ref{prop:cons}. Note that $\Vv_t(\cdot)$ has uniformly bounded $2r$-th moments, because
	\begin{equation*}
		\E\bigg[\sup_{\vtheta^v\in\mTheta^v}\big\Vert\Vv_t(\vtheta^v)\big\Vert^{2r}\bigg] \leq \E\big[V_1^{2r}(\mZ_t)\big]\leq C <\infty.
	\end{equation*}
	
	From Lemma~\ref{lem:type IV Vv} and Theorem~5 of \citet{And94} it follows that $\Vv_t(\cdot)$ satisfies ``Ossiander's $L^{2r}$-entropy condition'' (with ``$L^{2r}$-envelope'' given by the supremum $\sup_{\vtheta^v\in\mTheta^v}\Vert\Vv_t(\vtheta^v)\Vert$). 
	
	Moreover, the mixing condition of Theorem~1 in \citet{DMR95} is satisfied due to $\sum_{t=1}^{\infty}t^{1/(r-1)}\beta(t)<\infty$ from Assumption~\ref{ass:2} (and standard mixing inequalities from, e.g., \citet{Bra05}).
	
	Therefore, the desired FCLT follows from Theorem~1 and Application~1 of \citet{DMR95}.
\end{proof}

\subsection{Asymptotic Normality of the MES Parameter Estimator}\label{Asymptotic Normality of the MES Parameter Estimator}

Proving asymptotic normality of $\widehat{\vtheta}_n^m$ requires some further notation and lemmas. To make the analogy to the proof of the asymptotic normality of $\widehat{\vtheta}_n^v$ more explicit, we label the lemmas as Lemma~M.\ref{lem:1+ tilde}--M.\ref{lem:3 tilde}, which play the analogous roles of Lemmas~V.\ref{lem:1+}--V.\ref{lem:3}.

Define
\begin{equation}\label{eq:(12+)}
	\Mm_{t}(\vtheta^m, \vtheta^v):= \1_{\{X_t> v_t(\vtheta^v)\}}\nabla m_t(\vtheta^m)\big[m_t(\vtheta^m)-Y_t\big].
\end{equation}
Similarly as $\Vv_t(\vtheta^v)$ from \eqref{eq:gt} in Section~\ref{Asymptotic Normality of the VaR Parameter Estimator}, this quantity has the interpretation of the MES score with respect to the MES parameters, i.e., $\frac{\partial}{\partial \vtheta^m}S^{\MES}\big((v_t(\vtheta^v), m_t(\vtheta^m))^\prime,(X_t,Y_t)^\prime\big)$.
To see this, observe that
\begin{equation*}
	\frac{\partial}{\partial m}S^{\MES}\big((v,m)^\prime,(x,y)^\prime\big) = \1_{\{x> v\}}\big[m-y\big],
\end{equation*}
such that by the chain rule,
\begin{equation*}
	\frac{\partial}{\partial \vtheta^m}S^{\MES}\big((v_t(\vtheta^v), m_t(\vtheta^m))^\prime,(X_t,Y_t)^\prime\big) = \1_{\{X_t> v_t(\vtheta^v)\}}\nabla m_t(\vtheta^m)\big[m_t(\vtheta^m)-Y_t\big]=\Mm_{t}(\vtheta^m, \vtheta^v).
\end{equation*}
By the LIE, \eqref{eq:(12+)} implies that
\begin{align*}
	\E\big[\Mm_{t}(\vtheta^m, \vtheta^v)\big] &= \E\Big[\1_{\{X_t> v_t(\vtheta^v)\}}\nabla m_t(\vtheta^m)\big\{m_t(\vtheta^m)-Y_t\big\}\Big]\\
	&= \E\bigg[\nabla m_t(\vtheta^m)\Big\{m_t(\vtheta^m)\E_{t}\big[\1_{\{X_t> v_t(\vtheta^v)\}}\big] - \E_{t}\big[\1_{\{X_t> v_t(\vtheta^v)\}}Y_t\big]\Big\} \bigg]\\
	&= \E\bigg[\nabla m_t(\vtheta^m)\Big\{ m_t(\vtheta^m)\overline{F}_{t}^{X}\big(v_t(\vtheta^v)\big) - \int_{v_t(\vtheta^v)}^{\infty}\int_{-\infty}^{\infty} y f_t(x,y)\D y\D x\Big\}\bigg],
\end{align*}
where $\overline{F}_{t}^{X}(\cdot)=1-F_{t}^{X}(\cdot)$ denotes the survivor function.
Finally, Assumptions~\ref{ass:4}--\ref{ass:5} and the DCT allow us to interchange differentiation and expectation to yield that
\begin{align}
	\frac{\partial}{\partial \vtheta^m}\E\big[\Mm_{t}(\vtheta^m, \vtheta^v)\big] &= \E\Big[\big\{\nabla^2 m_t(\vtheta^m)m_t(\vtheta^m) + \nabla m_t(\vtheta^m)\nabla^\prime m_t(\vtheta^m)\big\}\overline{F}_{t}^{X}\big(v_t(\vtheta^v)\big)\Big] \notag\\
	& \hspace{2cm}- \E\bigg[\nabla^2 m_t(\vtheta^m)\int_{v_t(\vtheta^v)}^{\infty}\int_{-\infty}^{\infty} y f_t(x,y)\D y\D x\bigg],\label{eq:(p.15)}\\
	\frac{\partial}{\partial \vtheta^v}\E\big[\Mm_{t}(\vtheta^m, \vtheta^v)\big] &= \E\bigg[\nabla m_t(\vtheta^m)\nabla^\prime v_t(\vtheta^v)\Big\{\int_{-\infty}^{\infty} y f_t\big(v_t(\vtheta^v),y\big)\D y - m_t(\vtheta^m) f_{t}^{X}\big(v_t(\vtheta^v)\big)\Big\}\bigg].\label{eq:(p.16)}
\end{align}
Evaluating these quantities at the true parameters gives
\begin{align*}
	\mLambda_{(1)} &= \frac{\partial}{\partial \vtheta^m}\E\big[\Mm_{t}(\vtheta^m, \vtheta^v)\big]\Big\vert_{\substack{\vtheta^m=\vtheta_{0}^{m}\\ \vtheta^v=\vtheta_0^v}} \\
	&= (1-\b)\E\Big[\nabla m_t(\vtheta_0^m)\nabla^\prime m_t(\vtheta_0^m)\big\}\Big],\\
	\mLambda_{(2)} &= \frac{\partial}{\partial \vtheta^v}\E\big[\Mm_{t}(\vtheta^m, \vtheta^v)\big]\Big\vert_{\substack{\vtheta^m=\vtheta_{0}^{m}\\ \vtheta^v=\vtheta_0^v}}\\
	&= \E\bigg[\nabla m_t(\vtheta_0^m)\nabla^\prime v_t(\vtheta_0^v) \Big\{\int_{-\infty}^{\infty} y f_t\big(v_t(\vtheta_0^v),y\big)\D y - m_t(\vtheta_0^m) f_{t}^{X}\big(v_t(\vtheta_0^v)\big)\Big\}\bigg],
\end{align*}
where we used for $\mLambda_{(1)}$ that $\overline{F}_t^{X}\big(v_t(\vtheta_0^v)\big)=1-\b$ and 
\[
\int_{v_t(\vtheta_0^v)}^{\infty}\int_{-\infty}^{\infty} y f_t(x,y)\D y\D x=\E_t\big[Y_t\mid X_t\geq v_t(\vtheta_0^v)\big](1-\b)=m_t(\vtheta_0^m)(1-\b).
\]
By virtue of Assumption~\ref{ass:2}, $\mLambda_{(1)}$ and $\mLambda_{(2)}$ do not depend on $t$.

The crucial step in the proof is once again to apply the FCLT of \citet{DMR95}. For this, we require some additional notation:
\begin{align*}
	\widehat{\vlambda}_n(\vtheta^m,\vtheta^v) &= \frac{1}{n}\sum_{t=1}^{n}\Mm_{t}(\vtheta^m,\vtheta^v),\\
	\vlambda(\vtheta^m,\vtheta^v) &= \E\big[\Mm_{t}(\vtheta^m,\vtheta^v)\big],\\
	\mLambda_{(1)}(\vtheta^{m\ast},\vtheta^{v\ast}) &= \frac{\partial}{\partial\vtheta^m}\E\big[\Mm_{t}(\vtheta^m,\vtheta^v)\big]\Big\vert_{\substack{\vtheta^m=\vtheta^{m\ast}\\ \vtheta^v=\vtheta^{v\ast}}},\\
	\mLambda_{(2)}(\vtheta^{m\ast},\vtheta^{v\ast}) &= \frac{\partial}{\partial\vtheta^v}\E\big[\Mm_{t}(\vtheta^m,\vtheta^v)\big]\Big\vert_{\substack{\vtheta^m=\vtheta^{m\ast}\\  \vtheta^v=\vtheta^{v\ast}}},
\end{align*}
and note that $\mLambda_{(1)}=\mLambda_{(1)}(\vtheta_0^{m},\vtheta_0^{v})$ and $\mLambda_{(2)}=\mLambda_{(2)}(\vtheta_0^{m},\vtheta_0^{v})$.

We again give a general outline of the proof with the required Lemmas~M.\ref{lem:1+ tilde}--M.\ref{lem:3 tilde} to be found in Section~\ref{Supplementary Results for MES Parameter Estimator}.

\begin{proof}[{\textbf{Proof of Theorem~\ref{thm:an} (Asymptotic normality of $\widehat{\vtheta}_n^m$):}}]
	The MVT (applied around $\vtheta_0^m$ in the first line and around $\vtheta_0^v$ in the second line) and $\vlambda(\vtheta_0^{m}, \vtheta_0^{v})=\vzeros$ imply that
	\begin{align*}
		\vlambda(\widehat{\vtheta}_n^{m}, \widehat{\vtheta}_n^{v}) &= \vlambda(\vtheta_0^{m}, \widehat{\vtheta}_n^{v}) + \mLambda_{(1)}(\vtheta^{m\ast}, \widehat{\vtheta}_n^{v}) (\widehat{\vtheta}_n^m- \vtheta_0^m)\\
		&=\vlambda(\vtheta_0^{m}, \vtheta_0^{v}) + \mLambda_{(2)}(\vtheta_0^m,\vtheta^{v\ast})(\widehat{\vtheta}_n^v- \vtheta_0^v) + \mLambda_{(1)}(\vtheta^{m\ast}, \widehat{\vtheta}_n^{v}) (\widehat{\vtheta}_n^m- \vtheta_0^m)\\
		&=\mLambda_{(2)}(\vtheta_0^m,\vtheta^{v\ast})(\widehat{\vtheta}_n^v- \vtheta_0^v) + \mLambda_{(1)}(\vtheta^{m\ast}, \widehat{\vtheta}_n^{v}) (\widehat{\vtheta}_n^m- \vtheta_0^m)
	\end{align*}
	for some $\vtheta^{m\ast}$ ($\vtheta^{v\ast}$) on the line connecting $\widehat{\vtheta}_n^{m}$ and $\vtheta_0^{m}$ ($\widehat{\vtheta}_n^{v}$ and $\vtheta_0^{v}$). 
	(Recall our convention to slightly abuse notation when applying the mean value theorem to multivariate functions. Again, the argument goes through componentwise, yet the notation would be more complicated.) 
	Expanding the left-hand side of the above display gives that
	\begin{align}
		&\widehat{\vlambda}_n(\widehat{\vtheta}_n^{m}, \widehat{\vtheta}_n^{v}) - \Big\{\big[\widehat{\vlambda}_n(\widehat{\vtheta}_n^{m}, \widehat{\vtheta}_n^{v}) - \vlambda(\widehat{\vtheta}_n^{m}, \widehat{\vtheta}_n^{v})\big] - \big[\widehat{\vlambda}_n(\vtheta_0^{m}, \vtheta_0^{v}) - \vlambda(\vtheta_0^{m}, \vtheta_0^{v})\big]\Big\} \notag\\
		&\hspace{6cm}- \big[\widehat{\vlambda}_n(\vtheta_0^{m}, \vtheta_0^{v}) - \vlambda(\vtheta_0^{m}, \vtheta_0^{v})\big] \notag\\
		&=\mLambda_{(2)}(\vtheta_0^m,\vtheta^{v\ast})(\widehat{\vtheta}_n^v- \vtheta_0^v) + \mLambda_{(1)}(\vtheta^{m\ast}, \widehat{\vtheta}_n^{v}) (\widehat{\vtheta}_n^m- \vtheta_0^m)\notag\\
		&=\mLambda_{(2)}(\vtheta_0^m,\vtheta^{v\ast})\big[-\mLambda^{-1}+o_{\P}(1)\big]\Big[\widehat{\vlambda}_n(\vtheta_0^v) - \vlambda(\vtheta_0^v) + o_{\P}(1/\sqrt{n})\Big] + \mLambda_{(1)}(\vtheta^{m\ast}, \widehat{\vtheta}_n^{v}) (\widehat{\vtheta}_n^m- \vtheta_0^m),\label{eq:expan theta hat m}
	\end{align} 
	where we used \eqref{eq:exp VaR} and the subsequent items (i)--(iv) in the final step (and recalling that $\vlambda(\vtheta_0^v)=\vzeros$).

	To establish the asymptotic normality of $\widehat{\vtheta}_n^m$, we therefore show that
	\begin{enumerate}
		\item[(i)] $\mLambda_{(1)}^{-1}(\vtheta^{m\ast}, \widehat{\vtheta}_n^{v})=\mLambda_{(1)}^{-1} + o_{\P}(1)$;
		\item[(ii)] $\mLambda_{(2)}(\vtheta_{0}^{m},\vtheta^{v\ast})=\mLambda_{(2)} + o_{\P}(1)$;
		\item[(iii)] $\sqrt{n}\widehat{\vlambda}_n(\widehat{\vtheta}_n^{m}, \widehat{\vtheta}_n^{v})=o_{\P}(1)$;
		\item[(iv)] $\sqrt{n}\Big\{\big[\widehat{\vlambda}_n(\widehat{\vtheta}_n^{m}, \widehat{\vtheta}_n^{v}) - \vlambda(\widehat{\vtheta}_n^{m}, \widehat{\vtheta}_n^{v})\big] - \big[\widehat{\vlambda}_n(\vtheta_0^{m}, \vtheta_0^{v}) - \vlambda(\vtheta_0^{m}, \vtheta_0^{v})\big]\Big\}=o_{\P}(1)$;
		\item[(v)] As $n\to\infty$,
		\begin{equation}\label{eq:joint lambda}
			\sqrt{n}\begin{pmatrix}
				\widehat{\vlambda}_n(\vtheta_0^{v}) - \vlambda(\vtheta_0^{v})\\
				\widehat{\vlambda}_n(\vtheta_0^{m}, \vtheta_0^{v}) - \vlambda(\vtheta_0^{m}, \vtheta_0^{v})
			\end{pmatrix}
			\overset{d}{\longrightarrow}N(\vzeros, \mM).
		\end{equation}
	\end{enumerate}
	
	Items (i) and (ii) are established in Lemma~M.\ref{lem:1+ tilde}. For this, note that $\vtheta^{m\ast}\overset{\P}{\longrightarrow}\vtheta_0^m$ ($\vtheta^{v\ast}\overset{\P}{\longrightarrow}\vtheta_0^v$), because $\vtheta^{m\ast}$ ($\vtheta^{v\ast}$) lies on the line connecting $\widehat{\vtheta}_n^{m}$ and $\vtheta_0^{m}$ ($\widehat{\vtheta}_n^{v}$ and $\vtheta_0^{v}$), and $\widehat{\vtheta}_n^{m}\overset{\P}{\longrightarrow}\vtheta_0^m$ ($\widehat{\vtheta}_n^{v}\overset{\P}{\longrightarrow}\vtheta_0^v$) from Proposition~\ref{prop:cons}. 
	Claim (iii) is immediate from Lemma~M.\ref{lem:2 tilde} and claim (iv) follows from the FCLT of Lemma~M.\ref{lem:3 tilde} (similarly as \eqref{eq:(p.16)} followed from the FCLT of Lemma~V.\ref{lem:3}). 
	
	It remains to show (v). 
	From Lemmas~\ref{lem:type IV Vv} and \ref{lem:type IV Mm} it follows that $\big(\Vv_t^\prime(\cdot), \Mm_t^\prime(\cdot,\cdot)\big)^\prime$ forms a type IV class with index $p=2r$. 
	Then, arguing as for Lemma~V.\ref{lem:3} or Lemma~M.\ref{lem:3 tilde}, an FCLT holds for
	\[
	\sqrt{n}\begin{pmatrix} \widehat{\vlambda}_n(\cdot) - \vlambda(\cdot)\\
		\widehat{\vlambda}_n(\cdot,\cdot) - \vlambda(\cdot,\cdot)\end{pmatrix}.
	\]
	In particular,
	\begin{equation}\label{eq:in part}
		\Bigg\Vert\sqrt{n}\begin{pmatrix} \widehat{\vlambda}_n(\vtheta_0^v) - \vlambda(\vtheta_0^v)\\
			\widehat{\vlambda}_n(\vtheta_0^m,\vtheta_0^v) - \vlambda(\vtheta_0^m,\vtheta_0^v)\end{pmatrix}
		-\overline{\mB}^{(n)}(\vtheta_0^m,\vtheta_0^v) \Bigg\Vert=o_{\P}(1),
	\end{equation}
	where the covariance function of the zero-mean Gaussian process 
	\[
	\overline{\mB}^{(n)}(\vtheta_0^m,\vtheta_0^v)=\begin{pmatrix}\mB^{(n)}(\vtheta_0^v)\\ \mB^{(n)}(\vtheta_0^m,\vtheta_0^v)\end{pmatrix}
	\]
	is given by
	\begin{align*}
		\overline{\mGamma}\big((\vtheta_0^{m\prime}, \vtheta_0^{v\prime})^\prime, (\vtheta_0^{m\prime}, \vtheta_0^{v\prime})^\prime\big)	& = \begin{pmatrix}
			\mGamma_{11} & \mGamma_{12}\\
			\mGamma^{\prime}_{12}  & \mGamma_{22}
		\end{pmatrix}\\	
		&= \E\Bigg[\begin{pmatrix}\Vv_1(\vtheta_0^v)\\ \Mm_1(\vtheta_0^m,\vtheta_0^v)\end{pmatrix}\begin{pmatrix}\Vv_1^\prime(\vtheta_0^v) & \Mm_1^\prime(\vtheta_0^m,\vtheta_0^v)\end{pmatrix}\Bigg]\\
		&\hspace{1cm}+ \sum_{j=1}^{\infty}\E\Bigg[\begin{pmatrix}\Vv_1(\vtheta_0^v)\\ \Mm_1(\vtheta_0^m,\vtheta_0^v)\end{pmatrix}\begin{pmatrix}\Vv_{1+j}^\prime(\vtheta_0^v) & \Mm_{1+j}^\prime(\vtheta_0^m,\vtheta_0^v)\end{pmatrix}\Bigg]\\
		&\hspace{1.6cm}+ \E\Bigg[\begin{pmatrix}\Vv_{1+j}(\vtheta_0^v)\\ \Mm_{1+j}(\vtheta_0^m,\vtheta_0^v)\end{pmatrix}\begin{pmatrix}\Vv_{1}^\prime(\vtheta_0^v) & \Mm_{1}^\prime(\vtheta_0^m,\vtheta_0^v)\end{pmatrix}\Bigg].
	\end{align*}
	We now derive explicit expressions for $\mGamma_{11}$, $\mGamma_{12}$ and $\mGamma_{22}$.
	It follows from \eqref{eq:(12+)} that
	\begin{align*}
		\E\big[\Mm_1(\vtheta_0^m, \vtheta_0^v)\Mm_1^\prime(\vtheta_0^m, \vtheta_0^v)\big] &= \E\Big[\1_{\{X_1> v_1(\vtheta_0^v)\}}\big\{Y_1 - m_1(\vtheta_0^m)\big\}^2\nabla m_1(\vtheta_0^m)\nabla^\prime m_1(\vtheta_0^m) \Big]\\
		&=\mM^{\ast}
	\end{align*}
	and, for $j>0$,
	\begin{align*}
		\E\big[\Mm_1(\vtheta_0^m, \vtheta_0^v)\Mm_{1+j}^\prime(\vtheta_0^m, \vtheta_0^v)\big] &= \E\Big[\Mm_1(\vtheta_0^m, \vtheta_0^v)\E_{1+j}\big\{\Mm_{1+j}^\prime(\vtheta_0^m, \vtheta_0^v)\big\}\Big]\\
		&=\E\big[\Mm_1(\vtheta_0^m, \vtheta_0^v)\vzeros\big]=\vzeros,
	\end{align*}
	such that $\mGamma_{22}=\mM^{\ast}$. 
	
	By \eqref{eq:V lrv 1}--\eqref{eq:V lrv 2}, we obtain that $\mGamma_{11}=\mV$.  
	
	It remains to compute $\mGamma_{12}$. 
	To do so, use the LIE, \eqref{eq:gt} and \eqref{eq:(12+)} to write
	\begin{align*}
		\E&\big[\Vv_1(\vtheta_0^v)\Mm_1^\prime(\vtheta_0^m,\vtheta_0^v)\big]\\
		&= \E\Big[\nabla v_1(\vtheta_0^v) \big(\1_{\{X_1\leq v_1(\vtheta_0^v)\}}-\beta\big)\1_{\{X_1> v_1(\vtheta_0^v)\}}\nabla^\prime m_1(\vtheta_0^m)\big(m_1(\vtheta_0^m)-Y_1\big)\Big]\\
		&= \E\Big[\nabla v_1(\vtheta_0^v) \nabla^\prime m_1(\vtheta_0^m) \E_1\Big\{\big(\1_{\{X_1\leq v_1(\vtheta_0^v)\}}-\beta\big)\1_{\{X_1> v_1(\vtheta_0^v)\}}\big(m_1(\vtheta_0^m)-Y_1\big)\Big\}\Big].
	\end{align*}
	Since
	\begin{align*}
		\E_1&\Big\{\big(\1_{\{X_1\leq v_1(\vtheta_0^v)\}}-\beta\big)\1_{\{X_1> v_1(\vtheta_0^v)\}}\big(m_1(\vtheta_0^m)-Y_1\big)\Big\}\\
		&=-\beta\E_1\Big\{\1_{\{X_1> v_1(\vtheta_0^v)\}}\big(m_1(\vtheta_0^m)-Y_1\big)\Big\}\\
		&= -\beta \Big[m_1(\vtheta_0^m)\E_1\big\{\1_{\{X_1> v_1(\vtheta_0^v)\}}\big\} - \E_1\big\{\1_{\{X_1> v_1(\vtheta_0^v)\}}Y_1\big\}  \Big]\\
		&=-\beta \Big[m_1(\vtheta_0^m)(1-\beta) - (1-\beta)m_1(\vtheta_0^m)  \Big]\\
		&=0,
	\end{align*}
	we have that
	\[
	\E\big[\Vv_1(\vtheta_0^v)\Mm_1^\prime(\vtheta_0^m,\vtheta_0^v)\big]=\vzeros.
	\]
	Furthermore, by the LIE,
	\begin{align*}
		\E\big[\Vv_{1+j}(\vtheta_0^v)\Mm_1^\prime(\vtheta_0^m,\vtheta_0^v)\big] &= \E\Big[\E_{1+j}\big\{\Vv_{1+j}(\vtheta_0^v)\big\}\Mm_1^\prime(\vtheta_0^m,\vtheta_0^v)\Big]\\
		&=\E\Big[\vzeros\cdot\Mm_1^\prime(\vtheta_0^m,\vtheta_0^v)\Big]\\
		&=\vzeros,
	\end{align*}
	such that overall $\mGamma_{12}=\vzeros$.
	
	Taken together we have that $\overline{\mGamma}\big(\vtheta_0=(\vtheta_0^{m\prime}, \vtheta_0^{v\prime})^\prime, \vtheta_0=(\vtheta_0^{m\prime}, \vtheta_0^{v\prime})^\prime\big) = \mM$. Combining this with \eqref{eq:in part}, item (v) easily follows.
\end{proof}

\subsubsection{Supplementary Results for MES Parameter Estimator}\label{Supplementary Results for MES Parameter Estimator}

\begin{lemM}\label{lem:1+ tilde}
	Suppose Assumptions~\ref{ass:1}--\ref{ass:7} hold. Then, as $n\to\infty$, 
	\begin{align*}
		\mLambda_{(1)}^{-1}(\vtheta_n^m,\vtheta_n^v)	&\overset{\P}{\longrightarrow}\mLambda_{(1)}^{-1}\qquad\text{for any}\ \vtheta_n^m\overset{\P}{\longrightarrow}\vtheta_0^m\quad \text{and}\quad \vtheta_n^v\overset{\P}{\longrightarrow}\vtheta_0^v, \\
		\mLambda_{(2)}(\vtheta_0^m,\vtheta_n^v)	&\overset{\P}{\longrightarrow}\mLambda_{(2)}\qquad\text{for any}\ \vtheta_n^v\overset{\P}{\longrightarrow}\vtheta_0^v.
	\end{align*}
\end{lemM}

\begin{proof}
	We begin with the second statement. Our first goal is to show that
	\[
	\big\Vert\mLambda_{(2)}(\vtheta_0^m,\vtheta) - \mLambda_{(2)}\big\Vert \leq C \norm{\vtheta - \vtheta_0^v}
	\]
	for all $\vtheta\in\mTheta^v$.
	Exploit \eqref{eq:(p.16)} to write 
	\begin{align*}
		&\big\Vert\mLambda_{(2)}(\vtheta_0^m,\vtheta) - \mLambda_{(2)}\big\Vert =\big\Vert\mLambda_{(2)}(\vtheta_0^m,\vtheta) - \mLambda_{(2)}(\vtheta_0^m, \vtheta_0^v)\big\Vert\\
		&= \bigg\Vert \E\bigg[\nabla m_t(\vtheta_0^m)\Big\{\nabla^\prime v_t(\vtheta)\int_{-\infty}^{\infty} y f_t\big(v_t(\vtheta),y\big)\D y - \nabla^\prime v_t(\vtheta_0^v)\int_{-\infty}^{\infty} y f_t\big(v_t(\vtheta_0^v),y\big)\D y\Big\}\bigg]\\
		& \hspace{2cm} -\E\bigg[\nabla m_t(\vtheta_0^m)m_t(\vtheta_0^m) \Big\{ \nabla^\prime v_t(\vtheta) f_{t}^{X}\big(v_t(\vtheta)\big)- \nabla^\prime v_t(\vtheta_0^v) f_{t}^{X}\big(v_t(\vtheta_0^v)\big) \Big\}\bigg]\bigg\Vert\\
		&\leq \E\bigg\Vert\nabla m_t(\vtheta_0^m)\Big\{\nabla^\prime v_t(\vtheta)\int_{-\infty}^{\infty} y f_t\big(v_t(\vtheta),y\big)\D y - \nabla^\prime v_t(\vtheta_0^v)\int_{-\infty}^{\infty} y f_t\big(v_t(\vtheta_0^v),y\big)\D y\Big\}\bigg\Vert\\
		&\hspace{1.27cm}+ \E\bigg\Vert\nabla m_t(\vtheta_0^m)m_t(\vtheta_0^m) \Big\{ \nabla^\prime v_t(\vtheta) f_{t}^{X}\big(v_t(\vtheta)\big)- \nabla^\prime v_t(\vtheta_0^v) f_{t}^{X}\big(v_t(\vtheta_0^v)\big) \Big\}\bigg\Vert\\
		&=:A_{6} + B_{6}.
	\end{align*}
	First, consider $B_{6}$. A mean value expansion around $\vtheta_0^v$ implies
	\begin{align*}
		B_{6} &=  \E\bigg\Vert\nabla m_t(\vtheta_0^m) m_t(\vtheta_0^m)\big[\nabla^\prime v_t(\vtheta) - \nabla^\prime v_t(\vtheta_0^v)\big] f_t^X\big(v_t(\vtheta)\big) \\
		& \hspace{3cm} + \nabla m_t(\vtheta_0^m) m_t(\vtheta_0^m)\nabla^\prime v_t(\vtheta_0^v)\Big[f_t^X\big(v_t(\vtheta) - f_t^X\big(v_t(\vtheta_0^v)\big)\Big]\bigg\Vert\\
		&\leq \E\Big\Vert\nabla m_t(\vtheta_0^m) m_t(\vtheta_0^m)\big[\nabla^\prime v_t(\vtheta) - \nabla^\prime v_t(\vtheta_0^v)\big] f_t^X\big(v_t(\vtheta)\big)\Big\Vert \\
		& \hspace{3cm} + K\E\Big\Vert\nabla m_t(\vtheta_0^m) m_t(\vtheta_0^m)\nabla^\prime v_t(\vtheta_0^v)\big[v_t(\vtheta) - v_t(\vtheta_0^v)\big]\Big\Vert\\
		&\leq  K\E\big[M_1(\mZ_t)M(\mZ_t) V_2(\mZ_t)\big] \big\Vert\vtheta - \vtheta_0^v\big\Vert \\
		& \hspace{3cm} + K\E\big[M_1(\mZ_t)M(\mZ_t) V_1^2(\mZ_t)\big] \big\Vert\vtheta - \vtheta_0^v\big\Vert\\
		&\leq K\Vert M_1(\mZ_t)\Vert_{4}\Vert M(\mZ_t)\Vert_{4} \Vert V_2(\mZ_t)\Vert_{2} \big\Vert\vtheta - \vtheta_0^v\big\Vert \\
		& \hspace{3cm} + K\Vert M_1(\mZ_t)\Vert_{4} \Vert M(\mZ_t)\Vert_{4} \Vert V_1^2(\mZ_t)\Vert_{2} \big\Vert\vtheta - \vtheta_0^v\big\Vert\\
		&\leq C \norm{\vtheta - \vtheta_0^v},
	\end{align*}
	where we used (the generalized) H\"{o}lder's inequality in the penultimate step, and Assumption~\ref{ass:5} in the last step.
	
	Using similar arguments, we may also deduce that 
	\begin{align*}
		A_{6} &=  \E\bigg\Vert\nabla m_t(\vtheta_0^m)\big[\nabla^\prime v_t(\vtheta)-\nabla^\prime v_t(\vtheta_0^v)\big]\int_{-\infty}^{\infty} y f_t\big(v_t(\vtheta),y\big)\D y \\
		&\hspace{3cm} + \nabla m_t(\vtheta_0^m)\nabla^\prime v_t(\vtheta_0^v)\Big[\int_{-\infty}^{\infty} y f_t\big(v_t(\vtheta^{\ast}),y\big)\D y - \int_{-\infty}^{\infty} y f_t\big(v_t(\vtheta_0^v),y\big)\D y\Big]\bigg\Vert\\
		&\leq  \E\bigg\Vert\nabla m_t(\vtheta_0^m)\big[\nabla^\prime v_t(\vtheta)-\nabla^\prime v_t(\vtheta_0^v)\big]\int_{-\infty}^{\infty} y f_t\big(v_t(\vtheta),y\big)\D y \bigg\Vert\\
		&\hspace{3cm} + \E\bigg\Vert\nabla m_t(\vtheta_0^m)\nabla^\prime v_t(\vtheta_0^v)\sup_{x\in\mathbb{R}}\Big|\int_{-\infty}^{\infty} y \partial_1 f_t\big(x,y\big)\D y\Big|\cdot\big[v_t(\vtheta) - v_t(\vtheta_0^v)\big]\bigg\Vert\\
		&\leq  \E\big[ M_1(\mZ_t)V_2(\mZ_t)F_{1}(\mathcal{F}_t)\big]\big\Vert\vtheta^{\ast}-\vtheta_0^v\big\Vert \\
		&\hspace{3cm} + \E\big[M_1(\mZ_t)V_1^2(\mZ_t)F(\mathcal{F}_t)\big] \big\Vert \vtheta - \vtheta_0^v\big\Vert\\
		&\leq  \big\Vert M_1(\mZ_t)\big\Vert_{4r}\big\Vert V_2(\mZ_t)\big\Vert_{2r}\big\Vert F_{1}(\mathcal{F}_t)\big\Vert_{4r/(4r-3)}\big\Vert\vtheta^{\ast}-\vtheta_0^v\big\Vert \\
		&\hspace{3cm} + \big\Vert M_1(\mZ_t)\big\Vert_{4r} \big\Vert V_1^2(\mZ_t)\big\Vert_{2r} \big\Vert F(\mathcal{F}_t)\big\Vert_{4r/(4r-3)}\big\Vert \vtheta - \vtheta_0^v\big\Vert\\
		&\leq C\big\Vert \vtheta - \vtheta_0^v\big\Vert,
	\end{align*}
	where $\vtheta^\ast$ is a mean value on the line connecting $\vtheta$ and $\vtheta_0^v$.
	Hence, 
	\begin{equation}\label{eq:(2.11)}
		\big\Vert\mLambda_{(2)}(\vtheta_0^m,\vtheta_n^v) - \mLambda_{(2)}\big\Vert \leq C \norm{\vtheta_n^v - \vtheta_0^v} = o_{\P}(1)
	\end{equation}
	by the assumption that $\vtheta_n^v \overset{\P}{\longrightarrow} \vtheta_0^v$. The second statement of the lemma follows.
	
	To prove the first claim, consider the neighborhoods
	\[
	\mathcal{N}(\vtheta_0^v):=\big\{\vtheta\in\mTheta^v\colon\big\Vert\vtheta-\vtheta_0^v\big\Vert\leq\delta\big\}\qquad\text{and}\qquad \mathcal{N}(\vtheta_0^m):=\big\{\vtheta\in\mTheta^m\colon\big\Vert\vtheta-\vtheta_0^m\big\Vert\leq\delta\big\}
	\]
	with $\delta>0$ chosen sufficiently small, such that Assumption~\ref{ass:4}~\ref{it:4v} holds. Let $\vtheta^v\in\mathcal{N}(\vtheta_0^v)$ and $\vtheta^m\in\mathcal{N}(\vtheta_0^m)$. Use the triangle inequality for the bound
	\begin{align}
		\big\Vert\mLambda_{(1)} - \mLambda_{(1)}(\vtheta^{m},\vtheta^{v})\big\Vert &\leq \big\Vert\mLambda_{(1)} - \mLambda_{(1)}(\vtheta_0^m,\vtheta^{v})\big\Vert + \big\Vert\mLambda_{(1)}(\vtheta_0^m,\vtheta^{v}) - \mLambda_{(1)}(\vtheta^{m},\vtheta^{v})\big\Vert\notag\\
		&=:A_{7} + B_{7}. \label{eq:(3.1)}
	\end{align}
	First, consider $A_{7}$. Use \eqref{eq:(p.15)} to write
	\begin{align*}
		A_{7} &= \big\Vert \mLambda_{(1)}(\vtheta_{0}^m,\vtheta_0^v) - \mLambda_{(1)}(\vtheta_{0}^m,\vtheta^{v})\big\Vert\\
		&= \bigg\Vert \E\Big[\big\{\nabla^2 m_t(\vtheta_{0}^m)m_t(\vtheta_{0}^m) + \nabla m_t(\vtheta_{0}^m)\nabla^\prime m_t(\vtheta_{0}^m)\big\}\overline{F}_{t}^{X}\big(v_t(\vtheta_{0}^v)\big)\Big]\\
		&\hspace{1.3cm} - \E\Big[\big\{\nabla^2 m_t(\vtheta_{0}^m)m_t(\vtheta_{0}^m) + \nabla m_t(\vtheta_{0}^m)\nabla^\prime m_t(\vtheta_{0}^m)\big\}\overline{F}_{t}^{X}\big(v_t(\vtheta^{v})\big)\Big]\\
		&\hspace{1.3cm} + \E\bigg[\nabla^2 m_t(\vtheta_0^m)\int_{v_t(\vtheta^{v})}^{\infty}\int_{-\infty}^{\infty} y f_t(x,y)\D y\D x\bigg]\\
		&\hspace{1.3cm} - \E\bigg[\nabla^2 m_t(\vtheta_0^m)\int_{v_t(\vtheta_0^v)}^{\infty}\int_{-\infty}^{\infty} y f_t(x,y)\D y\D x\bigg]\bigg\Vert\\
		&= \bigg\Vert \E\Big[\nabla^2 m_t(\vtheta_{0}^m)m_t(\vtheta_{0}^m)\big\{F_{t}^{X}\big(v_t(\vtheta^{v})\big) - F_{t}^{X}\big(v_t(\vtheta_{0}^v)\big)\big\}\Big]\\
		&\hspace{1.3cm} + \E\Big[\nabla m_t(\vtheta_{0}^m)\nabla^\prime m_t(\vtheta_{0}^m)\big\{F_{t}^{X}\big(v_t(\vtheta^{v})\big) - F_{t}^{X}\big(v_t(\vtheta_{0}^v)\big)\big\}\Big]\\
		&\hspace{1.3cm}+ \E\bigg[\nabla^2 m_t(\vtheta_0^m)\int_{v_t(\vtheta_0^v)}^{v_t(\vtheta^{v})}\int_{-\infty}^{\infty} y f_t(x,y)\D y\D x\bigg] \bigg\Vert.
	\end{align*}
	
	\textbf{First term ($A_{7}$):} A mean value expansion around $\vtheta_0^v$ gives
	\begin{align*}
		\bigg\Vert\E&\Big[\nabla^2 m_t(\vtheta_{0}^m)m_t(\vtheta_{0}^m)\big\{F_{t}^{X}\big(v_t(\vtheta^{v})\big) - F_{t}^{X}\big(v_t(\vtheta_{0}^v)\big)\big\}\Big]\bigg\Vert\\
		&\leq \E\norm{\nabla^2 m_t(\vtheta_{0}^m)m_t(\vtheta_{0}^m)f_{t}^{X}\big(v_t(\vtheta^{v\ast})\big)\nabla v_t(\vtheta^{v\ast})}\big\Vert\vtheta^{v} - \vtheta_{0}^v\big\Vert\\
		&\leq  K\E\big[M_2(\mZ_t) M(\mZ_t)V_1(\mZ_t)\big] \big\Vert\vtheta^{v} - \vtheta_0^v\big\Vert\\
		&\leq K \big\Vert M_2(\mZ_t)\big\Vert_{3} \big\Vert M(\mZ_t)\big\Vert_{3} \big\Vert V_1(\mZ_t)\big\Vert_{3} \big\Vert\vtheta^{v} - \vtheta_0^v\big\Vert\\
		&\leq C \big\Vert\vtheta^{v} - \vtheta_0^v\big\Vert\\
		&\leq C\delta,
	\end{align*}
	where $\vtheta^{v\ast}$ is some value on the line connecting $\vtheta_0^v$ and $\vtheta^{v}$, and the third step follows from (the generalized) H\"{o}lder's inequality.

	\textbf{Second term ($A_{7}$):} A mean value expansion around $\vtheta_0^v$ implies
	\begin{align*}
		\bigg\Vert\E&\Big[\nabla m_t(\vtheta_{0}^m)\nabla^\prime m_t(\vtheta_{0}^m)\big\{F_{t}^{X}\big(v_t(\vtheta^{v})\big) - F_{t}^{X}\big(v_t(\vtheta_{0}^v)\big)\big\}\Big]\bigg\Vert\\
		&\leq \E\Big\Vert\nabla m_t(\vtheta_{0}^m)\nabla^\prime m_t(\vtheta_{0}^m)f_{t}^{X}\big(v_t(\vtheta^{v\ast})\big)\nabla v_t(\vtheta^{v\ast})\Big\Vert\big\Vert\vtheta^{v} - \vtheta_{0}^v\big\Vert\\
		&\leq K\E\big[M_1^2(\mZ_t) V_1(\mZ_t)\big] \big\Vert\vtheta^{v} - \vtheta_0^v\big\Vert\\
		&\leq K\big\Vert M_1^2(\mZ_t)\big\Vert_{2} \big\Vert V_1(\mZ_t)\big\Vert_{2}\big\Vert\vtheta^{v} - \vtheta_0^v\big\Vert\\
		&\leq C \big\Vert\vtheta^{v} - \vtheta_0^v\big\Vert\\
		&\leq C\delta,
	\end{align*}
	where $\vtheta^{v\ast}$ is some value on the line connecting $\vtheta_0^v$ and $\vtheta^{v}$, and the third step follows from H\"{o}lder's inequality.

	\textbf{Third term ($A_{7}$):} Define
	\begin{align*}
		\underline{\vtheta}^v &=\argmin_{\Vert\vtheta^v-\vtheta_0^v\Vert\leq \delta}v_t(\vtheta^v),\\
		\overline{\vtheta}^v &=\argmax_{\Vert\vtheta^v-\vtheta_0^v\Vert\leq \delta}v_t(\vtheta^v).
	\end{align*}
	Note that $\underline{\vtheta}^v\in\mathcal{N}(\vtheta_0^v)$ and $\overline{\vtheta}^v\in\mathcal{N}(\vtheta_0^v)$ by construction.
	We have 
	\begin{align*}
		\bigg\Vert\E&\bigg[\nabla^2 m_t(\vtheta_0^m)\int_{v_t(\vtheta_0^v)}^{v_t(\vtheta^{v})}\int_{-\infty}^{\infty} y f_t(x,y)\D y\D x\bigg]\bigg\Vert\\
		&\leq \E\bigg[\norm{\nabla^2 m_t(\vtheta_0^m)}\int_{v_t(\vtheta_0^v)}^{v_t(\vtheta^{v})}\int_{-\infty}^{\infty} |y| f_t(x,y)\D y\D x\bigg]\\
		&\leq \E\bigg[M_2(\mZ_t)\int_{v_t(\underline{\vtheta}^v)}^{v_t(\overline{\vtheta}^v)}\int_{-\infty}^{\infty} |y| f_t(x,y)\D y\D x\bigg]\\
		&= \E\bigg[M_2(\mZ_t)\E_t\Big\{|Y_t|\big(\1_{\{X_t\leq v_t(\overline{\vtheta}^v)\}} - \1_{\{X_t\leq v_t(\underline{\vtheta}^v)\}}\big)\Big\}\bigg]\\
		&= \E\bigg[M_2(\mZ_t)|Y_t|\big(\1_{\{X_t\leq v_t(\overline{\vtheta}^v)\}} - \1_{\{X_t\leq v_t(\underline{\vtheta}^v)\}}\big)\bigg]\\
		&\leq \big\Vert M_2(\mZ_t)\big\Vert_{4r}\big\Vert Y_t\big\Vert_{4r}\big\Vert \1_{\{X_t\leq v_t(\overline{\vtheta}^v)\}} - \1_{\{X_t\leq v_t(\underline{\vtheta}^v)\}}\big\Vert_p,
	\end{align*}
	where $p=2r/(2r-1)$. Now, we may use \eqref{eq:help ind1} to deduce that
	\[
	\norm{\E\bigg[\nabla^2 m_t(\vtheta_0^m)\int_{v_t(\vtheta_0^v)}^{v_t(\vtheta^{v})}\int_{-\infty}^{\infty} y f_t(x,y)\D y\D x\bigg]}\leq C \delta^{1/p}.
	\]
	
	Combining the results for the different terms, we conclude that
	\begin{equation}\label{eq:A34}
		A_{7}\leq C\delta^{1/p}.
	\end{equation}
	
	Now, consider $B_{7}$. For this, use \eqref{eq:(p.15)} to write
	\begin{align*}
		& \mLambda_{(1)}(\vtheta^{m},\vtheta^{v}) - \mLambda_{(1)}(\vtheta_{0}^m,\vtheta^{v}) \\
		&= \E\Big[\big\{\nabla^2 m_t(\vtheta^{m})m_t(\vtheta^{m}) - \nabla^2 m_t(\vtheta_{0}^m)m_t(\vtheta_{0}^m)\big\}\overline{F}_{t}^{X}\big(v_t(\vtheta^{v})\big)\Big]\\
		&\hspace{1.3cm} + \E\Big[\big\{\nabla m_t(\vtheta^{m})\nabla^\prime m_t(\vtheta^{m}) - \nabla m_t(\vtheta_{0}^m)\nabla^\prime m_t(\vtheta_{0}^m)\big\}\overline{F}_{t}^{X}\big(v_t(\vtheta^{v})\big)\Big]\\
		&\hspace{1.3cm} - \E\bigg[\big\{\nabla^2 m_t(\vtheta^{m}) - \nabla^2 m_t(\vtheta_0^m)\big\}\int_{v_t(\vtheta^{v})}^{\infty}\int_{-\infty}^{\infty} y f_t(x,y)\D y\D x\bigg].
	\end{align*}

	\textbf{First term ($B_{7}$):} Using a mean value expansion around $\vtheta_{0}^m$ and Assumptions~\ref{ass:4}--\ref{ass:5}, we obtain that
	\begin{align*}
		&\bigg\Vert \E\Big[\big\{\nabla^2 m_t(\vtheta^{m})m_t(\vtheta^{m}) - \nabla^2 m_t(\vtheta_{0}^m)m_t(\vtheta_{0}^m)\big\}\overline{F}_{t}^{X}\big(v_t(\vtheta^{v})\big)\Big] \bigg\Vert\\
		&\leq \E\norm{\big[\nabla^2 m_t(\vtheta^{m}) - \nabla^2 m_t(\vtheta_0^m)\big]m_t(\vtheta^{m}) \overline{F}_{t}^{X}\big(v_t(\vtheta^{v})\big) }\\
		& \hspace{2cm} + \E\norm{\nabla^2 m_t(\vtheta_0^m)\big[m_t(\vtheta^{m})- m_t(\vtheta_0^m)\big] \overline{F}_{t}^{X}\big(v_t(\vtheta^{v})\big)}\\
		& \leq \E\Big[M_3(\mZ_t)M(\mZ_t) \overline{F}_{t}^{X}\big(v_t(\vtheta^{v})\big) \Big]\norm{\vtheta^{m} - \vtheta_0^m}\\
		& \hspace{2cm} + \E\Big[M_2(\mZ_t)\nabla m_t(\vtheta^{m\ast}) \overline{F}_{t}^{X}\big(v_t(\vtheta^{v})\big)\Big]\norm{\vtheta^{m} - \vtheta_0^m}\\
		& \leq \E\big[M_3(\mZ_t)M(\mZ_t)\big]\norm{\vtheta^{m} - \vtheta_0^m}\\
		& \hspace{2cm} + \E\big[M_2(\mZ_t)M_1(\mZ_t)\big]\norm{\vtheta^{m} - \vtheta_0^m}\\
		&\leq \big\Vert M_3(\mZ_t)\big\Vert_{4r/(4r-1)}\big\Vert M(\mZ_t)\big\Vert_{4r}\norm{\vtheta^{m} - \vtheta_0^m}\\
		& \hspace{2cm} + \big\Vert M_2(\mZ_t)\big\Vert_2 \big\Vert M_1(\mZ_t)\big\Vert_2 \norm{\vtheta^{m} - \vtheta_0^m}\\
		&\leq C\norm{\vtheta^{m} - \vtheta_0^m}\\
		&\leq C\delta,
	\end{align*}
	where $\vtheta^{m\ast}$ is some mean value between $\vtheta^{m}$ and $\vtheta_{0}^m$.

	\textbf{Second term ($B_{7}$):} Using by now familiar arguments, we expand
	\begin{align*}
		&\bigg\Vert \E\Big[\big\{\nabla m_t(\vtheta^{m})\nabla^\prime m_t(\vtheta^{m}) - \nabla m_t(\vtheta_{0}^m)\nabla^\prime m_t(\vtheta_{0}^m)\big\}\overline{F}_{t}^{X}\big(v_t(\vtheta^{v})\big)\Big] \bigg\Vert\\
		&\leq \E\norm{\big[\nabla m_t(\vtheta^{m}) - \nabla m_t(\vtheta_0^m)\big]\nabla^\prime m_t(\vtheta^{m}) \overline{F}_{t}^{X}\big(v_t(\vtheta^{v})\big) }\\
		& \hspace{2cm} + \E\norm{\nabla m_t(\vtheta_0^m)\big[\nabla^\prime m_t(\vtheta^{m})- \nabla^\prime m_t(\vtheta_0^m)\big] \overline{F}_{t}^{X}\big(v_t(\vtheta^{v})\big)}\\
		& \leq \E\norm{\nabla^2 m_t(\vtheta^{m\ast})\big[\vtheta^{m} - \vtheta_0^m\big]\nabla^\prime m_t(\vtheta^{m}) \overline{F}_{t}^{X}\big(v_t(\vtheta^{v})\big) }\\
		& \hspace{2cm} + \E\norm{\nabla m_t(\vtheta_0^m) \big[\vtheta^{m} - \vtheta_0^m\big]^\prime \nabla^{2\prime} m_t(\vtheta^{m\ast}) \overline{F}_{t}^{X}\big(v_t(\vtheta^{v})\big)}\\
		& \leq \E\big[M_2(\mZ_t)M_1(\mZ_t)\big]\norm{\vtheta^{m} - \vtheta_0^m}\\
		& \hspace{2cm} + \E\big[M_2(\mZ_t)M_1(\mZ_t)\big]\norm{\vtheta^{m} - \vtheta_0^m}\\
		&\leq C\norm{\vtheta^{m} - \vtheta_0^m}\\
		&\leq C\delta,
	\end{align*}
	where $\vtheta^{m\ast}$ is some mean value between $\vtheta^{m}$ and $\vtheta_{0}^{m}$ that may be different in different places.
	
	\textbf{Third term ($B_{7}$):} Use the moment bounds of Assumption~\ref{ass:5} to conclude that
	\begin{align*}
		& \Bigg\Vert\E\bigg[\big\{\nabla^2 m_t(\vtheta^{m}) - \nabla^2 m_t(\vtheta_0^m)\big\}\int_{v_t(\vtheta^{v})}^{\infty}\int_{-\infty}^{\infty} y f_t(x,y)\D y\D x\bigg]\Bigg\Vert\\
		&\leq \E\bigg[\big\Vert\nabla^2 m_t(\vtheta^{m}) - \nabla^2 m_t(\vtheta_0^m)\big\Vert\int_{v_t(\vtheta^{v})}^{\infty}\int_{-\infty}^{\infty} |y| f_t(x,y)\D y\D x\bigg]\\
		&\leq  \E\bigg[M_3(\mZ_t)\int_{v_t(\vtheta^{v})}^{\infty}\int_{-\infty}^{\infty} |y| f_t(x,y)\D y\D x\bigg]\norm{\vtheta^{m} - \vtheta_0^m}\\
		&\leq \E\big[M_3(\mZ_t)|Y_t|\big] \norm{\vtheta^{m} - \vtheta_0^m}\\
		&\leq \big\Vert M_3(\mZ_t)\big\Vert_{4r/(4r-1)} \big\Vert Y_t\big\Vert_{4r} \norm{\vtheta^{m} - \vtheta_0^m}\\
		&\leq C\norm{\vtheta^{m} - \vtheta_0^m}\\
		&\leq C\delta,
	\end{align*}
	where we used the LIE and that
	\begin{align*}
		\int_{v_t(\vtheta^{v})}^{\infty}\int_{-\infty}^{\infty} |y| f_t(x,y)\D y\D x &\leq \int_{-\infty}^{\infty}|y|\Big\{\int_{-\infty}^{\infty}  f_t(x,y)\D x\Big\}\D y\\
		&= \int_{-\infty}^{\infty}|y|f_t^{Y}(y)\D y\\
		&=\E_t\big[|Y_t|\big]
	\end{align*}
	in the third step.

	Thus, we have shown that
	\begin{equation}\label{eq:B3}
		B_{7}=\big\Vert\mLambda_{(1)}(\vtheta^{m},\vtheta^{v}) - \mLambda_{(1)}(\vtheta_{0}^m,\vtheta^{v})\big\Vert\leq C \big\Vert\vtheta^{m} - \vtheta_{0}^m\big\Vert\leq C\delta.
	\end{equation}
	
	Plugging \eqref{eq:A34} and \eqref{eq:B3} into \eqref{eq:(3.1)} gives
	\begin{equation}\label{eq:(A.43)}
		\big\Vert\mLambda_{(1)} - \mLambda_{(1)}(\vtheta^{m},\vtheta^{v})\big\Vert\leq C\delta^{1/p}.
	\end{equation}
	Using consistency of the parameter estimators, proves that $\mLambda_{(1)}(\vtheta_n^m,\vtheta_n^v)\overset{\P}{\longrightarrow}\mLambda_{(1)}$. By positive definiteness of $\mLambda_{(1)}=\mLambda_{(1)}(\vtheta_0^m, \vtheta_0^v)$ and continuity of $\mLambda_{(1)}(\cdot,\cdot)$ (in a neighborhood of the true parameters; see \eqref{eq:(A.43)}), $\mLambda_{(1)}^{-1}(\vtheta_n^m,\vtheta_n^v)\overset{\P}{\longrightarrow}\mLambda_{(1)}^{-1}$ follows from this by the CMT.
\end{proof}

\begin{lemM}\label{lem:2 tilde}
	Suppose Assumptions~\ref{ass:1}--\ref{ass:7} hold. Then, as $n\to\infty$,
	\[
	\sqrt{n}\widehat{\vlambda}_n(\widehat{\vtheta}_n^{m}, \widehat{\vtheta}_n^{v})=o_{\P}(1).
	\]
\end{lemM}

\begin{proof}
	The proof is similar to that of Lemma~V.\ref{lem:2}.
	Recall from Assumption~\ref{ass:1}~\ref{it:1ii} that $\mTheta^m\subset\mathbb{R}^{q}$, such that $\vtheta^m$ is a $q$-dimensional parameter vector. Let $\ve_1,\ldots,\ve_q$ denote the standard basis of $\mathbb{R}^{q}$. Define
	\[
	S_{j,n}^{\MES}(a) := \frac{1}{\sqrt{n}}\sum_{t=1}^{n}S^{\MES}\Big(\big(v_t(\widehat{\vtheta}_n^{v}), m_t(\widehat{\vtheta}_n^{m}+a \ve_j)\big)^\prime, (X_t,Y_t)^\prime\Big),\qquad j=1,\ldots,q,
	\]
	where $a\in\mathbb{R}$. Let $G_{j,n}(a)$ be the right partial derivative of $S_{j,n}^{\MES}(a)$, such that (see \eqref{eq:(12+)})
	\[
	G_{j,n}(a)=\frac{1}{\sqrt{n}}\sum_{t=1}^{n}\1_{\{X_t> v_t(\widehat{\vtheta}_n^{v})\}}\nabla_j m_t(\widehat{\vtheta}_n^{m}+a \ve_j)\big[m_t(\widehat{\vtheta}_n^{m}+a \ve_j)-Y_t\big]
	,
	\]
	where $\nabla_j m_t(\cdot)$ is the $j$-th component of $\nabla m_t(\cdot)$. Then, $G_{j,n}(0)=\lim_{\xi\downarrow0}G_{j,n}(\xi)$ is the right partial derivative of
	\[
	S_n^{\MES}(\vtheta^{m}) := \frac{1}{\sqrt{n}}\sum_{t=1}^{n}S^{\MES}\Big(\big(v_t(\widehat{\vtheta}_n^v), m_t(\vtheta^m)\big)^\prime, (X_t,Y_t)^\prime\Big)
	\]
	at $\widehat{\vtheta}_n^{m}$ in the direction $\theta_j^{m}$, where $\vtheta^{m}=(\theta_1^{m},\ldots,\theta_q^{m})^\prime$. Correspondingly, $\lim_{\xi\downarrow0}G_{j,n}(-\xi)$ is the left partial derivative. Because $S_n^{\MES}(\cdot)$ achieves its minimum at $\widehat{\vtheta}_n^{m}$, the left derivative must be non-positive and the right derivative must be non-negative. Thus,
	\begin{align*}
		\big|G_{j,n}(0)\big| &\leq G_{j,n}(\xi) - G_{j,n}(-\xi)\\
		&= \frac{1}{\sqrt{n}}\sum_{t=1}^{n}\1_{\{X_t> v_t(\widehat{\vtheta}_n^{v})\}}\nabla_j m_t(\widehat{\vtheta}_n^{m}+\xi \ve_j)\big[m_t(\widehat{\vtheta}_n^{m}+\xi \ve_j)-Y_t\big]\\
		&\hspace{2cm} - \frac{1}{\sqrt{n}}\sum_{t=1}^{n}\1_{\{X_t> v_t(\widehat{\vtheta}_n^{v})\}}\nabla_j m_t(\widehat{\vtheta}_n^{m}-\xi \ve_j)\big[m_t(\widehat{\vtheta}_n^{m}-\xi \ve_j)-Y_t\big].
	\end{align*}
	By continuity of $m_t(\cdot)$ and $\nabla m_t(\cdot)$ (see Assumption~\ref{ass:4} \ref{it:4i}) it follows upon letting $\xi\to0$ that $|G_{j,n}(0)|=0$; in particular, $G_{j,n}(0)=o_{\P}(1)$. As this holds for every $j=1,\ldots,q$, we get that 
	\[
	\frac{1}{\sqrt{n}}\sum_{t=1}^{n}\Mm_{t}(\widehat{\vtheta}_n^{m}, \widehat{\vtheta}_n^{v})\overset{\eqref{eq:(12+)}}{=}\frac{1}{\sqrt{n}}\sum_{t=1}^{n}\1_{\{X_t> v_t(\widehat{\vtheta}_n^{v})\}}\nabla m_t(\widehat{\vtheta}_n^{m})\big[m_t(\widehat{\vtheta}_n^m) - Y_t\big]=o_{\P}(1),
	\]
	which is just the conclusion.
\end{proof}

For Lemma~M.\ref{lem:3 tilde}, we require the following preliminary result as an analog to Lemma~\ref{lem:type IV Vv}.

\begin{lem}\label{lem:type IV Mm}
	Suppose Assumptions~\ref{ass:1}--\ref{ass:7} hold. Then, $\Mm_t(\cdot,\cdot)$ forms a type IV class in the sense of \citet[p.~2278]{And94} with index $p=2r$ (where $r>1$ is from Assumption~\ref{ass:2}).
\end{lem}

\begin{proof}
	Write $\Mm_t(\cdot,\cdot)=\big(\mathsf{m}_t^{(1)}(\cdot,\cdot), \ldots, \mathsf{m}_t^{(q)}(\cdot,\cdot)\big)^\prime$.
	To prove the lemma, we have to show that there exist constants $C>0$ and $\psi>0$, such that
	\begin{equation}\label{eq:TSL4}
		\sup_{t\in\mathbb{N}}\bigg\Vert\sup_{\Vert\widetilde{\vtheta} - \vtheta\Vert\leq\delta}\big|\mathsf{m}_t^{(i)}(\widetilde{\vtheta}^m,\widetilde{\vtheta}^v) - \mathsf{m}_t^{(i)}(\vtheta^m,\vtheta^v)\big|\bigg\Vert_{2r}\leq C\delta^\psi
	\end{equation}
	for all $i=1,\ldots,q$, all $\vtheta=(\vtheta^{m\prime}, \vtheta^{v\prime})^\prime\in\mTheta$ and all $\delta>0$ in some neighborhood of 0.
	Fix $i\in\{1,\ldots,q\}$, $\vtheta\in\mTheta$ and $\delta>0$ such that $\{\widetilde{\vtheta}\in\mathbb{R}^{p+q}\colon \Vert\widetilde{\vtheta}-\vtheta\Vert\leq\delta\}\subset\mTheta$. Then,
	\begin{align*}
		&\bigg\Vert\sup_{\Vert\widetilde{\vtheta} - \vtheta\Vert\leq\delta}\big|\mathsf{m}_t^{(i)}(\widetilde{\vtheta}^m, \widetilde{\vtheta}^v) - \mathsf{m}_t^{(i)}(\vtheta^m,\vtheta^v)\big|\bigg\Vert_{2r} \\
		&=  \bigg\Vert\sup_{\Vert\widetilde{\vtheta} - \vtheta\Vert\leq\delta}\Big|\1_{\{X_t> v_t(\widetilde{\vtheta}^v)\}}\nabla_i m_t(\widetilde{\vtheta}^m)\big[m_t(\widetilde{\vtheta}^m)-Y_t\big] - \1_{\{X_t> v_t(\vtheta^v)\}}\nabla_i m_t(\vtheta^m)\big[m_t(\vtheta^m)-Y_t\big]\Big|\bigg\Vert_{2r} \\
		&= \bigg\Vert\sup_{\Vert\widetilde{\vtheta} - \vtheta\Vert\leq\delta}\Big| \big[\1_{\{X_t> v_t(\widetilde{\vtheta}^v)\}}-\1_{\{X_t> v_t(\vtheta^v)\}}\big]\nabla_i m_t(\widetilde{\vtheta}^m)m_t(\widetilde{\vtheta}^m)\\
		&\hspace{3cm} + \big[\1_{\{X_t> v_t(\widetilde{\vtheta}^v)\}}-\1_{\{X_t> v_t(\vtheta^v)\}}\big]\nabla_i m_t(\widetilde{\vtheta}^m)Y_t\\
		&\hspace{3cm} + \1_{\{X_t> v_t(\vtheta^v)\}}\Big\{\nabla_i m_t(\widetilde{\vtheta}^m)m_t(\widetilde{\vtheta}^m) - \nabla_i m_t(\vtheta^m)m_t(\vtheta^m)\Big\} \\
		&\hspace{3cm}+ \1_{\{X_t> v_t(\vtheta^v)\}}\Big\{\nabla_i m_t(\widetilde{\vtheta}^m) - \nabla_i m_t(\vtheta^m)\Big\} Y_t \Big|\bigg\Vert_{2r} \\	
		&\leq \bigg\Vert\sup_{\Vert\widetilde{\vtheta} - \vtheta\Vert\leq\delta}\Big|\big[\1_{\{X_t> v_t(\widetilde{\vtheta}^v)\}}-\1_{\{X_t> v_t(\vtheta^v)\}}\big]\nabla_i m_t(\widetilde{\vtheta}^m)m_t(\widetilde{\vtheta}^m)\Big|\bigg\Vert_{2r}\\
		&\hspace{3cm} + \bigg\Vert\sup_{\Vert\widetilde{\vtheta} - \vtheta\Vert\leq\delta}\Big|\big[\1_{\{X_t> v_t(\widetilde{\vtheta}^v)\}}-\1_{\{X_t> v_t(\vtheta^v)\}}\big]\nabla_i m_t(\widetilde{\vtheta}^m)Y_t\Big|\bigg\Vert_{2r}\\
		&\hspace{3cm} + \bigg\Vert\sup_{\Vert\widetilde{\vtheta} - \vtheta\Vert\leq\delta}\Big|\1_{\{X_t> v_t(\vtheta^v)\}}\Big\{\nabla_i m_t(\widetilde{\vtheta}^m)m_t(\widetilde{\vtheta}^m) - \nabla_i m_t(\vtheta^m)m_t(\vtheta^m)\Big\}\Big|\bigg\Vert_{2r}\\
		&\hspace{3cm} + \bigg\Vert\sup_{\Vert\widetilde{\vtheta} - \vtheta\Vert\leq\delta}\Big|\1_{\{X_t> v_t(\vtheta^v)\}}\Big\{\nabla_i m_t(\widetilde{\vtheta}^m) - \nabla_i m_t(\vtheta^m)\Big\} Y_t\Big|\bigg\Vert_{2r}\\
		&=:A_{8t} + B_{8t} + C_{8t} + D_{8t},
	\end{align*}
	where the third step follows from Minkowski's inequality.
	
	Consider the four terms separately. For the first term, we obtain from the generalized H\"{o}lder inequality that
	\begin{eqnarray*}
		A_{8t} &\leq& \bigg\Vert\sup_{\Vert\widetilde{\vtheta}^v - \vtheta^v\Vert\leq\delta}\big|\1_{\{X_t> v_t(\widetilde{\vtheta}^v)\}}-\1_{\{X_t> v_t(\vtheta^v)\}}\big|\sup_{\vtheta^m\in\mTheta^m}\norm{\nabla m_t(\vtheta^m)}\sup_{\vtheta^m\in\mTheta^m}\big|m_t(\vtheta^m)\big|\bigg\Vert_{2r}\\
		&\overset{\eqref{eq:theta v ast}}{\leq}& \bigg\Vert \big(\1_{\{X_t> v_t(\overline{\vtheta}^{v})\}}-\1_{\{X_t> v_t(\underline{\vtheta}^v)\}}\big)M_1(\mZ_t)M(\mZ_t)\bigg\Vert_{2r}\\
		&\leq& \big\Vert\1_{\{X_t> v_t(\overline{\vtheta}^{v})\}}-\1_{\{X_t> v_t(\underline{\vtheta}^v)\}}\big\Vert_{p}\big\Vert M_1(\mZ_t)\big\Vert_{4r+\iota}\big\Vert M(\mZ_t)\big\Vert_{4r+\iota}\\
		&\overset{\eqref{eq:help ind1}}{\leq}& C \delta^{1/p} \big\Vert M_1(\mZ_t)\big\Vert_{4r+\iota}\big\Vert M(\mZ_t)\big\Vert_{4r+\iota}\\
		&\leq& C \delta^{1/p},
	\end{eqnarray*}
	where $p=2r(4r+\iota)/\iota$.
	Arguing similarly for the second term, we get that 
	\begin{align*}
		B_{8t} &\leq \bigg\Vert \sup_{\Vert\widetilde{\vtheta}^v - \vtheta^v\Vert\leq\delta}\big|\1_{\{X_t> v_t(\widetilde{\vtheta}^v)\}}-\1_{\{X_t> v_t(\vtheta^v)\}}\big|\sup_{\vtheta^m\in\mTheta^m}\norm{\nabla m_t(\vtheta^m)}Y_t\bigg\Vert_{2r}\\
		&\leq \bigg\Vert\big(\1_{\{X_t> v_t(\overline{\vtheta}^{v})\}}-\1_{\{X_t> v_t(\underline{\vtheta}^v)\}}\big)M_1(\mZ_t)Y_t\bigg\Vert_{2r}\\
		&\leq \big\Vert\1_{\{X_t> v_t(\overline{\vtheta}^{v})\}}-\1_{\{X_t> v_t(\underline{\vtheta}^v)\}}\big\Vert_{p}\big\Vert M_1(\mZ_t)\big\Vert_{4r+\iota}\big\Vert Y_t\big\Vert_{4r+\iota} \\
		&\leq C \delta^{1/p} \big\Vert M_1(\mZ_t)\big\Vert_{4r+\iota} \big\Vert Y_t\big\Vert_{4r+\iota}\\
		&\leq C\delta^{1/p}.
	\end{align*}
	Moreover, from Minkowski's inequality,
	\begin{align*}
		C_{8t} &\leq \bigg\Vert\sup_{\Vert\widetilde{\vtheta}^m - \vtheta^m\Vert\leq\delta}\big|\nabla m_t(\widetilde{\vtheta}^m)m_t(\widetilde{\vtheta}^m) - \nabla m_t(\vtheta^m)m_t(\vtheta^m)\big|\bigg\Vert_{2r}\\
		&=\bigg\Vert\sup_{\Vert\widetilde{\vtheta}^m - \vtheta^m\Vert\leq\delta}\Big|\big\{\nabla m_t(\widetilde{\vtheta}^m) - \nabla m_t(\vtheta^m)\big\}m_t(\widetilde{\vtheta}^m) + \nabla m_t(\vtheta^m)\big\{m_t(\widetilde{\vtheta}^m) - m_t(\vtheta^m)\big\}\Big|\bigg\Vert_{2r}\\
		&\leq \bigg\Vert\sup_{\Vert\widetilde{\vtheta}^m - \vtheta^m\Vert\leq\delta}\Big|\big\{\nabla m_t(\widetilde{\vtheta}^m) - \nabla m_t(\vtheta^m)\big\}m_t(\widetilde{\vtheta}^m)\Big|\bigg\Vert_{2r}\\
		&\hspace{2cm} + \bigg\Vert\sup_{\Vert\widetilde{\vtheta}^m - \vtheta^m\Vert\leq\delta}\Big|\nabla m_t(\vtheta^m)\big\{m_t(\widetilde{\vtheta}^m) - m_t(\vtheta^m)\big\}\Big|\bigg\Vert_{2r}\\
		&\leq \big\Vert M_2(\mZ_t)M(\mZ_t)\big\Vert_{2r}\delta\\
		&\hspace{2cm} + \big\Vert M_1(\mZ_t)M_1(\mZ_t)\big\Vert_{2r}\delta\\
		&\leq C\delta.
	\end{align*}
	Finally, from H\"{o}lder's inequality and Assumption~\ref{ass:5},
	\begin{align*}
		D_{8t} &\leq \bigg\Vert Y_t \sup_{\Vert\widetilde{\vtheta}^m - \vtheta^m\Vert\leq\delta}\Big|\nabla m_t(\widetilde{\vtheta}^m) - \nabla m_t(\vtheta^m)\Big|\bigg\Vert_{2r}\\
		&\leq  \big\Vert Y_t M_2(\mZ_t)\big\Vert_{2r}\delta\\
		&\leq \big\Vert Y_t\big\Vert_{4r} \big\Vert M_2(\mZ_t)\big\Vert_{4r}\delta\\
		&\leq C\delta.
	\end{align*}
	Overall, \eqref{eq:TSL4} follows with $\psi=1/p$.
\end{proof}

\begin{lemM}\label{lem:3 tilde}
	Suppose Assumptions~\ref{ass:1}--\ref{ass:7} hold. Then, (on a possibly enlarged probability space) there exists a sequence of zero-mean Gaussian processes $\mB^{(n)}(\cdot,\cdot)$ with a.s.~continuous sample paths and covariance function $\mGamma(\cdot,\cdot)$, such that
	\begin{equation*}
		\sup_{\vtheta=(\vtheta^{m\prime},\vtheta^{v\prime})^\prime\in\mTheta}\Big\Vert\sqrt{n} \big[\widehat{\vlambda}_n(\vtheta^m,\vtheta^v) - \vlambda(\vtheta^m,\vtheta^v)\big]- \mB^{(n)}(\vtheta^m,\vtheta^v)\Big\Vert=o_{\P}(1),
	\end{equation*}
	where the covariance function satisfies
	\begin{multline*}
		\mGamma(\vtheta,\vtheta)=\E\big[\Mm_1(\vtheta^m,\vtheta^v)\Mm_1^\prime(\vtheta^m,\vtheta^v)\big] \\
		+ \sum_{j=1}^{\infty}\Big\{\E\big[\Mm_1(\vtheta^m,\vtheta^v)\Mm_{1+j}^\prime(\vtheta^m,\vtheta^v)\big] + \E\big[\Mm_{1+j}(\vtheta^m,\vtheta^v)\Mm_1^\prime(\vtheta^m,\vtheta^v)\big]\Big\}.
	\end{multline*}
\end{lemM}

\begin{proof}
	The arguments resemble those used to prove Lemma~V.\ref{lem:3}. Note that $\Mm_t(\cdot,\cdot)$ has uniformly bounded $2r$-th moments, because (using \eqref{eq:(12+)} and the $c_r$-inequality)
	\begin{align*}
		\E\bigg[\sup_{\vtheta\in\mTheta}\big\Vert\Mm_t(\vtheta^m,\vtheta^v)\big\Vert^{2r}\bigg] &\leq \E\Big[M_1^{2r}(\mZ_t)\big\{\big|M(\mZ_t)\big|+|Y_t|\big\}^{2r}\Big]\\
		&\leq C\Big\{\E\big[M_1^{2r}(\mZ_t)M^{2r}(\mZ_t)\big] + \E\big[M_1^{2r}(\mZ_t)|Y_t|^{2r}\big]\Big\}\\
		&\leq C <\infty.
	\end{align*}
	
	From Lemma~\ref{lem:type IV Mm} and Theorem~5 of \citet{And94} it follows that $\Mm_t(\cdot,\cdot)$ satisfies ``Ossiander's $L^{2r}$-entropy condition'' (with ``$L^{2r}$-envelope'' given by the supremum $\sup_{\vtheta\in\mTheta}\Vert\Mm_t(\vtheta^m,\vtheta^v)\Vert$). 
	
	Moreover, the mixing condition of Theorem~1 in \citet{DMR95} is satisfied due to $\sum_{t=1}^{\infty}t^{1/(r-1)}\beta(t)<\infty$ from Assumption~\ref{ass:2} (and standard mixing inequalities from, e.g., \citet{Bra05}).
	
	Therefore, the desired FCLT once again follows from Theorem~1 and Application~1 of \citet{DMR95}.
\end{proof}

\subsection{Joint Asymptotic Normality of the VaR and MES Parameter Estimator}\label{Joint Asymptotic Normality of the VaR and MES Parameter Estimator}

\begin{proof}[{\textbf{Proof of Theorem~\ref{thm:an} (Joint asymptotic normality of $(\widehat{\vtheta}_n^v, \widehat{\vtheta}_n^m)$):}}]
	Rewrite \eqref{eq:exp VaR} using the subsequent items (i)--(iv) to get that
	\begin{equation}\label{eq:decomp VaR simple}
		\sqrt{n}(\widehat{\vtheta}_n^v-\vtheta_0^v)=\big[-\mLambda^{-1}+o_{\P}(1)\big]\Big\{\sqrt{n}\big[\widehat{\vlambda}_n(\vtheta_0^v)-\vlambda(\vtheta_0^v)\big] + o_{\P}(1)\Big\}.
	\end{equation}
	Similarly, we obtain from \eqref{eq:expan theta hat m} and the subsequent items (i)--(v) that
	\begin{align*}
		\sqrt{n}(\widehat{\vtheta}_n^m-\vtheta_0^m) & = \mLambda_{(1)}^{-1}\mLambda_{(2)}\mLambda^{-1}\sqrt{n}\big[\widehat{\vlambda}_n(\vtheta_0^v)-\vlambda(\vtheta_0^v)\big] \\
		&\hspace{2cm} - \mLambda_{(1)}^{-1}\sqrt{n}\big[\widehat{\vlambda}_n(\vtheta_0^m,\vtheta_0^v)-\vlambda(\vtheta_0^m,\vtheta_0^v)\big] + o_{\P}(1).
	\end{align*}
	Combining this with \eqref{eq:decomp VaR simple} yields that
	\begin{align*}
		\sqrt{n}\begin{pmatrix}\widehat{\vtheta}_n^v-\vtheta_0^v\\
			\widehat{\vtheta}_n^m-\vtheta_0^m\end{pmatrix}&=
		\begin{pmatrix}
			-\mLambda^{-1} & \vzeros\\
			\mLambda_{(1)}^{-1}\mLambda_{(2)}\mLambda^{-1} & -\mLambda_{(1)}^{-1}
		\end{pmatrix}
		\sqrt{n}
		\begin{pmatrix}
			\widehat{\vlambda}_n(\vtheta_0^v)-\vlambda(\vtheta_0^v)\\
			\widehat{\vlambda}_n(\vtheta_0^m,\vtheta_0^v)-\vlambda(\vtheta_0^m,\vtheta_0^v)
		\end{pmatrix}+o_{\P}(1)\\
		&=(-\mGamma)
		\sqrt{n}
		\begin{pmatrix}
			\widehat{\vlambda}_n(\vtheta_0^v)-\vlambda(\vtheta_0^v)\\
			\widehat{\vlambda}_n(\vtheta_0^m,\vtheta_0^v)-\vlambda(\vtheta_0^m,\vtheta_0^v)
		\end{pmatrix}+o_{\P}(1)
	\end{align*}
	The conclusion now follows from \eqref{eq:joint lambda} and the CMT.
\end{proof}

\section{Estimating the Asymptotic Variance of Theorem~\ref{thm:an}}
\label{sec:thm3}

\renewcommand{\theequation}{B.\arabic{equation}}	% different equation numbering in appendix
\setcounter{equation}{0}

For the asymptotic normality result of Theorem~\ref{thm:an} to be of use in inference, we have to estimate the asymptotic variance-covariance matrix $\mGamma \mM \mGamma^\prime$ consistently. 
To do so, we introduce the scalar bandwidth $c_n>0$, whose convergence rate is specified in Assumption~\ref{ass:8} below. 
We propose the following estimators:
\begin{align*}
	\widehat{\mV}_n &= \frac{\b(1-\b)}{n}\sum_{t=1}^{n}\nabla v_{t}(\widehat{\vtheta}_n^v)\nabla^\prime v_{t}(\widehat{\vtheta}_n^v),\\
	\widehat{\mLambda}_n &= \frac{1}{n}\sum_{t=1}^{n}(2c_n)^{-1}\1_{\{|X_t-v_t(\widehat{\vtheta}_n^v)|<c_n\}}\nabla v_t(\widehat{\vtheta}_n^v)\nabla^\prime v_t(\widehat{\vtheta}_n^v),\\
	\widehat{\mM}_n^{\ast} &= \frac{1}{n}\sum_{t=1}^{n}\big\{Y_t-m_t(\widehat{\vtheta}_n^m)\big\}^2\1_{\{X_t>v_t(\widehat{\vtheta}_n^v)\}}  \nabla m_t(\widehat{\vtheta}_n^m)\nabla^\prime m_t(\widehat{\vtheta}_n^m),\\
	\widehat{\mLambda}_{n,(1)} &= \frac{(1-\b)}{n}\sum_{t=1}^{n}\nabla m_t(\widehat{\vtheta}_n^m)\nabla^\prime m_t(\widehat{\vtheta}_n^m),\\
	\widehat{\mLambda}_{n,(2)} &=\frac{1}{n}\sum_{t=1}^{n}\big\{Y_t - m_t(\widehat{\vtheta}_n^m)\big\}(2c_n)^{-1}\1_{\{|X_t-v_t(\widehat{\vtheta}_n^v)|<c_n\}}\nabla m_t(\widehat{\vtheta}_n^m)\nabla^\prime v_t(\widehat{\vtheta}_n^v).
\end{align*}
The estimators $\widehat{\mV}_n$ and $\widehat{\mLambda}_n$ are standard estimators in the (non-linear) quantile regression literature, known from, e.g., \citet[Theorem~3]{EM04} and \citet[Theorem~4]{OH16}. The remaining estimators are specific to MES regressions, and are reasonably straightforward variations of the aforementioned estimators.

To prove consistency of these estimators, we require the following additional regularity conditions.

\setcounter{assumption}{7}	
\begin{assumption}\label{ass:9}
	\renewcommand{\theenumi}{(\roman{enumi})}
	\begin{enumerate}
		
		\item\label{it:9i} $\frac{1}{n}\sum_{t=1}^{n}\nabla v_t(\vtheta_0^v)\nabla^\prime v_t(\vtheta_0^v)f_t^{X}\big(v_t(\vtheta_0^v)\big) \overset{\P}{\longrightarrow} \E\big[\nabla v_t(\vtheta_0^v)\nabla^\prime v_t(\vtheta_0^v)f_t^{X}\big(v_t(\vtheta_0^v)\big)\big]$.
		
		\item\label{it:9ii} $\frac{1}{n}\sum_{t=1}^{n}\nabla m_t(\vtheta_0^m)\nabla^\prime v_t(\vtheta_0^v)\big\{\int_{-\infty}^{\infty} y f_t\big(v_t(\vtheta_0^v),y\big)\D y - m_t(\vtheta_0^m)f_t^{X}\big(v_t(\vtheta_0^v)\big)\big\} \overset{\P}{\longrightarrow}$\\	
		$\E\Big[\nabla m_t(\vtheta_0^v)\nabla^\prime v_t(\vtheta_0^v)\big\{\int_{-\infty}^{\infty} y f_t\big(v_t(\vtheta_0^v),y\big)\D y - m_t(\vtheta_0^m)f_t^{X}\big(v_t(\vtheta_0^v)\big)\big\}\Big]$.
		
	\end{enumerate}
\end{assumption}

\begin{assumption}\label{ass:10}
	$\sup_{x\in\mathbb{R}}\int_{-\infty}^{\infty}|y|^{2} f_t(x,y)\D y\leq F_{2}(\mathcal{F}_t)$ for some $\mathcal{F}_t$-measurable random variable $F_{2}(\mathcal{F}_t)$.
\end{assumption}

\begin{assumption}\label{ass:11}
	$\E\big[V_1^2(\mZ_{t})M^2(\mZ_{t})M_1^2(\mZ_{t})\big]\leq K$, $\E\big[M_1^2(\mZ_t)V_1^2(\mZ_t)F_2(\mathcal{F}_t)\big]\leq K$.
\end{assumption}

\begin{assumption}\label{ass:8} 
	The bandwidth $c_{n}>0$  satisfies that $c_{n}=o(1)$ and $c_{n}^{-1}=o\big(n^{1/2}\big)$, as $n\to\infty$.
\end{assumption}

The serial dependence in (functionals of) the conditional density $f_t(\cdot,\cdot)$ is restricted by ensuring that the laws of large numbers (LLN) in items~\ref{it:9i} and \ref{it:9ii} of Assumption~\ref{ass:9} hold.
To see this more clearly, note that, as a function of $\mZ_t$, $\big\{\nabla v_t(\vtheta_0^v)\nabla^\prime v_t(\vtheta_0^v)\big\}$ is $\beta$-mixing of size $-r/(r-1)$.
Therefore, $\big\{\nabla v_t(\vtheta_0^v)\nabla^\prime v_t(\vtheta_0^v)\big\}$ obeys a standard mixing LLN.
Now, Assumption~\ref{ass:9}~\ref{it:9i} restricts the serial dependence in $f_t^{X}\big(v_t(\vtheta_0^v)\big)$ by requiring that a LLN also holds for $\big\{\nabla v_t(\vtheta_0^v)\nabla^\prime v_t(\vtheta_0^v)f_t^{X}\big(v_t(\vtheta_0^v)\big)\big\}$.
A similar argument also applies to Assumption~\ref{ass:9}~\ref{it:9ii}.
Therefore, Assumption~\ref{ass:9} may be interpreted as bounding the time series dependence of $f_t(\cdot,\cdot)$.
Note that $f_t(\cdot,\cdot)$ does not have to be known for the above estimators to be consistent.
Related assumptions in the literature are Assumption VC3 in \citet{EM04} and Assumption~3 in \citet{PZC19}.
Assumption~\ref{ass:10} provides a bound for the conditional density, similar to Assumption~\ref{ass:3}~\ref{it:3iv}.
The final Assumption~\ref{ass:11} strengthens the moment conditions of Assumption~\ref{ass:5}.
Finally, Assumption~\ref{ass:8} is a standard condition for the bandwidth in quantile regression contexts \citep{EM04,PZC19}. 

\setcounter{rem}{3}	
\begin{rem}[Choice of Bandwidth $c_n$]
	The bandwidth $c_n$ from Assumption~\ref{ass:8} is only specified asymptotically.
	In finite samples, we suggest to choose $c_n$ following \citet{Koenker2005book} and \citet{MachadoSantosSilva2013}. 
	Formally, we let 
	\[
	c_n = \text{MAD} \Big[\big\{X_t-v_t(\widehat{\vtheta}_n^v) \big\}_{t=1,\ldots,n}\Big] \left[ \Phi^{-1} \big( \beta + m(n,\beta) \big) - \Phi^{-1} \big( \beta - m(n,\beta) \big) \right]
	\]
	with 
	\[
	m(n, \tau) =  n^{-1/3} \left( \Phi^{-1} (0.975)\right)^{2/3} \left( \frac{1.5 (\phi(\Phi^{-1}(\tau)))^2}{2(\Phi^{-1}(\tau))^2 +1} \right)^{1/3},
	\]
	where $\text{MAD}(\cdot)$ refers to the sample median absolute deviation, and $\phi(\cdot)$ and $\Phi^{-1}(\cdot)$ denote the density and quantile functions of the standard normal distribution. 
	In the simulations and the applications in the main paper (Sections~\ref{sec:Simulations} and \ref{sec:EmpiricalApplication}), we adopt the above choice of $c_n$.
\end{rem}

\setcounter{thm}{1}	
\begin{thm}\label{thm:lrv}
	Suppose Assumptions~\ref{ass:1}--\ref{ass:8} hold. Then, as $n\to\infty$, $\widehat{\mV}_n \overset{\P}{\longrightarrow}\mV$, $\widehat{\mLambda}_n \overset{\P}{\longrightarrow}\mLambda$, $\widehat{\mM}_n^\ast \overset{\P}{\longrightarrow}\mM^\ast$, $\widehat{\mLambda}_{n,(1)} \overset{\P}{\longrightarrow}\mLambda_{(1)}$ and $\widehat{\mLambda}_{n,(2)} \overset{\P}{\longrightarrow}\mLambda_{(2)}$.
\end{thm}

\begin{proof}[{\textbf{Proof of Theorem~\ref{thm:lrv} ($\widehat{\mLambda}_n\overset{\P}{\longrightarrow}\mLambda$, $\widehat{\mV}_n\overset{\P}{\longrightarrow}\mV$, $\widehat{\mLambda}_{n,(1)}\overset{\P}{\longrightarrow}\mLambda_{(1)}$):}}]
	We first show that $\widehat{\mLambda}_n\overset{\P}{\longrightarrow}\mLambda$. Define
	\begin{equation*}
		\widetilde{\mLambda}_n = \frac{1}{n}\sum_{t=1}^{n}(2c_n)^{-1}\1_{\{|X_t-v_t(\vtheta_0^v)|<c_n\}}\nabla v_t(\vtheta_0^v)\nabla^\prime v_t(\vtheta_0^v).
	\end{equation*}
	Then, it suffices to show that
	\begin{align}
		\widehat{\mLambda}_n - \widetilde{\mLambda}_n&=o_{\P}(1),\label{eq:Lambda1}\\
		\widetilde{\mLambda}_n - \mLambda&=o_{\P}(1).\label{eq:Lambda2}
	\end{align}
	First, we establish \eqref{eq:Lambda1}. Use the MVT with $\vtheta^{\ast}$ some mean value on the line connecting $\vtheta_0^v$ and $\widehat{\vtheta}_n^v$ to write
	\begin{align*}
		\big\Vert\widehat{\mLambda}_n - \widetilde{\mLambda}_n\big\Vert &= \bigg\Vert\frac{1}{n}\sum_{t=1}^{n}(2c_n)^{-1}\1_{\{|X_t-v_t(\vtheta_0^v)|<c_n\}}\Big\{ \big[\nabla v_t(\widehat{\vtheta}_n^v) - \nabla v_t(\vtheta_0^v)\big]\nabla^\prime v_t(\widehat{\vtheta}_n^v) \\
		&\hspace{6cm} + \nabla v_t(\vtheta_0^v)\big[\nabla^\prime v_t(\widehat{\vtheta}_n^v) - \nabla^\prime v_t(\vtheta_0^v)\big] \Big\} \\
		&\hspace{0.6cm} +\frac{1}{n}\sum_{t=1}^{n}(2c_n)^{-1}\big[\1_{\{|X_t-v_t(\widehat{\vtheta}_n^v)|<c_n\}} - \1_{\{|X_t-v_t(\vtheta_0^v)|<c_n\}}\big]\nabla v_t(\widehat{\vtheta}_n^v)\nabla^\prime v_t(\widehat{\vtheta}_n^v) \bigg\Vert\\
		&\leq \frac{1}{n}\sum_{t=1}^{n}(2c_n)^{-1}\1_{\{|X_t-v_t(\vtheta_0^v)|<c_n\}}\Big\Vert \nabla^2 v_t(\vtheta^{\ast})\big(\widehat{\vtheta}_n^v - \vtheta_0^v\big)\nabla^\prime v_t(\widehat{\vtheta}_n^v) \\
		&\hspace{6cm} + \nabla v_t(\vtheta_0^v)\big(\widehat{\vtheta}_n^v - \vtheta_0^v\big)^\prime \nabla^{2\prime} v_t(\vtheta^{\ast}) \Big\Vert \\
		&\hspace{0.6cm} +\frac{1}{n}\sum_{t=1}^{n}(2c_n)^{-1}\big|\1_{\{|X_t-v_t(\widehat{\vtheta}_n^v)|<c_n\}} - \1_{\{|X_t-v_t(\vtheta_0^v)|<c_n\}}\big|\cdot\big\Vert\nabla v_t(\widehat{\vtheta}_n^v)\nabla^\prime v_t(\widehat{\vtheta}_n^v) \big\Vert\\
		& \leq \frac{1}{n}\sum_{t=1}^{n}(2c_n)^{-1}\1_{\{|X_t-v_t(\vtheta_0^v)|<c_n\}}2V_2(\mZ_t) V_1(\mZ_t)\big\Vert \widehat{\vtheta}_n^v - \vtheta_0^v \big\Vert \\
		&\hspace{0.6cm} +\frac{1}{n}\sum_{t=1}^{n}(2c_n)^{-1}\big|\1_{\{|X_t-v_t(\widehat{\vtheta}_n^v)|<c_n\}} - \1_{\{|X_t-v_t(\vtheta_0^v)|<c_n\}}\big|V_1^2(\mZ_t)\\
		&=: A_{9n} + B_{9n}.
	\end{align*}
	To prove $A_{9n}=o_{\P}(1)$, we use the LIE to obtain that
	\begin{align*}
		\E&\bigg[\frac{1}{n}\sum_{t=1}^{n}(2c_n)^{-1}\1_{\{|X_t-v_t(\vtheta_0^v)|<c_n\}}2V_2(\mZ_t) V_1(\mZ_t)\bigg]\\
		&= \frac{1}{c_n n}\sum_{t=1}^{n}\E\Big\{V_2(\mZ_t) V_1(\mZ_t)\E_{t}\big[\1_{\{-c_n<X_t-v_t(\vtheta_0^v)<c_n\}}\big]\Big\}\\
		&=\frac{1}{c_n n}\sum_{t=1}^{n}\E\Big\{V_2(\mZ_t) V_1(\mZ_t)\big[F_t^{X}(v_t(\vtheta_0^v) + c_n) - F_t^{X}(v_t(\vtheta_0^v) - c_n)\big]\Big\}\\
		&\leq\frac{1}{c_n n}\sum_{t=1}^{n}\E\Big\{V_2(\mZ_t) V_1(\mZ_t)\sup_{x\in\mathbb{R}}f_t^{X}(x)2c_n\Big\}\\
		&\leq \frac{2K}{n}\sum_{t=1}^{n}\big\Vert V_2(\mZ_t)\big\Vert_2 \big\Vert V_1(\mZ_t)\big\Vert_2\\
		&\leq C<\infty.
	\end{align*}
	Hence, by Markov's inequality, $\frac{1}{n}\sum_{t=1}^{n}(2c_n)^{-1}\1_{\{|X_t-v_t(\vtheta_0^v)|<c_n\}}2V_2(\mZ_t) V_1(\mZ_t)=O_{\P}(1)$. Combining this with $\big\Vert \widehat{\vtheta}_n^v - \vtheta_0^v \big\Vert=o_{\P}(1)$ (from Proposition~\ref{prop:cons}), $A_{9n}=O_{\P}(1)o_{\P}(1)=o_{\P}(1)$ follows.
	
	To show $B_{9n}=o_{\P}(1)$, it is easy to check that
	\begin{multline}\label{eq:ind ineq}
		\big|\1_{\{|X_t-v_t(\widehat{\vtheta}_n^v)|<c_n\}} - \1_{\{|X_t-v_t(\vtheta_0^v)|<c_n\}}\big|\leq \1_{\{|X_t-v_t(\vtheta_0^v)-c_n|<|v_t(\widehat{\vtheta}_n^v) - v_t(\vtheta_0^v)|\}}\\
		+\1_{\{|X_t-v_t(\vtheta_0^v)+c_n|<|v_t(\widehat{\vtheta}_n^v) - v_t(\vtheta_0^v)|\}}.
	\end{multline}
	Theorem~\ref{thm:an} and Assumption~\ref{ass:8} imply that for any $d>0$ the event $\mathsf{E}_{1n}:=\big\{c_n^{-1}\Vert\widehat{\vtheta}_n^v - \vtheta_0^v\Vert\leq d\big\}$ occurs w.p.a.~1, as $n\to\infty$. Therefore, 
	\begin{align}
		&\P\big\{B_{9n}>\varepsilon\big\}\notag\\
		&\leq \P\big\{B_{9n}>\varepsilon,\, \mathsf{E}_{1n}\big\} + \P\big\{\mathsf{E}_{1n}^{C}\big\}\notag\\
		&\leq \P\bigg\{\frac{1}{n}\sum_{t=1}^{n}(2c_n)^{-1}\big|\1_{\{|X_t-v_t(\widehat{\vtheta}_n^v)|<c_n\}} - \1_{\{|X_t-v_t(\vtheta_0^v)|<c_n\}}\big|V_1^2(\mZ_t)>\varepsilon,\, \mathsf{E}_{1n}\bigg\} + o(1)\notag\\
		&\leq \frac{1}{2\varepsilon c_n n}\sum_{t=1}^{n}\bigg\{\E\Big[\1_{\{|X_t-v_t(\vtheta_0^v)-c_n|<V_1(\mZ_t)dc_n\}}V_1^2(\mZ_t)\Big] + \E\Big[\1_{\{|X_t-v_t(\vtheta_0^v)+c_n|<V_1(\mZ_t)dc_n\}}V_1^2(\mZ_t)\Big]\bigg\}\notag\\
		&\hspace{13.5cm}+o(1),\label{eq:B9n}
	\end{align}
	where we used \eqref{eq:ind ineq} for the final inequality and $\mathsf{E}_{1n}^{C}$ denotes the complement of $\mathsf{E}_{1n}$.
	Use the LIE to get that
	\begin{align*}
		\E&\Big[\1_{\{|X_t-v_t(\vtheta_0^v)-c_n|<V_1(\mZ_t)dc_n\}}V_1^2(\mZ_t)\Big]\\
		&= \E\Big[V_1^2(\mZ_t) \E_{t}\big\{\1_{\{-V_1(\mZ_t)dc_n<X_t-v_t(\vtheta_0^v)-c_n<V_1(\mZ_t)dc_n\}}\big\}\Big]\\
		&= \E\Big[V_1^2(\mZ_t) \Big\{F_t^{X}\big(v_t(\vtheta_0^v)+c_n+V_1(\mZ_t)dc_n\big) - F_t^{X}\big(v_t(\vtheta_0^v)+c_n-V_1(\mZ_t)dc_n\big)\Big\}\Big]\\
		&\leq \E\Big[V_1^2(\mZ_t) \sup_{x\in\mathbb{R}}\big\{f_t^{X}(x)\big\}2V_1(\mZ_t)dc_n \Big]\\
		&\leq 2Kd c_n\E\big[V_1^3(\mZ_t)\big]\\
		&\leq Cdc_n,
	\end{align*}
	with an identical bound in force for the other expectation in \eqref{eq:B9n}. Therefore, from \eqref{eq:B9n}, 
	\[
	\P\big\{B_{9n}>\varepsilon\big\}\leq Cd + o(1).
	\]
	Since $d>0$ can be chosen arbitrarily small, $B_{9n}=o_{\P}(1)$ follows. This establishes \eqref{eq:Lambda1}.
	
	It remains to prove \eqref{eq:Lambda2}. Write
	\begin{align*}
		\widetilde{\mLambda}_n - \mLambda &= \frac{1}{n}\sum_{t=1}^{n}\nabla v_t(\vtheta_0^v)\nabla^\prime v_t(\vtheta_0^v) (2c_n)^{-1}\Big\{\1_{\{|X_t-v_t(\vtheta_0^v)|<c_n\}} - \E_t\big[\1_{\{|X_t-v_t(\vtheta_0^v)|<c_n\}}\big] \Big\}\\
		&\hspace{0.5cm} +\frac{1}{n}\sum_{t=1}^{n}\nabla v_t(\vtheta_0^v)\nabla^\prime v_t(\vtheta_0^v) \Big\{(2c_n)^{-1}\E_t\big[\1_{\{|X_t-v_t(\vtheta_0^v)|<c_n\}}\big] - f_t^{X}\big(v_t(\vtheta_0^v)\big) \Big\}\\
		&\hspace{0.5cm} +\frac{1}{n}\sum_{t=1}^{n}\Big\{\nabla v_t(\vtheta_0^v)\nabla^\prime v_t(\vtheta_0^v)f_t^{X}\big(v_t(\vtheta_0^v)\big) - \E\big[\nabla v_t(\vtheta_0^v)\nabla^\prime v_t(\vtheta_0^v)f_t^{X}\big(v_t(\vtheta_0^v)\big)\big]  \Big\} \\
		&=:A_{10n} + B_{10n} + C_{10n}.
	\end{align*}
	We show that each of these terms is $o_{\P}(1)$. 
	
	By the LIE, $\E\big[A_{10n}\big]=\vzeros$. Hence, for any $(i,j)$-th element of $A_{10n}$, denoted by $A_{10n,ij}$, we have that
	\begin{align*}
		\Var\big(A_{10n,ij}\big) &= \E\big[A_{10n,ij}^2\big]\\
		&= (2c_n n)^{-2}\E\bigg[\sum_{t=1}^{n}\nabla_i v_t(\vtheta_0^v)\nabla_j v_t(\vtheta_0^v) (2c_n)^{-1}\Big\{\1_{\{|X_t-v_t(\vtheta_0^v)|<c_n\}} - \E_t\big[\1_{\{|X_t-v_t(\vtheta_0^v)|<c_n\}}\big] \Big\}\bigg]^2\\
		&= (2c_n n)^{-2}\sum_{t=1}^{n}\E\bigg[\Big\{\nabla_i v_t(\vtheta_0^v)\nabla_j v_t(\vtheta_0^v)\Big\}^2 \Big\{\1_{\{|X_t-v_t(\vtheta_0^v)|<c_n\}} - \E_t\big[\1_{\{|X_t-v_t(\vtheta_0^v)|<c_n\}}\big] \Big\}^2\bigg]\\
		&\leq (2c_n n)^{-2}\sum_{t=1}^{n}\E\big[V_1^4(\mZ_t)\big]\\
		&\leq C c_n^{-2} n^{-1}=o(1),
	\end{align*}
	where we used for the third equality that all covariance terms are zero by the LIE, and the final equality follows from Assumption~\ref{ass:8}. Chebyshev's inequality now implies that $A_{10n}=o_{\P}(1)$.
	
	For $B_{10n}$, note that by the MVT for some $c^{\ast}\in[-c_n, c_n]$
	\begin{align*}
		\E_t\big[\1_{\{|X_t-v_t(\vtheta_0^v)|<c_n\}}\big] &= \E_t\big[\1_{\{-c_n<X_t-v_t(\vtheta_0^v)<c_n\}}\big]\\
		&= F_t^{X}\big(v_t(\vtheta_0^v) + c_n\big) - F_t^{X}\big(v_t(\vtheta_0^v) - c_n\big)\\
		&=  f_t^X\big(v_t(\vtheta_0^v) + c^{\ast}\big) 2c_n,
	\end{align*}
	such that (also using Assumption~\ref{ass:3}~\ref{it:3ii})
	\begin{align*}
		\E\big\Vert B_{10n}\big\Vert &\leq \frac{1}{n}\sum_{t=1}^{n}\E\bigg[V_1^2(\mZ_t) \Big|(2c_n)^{-1}\E_t\big[\1_{\{|X_t-v_t(\vtheta_0^v)|<c_n\}}\big] - f_t^{X}\big(v_t(\vtheta_0^v)\big) \Big|\bigg]\\
		&\leq \frac{1}{n}\sum_{t=1}^{n}\E\bigg[V_1^2(\mZ_t) \Big|f_t^X\big(v_t(\vtheta_0^v) + c^{\ast}\big) - f_t^{X}\big(v_t(\vtheta_0^v)\big) \Big|\bigg]\\
		&\leq\frac{1}{n}\sum_{t=1}^{n}\E\big[V_1^2(\mZ_t) \big] Kc_n\\
		&\leq C c_n=o(1).
	\end{align*}
	Therefore, by Markov's inequality, $B_{10n}=o_{\P}(1)$ follows. 
	
	Finally, $C_{10n}=o_{\P}(1)$ is immediate from Assumption~\ref{ass:9}~\ref{it:9i}.
	
	This establishes \eqref{eq:Lambda2}, such that $\widehat{\mLambda}_n\overset{\P}{\longrightarrow}\mLambda$ follows.
	
	%Finally, $C_{10n}=o_{\P}(1)$ follows from the ergodic theorem \citep[e.g.,][Theorem 3.34]{Whi01} as follows. Recall from the proof of Proposition~\ref{prop:cons} that $\{(X_t, Y_t, \mZ_t^\prime)^\prime\}_t$ is $\alpha$-mixing of size $-r/(r-1)$. Then, since the conditional (on $\mZ_t$) density $f_t^X\big(v_t(\vtheta_0^v)\big)$ is $\mZ_t$-measurable and (by Assumption~\ref{ass:cons} \ref{it:str stat}) stationary, $\{f_t^X\big(v_t(\vtheta_0^v)\big)\}_{t\in\mathbb{N}}$ is also $\alpha$-mixing of size $-r/(r-1)$ as a function of $\alpha$-mixing random variables of size $-r/(r-1)$ \citep[e.g.,][Theorem~3.49]{Whi01}. By the same arguments, $\big\{\nabla v_t(\vtheta_0^v)\nabla^\prime v_t(\vtheta_0^v)f_t^{X}\big(v_t(\vtheta_0^v)\big)\big\}_{t\in\mathbb{N}}$ is also $\alpha$-mixing of size $-r/(r-1)$ and stationary. By \citet[Proposition~3.44]{Whi01}, this implies ergodicity of $\big\{\nabla v_t(\vtheta_0^v)\nabla^\prime v_t(\vtheta_0^v)f_t^{X}\big(v_t(\vtheta_0^v)\big)\big\}_{t\in\mathbb{N}}$. Since also 
	%\[
	%\E\Big\Vert \nabla v_t(\vtheta_0^v) \nabla^\prime v_t(\vtheta_0^v) f_t^{X}\big(v_t(\vtheta_0^v)\big)\Big\Vert\leq K\E[V_1^2(\mZ_t)]\leq C<\infty,
	%\]
	%the ergodic theorem \citep[e.g.,][Theorem 3.34]{Whi01} immediately implies that $C_{10n}=o_{\P}(1)$, which establishes \eqref{eq:Lambda2} and, hence, $\widehat{\mLambda}_n\overset{\P}{\longrightarrow}\mLambda$.
	
	The proofs that $\widehat{\mV}_n\overset{\P}{\longrightarrow}\mV$, $\widehat{\mLambda}_{n,(1)}\overset{\P}{\longrightarrow}\mLambda_{(1)}$ are similar and, hence, omitted.
\end{proof}

\begin{proof}[{\textbf{Proof of Theorem~\ref{thm:lrv} ($\widehat{\mM}_n^{\ast}\overset{\P}{\longrightarrow}\mM^\ast$):}}]
	For $\widehat{\mM}_n^{\ast}\overset{\P}{\longrightarrow}\mM^\ast$ it suffices to show that
	\begin{align*}
		\widehat{\mM}_{1n}^{\ast}&:=\frac{1}{n}\sum_{t=1}^{n}\1_{\{X_t>v_t(\widehat{\vtheta}_n^v)\}} Y_t^2 \nabla m_t(\widehat{\vtheta}_n^m)\nabla^\prime m_t(\widehat{\vtheta}_n^m)\\
		&\hspace{4cm}\overset{\P}{\longrightarrow}\E\big[\1_{\{X_t>v_t(\vtheta_0^v)\}} Y_t^2 \nabla m_t(\vtheta_0^m)\nabla^\prime m_t(\vtheta_0^m)\big]=:\mM_{1}^\ast,\\
		\widehat{\mM}_{2n}^{\ast}&:=\frac{1}{n}\sum_{t=1}^{n}\1_{\{X_t>v_t(\widehat{\vtheta}_n^v)\}} Y_t m_t(\widehat{\vtheta}_n^m)\nabla m_t(\widehat{\vtheta}_n^m)\nabla^\prime m_t(\widehat{\vtheta}_n^m)\\
		&\hspace{4cm}\overset{\P}{\longrightarrow}\E\big[\1_{\{X_t>v_t(\vtheta_0^v)\}} Y_t m_t(\vtheta_0^m) \nabla m_t(\vtheta_0^m)\nabla^\prime m_t(\vtheta_0^m)\big]=:\mM_{2}^\ast,\\
		\widehat{\mM}_{3n}^{\ast}&:=\frac{1}{n}\sum_{t=1}^{n}\1_{\{X_t>v_t(\widehat{\vtheta}_n^v)\}} m_t^2(\widehat{\vtheta}_n^m)\nabla m_t(\widehat{\vtheta}_n^m)\nabla^\prime m_t(\widehat{\vtheta}_n^m)\\
		&\hspace{4cm}\overset{\P}{\longrightarrow}\E\big[\1_{\{X_t>v_t(\vtheta_0^v)\}} m_t^2(\vtheta_0^m) \nabla m_t(\vtheta_0^m)\nabla^\prime m_t(\vtheta_0^m)\big]=:\mM_{3}^\ast.
	\end{align*}
	We only prove $\widehat{\mM}_{1n}^{\ast}\overset{\P}{\longrightarrow}\mM_{1}^\ast$, because the other convergences follow along similar lines.
	Define
	\begin{equation*}
		\widetilde{\mM}_{1n}^{\ast} := \frac{1}{n}\sum_{t=1}^{n}\1_{\{X_t>v_t(\vtheta_0^v)\}} Y_t^2 \nabla m_t(\vtheta_0^m)\nabla^\prime m_t(\vtheta_0^m).
	\end{equation*}
	Then, it suffices to show that
	\begin{align}
		\widehat{\mM}_{1n}^{\ast} - \widetilde{\mM}_{1n}^{\ast}&=o_{\P}(1),\label{eq:M1}\\
		\widetilde{\mM}_{1n}^{\ast} - \mM_1^{\ast}&=o_{\P}(1).\label{eq:M2}
	\end{align}
	First, we establish \eqref{eq:M1}. Use the MVT with $\vtheta^{\ast}$ some mean value on the line connecting $\vtheta_0^m$ and $\widehat{\vtheta}_n^m$ to write
	\begin{align*}
		\big\Vert\widehat{\mM}_{1n}^{\ast} - \widetilde{\mM}_{1n}^{\ast}\big\Vert &= \bigg\Vert\frac{1}{n}\sum_{t=1}^{n}\1_{\{X_t>v_t(\vtheta_0^v)\}}Y_t^2\Big\{ \big[\nabla m_t(\widehat{\vtheta}_n^m) - \nabla m_t(\vtheta_0^m)\big]\nabla^\prime m_t(\widehat{\vtheta}_n^m) \\
		&\hspace{6cm} + \nabla m_t(\vtheta_0^m)\big[\nabla^\prime m_t(\widehat{\vtheta}_n^m) - \nabla^\prime m_t(\vtheta_0^m)\big] \Big\} \\
		&\hspace{0.6cm} +\frac{1}{n}\sum_{t=1}^{n}\big[\1_{\{X_t>v_t(\widehat{\vtheta}_n^v)\}} - \1_{\{X_t>v_t(\vtheta_0^v)\}}\big]Y_t^2\nabla m_t(\widehat{\vtheta}_n^m)\nabla^\prime m_t(\widehat{\vtheta}_n^m) \bigg\Vert\\
		&\leq \frac{1}{n}\sum_{t=1}^{n}\1_{\{X_t>v_t(\vtheta_0^v)\}}Y_t^2\Big\Vert \nabla^2 m_t(\vtheta^{\ast})\big(\widehat{\vtheta}_n^m - \vtheta_0^m\big)\nabla^\prime m_t(\widehat{\vtheta}_n^m) \\
		&\hspace{6cm} + \nabla m_t(\vtheta_0^m)\big(\widehat{\vtheta}_n^m - \vtheta_0^m\big)^\prime \nabla^{2\prime} m_t(\vtheta^{\ast}) \Big\Vert \\
		&\hspace{0.6cm} +\frac{1}{n}\sum_{t=1}^{n}\big|\1_{\{X_t>v_t(\widehat{\vtheta}_n^v)\}} - \1_{\{X_t>v_t(\vtheta_0^v)\}}\big|Y_t^2\big\Vert\nabla m_t(\widehat{\vtheta}_n^m)\nabla^\prime m_t(\widehat{\vtheta}_n^m) \big\Vert\\
		& \leq \frac{1}{n}\sum_{t=1}^{n}\1_{\{X_t>v_t(\vtheta_0^v)\}}Y_t^2 2M_2(\mZ_t) M_1(\mZ_t)\big\Vert \widehat{\vtheta}_n^m - \vtheta_0^m \big\Vert \\
		&\hspace{0.6cm} +\frac{1}{n}\sum_{t=1}^{n}\big|\1_{\{X_t>v_t(\widehat{\vtheta}_n^v)\}} - \1_{\{X_t>v_t(\vtheta_0^v)\}}\big|Y_t^2 M_1^2(\mZ_t)\\
		&=: A_{11n} + B_{11n}.
	\end{align*}
	To prove $A_{11n}=o_{\P}(1)$, use that
	\begin{align*}
		\E\bigg[\frac{1}{n}\sum_{t=1}^{n}\1_{\{X_t>v_t(\vtheta_0^v)\}}Y_t^2 2 M_2(\mZ_t) M_1(\mZ_t)\bigg]
		&\leq \frac{2}{n}\sum_{t=1}^{n}\E\big[Y_t^2M_2(\mZ_t) M_1(\mZ_t)\big]\\
		&\leq\frac{2}{n}\sum_{t=1}^{n}\big\Vert Y_t^2\big\Vert_{2} \big\Vert M_2(\mZ_t)\big\Vert_{4}  \big\Vert  M_1(\mZ_t)\big\Vert_{4}\\
		&\leq C <\infty.
	\end{align*}
	Hence, by Markov's inequality, $\frac{1}{n}\sum_{t=1}^{n}\1_{\{X_t>v_t(\vtheta_0^v)\}}Y_t^2 2 M_2(\mZ_t) M_1(\mZ_t)=O_{\P}(1)$. Combining this with $\big\Vert \widehat{\vtheta}_n^m - \vtheta_0^m \big\Vert=o_{\P}(1)$ (from Proposition~\ref{prop:cons}), $A_{11n}=O_{\P}(1)o_{\P}(1)=o_{\P}(1)$ follows.
	
	To show $B_{11n}=o_{\P}(1)$, consider the event $\mathsf{E}_{2n}:=\big\{\Vert\widehat{\vtheta}_n^v - \vtheta_0^v\Vert\leq \delta\big\}$ for $\delta>0$, which occurs w.p.a.~1, as $n\to\infty$, by Proposition~\ref{prop:cons}. Therefore, 
	\begin{align}
		&\P\big\{B_{11n}>\varepsilon\big\}\notag\\
		&\leq \P\big\{B_{11n}>\varepsilon,\, \mathsf{E}_{2n}\big\} + \P\big\{\mathsf{E}_{2n}^{C}\big\}\notag\\
		&\leq \P\bigg\{\frac{1}{n}\sum_{t=1}^{n}\big|\1_{\{X_t>v_t(\widehat{\vtheta}_n^v)\}} - \1_{\{X_t>v_t(\vtheta_0^v)\}}\big|Y_t^2 M_1^2(\mZ_t)>\varepsilon,\, \mathsf{E}_{2n}\bigg\} + o(1)\notag\\
		&\leq \frac{1}{\varepsilon n}\sum_{t=1}^{n}\E\bigg\{\sup_{\Vert\vtheta^v-\vtheta_0^v\Vert\leq\delta}\big|\1_{\{X_t>v_t(\vtheta^v)\}} - \1_{\{X_t>v_t(\vtheta_0^v)\}}\big| Y_t^2 M_1^2(\mZ_t)\bigg\}+o(1).\label{eq:B11n}
	\end{align}
	Use (an obvious adaptation of) \eqref{eq:theta v ast}, H\"{o}lder's inequality (with $p=(4r+\iota)/\iota$) and \eqref{eq:help ind1} to get that
	\begin{align*}
		\E&\bigg\{ \sup_{\Vert\vtheta^v-\vtheta_0^v\Vert\leq\delta}\big|\1_{\{X_t>v_t(\vtheta^v)\}} - \1_{\{X_t>v_t(\vtheta_0^v)\}}\big| Y_t^2 M_1^2(\mZ_t)\bigg\}\\
		&\leq \E\bigg\{ \big[\1_{\{X_t>v_t(\overline{\vtheta}^{v})\}} - \1_{\{X_t>v_t(\underline{\vtheta}^v)\}}\big] Y_t^2 M_1^2(\mZ_t)\bigg\}\\
		&\leq \big\Vert\1_{\{X_t>v_t(\overline{\vtheta}^{v})\}} - \1_{\{X_t>v_t(\underline{\vtheta}^v)\}}\big\Vert_{p} \big\Vert Y_t^2\big\Vert_{2+\iota/(2r)} \big\Vert M_1^2(\mZ_t)\big\Vert_{2+\iota/(2r)}\\
		&\leq C\delta^{1/p} \big\Vert Y_t^2\big\Vert_{2+\iota/2} \big\Vert M_1^2(\mZ_t)\big\Vert_{2+\iota/2}\\
		&\leq C\delta^{1/p}.
	\end{align*}
	Therefore, from \eqref{eq:B11n}, 
	\[
	\P\big\{B_{11n}>\varepsilon\big\}\leq C\delta^{1/p} + o(1).
	\]
	Since $\delta>0$ can be chosen arbitrarily small, $B_{11n}=o_{\P}(1)$ follows and, hence, also \eqref{eq:M1}.
	
	It remains to prove \eqref{eq:M2}, i.e.,
	\[
	\frac{1}{n}\sum_{t=1}^{n}\1_{\{X_t>v_t(\vtheta_0^v)\}} Y_t^2 \nabla m_t(\vtheta_0^m)\nabla^\prime m_t(\vtheta_0^m)\overset{\P}{\longrightarrow}\E\big[\1_{\{X_t>v_t(\vtheta_0^v)\}} Y_t^2 \nabla m_t(\vtheta_0^m)\nabla^\prime m_t(\vtheta_0^m)\big].
	\]
	This can be deduced from the ergodic theorem \citep[e.g.,][Theorem 3.34]{Whi01} as follows. Recall from the proof of Proposition~\ref{prop:cons} that $\{(X_t, Y_t, \mZ_t^\prime)^\prime\}_t$ is $\alpha$-mixing of size $-r/(r-1)$. Then, since $\1_{\{X_t<v_t(\vtheta_0^v)\}} Y_t^2 \nabla m_t(\vtheta_0^m)\nabla^\prime m_t(\vtheta_0^m)$ is $(X_t,Y_t,\mZ_t)$-measurable and (by Assumption~\ref{ass:2}) stationary, $\big\{\1_{\{X_t<v_t(\vtheta_0^v)\}} Y_t^2 \nabla m_t(\vtheta_0^m)\nabla^\prime m_t(\vtheta_0^m)\big\}_{t\in\mathbb{N}}$ is also $\alpha$-mixing of size $-r/(r-1)$ as a function of $\alpha$-mixing random variables of size $-r/(r-1)$ \citep[e.g.,][Theorem~3.49]{Whi01}. By \citet[Proposition~3.44]{Whi01}, this implies ergodicity of 
	\[
	\big\{\1_{\{X_t>v_t(\vtheta_0^v)\}} Y_t^2 \nabla m_t(\vtheta_0^m)\nabla^\prime m_t(\vtheta_0^m)\big\}_{t\in\mathbb{N}}.
	\]
	Since also 
	\[
	\E\big\Vert \1_{\{X_t>v_t(\vtheta_0^v)\}} Y_t^2 \nabla m_t(\vtheta_0^m)\nabla^\prime m_t(\vtheta_0^m)\big\Vert\leq \big\Vert Y_t^2\big\Vert_{2}\big\Vert M_1^2(\mZ_t)\big\Vert_{2}\leq C<\infty,
	\]
	the ergodic theorem of \citet[Theorem 3.34]{Whi01} immediately implies \eqref{eq:M2}. Thus, $\widehat{\mM}_{1n}^{\ast}\overset{\P}{\longrightarrow}\mM_{1}^{\ast}$. Overall, we obtain that $\widehat{\mM}_n^{\ast}\overset{\P}{\longrightarrow}\mM^\ast$.
\end{proof}

\begin{proof}[{\textbf{Proof of Theorem~\ref{thm:lrv} ($\widehat{\mLambda}_{n,(2)}\overset{\P}{\longrightarrow}\mLambda_{(2)}$):}}]
	Since the proof of
	\begin{multline*}
		\frac{1}{n}\sum_{t=1}^{n}m_t(\widehat{\vtheta}_n^m) (2c_n)^{-1} \1_{\{|X_t-v_t(\widehat{\vtheta}_n^v)|<c_n\}}\nabla m_t(\widehat{\vtheta}_n^m)\nabla^\prime v_t(\widehat{\vtheta}_n^v)\\
		\overset{\P}{\longrightarrow}\E\Big[m_t(\vtheta_0^m) f_{t}^{X}\big(v_t(\vtheta_0^v)\big) \nabla m_t(\vtheta_0^m)\nabla^\prime v_t(\vtheta_0^v)\Big]
	\end{multline*}
	follows along similar lines as that of $\widehat{\mLambda}_n\overset{\P}{\longrightarrow}\mLambda$, we only show that
	\begin{multline*}
		\widehat{\mLambda}_{22n}:=\frac{1}{n}\sum_{t=1}^{n}Y_t(2c_n)^{-1}\1_{\{|X_t-v_t(\widehat{\vtheta}_n^v)|<c_n\}} \nabla m_t(\widehat{\vtheta}_n^m)\nabla^\prime v_t(\widehat{\vtheta}_n^v)\\
		\overset{\P}{\longrightarrow}\E\bigg[\int_{-\infty}^{\infty} y f_t\big(v_t(\vtheta_0^v),y\big)\D y\nabla m_t(\vtheta_0^m)\nabla^\prime v_t(\vtheta_0^v)\bigg]=:\mLambda_{22}.
	\end{multline*}
	Define
	\begin{equation*}
		\widetilde{\mLambda}_{22n} := \frac{1}{n}\sum_{t=1}^{n}Y_t(2c_n)^{-1}\1_{\{|X_t-v_t(\vtheta_0^v)|<c_n\}}\nabla m_t(\vtheta_0^m)\nabla^\prime v_t(\vtheta_0^v).
	\end{equation*}
	Then, it suffices to show that
	\begin{align}
		\widehat{\mLambda}_{22n} - \widetilde{\mLambda}_{22n}&=o_{\P}(1),\label{eq:Lambda11}\\
		\widetilde{\mLambda}_{22n} - \mLambda_{22}&=o_{\P}(1).\label{eq:Lambda12}
	\end{align}
	First, we establish \eqref{eq:Lambda11}. Use the MVT with $\vtheta^{m\ast}$ ($\vtheta^{v\ast}$) some mean value on the line connecting $\vtheta_0^m$ and $\widehat{\vtheta}_n^m$ ($\vtheta_0^v$ and $\widehat{\vtheta}_n^v$) to write
	\begin{align*}
		\big\Vert\widehat{\mLambda}_{22n} - \widetilde{\mLambda}_{22n}\big\Vert &= \bigg\Vert\frac{1}{n}\sum_{t=1}^{n}Y_t(2c_n)^{-1}\1_{\{|X_t-v_t(\vtheta_0^v)|<c_n\}}\Big\{ \big[\nabla m_t(\widehat{\vtheta}_n^m) - \nabla m_t(\vtheta_0^m)\big]\nabla^\prime v_t(\widehat{\vtheta}_n^v) \\
		&\hspace{6cm} + \nabla m_t(\vtheta_0^v)\big[\nabla^\prime v_t(\widehat{\vtheta}_n^v) - \nabla^\prime v_t(\vtheta_0^v)\big] \Big\} \\
		&\hspace{0.6cm} +\frac{1}{n}\sum_{t=1}^{n}Y_t(2c_n)^{-1}\big[\1_{\{|X_t-v_t(\widehat{\vtheta}_n^v)|<c_n\}} - \1_{\{|X_t-v_t(\vtheta_0^v)|<c_n\}}\big]\nabla m_t(\widehat{\vtheta}_n^v)\nabla^\prime v_t(\widehat{\vtheta}_n^v) \bigg\Vert\\
		&\leq \frac{1}{n}\sum_{t=1}^{n}|Y_t|(2c_n)^{-1}\1_{\{|X_t-v_t(\vtheta_0^v)|<c_n\}}\Big\Vert \nabla^2 m_t(\vtheta^{m\ast})\big(\widehat{\vtheta}_n^m - \vtheta_0^m\big)\nabla^\prime v_t(\widehat{\vtheta}_n^v) \\
		&\hspace{6cm} + \nabla m_t(\vtheta_0^m)\big(\widehat{\vtheta}_n^v - \vtheta_0^v\big)^\prime \nabla^{2\prime} v_t(\vtheta^{v\ast}) \Big\Vert \\
		&\hspace{0.6cm} +\frac{1}{n}\sum_{t=1}^{n}|Y_t|(2c_n)^{-1}\big|\1_{\{|X_t-v_t(\widehat{\vtheta}_n^v)|<c_n\}} - \1_{\{|X_t-v_t(\vtheta_0^v)|<c_n\}}\big|\big\Vert\nabla m_t(\widehat{\vtheta}_n^m)\nabla^\prime v_t(\widehat{\vtheta}_n^v) \big\Vert\\
		& \leq \frac{1}{n}\sum_{t=1}^{n}|Y_t|(2c_n)^{-1}\1_{\{|X_t-v_t(\vtheta_0^v)|<c_n\}}M_2(\mZ_t) V_1(\mZ_t)\big\Vert \widehat{\vtheta}_n^m - \vtheta_0^m \big\Vert \\
		&\hspace{0.6cm} +\frac{1}{n}\sum_{t=1}^{n}|Y_t|(2c_n)^{-1}\1_{\{|X_t-v_t(\vtheta_0^v)|<c_n\}}M_1(\mZ_t) V_2(\mZ_t)\big\Vert \widehat{\vtheta}_n^v - \vtheta_0^v \big\Vert \\
		&\hspace{0.6cm} +\frac{1}{n}\sum_{t=1}^{n}|Y_t|(2c_n)^{-1}\big|\1_{\{|X_t-v_t(\widehat{\vtheta}_n^v)|<c_n\}} - \1_{\{|X_t-v_t(\vtheta_0^v)|<c_n\}}\big|M_1(\mZ_t) V_1(\mZ_t)\\
		&=: A_{12n} + B_{12n} + C_{12n}.
	\end{align*}
	The proofs that $A_{12n}=o_{\P}(1)$ and $B_{12n}=o_{\P}(1)$ are similar, so we only show the former result.
	Note for this that by the MVT there exists some $c^{\ast}\in[-c_n, c_n]$, such that
	\begin{align}
		\E_t\big[|Y_t|\1_{\{|X_t-v_t(\vtheta_0^v)|<c_n\}}\big] &= \E_t\big[|Y_t|\1_{\{-c_n<X_t-v_t(\vtheta_0^v)<c_n\}}\big]\notag\\
		&= \int_{v_t(\vtheta_0^v) - c_n}^{v_t(\vtheta_0^v) + c_n}\int_{-\infty}^{\infty}|y| f_t(x,y)\D y\D x\notag\\
		&=\int_{-\infty}^{\infty} |y| \int_{v_t(\vtheta_0^v) - c_n}^{v_t(\vtheta_0^v) + c_n}f_t(x,y)\D x\D y\notag\\
		&=\int_{-\infty}^{\infty} |y| f_t\big(v_t(\vtheta_0^v) +c^{\ast},y\big)2c_n\D y.\label{eq:(33.1)}
	\end{align}
	Use the LIE, \eqref{eq:(33.1)} and Assumption~\ref{ass:3}~\ref{it:3iv} to deduce that
	\begin{align*}
		\E&\bigg[\frac{1}{n}\sum_{t=1}^{n}|Y_t|(2c_n)^{-1}\1_{\{|X_t-v_t(\vtheta_0^v)|<c_n\}}M_2(\mZ_t) V_1(\mZ_t)\bigg]\\
		&= \frac{1}{2c_n n}\sum_{t=1}^{n}\E\Big\{ M_2(\mZ_t) V_1(\mZ_t)\E_{t}\big[|Y_t|\1_{\{-c_n<X_t-v_t(\vtheta_0^v)<c_n\}}\big]\Big\}\\
		&=\frac{1}{2c_n n}\sum_{t=1}^{n}\E\Big\{ M_2(\mZ_t) V_1(\mZ_t)\int_{-\infty}^{\infty} |y| f_t\big(v_t(\vtheta_0^v) +c^{\ast},y\big)2c_n\D y\Big\}\\
		&\leq\frac{1}{n}\sum_{t=1}^{n}\E\bigg\{ M_2(\mZ_t) V_1(\mZ_t)\sup_{x\in\mathbb{R}}\Big[\int_{-\infty}^{\infty} |y| f_t\big(x,y\big)\D y\Big]\bigg\}\\
		&\leq \frac{1}{n}\sum_{t=1}^{n} \big\Vert M_2(\mZ_t)\big\Vert_{4r} \big\Vert V_1(\mZ_t)\big\Vert_{4r} \big\Vert F_1(\mathcal{F}_{t})\big\Vert_{2r/(2r-1)}\\
		&\leq C<\infty.
	\end{align*}
	Hence, by Markov's inequality, $\frac{1}{n}\sum_{t=1}^{n}|Y_t|(2c_n)^{-1}\1_{\{|X_t-v_t(\vtheta_0^v)|<c_n\}}M_2(\mZ_t) V_1(\mZ_t)=O_{\P}(1)$. Combining this with $\big\Vert \widehat{\vtheta}_n^m - \vtheta_0^m \big\Vert=o_{\P}(1)$ (from Proposition~\ref{prop:cons}), $A_{12n}=O_{\P}(1)o_{\P}(1)=o_{\P}(1)$ follows.
	
	To show $C_{12n}=o_{\P}(1)$, we use the inequality in \eqref{eq:ind ineq} and (once again) the event $\mathsf{E}_{1n}=\big\{c_n^{-1}\Vert\widehat{\vtheta}_n^v - \vtheta_0^v\Vert\leq d\big\}$, which occurs w.p.a.~1, as $n\to\infty$. Doing so, we obtain that (similarly as for \eqref{eq:B9n})
	\begin{align}
		&\P\big\{C_{12n}>\varepsilon\big\}\notag\\
		&\leq \P\big\{C_{12n}>\varepsilon,\, \mathsf{E}_{1n}\big\} + \P\big\{\mathsf{E}_{1n}^{C}\big\}\notag\\
		&\leq \P\bigg\{\frac{1}{n}\sum_{t=1}^{n}|Y_t|(2c_n)^{-1}\big|\1_{\{|X_t-v_t(\widehat{\vtheta}_n^v)|<c_n\}} - \1_{\{|X_t-v_t(\vtheta_0^v)|<c_n\}}\big|M_1(\mZ_t) V_1(\mZ_t)>\varepsilon,\, \mathsf{E}_{1n}\bigg\} + o(1)\notag\\
		&\leq \frac{1}{2\varepsilon c_n n}\sum_{t=1}^{n}\bigg\{\E\Big[|Y_t|\1_{\{|X_t-v_t(\vtheta_0^v)-c_n|<V_1(\mZ_t)dc_n\}}M_1(\mZ_t) V_1(\mZ_t)\Big]\notag\\
		&\hspace{4cm}+ \E\Big[|Y_t|\1_{\{|X_t-v_t(\vtheta_0^v)+c_n|<V_1(\mZ_t)dc_n\}}M_1(\mZ_t) V_1(\mZ_t)\Big]\bigg\}+o(1).\label{eq:B12n}
	\end{align}
	Use the LIE to derive that
	\begin{align*}
		\E&\Big[|Y_t|\1_{\{|X_t-v_t(\vtheta_0^v)-c_n|<V_1(\mZ_t)dc_n\}}M_1(\mZ_t) V_1(\mZ_t)\Big]\\
		&= \E\Big[M_1(\mZ_t) V_1(\mZ_t) \E_{t}\big\{|Y_t|\1_{\{-V_1(\mZ_t)dc_n<X_t-v_t(\vtheta_0^v)-c_n<V_1(\mZ_t)dc_n\}}\big\}\Big]\\
		&= \E\Big[M_1(\mZ_t) V_1(\mZ_t) \int_{v_t(\vtheta_0^v)+c_n-V_1(\mZ_t)dc_n}^{v_t(\vtheta_0^v)+c_n+V_1(\mZ_t)dc_n}\int_{-\infty}^{\infty} |y|f_t(x,y)\D y\D x\Big]\\
		&\leq \E\Big[M_1(\mZ_t) V_1(\mZ_t) 2V_1(\mZ_t)dc_n F_1(\mathcal{F}_t) \Big]\\
		&\leq 2d c_n\E\big[M_1(\mZ_t)  V_1^2(\mZ_t) F_1(\mathcal{F}_t)\big]\\
		&\leq 2d c_n\big\Vert M_1(\mZ_t)\big\Vert_{4r} \big\Vert V_1^2(\mZ_t)\big\Vert_{2r} \big\Vert F_1(\mathcal{F}_t)\big\Vert_{4r/(4r-3)}\\
		&\leq C d c_n,
	\end{align*}
	where we used Assumption~\ref{ass:3}~\ref{it:3iv} for the first inequality.
	An identical bound obtains for the other expectation in \eqref{eq:B12n}.
	Therefore, from \eqref{eq:B12n}, 
	\[
	\P\big\{C_{12n}>\varepsilon\big\}\leq Cd + o(1).
	\]
	Since $d>0$ can be chosen arbitrarily small, $C_{12n}=o_{\P}(1)$ follows. This establishes \eqref{eq:Lambda11}.
	
	It remains to prove \eqref{eq:Lambda12}. Write
	\begin{align*}
		\widetilde{\mLambda}_{22n} - \mLambda_{22} &= \frac{1}{n}\sum_{t=1}^{n}(2c_n)^{-1}\Big\{Y_t\1_{\{|X_t-v_t(\vtheta_0^v)|<c_n\}} - \E_t\big[Y_t\1_{\{|X_t-v_t(\vtheta_0^v)|<c_n\}}\big] \Big\}\nabla m_t(\vtheta_0^m)\nabla^\prime v_t(\vtheta_0^v)\\
		&\hspace{0.5cm} +\frac{1}{n}\sum_{t=1}^{n}\Big\{(2c_n)^{-1}\E_t\big[Y_t\1_{\{|X_t-v_t(\vtheta_0^v)|<c_n\}}\big] - \int_{-\infty}^{\infty} y f_t\big(v_t(\vtheta_0^v),y\big)\D y\Big\}\nabla m_t(\vtheta_0^m)\nabla^\prime v_t(\vtheta_0^v)\\
		&\hspace{0.5cm} +\frac{1}{n}\sum_{t=1}^{n}\bigg\{\int_{-\infty}^{\infty} y f_t\big(v_t(\vtheta_0^v),y\big)\D y \nabla m_t(\vtheta_0^m)\nabla^\prime v_t(\vtheta_0^v) \\
		&\hspace{6.5cm}- \E\Big[\int_{-\infty}^{\infty} y f_t\big(v_t(\vtheta_0^v),y\big)\D y \nabla m_t(\vtheta_0^v)\nabla^\prime v_t(\vtheta_0^v)\Big]\bigg\}\\
		&=:A_{13n} + B_{13n} + C_{13n}.
	\end{align*}
	We show that each of these terms is $o_{\P}(1)$. 
	
	By the LIE, $\E\big[A_{13n}\big]=\vzeros$. Now, for any $(i,j)$-th element of $A_{13n}$, denoted by $A_{13n,ij}$, we have that
	\begin{align*}
		\Var&\big(A_{13n,ij}\big) = \E\big[A_{13n,ij}^2\big]\\
		&= (2c_n n)^{-2}\E\bigg[\sum_{t=1}^{n}\nabla_i m_t(\vtheta_0^m)\nabla_j v_t(\vtheta_0^v) \Big\{Y_t\1_{\{|X_t-v_t(\vtheta_0^v)|<c_n\}} - \E_t\big[Y_t\1_{\{|X_t-v_t(\vtheta_0^v)|<c_n\}}\big] \Big\}\bigg]^2\\
		&= (2c_n n)^{-2}\sum_{t=1}^{n}\E\bigg[\Big\{\nabla_i m_t(\vtheta_0^m)\nabla_j v_t(\vtheta_0^v)\Big\}^2 \Big\{Y_t\1_{\{|X_t-v_t(\vtheta_0^v)|<c_n\}} - \E_t\big[Y_t\1_{\{|X_t-v_t(\vtheta_0^v)|<c_n\}}\big] \Big\}^2\bigg]\\
		&=(2c_n n)^{-2}\sum_{t=1}^{n}\E\bigg[\Big\{\nabla_i m_t(\vtheta_0^m)\nabla_j v_t(\vtheta_0^v)\Big\}^2 \E_t\Big\{Y_t\1_{\{|X_t-v_t(\vtheta_0^v)|<c_n\}} - \E_t\big[Y_t\1_{\{|X_t-v_t(\vtheta_0^v)|<c_n\}}\big] \Big\}^2\bigg]\\
		&\leq (2c_n n)^{-2}\sum_{t=1}^{n}\E\bigg[M_1^2(\mZ_t)V_1^2(\mZ_t) \E_t\big[Y_t^2\1_{\{|X_t-v_t(\vtheta_0^v)|<c_n\}}\big] \bigg]\\
		&= (2c_n n)^{-2}\sum_{t=1}^{n}\E\bigg[M_1^2(\mZ_t)V_1^2(\mZ_t) \int_{v_t(\vtheta_0^v) - c_n}^{v_t(\vtheta_0^v) + c_n}\int_{-\infty}^{\infty}y^2 f_t(x,y)\D y\D x\bigg]\\
		&\leq (2c_n n)^{-2}\sum_{t=1}^{n}\E\big[M_1^2(\mZ_t)V_1^2(\mZ_t)F_2(\mathcal{F}_t)\big]2c_n\\
		&\leq C c_n^{-1} n^{-1}=o(1),
	\end{align*}
	where we have used for the third equality that all covariance terms are zero by the LIE, the first inequality uses the $c_r$-inequality, the second inequality exploits Assumption~\ref{ass:10}, the third inequality uses Assumption~\ref{ass:11}, and the final equality follows from Assumption~\ref{ass:8}.
	
	For $B_{13n}$, note that, similarly as in \eqref{eq:(33.1)},
	\begin{equation*}
		\E_t\big[Y_t\1_{\{|X_t-v_t(\vtheta_0^v)|<c_n\}}\big] =\int_{-\infty}^{\infty} y f_t\big(v_t(\vtheta_0^v) +c^{\ast},y\big)2c_n\D y
	\end{equation*}
	for some $c^{\ast}\in[-c_n, c_n]$, such that (using Assumption~\ref{ass:3}~\ref{it:3iii})
	\begin{align*}
		\E\big\Vert B_{13n}\big\Vert &\leq \frac{1}{n}\sum_{t=1}^{n}\E\bigg[M_1(\mZ_t) V_1(\mZ_t) \Big|(2c_n)^{-1}\int_{-\infty}^{\infty} y f_t\big(v_t(\vtheta_0^v) +c^{\ast},y\big)2c_n\D y - \int_{-\infty}^{\infty} y f_t\big(v_t(\vtheta_0^v),y\big)\D y \Big|\bigg]\\
		&= \frac{1}{n}\sum_{t=1}^{n}\E\bigg[M_1(\mZ_t) V_1(\mZ_t) \Big|\int_{-\infty}^{\infty} y f_t\big(v_t(\vtheta_0^v) +c^{\ast},y\big)\D y - \int_{-\infty}^{\infty} y f_t\big(v_t(\vtheta_0^v),y\big)\D y \Big|\bigg]\\
		&\leq \frac{1}{n}\sum_{t=1}^{n}\E\bigg[M_1(\mZ_t) V_1(\mZ_t) \sup_{x\in\mathbb{R}}\Big|\int_{-\infty}^{\infty} y \partial_1 f_t\big(x,y\big)\D y\Big| \bigg] c_n\\
		&\leq \frac{1}{n}\sum_{t=1}^{n}\E\big[M_1(\mZ_t) V_1(\mZ_t) F(\mathcal{F}_t)\big]c_n\\
		&\leq\frac{c_n}{n}\sum_{t=1}^{n}\big\Vert M_1(\mZ_t)\big\Vert_{4r} \big\Vert V_1(\mZ_t)\big\Vert_{4r}\big\Vert F(\mathcal{F}_t)\big\Vert_{2r/(2r-1)}\\
		&\leq C c_n=o(1).
	\end{align*}
	Therefore, $B_{13n}=o_{\P}(1)$ follows by Markov's inequality
	
	Finally, $C_{13n}=o_{\P}(1)$ is immediate from Assumption~\ref{ass:9} \ref{it:9ii}. 
	
	Overall, \eqref{eq:Lambda12} follows, in turn establishing $\widehat{\mLambda}_{n,(2)}\overset{\P}{\longrightarrow}\mLambda_{(2)}$.
\end{proof}

\section{Verification of Assumptions~\ref{ass:1}--\ref{ass:11} for a linear model}
\label{sec:ex verif}

\renewcommand{\theequation}{C.\arabic{equation}}	% different equation numbering in appendix
\setcounter{equation}{0}

In this section, we verify Assumptions~\ref{ass:1}--\ref{ass:11} for the linear model of Example~\ref{ex:1}. To do so, we have to impose some additional conditions beyond those introduced in Example~\ref{ex:1}.
We do not verify Assumption~\ref{ass:8} on the bandwidth choice here as this assumption is not related to a specific (linear or nonlinear) model.

%%%%%%%%%%%%%%
%%%% OLD %%%%%
%%%%%%%%%%%%%%
%\begin{assumptionM}\label{M:1}
%For all  $t\in\mathbb{N}$, the matrices 
%\[
%\begin{pmatrix}\mZ_1^{v\prime} \\ \vdots \\ \mZ_t^{v\prime} \end{pmatrix}\quad\text{and}\quad\begin{pmatrix}\mZ_1^{m\prime} \\ \vdots \\ \mZ_t^{m\prime} \end{pmatrix}
%\]
%have full rank, and $\E\big[\mZ_t^{v}\mZ_t^{v\prime}\big]$ and $\E\big[\mZ_t^{m}\mZ_t^{m\prime}\big]$ are positive definite.
%\end{assumptionM}

\begin{assumptionM}\label{M:1}
	For all $t\in\mathbb{N}$, the matrices $\E\big[\mZ_t^{v}\mZ_t^{v\prime}\big]$ and $\E\big[\mZ_t^{m}\mZ_t^{m\prime}\big]$ are positive definite.
\end{assumptionM}

\begin{assumptionM}\label{M:2}
	$(\vtheta_0^{v\prime},\vtheta_0^{m\prime})^\prime\in\inter(\mTheta)$, where $\mTheta=\mTheta^v\times\mTheta^m\subset\mathbb{R}^{p+q}$ is compact.
\end{assumptionM}

\begin{assumptionM}\label{M:3}
	$\big\{(X_t,Y_t,\mZ_t^{v\prime},\mZ_t^{m\prime})^\prime\big\}_{t\in\mathbb{N}}$ is strictly stationary and $\beta$-mixing with coefficients $\beta(\cdot)$ of size $-r/(r-1)$ for $r>1$.
\end{assumptionM}

\begin{assumptionM}\label{M:4}
	$\E|X_t|\leq K$, $\E|Y_t|^{4r+\iota}\leq K$, $\E\big[\Vert\mZ_t^v\Vert^{2}\Vert\mZ_t^m\Vert^{4}\big]\leq K$, $\E\Vert\mZ_t^v\Vert^{4r}\leq K$ and $\E\Vert\mZ_t^m\Vert^{4r+\iota}\leq K$ for some $\iota>0$ and $r>1$ from Assumption~LM.\ref{M:3}.
\end{assumptionM}

\begin{assumptionM}\label{M:5}
	The density $f_{\varepsilon_X,\varepsilon_Y}(\cdot,\cdot)$ of $(\varepsilon_{X,t},\varepsilon_{Y,t})^\prime$ is differentiable in the first argument, with derivative denoted by $\partial_1 f_{\varepsilon_X,\varepsilon_Y}(\cdot,\cdot)$. Moreover,
	\begin{align*}
		& f_{\varepsilon_X,\varepsilon_Y}(x,y)>0\ \text{for all}\ (x,y)^\prime\in\mathbb{R}^2\ \text{such that}\ F_{\varepsilon_X,\varepsilon_Y}(x,y)\in(0,1),\\
		& \big|f_{\varepsilon_X}(x)-f_{\varepsilon_X}(x^\prime)\big|\leq K|x-x^\prime|,\\
		&\sup_{x\in\mathbb{R}} f_{\varepsilon_X}(x)\leq K,\quad f_{\varepsilon_X}(0)>0,\\
		&\sup_{x\in\mathbb{R}} \int_{-\infty}^{\infty} |y|^{i} f_{\varepsilon_X,\varepsilon_Y}(x,y)\D y\leq K\text{ for } i=0,1,2,\\
		&\sup_{x\in\mathbb{R}}\Big| \int_{-\infty}^{\infty} y^{i}\partial_1 f_{\varepsilon_X,\varepsilon_Y}(x,y)\D y\Big|\leq K\text{ for } i=0,1.
	\end{align*}
\end{assumptionM}

\begin{assumptionM}\label{M:6}
	The distribution of $\varepsilon_{X,t}\mid\mathcal{F}_t$ possesses a continuous Lebesgue density.
\end{assumptionM}

The first Assumption~LM.\ref{M:1} is the usual regression condition ruling out linearly dependent regressors.
Assumption~LM.\ref{M:2} ensures the validity of Assumption~\ref{ass:1}~\ref{it:1ii}--\ref{it:1iii}. 
Assumption~LM.\ref{M:3} restricts the serial dependence in the observables and imposes a weak regularity condition on their joint distribution, while Assumption~LM.\ref{M:4} restricts their moments.
Of all the conditions, Assumption~LM.\ref{M:3} and LM.\ref{M:4} may be the hardest to check.
Therefore, below we give an example and references to the literature, where mixing and moment properties have been verified for a range of time series models.
Assumption~LM.\ref{M:5} imposes certain smoothness conditions on the distribution of the errors $(\varepsilon_{X,t}, \varepsilon_{Y,t})^\prime$. Finally, Assumption~LM.\ref{M:6} also imposes a sufficiently regular conditional distribution of $\varepsilon_{X,t}$. Note that Assumption~LM.\ref{M:6} is redundant if $\varepsilon_{X,t}$ is independent of $\mathcal{F}_t$, such that the conditional and unconditional distribution coincide and Assumption~LM.\ref{M:6} already holds by virtue of the Lipschitz continuity of $f_{\varepsilon_X}(\cdot)$ from Assumption~LM.\ref{M:5}.

\begin{prop}
	Suppose Assumptions~LM.\ref{M:1}--LM.\ref{M:6} hold for the linear model of Example~\ref{ex:1}. Then, Assumptions~\ref{ass:1}--\ref{ass:11} are in force.
\end{prop}

\begin{proof}
	We verify each assumption separately.
	
	\textbf{Assumption~\ref{ass:1}~\ref{it:1i}:} The existence of $\vtheta_0^v$ and $\vtheta_0^m$ follows from Example~\ref{ex:1}. 
	To prove uniqueness of $\vtheta_0^v$, suppose that $\P\big\{v_{t}(\vtheta^v)=v_{t}(\vtheta_0^v)\big\}=1$ for some $t\in\mathbb{N}$. 
	However, by definition of $v_{t}(\cdot)$, this is equivalent to
	\begin{align*}
		v_{t}(\vtheta^v)\overset{\text{a.s.}}{=}v_{t}(\vtheta_0^v) \quad&\Longleftrightarrow\quad \mZ_t^{v\prime}\vtheta^v\overset{\text{a.s.}}{=}\mZ_t^{v\prime}\vtheta_0^v\\
		&\Longleftrightarrow\quad \mZ_t^{v\prime}(\vtheta^v-\vtheta_0^v)\overset{\text{a.s.}}{=}0\\
		&\Longleftrightarrow\quad (\vtheta^v-\vtheta_0^v)^\prime\mZ_t^{v}\mZ_t^{v\prime}(\vtheta^v-\vtheta_0^v)\overset{\text{a.s.}}{=}0.
	\end{align*}
	The latter implies that
	\[
	(\vtheta^v-\vtheta_0^v)^\prime\E\big[\mZ_t^{v}\mZ_t^{v\prime}\big](\vtheta^v-\vtheta_0^v)=0.
	\]
	Since the matrix $\E\big[\mZ_t^{v}\mZ_t^{v\prime}\big]$ is positive definite by Assumption~LM.\ref{M:1}, it must be the case that $\vtheta^v=\vtheta_0^v$.
	
	The uniqueness of $\vtheta_0^m$ may be shown similarly.
	
	\textbf{Assumption~\ref{ass:1}~\ref{it:1ii}--\ref{it:1iii}:} These follow directly from Assumption~LM.\ref{M:2}.
	
	\textbf{Assumption~\ref{ass:2}:} Immediate from Assumption~LM.\ref{M:3}.

	\textbf{Assumption~\ref{ass:3}~\ref{it:3i}:} Write
	\begin{align*}
		F_t(x,y) &= \P_t\big\{X_t\leq x,\ Y_t\leq y\big\}\\
		&= \P_t\big\{\mZ_t^{v\prime}\vtheta_0^v+\varepsilon_{X,t}\leq x,\ \mZ_t^{m\prime}\vtheta_0^m+\varepsilon_{Y,t}\leq y\big\}\\
		&=\P_t\big\{\varepsilon_{X,t}\leq x-\mZ_t^{v\prime}\vtheta_0^v,\ \varepsilon_{Y,t}\leq y-\mZ_t^{m\prime}\vtheta_0^m\big\}\\
		&=F_{\varepsilon_X,\varepsilon_Y}(x-\mZ_t^{v\prime}\vtheta_0^v,\ y-\mZ_t^{m\prime}\vtheta_0^m).
	\end{align*}
	Hence, 
	\begin{equation}\label{eq:dens trafo}
		f_t(x,y)=f_{\varepsilon_X,\varepsilon_Y}(x-\mZ_t^{v\prime}\vtheta_0^v,\ y-\mZ_t^{m\prime}\vtheta_0^m),
	\end{equation}
	Therefore, if $F_t(x,y)=F_{\varepsilon_X,\varepsilon_Y}(x-\mZ_t^{v\prime}\vtheta_0^v,\ y-\mZ_t^{m\prime}\vtheta_0^m)\in(0,1)$, then---by Assumption~LM.\ref{M:5}---it follows that $f_t(x,y)=f_{\varepsilon_X,\varepsilon_Y}(x-\mZ_t^{v\prime}\vtheta_0^v,\ y-\mZ_t^{m\prime}\vtheta_0^m)>0$, as required.
	
	\textbf{Assumption~\ref{ass:3}~\ref{it:3ii}:} To show $\sup_{x\in\mathbb{R}}f_t^{X}(x)\leq K$, write
	\begin{align*}
		F_t^{X}(x) &= \P_t\{X_t\leq x\}\\
		&= \P_t\{\mZ_t^{v\prime}\vtheta_0^v+\varepsilon_{X,t}\leq x\}\\
		&= \P_t\{\varepsilon_{X,t}\leq x-\mZ_t^{v\prime}\vtheta_0^v\}\\
		&= F_{\varepsilon_X}(x-\mZ_t^{v\prime}\vtheta_0^v).
	\end{align*}
	Therefore, for all $x\in\mathbb{R}$,
	\begin{align}
		f_t^{X}(x)&=f_{\varepsilon_X}(x-\mZ_t^{v\prime}\vtheta_0^v)\label{eq:(5.11)}\\
		&\leq K\notag
	\end{align}
	by Assumption~LM.\ref{M:5}. 
	From \eqref{eq:(5.11)} and Assumption~LM.\ref{M:5},
	\begin{align*}
		\big|f_t^{X}(x)-f_t^{X}(x^\prime)\big| &=\big|f_{\varepsilon_X}(x-\mZ_t^{v\prime}\vtheta_0^v)-f_{\varepsilon_X}(x^\prime-\mZ_t^{v\prime}\vtheta_0^v)\big|\\
		&\leq K\big|(x-\mZ_t^{v\prime}\vtheta_0^v) - (x^\prime-\mZ_t^{v\prime}\vtheta_0^v)\big|\\
		& = K|x-x^\prime|.
	\end{align*}

	\textbf{Assumption~\ref{ass:3}~\ref{it:3iii}:} The differentiability requirement is satisfied by \eqref{eq:dens trafo} and Assumption~LM.\ref{M:5}.

	\textbf{Assumption~\ref{ass:3}~\ref{it:3iv}:} 
	Relation \eqref{eq:dens trafo} allows us to write (using an obvious substitution in the integral)
	\begin{align*}
		\int_{-\infty}^{\infty}y\partial_1 f_t(x,y)\D y &= \int_{-\infty}^{\infty}y\partial_1 f_{\varepsilon_X,\varepsilon_Y}(x-\mZ_t^{v\prime}\vtheta_0^v, y-\mZ_t^{m\prime}\vtheta_0^m)\D y\\
		&=\int_{-\infty}^{\infty}(y+\mZ_t^{m\prime}\vtheta_0^m)\partial_1 f_{\varepsilon_X,\varepsilon_Y}(x-\mZ_t^{v\prime}\vtheta_0^v, y)\D y\\
		&=\int_{-\infty}^{\infty}y\partial_1 f_{\varepsilon_X,\varepsilon_Y}(x-\mZ_t^{v\prime}\vtheta_0^v, y)\D y\\
		&\hspace{3cm}+\int_{-\infty}^{\infty}\mZ_t^{m\prime}\vtheta_0^m\partial_1 f_{\varepsilon_X,\varepsilon_Y}(x-\mZ_t^{v\prime}\vtheta_0^v, y)\D y.
	\end{align*}
	Thus,
	\begin{align}
		\sup_{x\in\mathbb{R}}\Big|\int_{-\infty}^{\infty}y\partial_1 f_t(x,y)\D y\Big| &\leq\sup_{x\in\mathbb{R}}\Big|\int_{-\infty}^{\infty}y\partial_1 f_{\varepsilon_X,\varepsilon_Y}(x, y)\D y\Big|\notag\\
		&\hspace{3cm}+\big|\mZ_t^{m\prime}\vtheta_0^m\big|\cdot\sup_{x\in\mathbb{R}}\Big|\int_{-\infty}^{\infty}\partial_1 f_{\varepsilon_X,\varepsilon_Y}(x, y)\D y\Big|\notag\\
		&=:F(\mathcal{F}_t).\label{eq:(8.11)}
	\end{align}
	
	The verification of the bound for $\int_{-\infty}^{\infty}|y| f_t(x,y)\D y$ follows along similar lines.
	Write using \eqref{eq:dens trafo} that
	\begin{align*}
		\int_{-\infty}^{\infty}|y| f_t(x,y)\D y &= \int_{-\infty}^{\infty}|y| f_{\varepsilon_X,\varepsilon_Y}(x-\mZ_t^{v\prime}\vtheta_0^v, y-\mZ_t^{m\prime}\vtheta_0^m)\D y\\
		&=\int_{-\infty}^{\infty}\big|y+\mZ_t^{m\prime}\vtheta_0^m\big|f_{\varepsilon_X,\varepsilon_Y}(x-\mZ_t^{v\prime}\vtheta_0^v, y)\D y\\
		&=\int_{-\infty}^{\infty}|y|f_{\varepsilon_X,\varepsilon_Y}(x-\mZ_t^{v\prime}\vtheta_0^v, y)\D y\\
		&\hspace{3cm}+\int_{-\infty}^{\infty}\big|\mZ_t^{m\prime}\vtheta_0^m\big| f_{\varepsilon_X,\varepsilon_Y}(x-\mZ_t^{v\prime}\vtheta_0^v, y)\D y.
	\end{align*}
	Therefore,
	\begin{align}
		\sup_{x\in\mathbb{R}}\int_{-\infty}^{\infty}|y| f_t(x,y)\D y &\leq\sup_{x\in\mathbb{R}}\int_{-\infty}^{\infty}|y| f_{\varepsilon_X,\varepsilon_Y}(x, y)\D y\notag\\
		&\hspace{3cm}+\big|\mZ_t^{m\prime}\vtheta_0^m\big|\cdot\sup_{x\in\mathbb{R}}\int_{-\infty}^{\infty}f_{\varepsilon_X,\varepsilon_Y}(x, y)\D y\notag\\
		&=:F_1(\mathcal{F}_t).\label{eq:(C.3+)}
	\end{align}

	\textbf{Assumption~\ref{ass:4}~\ref{it:4i}:} The differentiability of $v_t(\vtheta^v)=\mZ_t^{v\prime}\vtheta^v$ and $m_t(\vtheta^m)=\mZ_t^{m\prime}\vtheta^m$ with gradients
	\begin{equation}\label{eq:(4.1)}
		\nabla v_t(\vtheta^v)=\mZ_t^{v}\qquad\text{and}\qquad \nabla m_t(\vtheta^m)=\mZ_t^{m}
	\end{equation}
	and Hessians
	\begin{equation}\label{eq:(D.3+)}
		\nabla^2 v_t(\vtheta^v)=\vzeros\qquad\text{and}\qquad \nabla^2 m_t(\vtheta^m)=\vzeros
	\end{equation}
	is obvious.

	\textbf{Assumption~\ref{ass:4}~\ref{it:4ii}:} It holds that
	\begin{align*}
		\sup_{\vtheta^v\in\mTheta^v}\big|v_t(\vtheta^v)\big| &= \sup_{\vtheta^v\in\mTheta^v}\big|\mZ_t^{v\prime}\vtheta^v\big|\\
		&\leq \big\Vert\mZ_t^{v}\big\Vert \sup_{\vtheta^v\in\mTheta^v}\big\Vert\vtheta^v\big\Vert\\
		&\leq C \big\Vert\mZ_t^{v}\big\Vert=:V(\mZ_t),
	\end{align*}
	since $\mTheta^v$ is compact (by Assumption~LM.\ref{M:2}) and, hence, bounded.
	
	That $\sup_{\vtheta^m\in\mTheta^m}\big|m_t(\vtheta^m)\big|\leq C \big\Vert\mZ_t^{m}\big\Vert=:M(\mZ_t)$ follows similarly.

	\textbf{Assumption~\ref{ass:4}~\ref{it:4iii}:} From \eqref{eq:(4.1)}, we obtain that
	\begin{align}
		\sup_{\vtheta^v\in\mTheta^v}\big\Vert\nabla v_t(\vtheta^v)\big\Vert&=\big\Vert\mZ_t^v\big\Vert=:V_1(\mZ_t),\label{eq:(4.2)}\\
		\sup_{\vtheta^m\in\mTheta^m}\big\Vert\nabla m_t(\vtheta^m)\big\Vert&=\big\Vert\mZ_t^m\big\Vert=:M_1(\mZ_t).\label{eq:(4.3)}
	\end{align}
	
	\textbf{Assumption~\ref{ass:4}~\ref{it:4iv}--\ref{it:4v}:} By \eqref{eq:(D.3+)}, the assumptions are trivially satisfied for
	\begin{align}
		V_2(\mZ_t) &= M_2(\mZ_t) = 0,\label{eq:(6.1)}\\
		V_3(\mZ_t) &= M_3(\mZ_t) = 0.\label{eq:(6.2)}
	\end{align}
	
	\textbf{Assumption~\ref{ass:5}:} From \eqref{eq:(4.2)}--\eqref{eq:(6.2)} combined with Assumption~LM.\ref{M:4},
	\begin{align*}
		\E|X_t|&\leq K,\\
		\E|Y_t|^{4r+\iota}&\leq K,\\
		\E\big[V(\mZ_t)\big] &= C\E\big[\Vert\mZ_t^v\Vert\big] \leq K,\\
		\E\big[M^{4r+\iota}(\mZ_t)\big] &= C\E\big[\Vert\mZ_t^m\Vert^{4r+\iota}\big] \leq K,\\
		\E\big[V_1^{4r}(\mZ_t)\big]&=\E\big[\Vert\mZ_t^v\Vert^{4r}\big]\leq K,\\
		\E\big[M_1^{4r+\iota}(\mZ_t)\big]&=\E\big[\Vert\mZ_t^m\Vert^{4r+\iota}\big]\leq K,\\
		\E\big[V_2^{2r}(\mZ_t)\big]&=\E\big[V_3(\mZ_t)\big]=0\leq K,\\
		\E\big[M_2^{4r}(\mZ_t)\big]&=\E\big[M_3^{4r/(4r-1)}(\mZ_t)\big]=0\leq K.
	\end{align*}
	Finally, $\E\big[F^{4r/(4r-3)}(\mathcal{F}_t)\big]\leq K$ and $\E\big[F_1^{4r/(4r-3)}(\mathcal{F}_t)\big]\leq K$ follow from \eqref{eq:(8.11)}--\eqref{eq:(C.3+)} and Assumption~LM.\ref{M:4}.

	\textbf{Assumption~\ref{ass:6}:} From \eqref{eq:(5.11)},
	\begin{align}
		f_t^X\big(v_t(\vtheta_0^v)\big) &= f_{\varepsilon_X}\big(v_t(\vtheta_0^v)-\mZ_t^{v\prime}\vtheta_0^v\big)\notag\\
		&=f_{\varepsilon_X}\big(\mZ_t^{v\prime}\vtheta_0^v-\mZ_t^{v\prime}\vtheta_0^v\big)\notag\\
		&=f_{\varepsilon_X}(0).\label{eq:(p.59)}
	\end{align}
	Using this and \eqref{eq:(4.1)}, we obtain that
	\begin{align*}
		\mLambda &=\E\Big[f_t^{X}\big(v_t(\vtheta_0^v)\big)\nabla v_t(\vtheta_0^v)\nabla^\prime v_t(\vtheta_0^v)\Big]\\
		&= f_{\varepsilon_X}(0)\E\big[\mZ_t^{v}\mZ_t^{v\prime}\big].
	\end{align*}
	Since $f_{\varepsilon_X}(0)>0$ (by Assumption~LM.\ref{M:5}) and $\E\big[\mZ_t^{v}\mZ_t^{v\prime}\big]$ is positive definite (by Assumption~LM.\ref{M:1}), it follows that $\mLambda$ is positive definite.
	
	Again using \eqref{eq:(4.1)}, we get that
	\begin{equation*}
		\mLambda_{(1)} = (1-\b)\E\big[\nabla m_t(\vtheta_0^m)\nabla^\prime m_t(\vtheta_0^m)\big]=(1-\b)\E\big[\mZ_t^{m}\mZ_t^{m\prime}\big],
	\end{equation*}
	which is once again positive definite by Assumption~LM.\ref{M:1}.

	\textbf{Assumption~\ref{ass:7}:} We now show that
	$\sup_{\vtheta^v\in\mTheta^v}\sum_{t=1}^{n}\1_{\{X_t=v_t(\vtheta^v)\}}\leq p$ a.s.~for all $n\in\mathbb{N}$, where $p$ is the dimension of $\mTheta^v$.
	To do so, we proceed by induction over $n$.
	
	The induction hypothesis is satisfied because the statement is trivial for $n\in\{1,\ldots,p\}$.
	
	Now, for the induction step, we assume that the statement is true for some $n-1\geq p$, such that $\sup_{\vtheta^v\in\mTheta^v}\sum_{t=1}^{n-1}\1_{\{X_t=v_t(\vtheta^v)\}}\leq p$ a.s..
	Then, we need to show that a.s.,
	\begin{equation}\label{eq:to show}
		\sup_{\vtheta^v\in\mTheta^v}\sum_{t=1}^{n}\1_{\{X_t=v_t(\vtheta^v)\}}\leq p.
	\end{equation}
	If $\sup_{\vtheta^v\in\mTheta^v}\sum_{t=1}^{n-1}\1_{\{X_t=v_t(\vtheta^v)\}}< p$ a.s., then
	\[
	\sup_{\vtheta^v\in\mTheta^v}\sum_{t=1}^{n}\1_{\{X_t=v_t(\vtheta^v)\}}\leq \sup_{\vtheta^v\in\mTheta^v}\sum_{t=1}^{n-1}\1_{\{X_t=v_t(\vtheta^v)\}} + 1\leq p,
	\]
	such that \eqref{eq:to show} is immediate.
	
	So suppose now that $\sup_{\vtheta^v\in\mTheta^v}\sum_{t=1}^{n-1}\1_{\{X_t=v_t(\vtheta^v)\}}= p$ on some set $A \in \mathcal{F}$ with positive probability  $\P(A) > 0$, while $\sup_{\vtheta^v\in\mTheta^v}\sum_{t=1}^{n-1}\1_{\{X_t=v_t(\vtheta^v)\}} < p$ outside of $A$.
	Then, on $A$, there exist some $\mathcal{F}_n$-measurable random times $t_1 < \dots < t_p \le n-1$ and random variables $\vtheta_{\ast}^{v}  = \vtheta_{\ast}^{v}(t_1, \dots,t_p)  \in \mTheta^v$ such that $X_{t_i} = \mZ_{t_i}^{v\prime} \vtheta_{\ast}^{v}$ for all $i=1,\ldots,p$.
	Here, we use the compactness of $\mTheta^v$ to guarantee that $\vtheta_{\ast}^{v} \in \mTheta^v$.
	Further notice that at most ${n-1}\choose{p}$ combinations exist to choose $p$ time points out of the set  $\{1,\dots, n-1\}$.
	% and there exist at most ${n-1}\choose{p}$ different random variables $\vtheta_{\ast}^{v}(t_1, \dots,t_p)$ satisfying that condition.
	(Intuitively, the $\vtheta_{\ast}^{v}(t_1, \dots,t_p)$ are the regression coefficients that generate $p$-dimensional linear maps that pass through exactly $p$ (random) data points, which have ordered indices $t_1 < \dots < t_p$.)

	We now define $\overline{\mZ}_{\mathbf{t}_{1:p}} := (\mZ_{t_1}, \dots, \mZ_{t_p})' \in \mathbb{R}^{p \times p}$, $\overline{X}_{\mathbf{t}_{1:p}} := (X_{t_1}, \dots, X_{t_p})' \in \mathbb{R}^{p \times 1}$ and $\overline{\varepsilon}_{\mathbf{t}_{1:p}}  := (\varepsilon_{X,t_1}, \dots, \varepsilon_{X,t_p})' \in \mathbb{R}^{p \times 1}$ such that (given $t_1, \dots, t_p$), the random vector $\vtheta_{\ast}^{v}  = \vtheta_{\ast}^{v}(t_1, \dots,t_p)$ satisfies the linear system (on the set $A$ with positive probability $\P(A) >0$)
	\begin{align}
		\label{eqn:LinSystemProof}
		\overline{X}_{\mathbf{t}_{1:p}}  = \overline{\mZ}_{\mathbf{t}_{1:p}} \vtheta_{\ast}^{v}.
	\end{align}
	
	We now show by contradiction that equality in \eqref{eqn:LinSystemProof} implies that the matrix $\overline{\mZ}_{\mathbf{t}_{1:p}}$ has full rank $p$.
	For this, assume that $\rank \big( \overline{\mZ}_{\mathbf{t}_{1:p}} \big) < p$ with positive probability (formally: on some $\tilde A \subseteq A$ with $\tilde A \in \mathcal{F}$ and $\P(\tilde A)>0$).
	The Rouch\'e–Capelli theorem and the existence of a solution $\vtheta_{\ast}^{v}$ of the linear system in  \eqref{eqn:LinSystemProof} imply that $\rank \big( \overline{\mZ}_{\mathbf{t}_{1:p}} \big) = \rank \big( \overline{\mZ}_{\mathbf{t}_{1:p}} \mid  \overline{X}_{\mathbf{t}_{1:p}} \big)$, where the latter is the rank of the ``augmented coefficient matrix''.
	Thus, there exists a linear map such that $\overline{X}_{\mathbf{t}_{1:p}}  = \overline{\mZ}_{\mathbf{t}_{1:p}} \vbeta$ for some $\vbeta \in \mathbb{R}^p$, as otherwise, $\overline{X}_{\mathbf{t}_{1:p}}$ would be linearly independent of the vectors in $\overline{\mZ}_{\mathbf{t}_{1:p}}$, implying $\rank \big( \overline{\mZ}_{\mathbf{t}_{1:p}} \big) < \rank \big( \overline{\mZ}_{\mathbf{t}_{1:p}} \mid  \overline{X}_{\mathbf{t}_{1:p}} \big)(\leq p)$.
	As $\overline{X}_{\mathbf{t}_{1:p}} = \overline{\mZ}_{\mathbf{t}_{1:p}} \vtheta_0^v + \overline{\varepsilon}_{\mathbf{t}_{1:p}}$ however, and the unconditional distribution of $\overline{\varepsilon}_{\mathbf{t}_{1:p}}$ is absolutely continuous, such a linear map can only exist with probability zero such that on $A$ (except on a null set within $A$), it holds that $\rank \big( \overline{\mZ}_{\mathbf{t}_{1:p}} \big) \not= \rank \big( \overline{\mZ}_{\mathbf{t}_{1:p}} \mid  \overline{X}_{\mathbf{t}_{1:p}} \big)$.
	This, however, contradicts the rank equality, and we can conclude that the matrix $\overline{\mZ}_{\mathbf{t}_{1:p}}$ must have full rank on $A$ (except on a null set within $A$).
	Hence, on $A$ we can write
	\begin{align*}
		\vtheta_{\ast}^{v} = \vtheta_{0}^{v} + \big(\overline{\mZ}_{\mathbf{t}_{1:p}}\big)^{-1} \overline{\varepsilon}_{\mathbf{t}_{1:p}}.
	\end{align*}
	As  $\big(\overline{\mZ}_{\mathbf{t}_{1:p}}\big)^{-1}$ and $\overline{\varepsilon}_{\mathbf{t}_{1:p}}$ are $\mathcal{F}_n$-measurable, so is $\vtheta_{\ast}^{v} = \vtheta_{\ast}^{v}(t_1, \dots, t_p)$.

	For each of these finitely many $\vtheta_{\ast}^{v} = \vtheta_{\ast}^{v}(t_1, \dots,t_p)$, it now holds that
	\begin{align*}
		\P\big\{X_n =  \mZ_{n}^{v\prime} \vtheta_{\ast}^{v},\, A \big\}
		&= \P\big\{ \mZ_{n}^{v\prime}\vtheta_{0}^{v} + \varepsilon_{X,n} = \mZ_{n}^{v\prime}\vtheta_{\ast}^{v},\, A  \big\} \notag \\
		&= \P\big\{ \varepsilon_{X,n} = \mZ_{n}^{v\prime}(\vtheta_{\ast}^{v}-\vtheta_{0}^{v})  ,\, A \big\}\notag\\
		&=  \E\big[\1_{\{\varepsilon_{X,n} = \mZ_{n}^{v\prime}(\vtheta_{\ast}^{v}-\vtheta_{0}^{v}),\, A \}} \big]\notag\\
		&=  \E\Big[ \E_{n}\big\{\1_{\{\varepsilon_{X,n} = \mZ_{n}^{v\prime}(\vtheta_{\ast}^{v}-\vtheta_{0}^{v}),\, A \}}\big\} \Big]\notag\\
		&= \E\Big[ \P_{n}\big\{ \varepsilon_{X,n} = \mZ_{n}^{v\prime}(\vtheta_{\ast}^{v}-\vtheta_{0}^{v}),\, A  \big\} \Big] \\
		&= 0 
		% \label{eq:prob prof 1}
	\end{align*}
	by Assumption~LM.\ref{M:6}, where we have used that $\mZ_n^{v}$ and $\vtheta_{\ast}^{v}$ are $\mathcal{F}_{n}$-measurable.
	As there are only finitely many $\vtheta_{\ast}^{v}(t_1, \dots,t_p)$, this also implies that
	\begin{align*}
		\P \bigg\{ \bigcup_{t_1 < \cdots < t_p} \big\{X_n =  \mZ_{n}^{v\prime} \vtheta_{\ast}^{v}(t_1, \dots,t_p)\big\},\, A \bigg\} &= \P \bigg\{ \bigcup_{t_1 < \cdots < t_p} \big\{X_n =  \mZ_{n}^{v\prime} \vtheta_{\ast}^{v}(t_1, \dots,t_p),\, A\big\} \bigg\}\\
		&\leq \sum_{t_1 < \cdots < t_p}\P\big\{ X_n =  \mZ_{n}^{v\prime} \vtheta_{\ast}^{v}(t_1, \dots,t_p),\, A \big\} = 0.
	\end{align*}
	Therefore,
	\begin{align*}
		\P&\bigg\{\sup_{\vtheta^v\in\mTheta^v} \sum_{t=1}^{n}\1_{\{X_t=v_t(\vtheta^v)\}} > p\bigg\}\\
		&=\P\bigg\{\sup_{\vtheta^v\in\mTheta^v} \sum_{t=1}^{n}\1_{\{X_t=v_t(\vtheta^v)\}} > p,\, A\bigg\} + \P\bigg\{\sup_{\vtheta^v\in\mTheta^v} \sum_{t=1}^{n}\1_{\{X_t=v_t(\vtheta^v)\}} > p,\, A^{\complement}\bigg\}\\
		&=\P\bigg\{ \bigcup_{t_1 < \cdots < t_p} \big\{X_n =  \mZ_{n}^{v\prime} \vtheta_{\ast}^{v}(t_1, \dots,t_p)\big\},\, A \bigg\} + 0\\
		&=0
	\end{align*}
	or, in other words,
	\begin{equation*}
		\sup_{\vtheta^v\in\mTheta^v} \sum_{t=1}^{n}\1_{\{X_t=v_t(\vtheta^v)\}} \leq p\quad\text{a.s.},
	\end{equation*}
	such that the induction step is complete and the desired result, $\sup_{\vtheta^v\in\mTheta^v}\sum_{t=1}^{n}\1_{\{X_t=v_t(\vtheta^v)\}}\leq p$, follows.

	\textbf{Assumption~\ref{ass:9}~\ref{it:9i}:} From \eqref{eq:(4.1)} and \eqref{eq:(p.59)},
	\[
	\frac{1}{n}\sum_{t=1}^{n}\nabla v_t(\vtheta_0^v)\nabla^\prime v_t(\vtheta_0^v)f_t^{X}\big(v_t(\vtheta_0^v)\big) = f_{\varepsilon_X}(0)\bigg(\frac{1}{n}\sum_{t=1}^{n}\mZ_t^v\mZ_t^{v\prime}\bigg).
	\]
	By Assumptions~LM.\ref{M:3}--\ref{M:4}, a standard mixing LLN applies for $\frac{1}{n}\sum_{t=1}^{n}\mZ_t^v\mZ_t^{v\prime}$ \citep[e.g.,][Corollary~3.48]{Whi01}, such that Assumption~\ref{ass:9}~\ref{it:9i} is satisfied.

	\textbf{Assumption~\ref{ass:9}~\ref{it:9ii}:} Use \eqref{eq:dens trafo} and the definition of $v_t(\vtheta_0^v)$ to write
	\begin{align*}
		f_t\big(v_t(\vtheta_0^v), y\big) &= f_{\varepsilon_X,\varepsilon_Y}\big(v_t(\vtheta_0^v) - \mZ_t^{v\prime}\vtheta_0^v, y-\mZ_t^{m\prime}\vtheta_0^m\big)\\
		&= f_{\varepsilon_X,\varepsilon_Y}(0, y-\mZ_t^{m\prime}\vtheta_0^m).
	\end{align*}
	Therefore, using an obvious substitution,
	\begin{align*}
		\int_{-\infty}^{\infty} y f_t\big(v_t(\vtheta_0^v), y\big)\D y &= \int_{-\infty}^{\infty} y f_{\varepsilon_X,\varepsilon_Y}(0, y-\mZ_t^{m\prime}\vtheta_0^m)\D y\\
		&= \int_{-\infty}^{\infty} \big(y+\mZ_t^{m\prime}\vtheta_0^m\big) f_{\varepsilon_X,\varepsilon_Y}(0, y)\D y\\
		&= \int_{-\infty}^{\infty} y f_{\varepsilon_X,\varepsilon_Y}(0, y)\D y + \mZ_t^{m\prime}\vtheta_0^m\int_{-\infty}^{\infty} f_{\varepsilon_X,\varepsilon_Y}(0, y)\D y\\
		&= \int_{-\infty}^{\infty} y f_{\varepsilon_X,\varepsilon_Y}(0, y)\D y + \mZ_t^{m\prime}\vtheta_0^m f_{\varepsilon_X}(0).
	\end{align*}
	With this, \eqref{eq:(4.1)} and \eqref{eq:(p.59)},
	\begin{align*}
		\frac{1}{n}&\sum_{t=1}^{n}\nabla m_t(\vtheta_0^m)\nabla^\prime v_t(\vtheta_0^v)\bigg\{\int_{-\infty}^{\infty} y f_t\big(v_t(\vtheta_0^v),y\big)\D y - m_t(\vtheta_0^m)f_t^{X}\big(v_t(\vtheta_0^v)\big)\bigg\}\\
		&=\frac{1}{n}\sum_{t=1}^{n}\mZ_t^{m}\mZ_t^{v\prime}\bigg\{\int_{-\infty}^{\infty} y f_{\varepsilon_X,\varepsilon_Y}(0, y)\D y + \mZ_t^{m\prime}\vtheta_0^m f_{\varepsilon_X}(0) - \mZ_t^{m\prime}\vtheta_0^m f_{\varepsilon_X}(0)\bigg\}\\
		&= \bigg\{\int_{-\infty}^{\infty} y f_{\varepsilon_X,\varepsilon_Y}(0, y)\D y\bigg\}\frac{1}{n}\sum_{t=1}^{n}\mZ_t^{m}\mZ_t^{v\prime}.
	\end{align*}
	By Assumption~LM.\ref{M:3} and standard mixing inequalities \citep[e.g.,][Theorem~3.49]{Whi01}, the sequence $\big\{\mZ_t^{m}\mZ_t^{v\prime}\big\}_{t\in\mathbb{N}}$ is $\beta$-mixing of size $-r/(r-1)$ and, hence, in particular $\alpha$-mixing of size $-r/(r-1)$.
	Since also
	\[
	\E\big[\Vert\mZ_t^{m}\mZ_t^{v\prime}\Vert^{2r}\big]\leq \sqrt{\E\Vert\mZ_t^{m}\Vert^{4r}\E\Vert\mZ_t^{v}\Vert^{4r}}<\infty
	\]
	by Assumption~LM.\ref{M:4} and the Cauchy--Schwarz inequality, the $\alpha$-mixing LLN of \citet[Corollary~3.48]{Whi01} applies for $\frac{1}{n}\sum_{t=1}^{n}\mZ_t^{m}\mZ_t^{v\prime}$.
	This easily establishes Assumption~\ref{ass:9}~\ref{it:9ii}.

	\textbf{Assumption~\ref{ass:10}:} The verification is similar to that of Assumption~\ref{ass:3}~\ref{it:3iv}.
	It follows from \eqref{eq:dens trafo} and the Cauchy--Schwarz inequality that
	\begin{align}
		\int_{-\infty}^{\infty} |y|^2f_t(x,y)\D y &= \int_{-\infty}^{\infty} |y|^2f_{\varepsilon_X,\varepsilon_Y}(x-\mZ_t^{v\prime}\vtheta_0^v, y-\mZ_t^{m\prime}\vtheta_0^m)\D y \notag\\
		&= \int_{-\infty}^{\infty} (y+\mZ_t^{m\prime}\vtheta_0^m)^2f_{\varepsilon_X,\varepsilon_Y}(x-\mZ_t^{v\prime}\vtheta_0^v, y)\D y \notag\\
		&\leq 2 \int_{-\infty}^{\infty} y^2 f_{\varepsilon_X,\varepsilon_Y}(x-\mZ_t^{v\prime}\vtheta_0^v, y)\D y \notag\\
		&\hspace{2cm} + 2 \big|\mZ_t^{m\prime}\vtheta_0^m\big|^2 \int_{-\infty}^{\infty} f_{\varepsilon_X,\varepsilon_Y}(x-\mZ_t^{v\prime}\vtheta_0^v, y)\D y \notag\\
		&\leq 2\sup_{x\in\mathbb{R}} \int_{-\infty}^{\infty} y^2 f_{\varepsilon_X,\varepsilon_Y}(x, y)\D y + 2 \big|\mZ_t^{m\prime}\vtheta_0^m\big|^2 \sup_{x\in\mathbb{R}} \int_{-\infty}^{\infty} f_{\varepsilon_X,\varepsilon_Y}(x, y)\D y\notag\\
		&=:F_2(\mathcal{F}_t).\label{eq:(p.3.1)}
	\end{align}

	\textbf{Assumption~\ref{ass:11}:} From \eqref{eq:(4.2)} and \eqref{eq:(4.3)},
	\[
	\E\big[V_1^2(\mZ_{t})M^2(\mZ_{t})M_1^2(\mZ_{t})\big]=C\E\big[\Vert\mZ_t^{v}\Vert^2 \Vert\mZ_t^{m}\Vert^4\big]\leq C
	\]
	by Assumption~LM.\ref{M:4}.
	Similarly, from \eqref{eq:(4.2)}, \eqref{eq:(4.3)} and \eqref{eq:(p.3.1)},
	\[
	\E\big[M_1^2(\mZ_t)V_1^2(\mZ_t)F_2(\mathcal{F}_t)\big]\leq \E\big[\Vert\mZ_t^{v}\Vert^2 \Vert\mZ_t^{m}\Vert^2(C + C \Vert\mZ_t^{m}\Vert^2)\big]\leq C.
	\]

	We have now verified Assumptions~\ref{ass:1}--\ref{ass:11}, therefore ending the proof.
\end{proof}

We now show that Assumptions~LM.\ref{M:3}--LM.\ref{M:4} hold for $\{X_t\}$ and $\{Y_t\}$ generated by an empirically relevant AR--GARCH model.
In Example~\ref{ex:AR-GARCH}, we use the generic $Z_t$ to denote either $X_t$ or $Y_t$.

\setcounter{example}{1}
\begin{example}\label{ex:AR-GARCH}
	For $t\in\mathbb{Z}$, consider the AR($\overline{p}$)--GARCH($1,1$) model
	\begin{align*}
		Z_t &= \mu_t + \sigma_t\varepsilon_t,\\
		\mu_t &= \phi_0 + \phi_1 Z_{t-1} + \ldots + \phi_{\overline{p}} Z_{t-\overline{p}},\\
		\sigma_t^2 &= \omega^{\circ} + \alpha^{\circ} u_{t-1}^2 + \beta^{\circ} \sigma_{t-1}^2,
	\end{align*}
	where $u_t=Z_t - \mu_t$ is the (conditionally) demeaned observation and $\varepsilon_t$ is a sequence of i.i.d.~random variables with $\E[\varepsilon_t]=0$ and $\Var(\varepsilon_t)=1$, such that $\varepsilon_t$ is independent of $\{Z_s\}_{s<t}$.
	By construction, $\mu_t$ and $\sigma_t^2$ are the conditional mean and conditional variance of $Z_t$, i.e., $\E[Z_{t}\mid Z_{t-1},Z_{t-2}\ldots]=\mu_t$ and $\Var(Z_{t}\mid Z_{t-1},Z_{t-2}\ldots)=\sigma_t^2$.
	Now, impose the following conditions:
	\begin{enumerate}[(i)]
		
		\item\label{it:Lebesgue} The $\varepsilon_t$ have a Lebesgue density, which is positive and lower semicontinuous on $\mathbb{R}$.
		Furthermore, $\E\big[|\varepsilon_t|^{2s}\big]<\infty$ for some $s>0$.
		
		\item\label{it:ft} The spectral radius $\rho(\mA)$ of
		\[
		\mA=\begin{pmatrix}
			\phi_1 & \phi_{2} & \ldots & \phi_{\overline{p}-1} & \phi_{\overline{p}}\\
			1 &0 &\ldots & 0& 0\\
			0 &1 &\ldots & 0& 0\\
			\vdots & \vdots & \ddots & \vdots & \vdots \\
			0 &0& \ldots & 1 & 0
		\end{pmatrix}
		\]
		satisfies that $\rho(\mA)<1$.
		
		\item\label{it:exp} $\E\big[(\beta^{\circ} + \alpha^{\circ}\varepsilon_t^2)^s\big] < 1$ for $s>0$ from item \eqref{it:Lebesgue}.
		
	\end{enumerate}
	Item~\eqref{it:Lebesgue} imposes assumptions on the innovations $\{\varepsilon_t\}$, item~\eqref{it:ft} on the conditional mean function $\mu_t$, and item~\eqref{it:exp} on the conditional variance process $\sigma_t^2$.
	Proposition~1 of \citet{MS08} shows that under assumptions \eqref{it:Lebesgue}--\eqref{it:exp} the conditions of their Theorem~1 are satisfied for $Z_t$.
	By the discussion following that theorem, $Z_t$ is strictly stationary and geometrically $\beta$-mixing, such that Assumption~LM.\ref{M:3} holds for any mixing size.
	Moreover, $\E\big[|Z_t|^{2s}\big]<\infty$ for $s>0$ from item \eqref{it:Lebesgue}.
	Therefore, by a suitable choice of $s>0$, the existence of moments of any order can be ensured to satisfy Assumption~LM.\ref{M:4}.
\end{example}

Clearly, Proposition~1 of \citet{MS08} used in the above example also holds for nonlinear AR--GARCH models.
However, to simplify the exposition, we have presented the conditions for the linear case in Example~\ref{ex:AR-GARCH}.
There is a wealth of further results that can be used to verify mixing and moment conditions in time series models, such as \citet{CC02}, \citet{Lie05}, \citet{MS08a} and \citet{FS11}.
The work of \citet{Lie05} may be particularly helpful to establish mixing conditions for multivariate time series models appropriate for $\mZ_t^{v}$ and $\mZ_t^{m}$, such as certain vector-ARCH models; see in particular Section~4 in \citet{Lie05}.

\section{Relation to Regressions Containing the Distress Variable as a Covariate}
\label{sec:RelationOLSXCovariate}

One might wonder how the linear special case of our MES regressions from Example~\ref{ex:1} relates to OLS regressions of $Y_t$ containing $X_t$ as an additional explanatory variable,
\begin{align}
	\label{eq:MeanReg}
	Y_t  = X_t \eta + \mZ_t^{m\prime}\vtheta_0^{m} + u_t,
\end{align}
where $\eta \in \mathbb{R}$, $\E [ u_t \mid \mathcal{G}_t ]=0$, and $\mathcal{G}_t$ is an information set containing $X_t$,  $\mZ_t^m$ and possibly lagged values of  $X_t$, $\mZ_t^m$ and $Y_t$.
Instead of employing our MES regression, one could use \eqref{eq:MeanReg} to predict the conditional mean of $Y_t$ given that $X_t$ is ``extreme''. This would crucially differ from our MES regression in Example~\ref{ex:1} as follows:
First, \eqref{eq:MeanReg} treats $X_t$ as known at the time of prediction, whereas our MES regression treats $X_t \not\in \mathcal{F}_t$ as an outcome variable.
This is a critical difference when conditioning on future (extreme) outcomes of $X_t$, such as unfavorable economic conditions that are not observable at the time of prediction.
Second, \eqref{eq:MeanReg} only allows to condition on a fixed value $X_t = x \in \mathbb{R}$, as opposed to $X_t$ exceeding its conditional quantile.
Third, \eqref{eq:MeanReg} imposes stricter assumptions: Besides the imposed linearity of the conditional expectation of $Y_t$ given $X_t$, the linearity of $Y_t$ in $\mZ_t^m$ must hold for all values of $X_t$, whereas the MES regression of Example~\ref{ex:1} only imposes linearity of $Y_t$ in $\mZ_t^m$ in the tail of $X_t$.

\section{Managing and Allocating Portfolio Risks Using MES Regressions}
\label{sec:add appl}

\renewcommand{\theequation}{D.\arabic{equation}}	% different equation numbering in appendix
\setcounter{equation}{0}

Here, we illustrate the usefulness of our regressions under adverse conditions in managing portfolio risks.
Consider losses (i.e., negative gains) on a portfolio $X_t^{\boldsymbol{w}} = \sum_{d=1}^D w_{t,d} Y_{t,d}$ consisting of $D \in \mathbb{N}$ individual assets with possibly time-varying and non-negative portfolio weights $\boldsymbol{w}_t = (w_{t,1}, \dots, w_{t,D})^\prime \ge 0$ summing to one (i.e., $\Vert\boldsymbol{w}_t\Vert_1 = \sum_{d=1}^{D} w_{t,d} = 1$ for all $t=1,\dots, n$).
Of course, the weights $\boldsymbol{w}_t$ have to be chosen in advance at time $t-1$. 
For instance, we use equal weights (i.e., $w_{t,d}=1/D$), weights for the global minimum variance portfolio (GMVP), and weights such that the forecasted risk contribution for each portfolio component is equalized.

Under the Basel framework of the \citet{BCBS19}, market risks have to be measured by the conditional ES at level $\beta = 0.975$, i.e., $\ES_{t,\beta} = \ES_{\beta}(X_t^{\boldsymbol{w}} \mid \mathcal{F}_t) = \mathbb{E}_t [X_t^{\boldsymbol{w}} \mid X_t^{\boldsymbol{w}} \ge \VaR_{t,\beta}]$.
%It is often of great interest how the overall risk $\ES_{t,\beta} (X_t^{\boldsymbol{w}}) = \sum_{d=1}^{D} \RC_{t,d}$ disaggregates into additive risk contributions $\RC_{t,d}$, which are attributable to the different components (or investments) $Y_{t,d}$; see \citet{MFE15}.
We analyze how the overall risk $\ES_{\beta} (X_t^{\boldsymbol{w}} \mid \mathcal{F}_t) = \sum_{d=1}^{D} \RC_{t,d}$ disaggregates into additive risk contributions $\RC_{t,d}$, which are attributable to the different components $Y_{t,d}$ \citep[Section 8.5]{MFE15}.
These risk contributions quantify the risk exposure of the portfolio with respect to its individual components.
We show in Section \ref{sec:ERC_Portfolio} how the risk contributions can be used to construct a diversified portfolio where all components contribute equally to its overall riskiness.

%Specifically, we consider allocating the risk of the overall portfolio into its individual components $Y_{t,d}$ in the sense that $\ES_\beta(X_t) = \sum_{d=1}^{D} \RC_{t,d}$, where $\RC_{t,d}$ is the risk contribution allocated to investment $Y_{t,d}$.
The risk contribution of a single investment within a portfolio (i.e., $\RC_{t,d}$)  is not simply the stand-alone risk of that investment in isolation (i.e., $\ES_{\beta}(Y_{t,d} \mid \mathcal{F}_t)$).
This would neglect diversification benefits and give an inaccurate measure of the performance of an investment within the larger portfolio. 
Instead, the risk contribution for an investment should reflect the contribution of that investment to the overall riskiness of the portfolio. 
The \textit{Euler allocation rule} \citep[see][Eq.~(8.64)]{MFE15} states that for the risk measure ES, the risk contributions are 
%\yannick{In formula \eqref{eq:RiskContribution}, should it be $\VaR_{\b}$ instead of $\VaR_{t,\b}$ to be consistent with above definition $\ES_\beta(X_t^{\boldsymbol{w}}) = \mathbb{E} [X_t^{\boldsymbol{w}} \mid X_t^{\boldsymbol{w}} \ge \VaR_\beta]$?}
\begin{equation}
	\label{eq:RiskContribution}
	\RC_{t,d} = \E_t [ w_{t,d} Y_{t,d}\mid X_t^{\boldsymbol{w}} \ge \VaR_{t,\b} ],
\end{equation}
which equals the $\beta$-MES of the relative value of investment $d$ in the portfolio, $w_{t,d} Y_{t,d}$.\footnote{\citet[Section 8.5]{MFE15} illustrate the advantages of the Euler allocation rule in \eqref{eq:RiskContribution}. It is essentially the only allocation rule that guarantees diversification benefits and it is compatible with the RORAC (return on risk-adjusted capital) approach:
	Define the overall (conditional) $\RORAC(X_t^{\boldsymbol{w}}) := \frac{\E_t[-X_t^{\boldsymbol{w}}]}{\ES_{\beta}(X_t^{\boldsymbol{w}} \mid \mathcal{F}_t)}$ and the component-related (conditional)  $\RORAC(Y_{t,d}\mid X_t^{\boldsymbol{w}}):= \frac{\E_t[-Y_{t,d}]}{\RC_{t,d}}$. Then, the Euler allocation rule guarantees diversification benefits by ensuring that if component $d$ performs better than the portfolio in the RORAC metric, then the RORAC of the portfolio is increased (i.e., improved) if one slightly increases the weight of component $d$. Under weak conditions, this implication \textit{only} holds for $\RC_{t,d}$ chosen according to the Euler allocation rule, which---in this sense---is compatible with the RORAC approach.} 
Notice that in contrast  to \citet{MFE15} and much of the related literature, we consider \emph{conditional} (on $\mathcal{F}_{t}$) risk contributions that we estimate based on covariates using our MES regression.

We now estimate the conditional risk contributions of the equally weighted portfolio in Section \ref{sec:EW_Portfolio}, and construct a portfolio with equal risk contributions in Section \ref{sec:ERC_Portfolio}.

%\footnote{
	%	Defining the overall $\RORAC(X_t^{\boldsymbol{w}})$ and the portfolio-related $\RORAC(Y_{t,d}\mid X_t^{\boldsymbol{w}})$ of investment $Y_{t,d}$ by
	%	\begin{align*}
		%		\RORAC(X_t^{\boldsymbol{w}}):= \frac{\E[-X_t^{\boldsymbol{w}}]}{\ES_\beta(X_t^{\boldsymbol{w}})}.
		%		\qquad \text{ and } \qquad 
		%		\RORAC(Y_{t,d}\mid X_t^{\boldsymbol{w}}):= \frac{\E[-Y_{t,d}]}{\RC_{t,d}},
		%	\end{align*}	
	%	we get that if $\RC_{t,d}$ is chosen as in \eqref{eq:RiskContribution}, there exists some $\varepsilon>0$, such that
	%	\begin{equation*}
		%		%\label{eq:RORAC}
		%		\RORAC(Y_{t,d}\mid X_t^{\boldsymbol{w}}) > \RORAC(X_t^{\boldsymbol{w}}) \qquad\Longrightarrow\qquad \RORAC(X_t^{\boldsymbol{w}} + h Y_{t,d})  > \RORAC(X_t^{\boldsymbol{w}})
		%	\end{equation*}
	%	for all $0<h\leq\varepsilon$. 
	%	This implies that if investment opportunity $d$ performs better than the portfolio in the RORAC metric, then the RORAC of the portfolio is increased if one slightly increases the weight of unit $d$. 
	%	Thus, the Euler principle gives correct signals for investment decisions, which we illustrate in Section \ref{sec:ERC_Portfolio}.
	%}

\subsection{The Equally Weighted Portfolio}
\label{sec:EW_Portfolio}

We consider an equally weighted (EW) portfolio $X_t^{\boldsymbol{w}}$ (with $w_{t,d} = 1/5$) consisting of the five large cap assets Boeing (BA), JPMorgan Chase (JPM), Eli Lilly and Co.~(LLY), Microsoft (MSFT), and  Exxon Mobil (XOM).
Each company represents a different economic sector.
Their weekly log-losses between January 2nd, 1990 and May 29th, 2023 are denoted by $Y_{t,d}$ ($d=1,\dots,5$), resulting in a total of  $T = 1744$ trading weeks.
Despite its simplicity, the EW portfolio is known to perform well empirically, mainly due to its stability \citep{deMiguel2009}.

We model the risk contributions in \eqref{eq:RiskContribution} conditionally on an intercept and the lagged weekly average of the VIX index, $\mZ_t^v = \mZ_t^m = \mZ_t = (1, \VIX_{t-1})$. Specifically, we use the linear MES regression
\begin{align}
	\VaR_{t, \beta} = \VaR_{\beta}(X_t^{\boldsymbol{w}} \mid \mathcal{F}_t ) &= \theta_1^v + \theta_2^v \VIX_{t-1},\notag \\
	\RC_{t,d} = \E_t \big[ w_{t,d} Y_{t,d}\mid X_t^{\boldsymbol{w}} \ge \VaR_{t,\b} \big] &= \theta_1^m + \theta_2^m \VIX_{t-1}.\label{eq:RegModelVIX}
\end{align}
We forecast $\widehat{\RC}_{t,d}$ for all weeks $t \in \mathfrak{T} := \{S,\dots,T\}$ in the out-of-sample window
by using the model in \eqref{eq:RegModelVIX}, % , which can be expected to prove valuable in improving capital allocation in portfolio optimization.
% We use an \textcolor{red}}intercept and the lagged weekly averaged VIX index as explanatory variables and ...
which we re-estimate every week with an expanding window approach, i.e., using all data with time index $s < t$.
The first estimation window consists of data up to December 31, 2006, hence implying $S=888$.

\begin{figure}
\centering
\includegraphics[width=\linewidth]{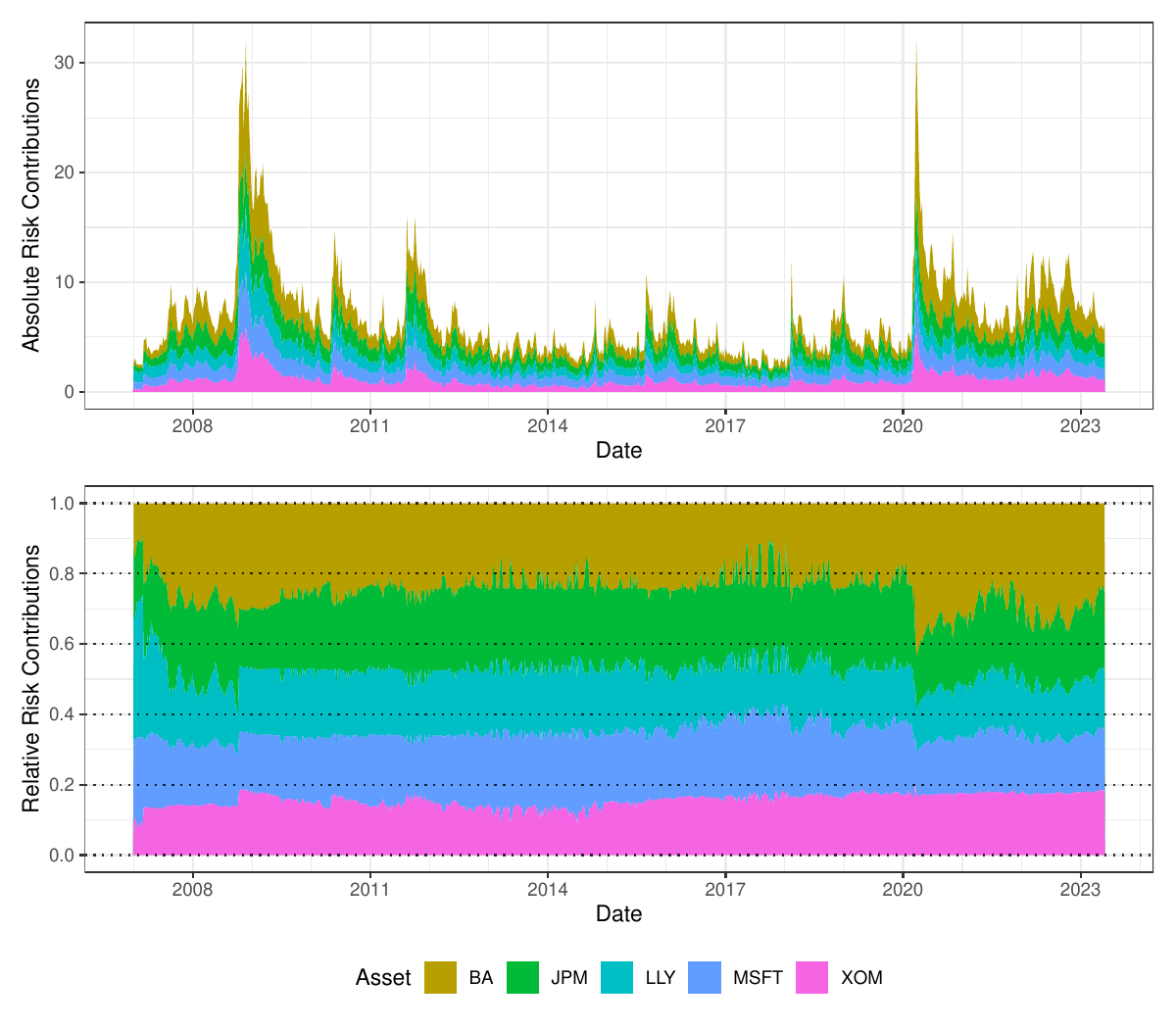}
\caption{Absolute (upper panel) and relative (lower panel) weekly risk contributions (MES forecasts) for the equally weighted portfolio during the out-of-sample period $\mathfrak{T}$.}
\label{fig:MES_EWPortfolio}
\end{figure}

We display the resulting  absolute and relative values of the forecasted MES contributions in Figure \ref{fig:MES_EWPortfolio}.
Besides the absolute time-variation of the overall portfolio risk (measured in terms of ES) that can be seen in the upper panel of the figure, the relative risk contribution of the individual assets changes substantially over the sample period.
While the EW portfolio is balanced in terms of the invested money (the weights), it is unbalanced with respect to the implied risk contributions.
Not surprisingly, the forecasted risk contribution of JPMorgan Chase clearly increases during the financial crisis of 2007--08, whereas with the start of the COVID pandemic, the risk contribution of the aircraft manufacturer Boeing almost doubles, while the risk of Microsoft as a technology corporation decreases.

\subsection{The Equal Risk Contribution Portfolio}
\label{sec:ERC_Portfolio}

The unequal and time-varying risk contributions to the overall portfolio risk observed in Figure \ref{fig:MES_EWPortfolio} motivate the construction of a portfolio with balanced risk contributions.
%, which is an (equally weighted) special case of the \textit{risk parity} portfolio, which is highly popular in practice 
\citet{Maillard2010} introduce the so-called risk parity or equal risk contribution (ERC) portfolio.
In the ERC portfolio, the weights are chosen such that each investment contributes equally to the overall portfolio risk.
For example, the relative risk contributions of the five assets in Figure~\ref{fig:MES_EWPortfolio} would be 0.2 in the ERC portfolio, instead of being time-varying.
Among many others, we refer to \citet{AFD12}, \citet{Bruder2012}, \citet{Roncalli2013Book} and \citet{Bea18} for the ERC and risk parity portfolio.
% approach, which aims at constructing a portfolio such that each component contributes equally to its risk.
While the previously mentioned authors use the variance as the target risk measure, we follow the Basel framework of the \citet{BCBS19} and focus on the ES risk measure.
For ES, very little work has been carried out in constructing ERC portfolios; see \citet{CesaroneColucci2017} and \citet{Mausser2018} for some exceptions. 
We crucially deviate from these papers by considering MES \emph{forecasts} issued from our covariate-based MES regressions.

In formulae, for all weeks in the out-of-sample window, $t \in \mathfrak{T}$, the weight vector $\boldsymbol{w}_t = (w_{t,1}, \dots, w_{t,D})^\prime \ge 0$ with $\Vert \boldsymbol{w}_t \Vert_1 = 1$  of the ERC portfolio $X_t^{\boldsymbol{w}} = \sum_{d=1}^D w_{t,d} Y_{t,d}$ is chosen such that forecasts for the risk contributions $\RC_{t,d}^w = \E_t [ w_{t,d} Y_{t,d} \mid X_t^{\boldsymbol{w}} \ge \VaR_{t,\b} ]$ are equal for all $d=1,\dots,D$.
As the portfolio weights must be allocated in advance, we forecast $\widehat{\RC}_{t,d}^w$ for all  $t \in \mathfrak{T}$ with the expanding window described above by estimating the MES regression in \eqref{eq:RegModelVIX}.
%\begin{align}
%	\label{eq:ERC_MES_formula}
%	\E_t [ w_{t,d} Y_{t,d} \mid X_t^{\boldsymbol{w}} \ge \VaR_{t,\b} ] = \theta_1 + \theta_2 \VIX_{t-1},
%\end{align}
We iterate over the possible weights $\boldsymbol{w}_t$ until the forecasted risk contributions $\widehat{\RC}_{t,d}^w$ are approximately equal by following Algorithm \ref{alg:ERCportfolio}.\footnote{We update the weights by an average of the reciprocal of the forecasted MES and the previous weights in lines 6 and 7 to avoid an ``overshooting'' that was found to occur without the averaging.}

\begin{algorithm}
\caption{ERC Portfolio Weights}
\label{alg:ERCportfolio}
\begin{algorithmic}[1]
	\For{$t \in \mathfrak{T}$}
	\State $\boldsymbol{w}_t \gets (1/D, \dots, 1/D)$
	\State $\widehat{\RC}_{t,d}^w \gets (1,0,\dots,0)$
	% \State  $\widehat{\RC}_{t,d}^w \gets 1/5$ for all $d$
	\While{$\max_{d,e \in \{1,\dots,D\} } \big| \widehat{\RC}_{t,d}^w -  \widehat{\RC}_{t,e}^w \big| > 0.01$}
	\State Estimate the MES model in \eqref{eq:RegModelVIX} based on all data $s < t$ and forecast $\widehat{\RC}_{t,d}^w$
	\State $\widetilde{w}_{t,d} \gets w_{t,d} / (2 \widehat{\RC}_{t,d}^w) + w_{t,d}/2$ for all $d$
	\State $w_{t,d} \gets \widetilde{w}_{t,d} / \Vert \widetilde{\boldsymbol{w}}_t \Vert_1$ for all $d$
	\EndWhile
	\EndFor
	\State \textbf{return} $(w_{t,d})_{t \in \mathfrak{T}}$
\end{algorithmic}
\end{algorithm}

Besides the EW portfolio, we also use the GMVP for comparison, which is an attractive portfolio strategy that is still actively investigated \citep{SSX23+, Reh2023}. The GMVP minimizes the overall (forecasted) risk---measured in terms of the variance---while ignoring expected returns entirely.
As expected returns are inherently difficult to estimate \citep{Jagannathan2003, Welch2008}, the latter can be interpreted as a robustness property of the GMVP, which it shares with the EW portfolio.
However, the overall risk minimization of the GMVP often leads to a drastic concentration in portfolio weight (and also risk contribution) on certain assets \citep{Maillard2010}. %that are forecasted to have a comparably low risk 
In our empirical analysis, we compute the GMVP weights based on the variance-covariance matrix forecasts from a state-of-the-art rolling window DCC--GARCH model of \cite{Eng02} with (estimated) Student's $t$-distributed innovations.\footnote{We normalize the GMVP weights to be non-negative by setting negative weights to zero and by re-balancing the remaining weights such that they add to one. The effect of this normalization on our results is almost negligible.}

\begin{figure}[tb]
\centering
\includegraphics[width=\linewidth]{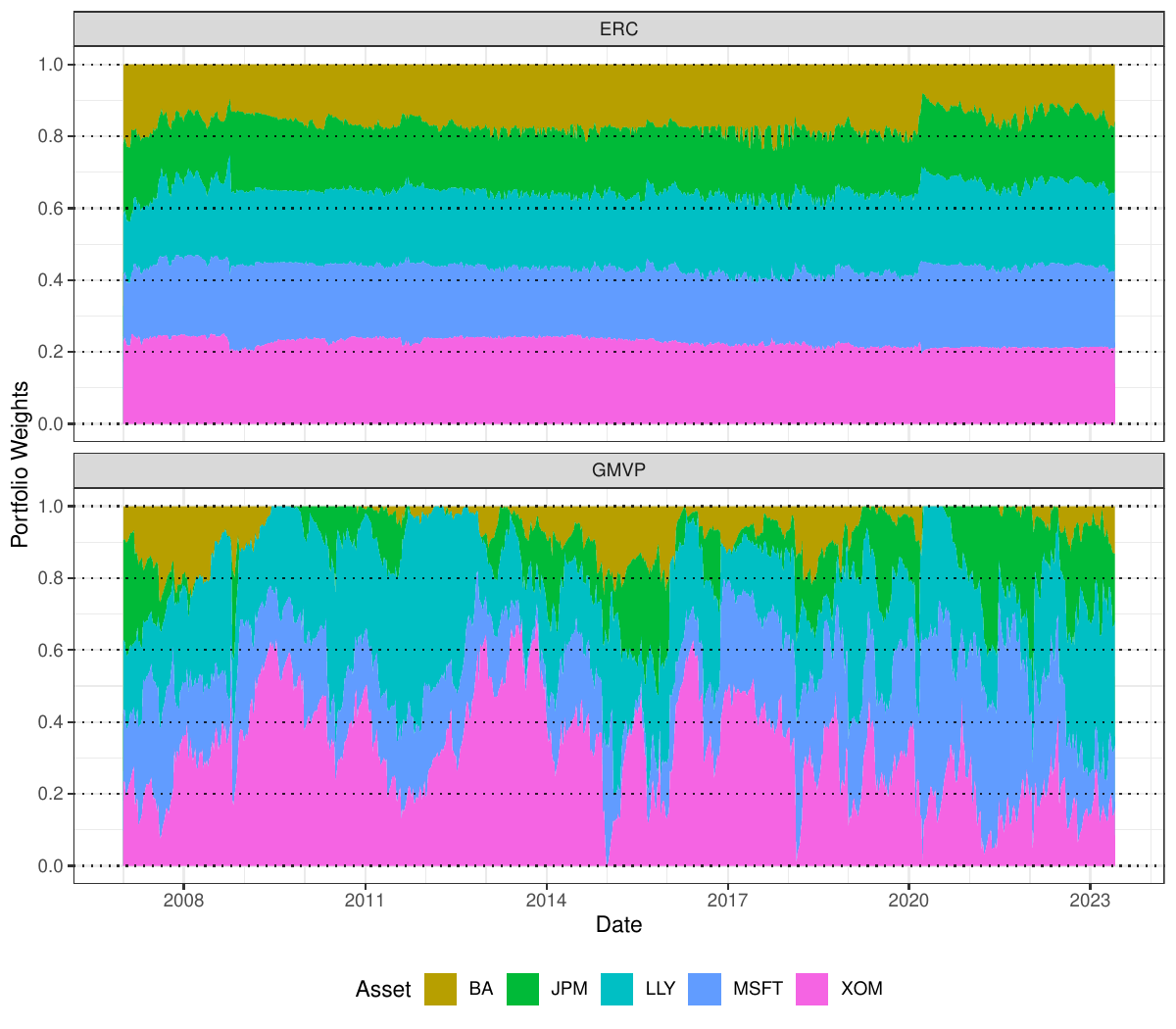}
\caption{Portfolio weights of the ERC portfolio (upper panel) and GMVP (lower panel) during the out-of-sample period $\mathfrak{T}$.}
\label{fig:ERCweights}
\end{figure}

Figure \ref{fig:ERCweights} shows the evolution of the GMVP and ERC portfolio weights over the out-of-sample window $\mathfrak{T}$.
Overall, the ERC portfolio contains relatively little shares in the risky stocks of Boeing and JPMorgan Chase, and more shares of the remaining three assets.
We can clearly see the effects of the financial crisis of 2008 and the COVID crisis in 2020, especially on the particularly exposed firms such as Boeing, whose share is drastically reduced in the beginning of 2020.
The ERC weights are however relatively stable over time, especially in comparison with the GMVP that reacts very strongly to a changing covariance structure within the assets, and shows a severe concentration on particular firms such that---at certain times---over 60\% of the money is invested in a single asset (Exxon Mobil).
Besides generating substantial re-balancing costs, this leads to a risk concentration in very few assets, thus defeating possible diversification benefits.
In contrast, the ERC portfolio weights are strikingly more stable (as opposed to the GMVP), while still being able to react to varying risks (as opposed to the EW portfolio).
Hence, the ERC portfolio can be interpreted as a balanced middle ground between the GMVP and the EW portfolio.

\begin{table}[tb]
\centering
\begin{tabular}{l c rrrrrr}
	\toprule
	&& \multicolumn{6}{c}{Empirical Portfolio Return Measures} \\
	\cmidrule{3-8}   
	Portfolio &&  Avg return   & Std &    VaR   &   ES  & Sharpe Ratio  &  RORAC \\
	\midrule
	ERC && 0.155 &  2.91 & 5.71  & 9.16  & 0.0532 & 0.0169 \\
	EW && 0.158   & 3.06 & 5.86 & 9.42 & 0.0517 & 0.0168 \\
	GMVP && 0.115  & 2.72 & 5.35 & 8.55 & 0.0424 & 0.0135 \\
	\bottomrule
\end{tabular}	
\caption{Performance results of the ERC, EW, and GMVP portfolios. The performance measures are computed as the respective functionals of the empirical distributions over the out-of-sample period $\mathfrak{T}$.}
\label{tab:ERCPortfolioPerformance}
\end{table}

Table~\ref{tab:ERCPortfolioPerformance} summarizes the overall performance of the GMVP, the EW and the ERC portfolios in terms of classical metrics that are computed as the respective functionals of the empirical distributions over the out-of-sample period $\mathfrak{T}$.
While the ERC and the EW portfolios achieve a comparable average return, the ERC portfolio is able to reduce its overall riskiness measured in terms of either standard deviation, the VaR or ES in comparison to the EW portfolio.
The further reduction in risk achieved by the GMVP comes at the cost of a substantially smaller average return.
Hence, the ERC portfolio achieves the highest Sharpe ratio and RORAC metric.
These results reconfirm the ERC portfolio as a robust middle ground between the GMVP and the EW portfolio, which combines the properties of achieving high returns while reducing (and in particular balancing) risks.
In this application, our regressions under adverse conditions prove to be a very valuable tool in assessing risk contributions and in constructing ERC portfolios.

\section{Additional Empirical Results}
\label{sec:AddEmpResults}

Here, we present additional details and results for the two applications in Section~\ref{sec:EmpiricalApplication} of the main article.

In particular, Table~\ref{tab:PredictiveRegression_AddBanks} reports results of our systemic risk regression analogous  to Table~\ref{tab:PredictiveRegression} in the main article, but for five additional systemically relevant US banks according to \cite{FSB22}.
These results qualitatively match the ones in the main article:
Most importantly, the apparent significance (in the ad hoc MES regression) of the boom and bust indicators cannot be verified by our MES regression model.

\begin{table}[tb]
\centering
% \scriptsize
\footnotesize
\resizebox{\columnwidth}{!}{
	\begin{tabular}{c c l c rrr c rrr c rrr}
		\toprule
		&&&& \multicolumn{3}{c}{VaR Model} &&  \multicolumn{3}{c}{MES Model} &&  \multicolumn{3}{c}{Ad Hoc MES  Regression} \\
		\cmidrule{5-7}   	\cmidrule{9-11}   	\cmidrule{13-15}  
		Bank & & Covariate && Est. & SE & $p$-val  &&  Est. & SE & $p$-val &&  Est. & SE & $p$-val   \\
		\midrule
		\addlinespace		
		\multirow{6}{*}{BK} &  & Intercept &  & $-$1.734 & 0.185 & 0.000 &  & $-$2.025 & 0.748 & 0.007 &  & $-$0.388 & 0.064 & 0.000 \\ 
		&  & ChangeSpread &  & 0.591 & 0.107 & 0.000 &  & 0.455 & 0.369 & 0.217 &  & 1.226 & 0.037 & 0.000 \\ 
		&  & TEDSpread &  & 1.751 & 0.182 & 0.000 &  & 3.532 & 1.269 & 0.005 &  & $-$1.013 & 0.051 & 0.000 \\ 
		&  & VIX &  & 0.092 & 0.015 & 0.000 &  & 0.135 & 0.040 & 0.001 &  & 0.057 & 0.004 & 0.000 \\ 
		&  & ST\_boom &  & $-$0.194 & 0.125 & 0.120 &  & $-$0.411 & 0.409 & 0.315 &  & 0.402 & 0.054 & 0.000 \\ 
		&  & ST\_bust &  & $-$0.571 & 0.169 & 0.001 &  & $-$1.520 & 0.587 & 0.010 &  & 0.381 & 0.074 & 0.000 \\ 
		\midrule
		\addlinespace		 
		\multirow{6}{*}{GS} &  & Intercept &  & $-$1.711 & 0.212 & 0.000 &  & $-$1.242 & 1.220 & 0.309 &  & 0.051 & 0.080 & 0.526 \\ 
		&  & ChangeSpread &  & 0.433 & 0.144 & 0.003 &  & 0.326 & 0.544 & 0.549 &  & 1.138 & 0.046 & 0.000 \\ 
		&  & TEDSpread &  & 1.796 & 0.403 & 0.000 &  & 2.114 & 0.638 & 0.001 &  & $-$0.300 & 0.055 & 0.000 \\ 
		&  & VIX &  & 0.115 & 0.018 & 0.000 &  & 0.148 & 0.041 & 0.000 &  & 0.021 & 0.004 & 0.000 \\ 
		&  & ST Boom &  & $-$0.411 & 0.333 & 0.218 &  & $-$0.048 & 1.179 & 0.967 &  & 0.593 & 0.099 & 0.000 \\ 
		&  & ST Bust &  & $-$0.772 & 0.174 & 0.000 &  & $-$1.133 & 0.794 & 0.154 &  & 0.336 & 0.079 & 0.000 \\ 
		\midrule
		\addlinespace		 
		\multirow{6}{*}{MS} &  & Intercept &  & $-$1.734 & 0.185 & 0.000 &  & $-$2.486 & 0.893 & 0.005 &  & $-$1.376 & 0.083 & 0.000 \\ 
		&  & ChangeSpread &  & 0.591 & 0.107 & 0.000 &  & 0.838 & 0.468 & 0.073 &  & 1.839 & 0.048 & 0.000 \\ 
		&  & TEDSpread &  & 1.751 & 0.182 & 0.000 &  & 2.657 & 1.358 & 0.050 &  & $-$1.115 & 0.066 & 0.000 \\ 
		&  & VIX &  & 0.092 & 0.015 & 0.000 &  & 0.170 & 0.050 & 0.001 &  & 0.075 & 0.005 & 0.000 \\ 
		&  & ST Boom &  & $-$0.194 & 0.125 & 0.120 &  & -0.629 & 0.566 & 0.267 &  & 0.614 & 0.070 & 0.000 \\ 
		&  & ST Bust &  & $-$0.571 & 0.169 & 0.001 &  & 0.053 & 0.876 & 0.951 &  & 0.343 & 0.096 & 0.000 \\ 
		\midrule
		\addlinespace		 
		\multirow{6}{*}{STT} &  & Intercept &  & $-$1.734 & 0.185 & 0.000 &  & $-$5.643 & 3.512 & 0.108 &  & $-$0.647 & 0.125 & 0.000 \\ 
		&  & ChangeSpread &  & 0.591 & 0.107 & 0.000 &  & 1.616 & 1.414 & 0.253 &  & 1.219 & 0.073 & 0.000 \\ 
		&  & TEDSpread &  & 1.751 & 0.182 & 0.000 &  & 2.840 & 2.254 & 0.208 &  & $-$1.951 & 0.100 & 0.000 \\ 
		&  & VIX &  & 0.092 & 0.015 & 0.000 &  & 0.199 & 0.046 & 0.000 &  & 0.109 & 0.007 & 0.000 \\ 
		&  & ST Boom &  & $-$0.194 & 0.125 & 0.120 &  & 0.187 & 0.935 & 0.842 &  & $-$0.360 & 0.105 & 0.001 \\ 
		&  & ST Bust &  & $-$0.571 & 0.169 & 0.001 &  & $-$1.312 & 1.143 & 0.251 &  & 0.657 & 0.145 & 0.000 \\ 
		\midrule
		\addlinespace		 
		\multirow{6}{*}{WFC} &  & Intercept &  & $-$1.734 & 0.185 & 0.000 &  & $-$4.271 & 1.053 & 0.000 &  & $-$0.295 & 0.120 & 0.014 \\ 
		&  & ChangeSpread &  & 0.591 & 0.107 & 0.000 &  & 1.706 & 0.460 & 0.000 &  & 1.105 & 0.069 & 0.000 \\ 
		&  & TEDSpread &  & 1.751 & 0.182 & 0.000 &  & 2.714 & 1.096 & 0.013 &  & $-$0.959 & 0.095 & 0.000 \\ 
		&  & VIX &  & 0.092 & 0.015 & 0.000 &  & 0.106 & 0.039 & 0.007 &  & 0.062 & 0.007 & 0.000 \\ 
		&  & ST Boom &  & $-$0.194 & 0.125 & 0.120 &  & 0.157 & 0.461 & 0.734 &  & $-$0.063 & 0.100 & 0.527 \\ 
		&  & ST Bust &  & $-$0.571 & 0.169 & 0.001 &  & $-$0.799 & 0.704 & 0.256 &  & 0.286 & 0.138 & 0.039 \\ 			
		\bottomrule
	\end{tabular}
}
\caption{VaR and MES regression results for $\beta = 0.95$ for $X_t$ being the log-losses of the S\&P~500 Financials and $Y_t$ being the log-losses of The Bank of New York Mellon Corporation (BK), Goldman Sachs (GS), Morgan Stanley (MS), State Street Corporation (STT) and Wells Fargo (WFC) in the respective panels. 
	Notice that the underlying data for GS has a reduced sample span starting on May 5th, 1999 when GS was first listed on the NYSE.
	The vertical panel entitled ``Ad Hoc MES  Regression'' reports the results of the method described around  \eqref{eq:Y_MES_transform}. The respective parameter estimates are given in the columns ``Est.'', their standard errors in the columns ``SE'', and associated $t$-test $p$-values in the columns ``$p$-val''.}
\label{tab:PredictiveRegression_AddBanks}
\end{table}

\begin{figure}[tb]
\centering
\includegraphics[width=\linewidth]{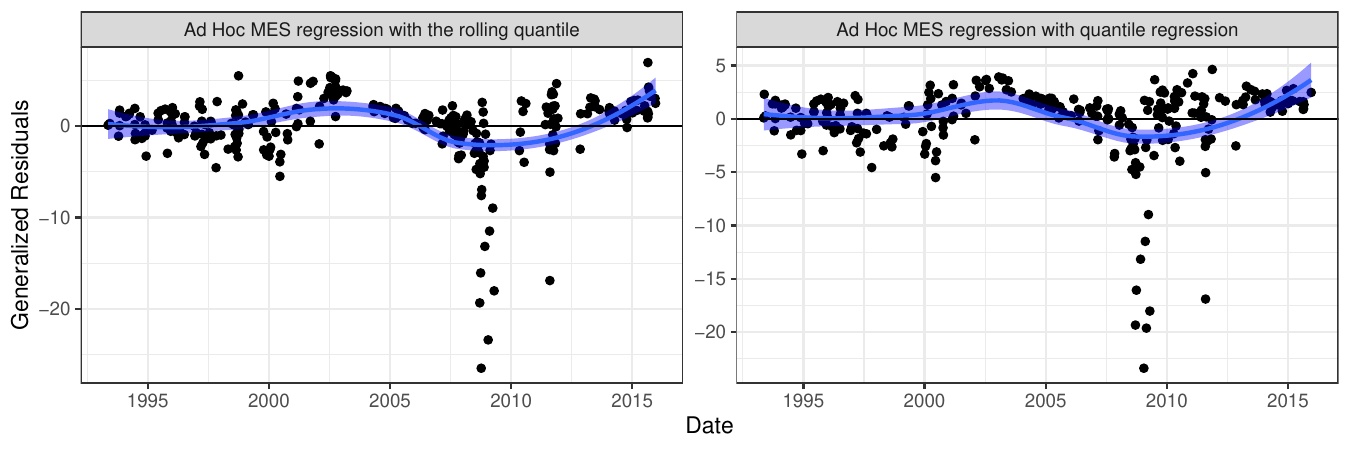}
\includegraphics[width=\linewidth]{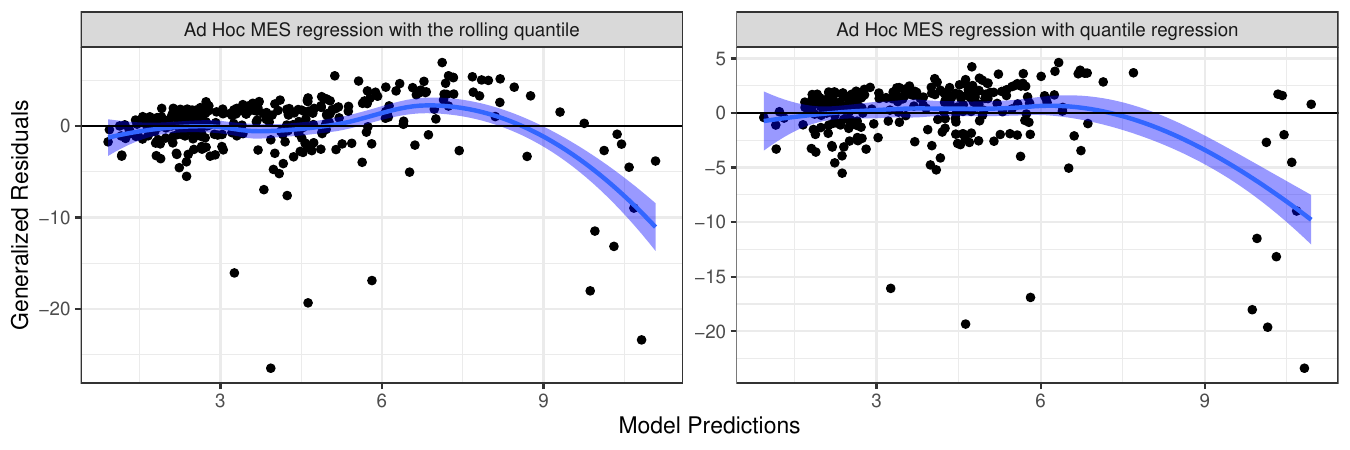}
\caption{In-sample model diagnostics of the MES (and VaR) predictions of the ``ad hoc MES regression'' used in \citet{BRS20, BDP20} for $Y_t$ being negative returns of the Bank of America (BAC), $X_t$  negative returns of the S\&P\,500 Financials and $\beta = 0.95$. 
	The generalized MES residuals are plotted and non-parametrically regressed (in blue) against time (upper row) and the model predictions (lower row).
	The plots on the left-hand side use the rolling  quantile $\widehat{Q}_\beta(X_{(t-S):t})$ for $v$ in the  generalized model residuals $V\big( (v,m)', (x,y)' \big) = \big(\1_{\{x\leq v\}}-\b, \, \1_{\{x>v\}}(m-y) \big)'$, whereas the plots on the right-hand side use a standard quantile regression of $X_t$ on the employed covariates.}
\label{fig:ModelDiagnostics_BRSDetails}
\end{figure}

Figure~\ref{fig:ModelDiagnostics_BRSDetails} provides additional model diagnostics of the ad hoc MES  regression.
While both components of the generalized model residuals $V\big( (v,m)', (x,y)' \big) = \big(\1_{\{x\leq v\}}-\b, \, \1_{\{x>v\}}(m-y) \big)'$ (that coincide with the identification functions in \citet{FH24}) necessarily require model predictions for the VaR, the ad hoc MES regression is unfortunately ambiguous about its VaR model.
The left-hand side of Figure~\ref{fig:ModelDiagnostics_BRSDetails} shows the generalized MES residuals of the  ad hoc MES  regression as in Figure~\ref{fig:ModelDiagnostics_MES_BRS_regressions} of the main article, i.e., based on the rolling  quantile $\widehat{Q}_\beta(X_{(t-S):t})$ for $v$ (which is however not conditional on the covariates).
The right-hand side of Figure~\ref{fig:ModelDiagnostics_BRSDetails} instead uses a quantile regression of $X_t$ on the employed covariates \citep{KB78}, which coincides with the first stage quantile estimation of our MES regression.
While the individual generalized residuals (shown as black points) of course differ between the two panels, the non-parametrically estimated conditional expectation functions are very similar for both variants in Figure~\ref{fig:ModelDiagnostics_BRSDetails}, hence underpinning the finding of a substantial model misspecification of the ad hoc MES  regression.

\begin{figure}[tb]
\centering
\includegraphics[width=\linewidth]{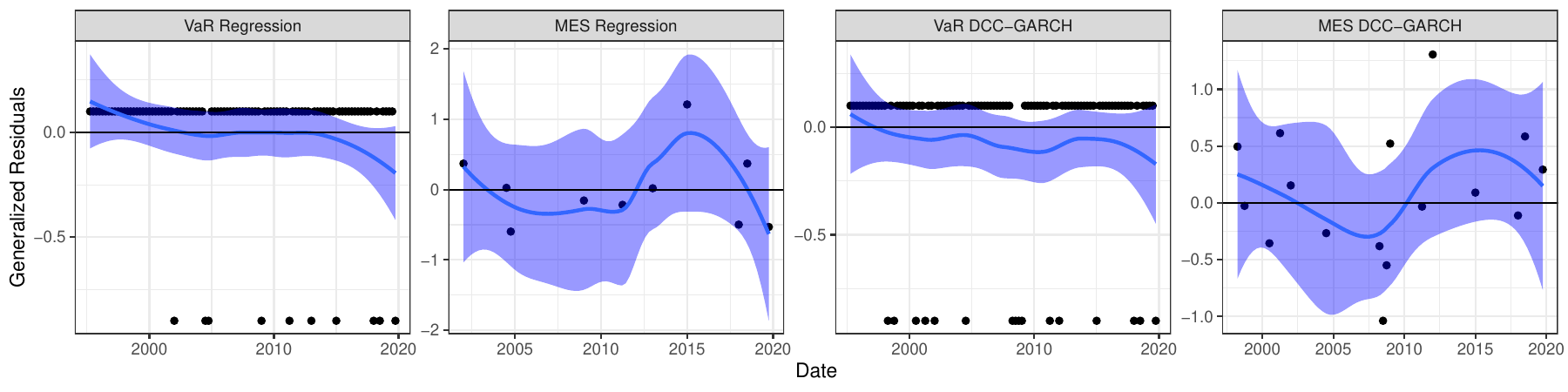}
\includegraphics[width=\linewidth]{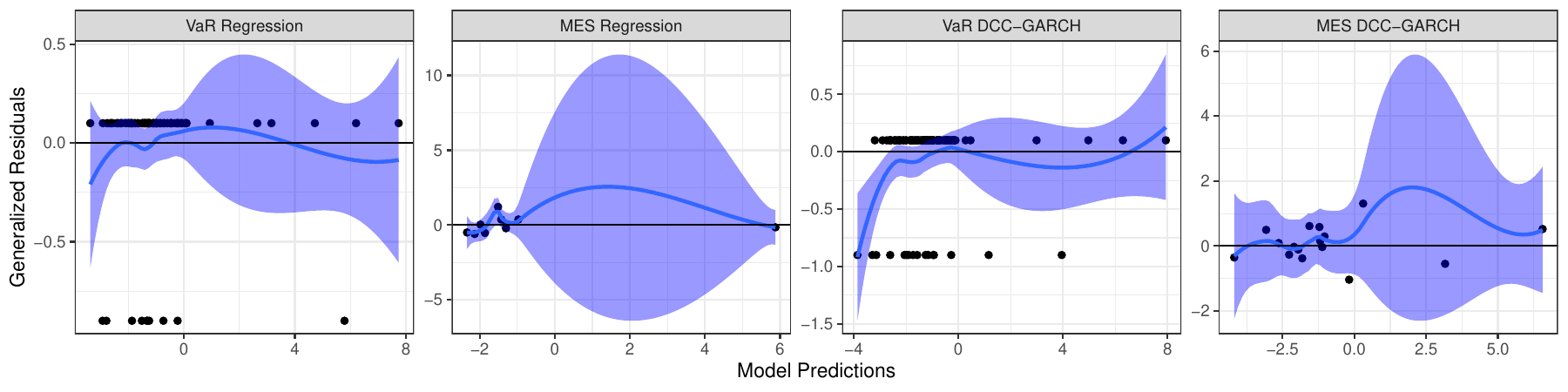}

\vspace{0.8cm}

\includegraphics[width=\linewidth]{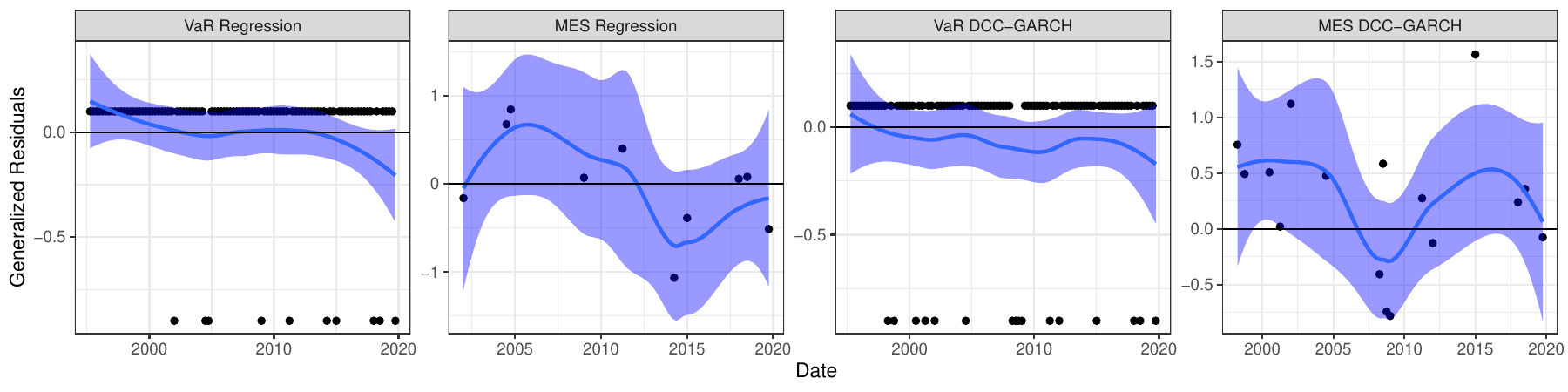}
\includegraphics[width=\linewidth]{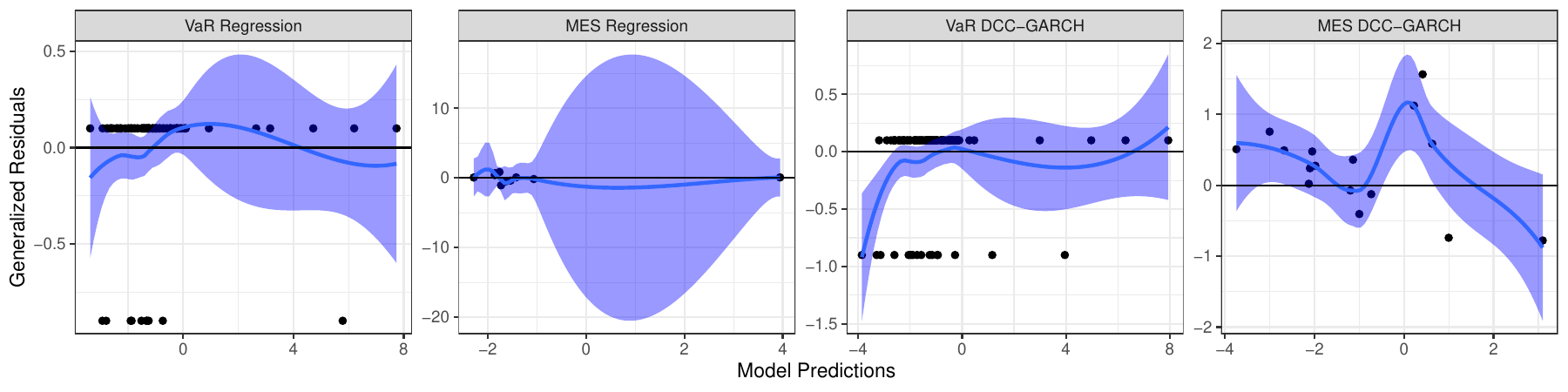}
\caption{In-sample model diagnostics of the VaR and MES predictions of our regression and a DCC--GARCH model for $Y_t$ being negative GDP growth of the United Kingdom (upper two rows) and France (lower two rows), $X_t$ the joint economic region and $\beta = 0.9$. 
	The generalized residuals are plotted and non-parametrically regressed (in blue) against time (upper row) and the model predictions (lower row).}
\label{fig:ModelDiagnostics_GDP_MES_DCC_FRA_GBR}
\end{figure}

Figure~\ref{fig:ModelDiagnostics_GDP_MES_DCC_FRA_GBR} visualizes VaR and MES model diagnostics in the vulnerable growth application for our regression and a DCC-GARCH model for the United Kingdom (upper two rows) and France (lower two rows), analogous to Figure~\ref{fig:ModelDiagnostics_GDP_MES_DCC} in the main article.
Similar to the model diagnostics for Germany shown in the main article, there is no meaningful difference in model fit between our MES regression and the DCC-GARCH model.

We finally note that by definition, the generalized MES residuals follow a mixture distribution that consists of approximately $\beta \times 100 \%$ zeros together with the remaining non-zero values $m-y$. 
In all figures in the main article as well as the supplementary material involving regression diagnostics for the generalized MES residuals, we model these non-zero values of the generalized residuals non-parametrically (with automated parameter choices in the \texttt{loess} function of the statistical software \texttt{R}).

%Second, notice that the MES part of the identification function $V\big( (v,m)', (x,y)' \big) = \big(\1_{\{x\leq v\}}-\b, \, \1_{\{x>v\}}(m-y) \big)'$ also requires the VaR model predictions $v$. For the ad hoc MES regression in \eqref{eq:Y_MES_transform}, we use $\widehat{Q}_\beta(X_{(t-S):t})$ for $v$. Using model predictions of a quantile regression on $\mZ_t^m$ instead results in model diagnostics with qualitatively very similar conclusions; see Figure~\ref{fig:ModelDiagnostics_BRSDetails} in Appendix~\ref{sec:AddEmpResults}.

\FloatBarrier
\singlespacing
\small 
\setstretch{1.0}
% \singlespacing

\bibliographystyle{jaestyle2}
\bibliography{bibCoQR}

\end{document}